%% file: root.tex
\documentclass[11pt]{book}
\usepackage{dependencies/simplethesis}
\usepackage{amsmath, amsthm, amssymb, amsfonts}
\usepackage{proof}
\usepackage{stmaryrd}
\usepackage{url}

\include{macros}

\usepackage{hyperref}

\Title{A Framework for Reasoning About LF Specifications}
\Author{Mary Southern}
\Month{April}
\Year{2021}
\Advisor{Gopalan Nadathur}
\Degree{DOCTOR OF PHILOSOPHY}
\degreelevel{doctoral}

\begin{document}
\prelimpages

\titlepage

\copyrightpage

\setcounter{page}{1}

\acknowledgments{\input{acknowledgments}}

\abstract{\input{abstract}}

\setcounter{tocdepth}{2}
\tableofcontents

\listoffigures

\textpages

\doublespace

\include{introduction/intro}
\include{LF/lf}

\include{logic/logic}

\include{proof-system/proof-system}

\include{Adelfa/adelfa}
\include{examples/adelfa-ex}

\include{comparisons/related-work}

\include{conclusion/conclusion}

\bibliographystyle{alpha}
\bibliography{references/master,references/local}

\end{document}

%% file: macros.tex
\newcommand{\dom}[1]{dom\left({#1}\right)}

\newcommand{\ie}{{\rm i.e.}}
\newcommand{\emptybb}{\cdot}
\newcommand{\emptycs}{\cdot}
\newcommand{\declinst}[4]{#1 \vdash #2 \leadsto_{\mbox{\sl \scriptsize dec}} #3 \bowtie #4}
\newcommand{\bsinst}[4]{#1 ; #2 \vdash #3 \leadsto_{\mbox{\sl \scriptsize bs}} #4}
\newcommand{\csinst}[4]{#1 ; #2 \vdash #3 \leadsto_{\mbox{\sl \scriptsize cs}} #4}
\newcommand{\csinstone}[4]{#1 ; #2 \vdash #3 \leadsto^1_{\mbox{\sl \scriptsize cs}} #4}

\newcommand{\emptycb}{\cdot}

\newcommand{\fvalid}[2]{#1 \vDash #2\ \mbox{\tt valid}}
\newcommand{\fatmann}[3]{{\fatm{#1}{#2}}^{#3}}
\newcount\counter
\newcommand{\repeatchar}[2]{
  \ifnum #2>-1
    \counter=0
    \loop
      \ifnum \counter<#2
        #1
        \advance \counter+1
    \repeat
  \fi
}
\newcommand{\ltann}[2][1]{{#2}^{\repeatchar{*}{#1}}}
\newcommand{\ltannaux}[2]{{#2}^{*^{#1}}}
\newcommand{\eqann}[2][1]{{#2}^{\repeatchar{@}{#1}}}
\newcommand{\eqannaux}[2]{{#2}^{@^{#1}}}
\newcommand{\assn}[3]{#1\lbrack{#2}\mapsfrom{#3}\rbrack}

\newcommand{\typedpi}[3]{\Pi {#1}{:}{#2}.{\mkern 3mu} #3}
\newcommand{\type}{\mbox{Type}}
\newcommand{\lflam}[2]{\lambda {#1}.{\mkern 3mu} #2}
\newcommand{\imp}[2]{#1\rightarrow #2}
\newcommand{\app}{\ }
\newcommand{\emptysig}{\cdot}
\newcommand{\emptyctx}{\cdot}

\newcommand{\lfprove}[3][\Sigma]{#2 \vdash_{#1} #3}
\newcommand{\lfsynthkind}[4][\Sigma]{\lfprove[#1]{#2}{#3\Rightarrow #4}}
\newcommand{\lfsynthtype}[4][\Sigma]{\lfprove[#1]{#2}{#3\Rightarrow #4}}
\newcommand{\lfchecktype}[4][\Sigma]{\lfprove[#1]{#2}{#3\Leftarrow #4}}
\newcommand{\lfkind}[3][\Sigma]{#2 \vdash_{#1} #3\ \mbox{\tt kind}}
\newcommand{\lftype}[3][\Sigma]{#2 \vdash_{#1} #3\ \mbox{\tt type}}
\newcommand{\lfctx}[2][\Sigma]{\vdash_{#1} #2\ \mbox{\tt ctx}}
\newcommand{\lfsig}[1]{\vdash {#1}\ \mbox{\tt sig}}

\newcommand{\domain}[1]{\mbox{\sl dom}(#1)}
\newcommand{\range}[1]{\mbox{\sl rng}(#1)}
\newcommand{\context}[1]{\mbox{\sl ctx}(#1)}
\newcommand{\supportof}[1]{\mbox{\sl supp}(#1)}
 
\newcommand{\subst}[2]{#2[#1]}
\newcommand{\hsubst}[2]{#2\llbracket {#1} \rrbracket} 
\newcommand{\hsubstseq}[3]{#3\llbracket {#2} \rrbracket_{#1}}
\newcommand{\hsubr}[4][]{#3\llbracket {#2} \rrbracket^r_{#1} = #4}
\newcommand{\hsub}[4][]{#3\llbracket {#2} \rrbracket_{#1} = #4}

\newcommand{\comp}[2]{#1\circ #2}

\newcommand{\seqsub}[2]{\langle{#1},{#2}\rangle}
\newcommand{\ctxvarminus}[2]{#1_{#2}}
\newcommand{\restrict}[2]{#1|_{#2}}
\newcommand{\extsubst}[2]{#2[#1]}

\newcommand{\STLCGamma}{\Theta}
\newcommand{\aritysum}[2]{#1 \uplus #2}
\newcommand{\stlcder}{\vdash_{\mbox{\sl \scriptsize at}}}
\newcommand{\stlctyjudg}[3]{#1 \stlcder #2 : #3}
\newcommand{\stlcderr}{\vdash^r_{\mbox{\sl\scriptsize at}}}
\newcommand{\stlctyjudgr}[3]{#1 \stlcderr #2 : #3}
\newcommand{\kindingder}{\vdash_{\mbox{\sl \scriptsize ak}}}
\newcommand{\wftype}[2]{#1 \kindingder #2\ \mbox{\sl type}}
\newcommand{\kindingderp}{\vdash^p_{\mbox{\sl \scriptsize ak}}}
\newcommand{\akindingp}[3]{#1 \kindingderp #2 : #3}
\newcommand{\declder}{\vdash_{\mbox{\sl \scriptsize dec}}}
\newcommand{\wfdecls}[3]{#1 \declder #2 \Rightarrow #3}

\newcommand{\abstyping}[1]{\vdash #1\ \mbox{\sl blk schema}}

\newcommand{\acstyping}[1]{\vdash #1\ \mbox{\sl ctx schema}}
\newcommand{\wfctx}[3]{#1;#2 \vdash #3\ \mbox{\sl context}}

\newcommand{\noms}{\mathcal{N}} 
\newcommand{\atyarr}{\rightarrow}
\newcommand{\arr}[2]{#1\rightarrow #2}
\newcommand{\erase}[1]{{(#1)}^{-}}
\newcommand{\oty}{o}
\newcommand{\ctxty}[2]{#1\lbrack {#2}\rbrack}
\newcommand{\ctxvarty}[3]{#1\!\uparrow\!#2 : #3}

\newcommand{\supp}[1]{\text{supp}\left({#1}\right)}

\newcommand{\initctx}{\STLCGamma_0}
\newcommand{\ctxsanstype}[1]{#1^-}

\newcommand{\of}[2]{#1:#2}
\newcommand{\emptyce}{\cdot}
\newcommand{\fatm}[2]{\left\{#1 \vdash #2\right\}}
\newcommand{\fand}[2]{#1\wedge #2}
\newcommand{\for}[2]{#1\vee #2}
\newcommand{\fimp}[2]{#1\supset #2}
\newcommand{\fall}[2]{\forall #1. #2}
\newcommand{\fexists}[2]{\exists #1. #2}
\newcommand{\genericq}{\mathcal{Q}}
\newcommand{\fgeneric}[2]{\genericq #1. #2}
\newcommand{\fctx}[3]{\Pi\,#1 : #2.#3}

\newcommand{\ftrue}{\top}
\newcommand{\ffalse}{\bot}

\newcommand{\typdecomp}[4]{\mbox{\sl Decompose}\left(#1;#2; #3; #4\right)}

\newcommand{\seq}[5][\mathbb{N}]{#1;#2;#3;#4\longrightarrow #5}
\newcommand{\seqsans}[3][\mathbb{N}]{#1;\emptyset; \emptyset; #2\longrightarrow #3}

\newcommand{\setand}[2]{#1, #2}

\newcommand{\formeq}[4]{#1 \vdash #3 \equiv_{#2} #4}
\newcommand{\ctxeq}[4]{#1 \vdash #3 \equiv_{#2} #4}

\newcommand{\streq}[4]{#1 \vdash #3 \succeq_{#2} #4}

\newcommand{\ind}{\mbox{\sl ind}}
\newcommand{\cut}{\mbox{\sl cut}}

\newcommand{\id}{\mbox{\sl id}}
\newcommand{\absR}{\mbox{\sl atm-abs-R}}
\newcommand{\absL}{\mbox{\sl atm-abs-L}}
\newcommand{\appR}{\mbox{\sl atm-app-R}}
\newcommand{\appL}{\mbox{\sl atm-app-L}}
\newcommand{\topR}{\mbox{\sl $\top$-R}}
\newcommand{\botL}{\mbox{\sl $\bot$-L}}
\newcommand{\ctxR}{\mbox{\sl $\Pi$-R}}
\newcommand{\ctxL}{\mbox{\sl $\Pi$-L}}
\newcommand{\allR}{\mbox{\sl $\forall$-R}}
\newcommand{\allL}{\mbox{\sl $\forall$-L}}
\newcommand{\existsR}{\mbox{\sl $\exists$-R}}
\newcommand{\existsL}{\mbox{\sl $\exists$-L}}
\newcommand{\impR}{\mbox{\sl $\supset$-R}}
\newcommand{\impL}{\mbox{\sl $\supset$-L}}
\newcommand{\andR}{\mbox{\sl $\wedge$-R}}
\newcommand{\andL}{\mbox{\sl $\wedge$-L$_i$}}
\newcommand{\orR}{\mbox{\sl $\vee$-R$_i$}}
\newcommand{\orL}{\mbox{\sl $\vee$-L}}
\newcommand{\lfwk}{\mbox{\sl LF-wk}}
\newcommand{\lfstr}{\mbox{\sl LF-str}}
\newcommand{\lfperm}{\mbox{\sl LF-perm}}
\newcommand{\lfinst}{\mbox{\sl LF-inst}}
\newcommand{\sweak}{\mbox{\sl ctx-wk}}
\newcommand{\sstr}{\mbox{\sl ctx-str}}
\newcommand{\weakening}{\mbox{\sl weak}}
\newcommand{\contraction}{\mbox{\sl cont}}

\newcommand{\emptyseq}{\mbox{\sl nil}}
\newcommand{\consseq}[2]{ #1 :: #2}

\newcommand{\addblksans}{\mbox{\sl AddBlock}}
\newcommand{\addblk}[7]{\addblksans(#1,#2,#3,#4,#5,#6,#7)}
\newcommand{\allblks}[3]{\mbox{\sl AllBlocks}(#1,#2,#3)}
\newcommand{\implheadssans}{\mbox{\sl ImplicitHeads}}
\newcommand{\implheads}[2]{\implheadssans(#1,#2)}
\newcommand{\namessans}{\mbox{\sl NamesLsts}}
\newcommand{\names}[3]{\namessans(#1,#2,#3)}
\newcommand{\headssans}{\mbox{\sl Heads}}
\newcommand{\hds}[2]{\headssans(#1,#2)}
\newcommand{\decompseqsans}{\mbox{\sl ReduceSeq}}
\newcommand{\decompseq}[2]{\decompseqsans\left({#2}, #1\right)}

\newcommand{\makecasessans}{\mbox{\sl Cases}}
\newcommand{\makecases}[3][\fatm{G}{R:P}]{\makecasessans\left({#2}, #1, #3\right)}
\newcommand{\casessans}{\mbox{\sl AllCases}}
\newcommand{\casesfn}[2][\fatm{G}{R:P}]{\casessans \left(#2, #1\right)}

\newcommand{\unif}[3]{\left\langle{#1}; {#2}; {#3}\right\rangle}
\newcommand{\eqn}[2]{{#1}={#2}}
\newcommand{\solun}[2]{\langle{#1},{#2}\rangle}
\newcommand{\permute}[2]{{#1}.{#2}}
\newcommand{\inv}[1]{{#1}^{-1}}

\newcommand{\wfform}[3]{#1;#2 \vdash #3\ \mbox{\sl fmla}}

\newcommand{\wfctxvarty}[3]{#1 ; #2 \vdash #3 \mbox{ ctx-ty}}

\newcommand{\ctxtyinst}[5]{#1;#2; #3 \vdash #4 \leadsto_{\mbox{\sl \scriptsize csty}} #5}

\newcommand{\sigempty}{\mbox{\sl\small SIG\_EMPTY }}
\newcommand{\sigterm}{\mbox{\sl\small SIG\_TERM}}
\newcommand{\sigfam}{\mbox{\sl\small SIG\_FAM}}
\newcommand{\ctxempty}{\mbox{\sl\small CTX\_EMPTY}}
\newcommand{\ctxterm}{\mbox{\sl\small CTX\_TERM}}
\newcommand{\canonkindtype}{\mbox{\sl\small CANON\_KIND\_TYPE}}
\newcommand{\canonkindpi}{\mbox{\sl\small CANON\_KIND\_PI}}
\newcommand{\canonfamatom}{\mbox{\sl\small CANON\_FAM\_ATOM}}
\newcommand{\canonfampi}{\mbox{\sl\small CANON\_FAM\_PI}}
\newcommand{\atomfamconst}{\mbox{\sl\small ATOM\_FAM\_CONST}}
\newcommand{\atomfamapp}{\mbox{\sl\small ATOM\_FAM\_APP}}
\newcommand{\canontermatom}{\mbox{\sl\small CANON\_TERM\_ATOM}}
\newcommand{\canontermlam}{\mbox{\sl\small CANON\_TERM\_LAM}}
\newcommand{\atomtermvar}{\mbox{\sl\small ATOM\_TERM\_VAR}}
\newcommand{\atomtermconst}{\mbox{\sl\small ATOM\_TERM\_CONST}}
\newcommand{\atomtermapp}{\mbox{\sl\small ATOM\_TERM\_APP}}

\newcommand{\tintros}{\mbox{\tt intros}}
\newcommand{\tsearch}{\mbox{\tt search}}
\newcommand{\tsplit}{\mbox{\tt split}}
\newcommand{\tleft}{\mbox{\tt left}}
\newcommand{\tright}{\mbox{\tt right}}
\newcommand{\tassert}{\mbox{\tt assert}}
\newcommand{\tapply}{\mbox{\tt apply}}
\newcommand{\tind}{\mbox{\tt induction}}
\newcommand{\texists}{\mbox{\tt exists}}
\newcommand{\tcase}{\mbox{\tt case}}
\newcommand{\tweak}{\mbox{\tt weaken}}
\newcommand{\tstr}{\mbox{\tt strengthen}}
\newcommand{\tpermute}{\mbox{\tt permutectx}}
\newcommand{\tinst}{\mbox{\tt inst}}

\newcommand{\tpty}{\mbox{\sl tp}}
\newcommand{\unittm}{\mbox{\sl unit}}
\newcommand{\arrtm}{\mbox{\sl arr}}
\newcommand{\tmty}{\mbox{\sl tm}}
\newcommand{\emptytm}{\mbox{\sl empty}}
\newcommand{\lamtm}{\mbox{\sl lam}}
\newcommand{\apptm}{\mbox{\sl app}}
\newcommand{\ofty}{\mbox{\sl of}}
\newcommand{\ofemptytm}{\mbox{\sl of\_empty}}
\newcommand{\ofapptm}{\mbox{\sl of\_app}}
\newcommand{\oflamtm}{\mbox{\sl of\_lam}}
\newcommand{\eqty}{\mbox{\sl eq}}
\newcommand{\refltm}{\mbox{\sl refl}} 

\newcommand{\natty}{\mbox{\sl nat}}
\newcommand{\plusty}{\mbox{\sl plus}}
\newcommand{\ztm}{\mbox{\sl z}}
\newcommand{\stm}{\mbox{\sl s}}
\newcommand{\plusztm}{\mbox{\sl plus\_z}}
\newcommand{\plusstm}{\mbox{\sl plus\_s}}

\newcommand{\intseq}[2]{#1 \Rightarrow #2}
\newcommand{\propty}{\mbox{\sl prop}}
\newcommand{\toptm}{\mbox{\sl top}}
\newcommand{\imptm}{\mbox{\sl imp}}
\newcommand{\andtm}{\mbox{\sl and}}
\newcommand{\hypty}{\mbox{\sl hyp}}
\newcommand{\concty}{\mbox{\sl conc}}
\newcommand{\inittm}{\mbox{\sl init}}
\newcommand{\toprtm}{\mbox{\sl topR}}
\newcommand{\andltm}{\mbox{\sl andL}}
\newcommand{\andrtm}{\mbox{\sl andR}}
\newcommand{\impltm}{\mbox{\sl impL}}
\newcommand{\imprtm}{\mbox{\sl impR}}

\newcommand{\tyty}{\mbox{\sl ty}}
\newcommand{\toptytm}{\mbox{\sl top}}
\newcommand{\arrowtm}{\mbox{\sl arrow}}
\newcommand{\alltm}{\mbox{\sl all}}
\newcommand{\boundty}{\mbox{\sl bound}}
\newcommand{\subty}{\mbox{\sl sub}}
\newcommand{\stoptm}{\mbox{\sl sa-top}}
\newcommand{\srefltm}{\mbox{\sl sa-refl-tvar}}
\newcommand{\stranstm}{\mbox{\sl sa-trans-tvar}}
\newcommand{\sarrowtm}{\mbox{\sl sa-arrow}}
\newcommand{\salltm}{\mbox{\sl sa-all}}
\newcommand{\varty}{\mbox{\sl var}}
\newcommand{\bvarty}{\mbox{\sl bound\_var}}

\newcommand{\hastype}{\mbox{\sl hastype}}
\newcommand{\unique}{\mbox{\sl unique}}

\newcommand{\redty}{\mbox{\sl reduce}}
\newcommand{\redappltm}{\mbox{\sl red-app-1}}
\newcommand{\redapprtm}{\mbox{\sl red-app-2}}
\newcommand{\redlamtm}{\mbox{\sl red-lam}}
\newcommand{\redbetatm}{\mbox{\sl red-beta}}
\newcommand{\appty}{\mbox{\sl append}}
\newcommand{\listty}{\mbox{\sl list}}
\newcommand{\niltm}{\mbox{\sl nil}}
\newcommand{\constm}{\mbox{\sl cons}}

\newcommand{\case}[1]{\noindent{\bf Case: #1}\\\noindent}
\newtheorem{theorem}{Theorem}[chapter]
\newtheorem{lemma}{Lemma}[chapter]
\theoremstyle{definition}
\newtheorem{definition}{Definition}[chapter]

%% file: acknowledgments.tex
\hspace{\parindent}I'd like to thank my advisor, Gopalan Nadathur, without whose knowledge and support
this dissertation would not have been completed. 
I would also like to thank my committee members, Eric Van Wyk, Andrew Gacek, and 
Paul Garrett, for the time and ideas they have shared in getting me to this point.

I am grateful for the support of my family, immediate and extended, for their
invaluable material and emotional support.
I want to especially thank my Mom for always being on my side and for making 
a bed available when I needed one most.
Also my sister Anna Mary, who can get me to smile, is always up for being together, and 
keeps me curious; I hope you will be as proud of me as I am of you.

I'd also like to thank my friends, who have each in their way provided support
and understanding.
I am grateful in particular to Estelle Smith for her resilience and sound advice.
I also want to thank Shelley Kandola and Maggie Ewing for bringing their knowledge
to our study of Homotopy Type Theory and the friendship that continued after.

Roller derby has been a place that taught me how strong I can be, and gave me
a space for self-reflection and understanding when I was lost.
Though there are many whom I could name, I want to particularly acknowledge
Syd, Ace, Robin, Patrick, and Amber for their friendship 
both within and outside derby.

The theater community has also given me a needed outlet and sense of self, and I want
to especially thank Steph and Varghese Alexander who are both talented performers and 
warm friends.

Finally, I want to express my most enduring gratitude to Ted for his unending love,
support, and squishes.
He has brought comfort, confidence, inspiration, joy, and laughter each in their time, 
and is the one on whom I rely most.
I could not have picked a better person to be stuck in quarantine with.

\vfill
This dissertation is based on work supported by the National Science Foundation under Grant No. CCF-1617771. Any opinions, findings, and conclusions or recommendations expressed in this material are those of the author and do not necessarily reflect the views of the National Science Foundation.

%% file: abstract.tex
With growing reliance on software in the modern world there is also a growing
interest in ensuring that these systems behave in desired ways.
Many researchers are interested in using formal specifications to develop 
correct systems, relying on the specifications as a means for reasoning about 
properties of systems and their behaviour.
In this thesis we focus on systems whose behaviour is described in a 
syntax-directed, rule-based fashion.
Such systems are typically encoded through a description of the relevant objects 
of the system along with some relations between these objects defined through a collection
of rules.
Properties of such systems are expressed through these object relations,
relating the validity of certain relations to the validity of others.

A specification language based on the dependently typed $\lambda$-calculus,
the Edinburgh Logical Framework, or LF, is often used for specifying such systems.
The objects of interest in the system are formalized through terms in
the specification, and the dependencies permitted in types provide a natural means 
for formalizing the relationships between objects of the system.
Under such an encoding, the terms inhabiting the dependent types of LF represent
valid derivations of the relation in the system and thus reasoning about 
type inhabitation in LF will correspond to reasoning about the validity of
relations in the system.

This thesis develops a framework for formalizing reasoning
about specifications of systems written in LF, with the ultimate
goal of formalizing the informal reasoning steps one would take in an LF setting.
This formalization will center around the development of a reasoning
logic that can express the sorts of properties which arise in
reasoning about such specifications. 
In this logic, type inhabitation judgements in LF serve as atomic formulas, and 
quantification is permitted over both contexts and terms in these judgements.
The logic permits arbitrary relations over derivations of LF judgements to be 
expressed using a collection of logical connectives, in contrast to
other systems for reasoning about LF specifications. 
Defining a semantics for these formulas raises issues which we must address,
such as how to interpret both term and context quantification as well as the
relation between atomic formulas and the LF judgements they are meant to encode.

This thesis also develops a proof system which
captures informal reasoning steps as sound inference rules for the logic.
To achieve this we develop a collection of proof rules including mechanisms for
both case analysis and inductive reasoning over the derivations of judgements in LF.
The proof system also supports applying LF meta-theorems through proof
rules that enforce the requirements of the LF meta-theorem that cannot
be expressed in the logic.

We also implement a proof assistant called Adelfa that provides a
means for mechanizing the approach to reasoning about specifications
written in LF that is the subject of this thesis.
A characteristic of this proof assistant is that it uses the proof
rules that complement the logic to describe a collection of tactics
that can be used to develop proofs in goal-driven fashion.
One of the problems to be solved in this context is that of realizing
the rule in the proof system that enables the analysis of LF typing
judgements that appear as assumption formulas.
We show that a form of unification called higher-order pattern
unification can provide the basis for such a realization.
The Adelfa system is used to develop a collection of examples which demonstrate
the effectiveness of the framework and to showcase how informal reasoning about
specifications written in LF can be formalized using the logic and associated
proof system.

%% file: introduction/intro.tex
\chapter{Introduction}
\label{ch:intro}

With the proliferation of software in the modern world there is also a
growing interest in ensuring that these artefacts will work as
desired and, specifically, not cause harm to people or property.
These notions can be formalized through the use of specifications,
mathematical descriptions of the desired behaviour for a system.
Many researchers and developers are interested in using formal
specifications for various tasks related to the development of correct
systems.
Thus, such specifications have been used in the past as the basis for
building prototypes of relevant systems and as a means for reasoning
about deeper properties concerning the behaviour of the systems.

This thesis is motivated by the latter concern, i.e. reasoning
about computational systems through their specifications.
A key ingredient to the formal description of such systems is the
language in which the specifications are presented.
The focus in this thesis in on a particular specification
language based on a dependently typed $\lambda$-calculus
called the Edinburgh Logical Framework~\cite{harper93jacm} or LF.
In the typical situation, the terms of the specification language are
used to provide encodings of the objects that are of interest in the
system being formalized.
Dependent types, which effectively relate terms, then provide a useful
and convenient means for encoding relationships between the objects of the
system. 
This can lead to very natural encodings of rule-based systems where 
relations are captured as dependent types and the rules defining this relation 
become the constructors for expressions of this type.
Under this interpretation, 
the terms of the LF specification represent valid derivations in the 
encoded system, and reasoning about the derivability of typing judgements
in LF corresponds to reasoning about the validity of relations in the system.

In this thesis we present a framework for reasoning about
systems through specifications written in LF.
This framework comprises a logic we have developed for reasoning
about such specifications, an associated proof system for formalizing 
the construction of proofs in this logic, and an implementation of
the proof system mechanizing the construction of such proofs.

\section{Specification and Reasoning about Systems}
Our focus in this thesis is on reasoning about object systems that are
described in a syntax-directed and rule-based fashion.
In the specification of such systems, the syntactic structure of the
expressions describing the objects of interest is used to present
rules that define relations between the objects. 
The typing of terms in the simply-typed $\lambda$-calculus
would be a system of this sort.
As an illustration, we may consider a version of the calculus in which
there a single base type, $\unittm$; other types
are constructed using the function type constructor $\rightarrow$.
The lambda terms are constructed from variables and the constant
$\unittm$ using applications and abstractions. 
An important relation in this context is that between lambda terms and
types relative to an assignment of types to variables.
Below we give rules that define this relation.
Observe that these rules are driven by the syntax of lambda terms. 

\[
\begin{array}{cccc}
\infer{\Gamma \vdash \langle\rangle : \unittm}
      {}
  \quad&
\infer{\Gamma \vdash x : \tau}
      {x:\tau \in\Gamma}
  \quad&
\infer{\Gamma \vdash \lambda x:\tau. t:\tau\rightarrow \tau'}
      {\Gamma,x:\tau \vdash t : \tau'}
  \quad&
\infer{\Gamma \vdash t_1\ t_2: \tau}
      {\Gamma \vdash t_1 : \tau'\rightarrow \tau
         &
       \Gamma \vdash t_2 : \tau'}
\end{array}
\]

A first requirement in transforming such presentations into a formal
specification in LF is being able to represent relations.
Use is made in this context of the fact that LF types can depend on LF
terms.
In particular, such \emph{dependent types} are used to encode
relations. 
For example suppose that we have described representations in LF for the
types and terms in the simply typed $\lambda$-calculus.\footnote{An
  interesting aspect of such representations is the possibility of
  using the technique of higher-order abstract syntax or
  HOAS~\cite{miller87slp,pfenning88pldi} in encoding abstraction in
  the object system by abstraction in the meta-language, i.e. LF. We
  elide a discussion of this aspect because it is orthogonal to our
  present focus.}
Then, writing $\overline{e}$ to denote the LF representation of the object
language expression $e$, the typing relation between a simply typed
$\lambda$-calculus term $t$ and a type $\tau$ can be encoded by a
dependent type of the form
$(\hastype\app \overline{t}\app \overline{\tau})$; here, $\hastype$ is
a type constant in LF that takes two terms as arguments and is
designated to represent the typing relation in the object language.
An interesting aspect of the LF calculus is that the rules describing
the encoded relations may themselves be characterized by suitably
typed LF constants that provide a means for constructing objects of
the LF type encoding the conclusion relation of the rule.
For example, consider the typing rule for an application term in the
simply typed $\lambda$-calculus. 
Ignoring the typing context $\Gamma$ which is treated implicitly via
LF typing contexts in the standard encoding of $\lambda$-terms, we see
that this rule has four schematic variables---$t_1$, $t_2$, $\tau$ and
$\tau'$---and two premises.
Accordingly, the rule can be represented by a designated term-level
constant {\sl ofapp} in LF that takes six arguments whose types
correspond to those of the representations of the four schematic
variables and the two premises and that yields an object that has the
type $(\hastype\app\overline{(t_1\app t_2)}\app\overline{\tau})$.

Specifications developed in this way can be used to determine if
particular relations hold between relevant objects in the object
system.
For example, assuming an encoding of the syntax and typing rules for
the simply typed $\lambda$-calculus of the kind described above, we
might want to determine if the object language $(\lambda
x:\unittm. x)$ has the type $(\unittm\rightarrow \unittm)$.
Such a question translates into one about the \emph{inhabitation} of a
type relative to the LF specification.
Thus, in the example under consideration, this becomes a question
about whether there is an LF term of the type
$(\hastype\app\overline{(\lambda x:\unittm. x)}\app\overline{(\unittm\rightarrow\unittm)})$ 
There is in fact such a term and it can be seen that that term will
essentially encode a derivation of the typing judgement using the
rules defining such judgements.
More broadly, answering the inhabitation question mirrors a search for
a derivation in the object system.

Specifications in LF provide the basis also for stating and reasoning
about properties of an encoded system that are more general than
simply verifying if a particular relation holds.
For example, in the simply-typed $\lambda$-calculus an interesting property
is that when a term is typeable then that type is unique, i.e. that if $t: \tau$ and 
$t:\tau'$ are both derivable, then $\tau$ and $\tau'$ must
in fact be the same type.
This property can be captured relative to the LF specification by an
assertion that if the dependent types $(\hastype\app t\app\tau)$ and
$(\hastype\app t\app\tau')$ are inhabited for any choice of term $t$,
then it must be the cast that $\tau$ and $\tau'$ are the same.
It is the statement of this more general form of properties and the
process of reasoning about them that is the focus of this thesis.

\section{Existing Approaches to Reasoning about LF Specifications}
One approach to reasoning about specifications written in LF is to
encode properties of the system also as (dependent) types.
This approach is the basis of reasoning in the Twelf
system~\cite{Pfenning02guide} and it's related logic,
$\mathcal{M}_2^+$~\cite{schurmann00phd}.
Consider the property of type uniqueness; it is essentially a relation
between terms of particular dependent types, and thus can itself be
encoded as a dependent type in LF which takes these terms as
parameters.
To encode this property as a type we need first encode equality as an
LF type $\eqty$ taking two types as parameters, and encode that only
identical types are equal via a constructor for the type
$(\eqty\app\tau\app\tau)$.
Thus for type uniqueness we define a new dependent type $\unique$
which takes as parameters a term $t$, two types $\tau$ and
$\tau'$, a term of type $(\hastype\app t\app\tau)$, and a term of
type $(\hastype\app t\app\tau')$ and will return a term of type
$(\eqty\app\tau\app\tau')$.
A key to this approach to reasoning is the observation that if a
function of this type were total, then one can conclude that the
property it encodes holds of the specification, and thus the system.
A ``proof'' in the Twelf system amounts to presenting a term of the
described type and then demonstrating via an external process that
that term is in fact total; in the example, that it works, no matter
what actual term is chosen for $t$.
The logic $\mathcal{M}_2^+$ provides a means for making explicit the
reasoning steps used by the external process in Twelf.

The approach to reasoning that underlies the Twelf system has the
drawback that, at the end of the process, it does not provide a
witness in the form of a proof by which a conclusion was arrived at.
The ideas underlying Twelf are also germane to the system
Beluga~\cite{pientka10ijcar}.
While Beluga uses a richer type system based on Contextual Modal Type
Theory~\cite{nanevski08tocl} to overcome some of the issues of
expressivity with Twelf, it continues to have the basic limitation of
Twelf described above. 
The logic $\mathcal{M}_2^+$ addresses this criticism.
However, it shares with Twelf and Beluga the limitation that
the properties that can be expressed and reasoned about have a
fundamentally functional structure, that is, quantifiers in them take
the form of a prefix with a $\forall\ldots\exists\ldots$ structure. 

A second approach to reasoning about specifications written in LF is to
translate the specification into a predicate logic and use tools for
reasoning about the translation to construct a proof.
This approach was the motivation behind the the Abella-LF
system~\cite{southern2014}.
The translation approach utilizes a connection between dependencies in
types and relations between terms to map a given LF specification into
one which is written in a predicate logic.
The dependent types and terms of LF are translated into simple types and terms
via erasure of dependencies; since this is clearly a lossy mapping the erased
typing information is then encoded using predicates.
For example, with the type uniqueness of the simply typed $\lambda$-calculus, 
the type $\hastype$ would be translated into a relation in 
predicate logic that relates a translated term with a translated type.
The constructors for the dependent type are translated into clauses 
defining the relation via the same encoding.
Due to the lossy nature of the term encoding, this translation is not
one-to-one when the terms are not in a form which contains no $\beta$-redexes,
or normal form.
This is because the erasure of dependencies in the types of abstraction
terms means there is a many to one mapping of such terms, and there will not
be sufficient information to uniquely identify which of these terms is 
represented using only the translation of the LF type 
in the case of a non-canonical term structure.
Further complicating this approach, it is unclear what proofs about 
the translated specification mean in the context of LF.
Lifting of results proved in this way to LF (and hence the system of interest)
is suspect without such understanding.
A variant of this approach based on only canonical forms may address the
former issue, but would not in itself address the latter.

\section{The Approach to Reasoning Developed in this Thesis}
We approach reasoning from the perspective of understanding LF and the
derivability of judgements relative to a given LF specification.
Our goal is to provide a formalization of informal reasoning as
it is performed in the LF setting, and provide a reasoning logic which
is based on understanding LF derivability.
Our approach focuses on constructing proofs
explicitly, using reasoning steps which naturally correspond to 
the structure of informal arguments about LF, rather than relying on
external analyses.
The work in this thesis would, for example, provide a theoretical basis for the
translation approach to interpret reasoning steps done over the translation
as LF reasoning steps and thus lift the proofs in a sensible way.

The sorts of formulas we are interested in will center around the
LF typing judgements as atomic formulas, which are intended to represent the
derivation of that judgement in LF.
We will use the notation $\fatm{G}{M:A}$ to represent a derivation in LF
that the term $M$ inhabits the dependent type $A$ under a context $G$.
To express the relations over these derivations we permit formulas
to be constructed using logical connectives as well as quantification
over both the terms and contexts.
For example, for the property of type uniqueness 
we would want to construct a formula which quantified over the term $t$, the 
types $\tau_1$ and $\tau_2$, and the LF terms inhabiting the types 
$(\hastype\app t\app\tau_1)$ and $(\hastype\app t\app\tau_2)$.
Since the typing rules can also extend the context as we descend under bindings,
the formula representing this property should also include a quantification
over the LF context appearing in the atomic formulas.
The result might be something of the form
\begin{tabbing}
\qquad\quad\=\qquad\=\kill
\>$\fall{t}{\fall{\tau_1}{\fall{\tau_2}{\fall{d_1}{\fall{d_2}{\Pi\Gamma.}}}}}$\\
\>\>$\fimp{\fatm{\Gamma}{d_1:\hastype\app t\app \tau_1}}{\fimp{\fatm{\Gamma}{d_2:\hastype\app t\app \tau_2}}{\fexists{d_3}{\fatm{\Gamma}{d_3:\eqty\app \tau_1\app \tau_2}}}}.$
\end{tabbing}

The design of such a logic raises some interesting questions.
We do not intend for the quantification over $t$, for example, to mean for 
every single possible term which can be constructed.
On the other hand, attempting to restrict these terms using an LF type will
be problematic; the occurrences of a quantified variable may be under different
contexts or no context and it is unclear how such typing could be 
made sensible at the quantifier.
Further, using LF types would inhibit the meaning of atomic formulas like 
$\fatm{\Gamma}{d_1:\hastype\app t\app \tau_1}$ 
as we would already have identified $d_1$ as inhabiting a particular LF type.
Our approach addresses these issues by using simple arity types for quantifiers
which capture the functional character of the terms.

We also must make sense of the quantification over contexts and its meaning.
In our example property, the contexts which are relevant to reasoning are
those whose bindings might play a role in determining the inhabitation
of a $\hastype$ type, i.e. ones which are the encoding of the contexts
of the object system.
In this example, the system contexts consist only of bindings of the form
$x:\tau$, and so we are interested in those LF contexts which
are constructed using the bindings
$(x:\mbox{\sl term},y:\hastype\app x\app\tau)$ which encode this assumption.
Unconstrained quantification would mean contexts of any form must
be considered, which does not capture this notion.
Further, such an interpretation leads to a system where analysis based on
the derivability of typing judgements in LF would not be effective.
To address this, we use the idea of context schemas to describe a regular
structure for context variables and use these as types for the context
quantifiers.

The questions we have discussed above pertain to the structure of the
formulas in the logic and the semantics that governs them.
A separate question is that of providing a basis for mechanized
reasoning based on the semantics.
In particular, we would desire to complement the description of the
logic with a sound and effective proof system.
While soundness is a question that can be settled theoretically, the
demonstration of effectiveness requires an implementation and the
experimentation with this implementation as well.
Thus, these are also matters to be addressed in the successful
development of the approach that we have described here.

\section{The Contributions of this Thesis}
There are three main contributions of this thesis.
First is the definitions of a logic
rooted in an understanding of reasoning about LF derivability.
The atomic formulas of this logic represent LF derivations, and formulas
are constructed to capture relations about derivability in LF.
The logic consists of these formulas along with a semantics 
based on interpreting the derivability of judgements in LF.

Second is the development of a proof system for the logic which formalizes the
construction of arguments of validity based on this semantics.
Reasoning steps which are based on properties of LF derivability are
encoded as proof rules in the proof system allowing, for informal arguments
to be captured naturally as derivations.

Finally, we mechanize the construction of derivations in the proof system
through the implementation of the Adelfa theorem prover.
The effectiveness of this system is demonstrated through a collection
of examples which have all been formalized in Adelfa.

\section{Overview of the Thesis}
In Chapter~\ref{ch:lf} we present the specifics of the specification language
LF.
In this thesis we choose to work with the Canonical LF~\cite{harper07jfp}
formulation which permits only canonical form terms to be typed.
We will use the more detailed understanding of LF from this chapter to
provide an overview of the structure reasoning about such specifications
takes.

In Chapter~\ref{ch:logic} we present a logic for reasoning about LF derivability
based on an interpretation of atomic formulas as derivations in LF.
We describe the sorts of formulas of interest and provide their semantics
which relies on checking the derivability of LF judgements.

In Chapter~\ref{ch:proof-system} we propose a proof system to formalize 
reasoning based on the logic.
The proof rules in this system capture the sort of reasoning steps
which are used in informal reasoning about the validity of formulas.
The rules that we describe come in two forms: those that interpret the
meanings of the logical connectives and those that build in an
understanding of LF derivability that is embodied in atomic formulas.

In Chapter~\ref{ch:adelfa} we describe an implementation of the proof
system that we call Adelfa.
We describe how the system is used and specifics of how some of the
more complex rules in the proof system are implemented. 

In Chapter~\ref{ch:examples} we present a collection of examples which showcase 
the effectiveness of using Adelfa for reasoning about LF specifications.
These examples cover a variety of different kinds of systems which have been
specified using LF.

In Chapter~\ref{ch:comp} we contrast this work with previously developed approaches
to reasoning about LF in more detail.

We conclude in Chapter~\ref{ch:future-work} and include discussion
of future avenues of work involving the logic.

%% file: LF/lf.tex
\chapter{Canonical LF and the Specification of Object Systems}
\label{ch:lf}

The methodology for modelling object systems in a specification
language depends on there being a one-to-one correspondence between
the objects to be described and the expressions that are used to
describe them.
The existence of such a correspondence is the substance of the so-called
\emph{adequacy theorems}.
When LF is used as the specification language, the adequacy theorems
typically rely on limiting attention to normal forms with respect to
the $\beta$- and $\eta$-conversion rules in the $\lambda$-calculus;
these normal forms are referred to as the \emph{canonical terms} of
the language. 
The original presentation of LF~\cite{harper93jacm} includes terms
in both canonical and non-canonical form.
Such a presentation simplifies the treatment of substitution but at
the price of complicating arguments concerning adequacy and LF
derivability.  
In light of this, an alternative treatment of LF has been proposed 
that admits only terms that are in $\beta$-normal form and that are
well-typed only if they are additionally in $\eta$-long
form~\cite{harper07jfp,watkins03tr}.
We use this presentation of LF, called \emph{canonical LF}, as the
basis for this work. 
The first section recalls this presentation and develops
notions related to it that will be used in the later parts of this
thesis.  
Towards motivating the development of a reasoning logic, we then discuss 
the use of LF in representing
object systems and in reasoning about them at an informal level.
The chapter concludes with an identification of meta-theorems related
to derivability in LF that are useful in informal arguments concerning
this relation. 

\input{LF/canon-lf}
\input{LF/lf-ex}
\input{LF/lf-metathm}

%% file: LF/canon-lf.tex
\section{Canonical LF}
\label{sec:canon-lf}

Our presentation of canonical LF, henceforth referred to simply as LF,
differs from that in ~\cite{harper07jfp} in two respects.
First, we elide the subordination relation in typing judgements since
it is orthogonal to the thrust of this thesis.
Second, we treat substitution independently of LF typing judgements
and we also extend the notion to include the simultaneous
replacement of multiple variables. 
The elaboration below builds in these ideas.

\begin{figure}[htpb]
\[
\begin{array}{r r c l}
  \mbox{\bf Kinds} & K & ::= & \type\ |\ \typedpi{x}{A}{K}\\[5pt]

  \mbox{\bf Canonical Type Families} & A,B & ::= &
           P\ |\ \typedpi{x}{A}{B}\\
  \mbox{\bf Atomic Type Families} & P & ::= & a\ |\ P\app M\\[5pt]
  \mbox{\bf Canonical Terms} & M,N & ::= & R\ |\ \lflam{x}{M}\\
  \mbox{\bf Atomic Terms} & R & ::= & c\ |\ x\ |\ R\app M\\[5pt]
  \mbox{\bf Signatures} & \Sigma & ::= &
  \emptysig\ |\ \Sigma,c:A\ |\ \Sigma,a:K\\[5pt]
  \mbox{\bf Contexts} & \Gamma & ::= & \emptyctx\ |\ \Gamma,x:A
\end{array}
\]
\caption{The Syntax of LF Expressions}
\label{fig:lf-terms}
\end{figure}

\subsection{The Syntax}\label{ssec:lf-syntax}

The syntax of LF expressions is described in
Figure~\ref{fig:lf-terms}.
The primary interest is in three categories of expressions:
kinds, types which are indexed by kinds, and terms which are indexed
by types. 
In these expressions, $\lambda$ and $\Pi$ are binding or abstraction
operators.
Relative to these operators, we assume the principle of
equivalence under renaming that is applied as needed. 
We also assume as understood the notions of free and bound
variables that are usual to expressions involving such operators.
To ensure the absence of $\beta$-redexes, terms are stratified into
\emph{canonical} and \emph{atomic} forms.  
A similar stratification is used with types that is exploited by the
formation rules to force all well-typed terms to be in $\eta$-long
form.
We use $x$ and $y$ to represent term-level variables, which are bound
by abstraction operators or in the contexts that are associated with
terms. 
Further, we use $c$ and $d$ for term-level constants, and $a$
and $b$ for type-level constants, both of which are typed in
signatures. 
The expression $\imp{A_1}{A_2}$ is used as an alternative notation
for the type family $\typedpi{x}{A_1}{A_2}$ when $x$ does
not appear free in $A_2$. 
An atomic term has the form $(h\app M_1\app \ldots\app M_n)$ where $h$
is a variable or a constant.
We refer to $h$ as the \emph{head symbol} of such a term. 

\subsection{Simultaneous Hereditary Substitution}\label{ssec:hsub}

We will need to consider substitution into LF expressions when
explicating typing and other logical notions related to these
expressions.
To preserve the form of these expressions, it is necessary to build
$\beta$-reduction into the application of such substitutions.
An important consideration in this context is that substitution
application must be a terminating operation. 
Towards ensuring this property, substitutions are indexed by types
that are eventually intended to characterize the functional structure
of expressions. 

\begin{definition}[Arity Types]\label{def:aritytypes}
The collection of expressions that are obtained from the constant
$\oty$ using the binary infix constructor $\atyarr$ constitute the
{\it arity types}.
Corresponding to each canonical type $A$, there is an arity type
called its {\it erased form} and denoted by $\erase{A}$ and given as
follows: $\erase{P} = \oty$ and $\erase{\typedpi{x}{A_1}{A_2}} =
\arr{\erase{A_1}}{\erase{A_2}}$.  

\end{definition}

\begin{definition}[Substitutions]\label{def:substitution}
A variable substitution $\theta$ is a finite set of tuples of 
the form $\{\langle x_1,M_1,\alpha_1 \rangle, \ldots, \langle x_n,
M_n, \alpha_n \rangle \}$, where, for $1 \leq i \leq n$, $x_i$ is a
distinct variable, $M_i$ is a canonical term and $\alpha_i$ is an
arity type.\footnote{Note that by a 
systematic abuse of notation, $n$ may be less than $m$ in a sequence written in
the form $s_m,\ldots,s_n$, in which case the empty sequence is
denoted. In this particular instance, a substitution can be an empty
set of triples.}
Given such a substitution, $\domain{\theta}$ denotes the set 
$\{x_1,\ldots,x_n\}$ and $\range{\theta}$ denotes the set
$\{M_1,\ldots,M_n\}$.
\end{definition}
 
Given a substitution $\theta$ and an expression $E$ that is a kind, a
type, a canonical term or a context, we wish the expression
$\hsubst{\theta}{E}$ notionally to denote the application of $\theta$
to $E$.
However, such an application is not guaranteed to exist.
We therefore use the expression $\hsub{\theta}{E}{E'}$ to indicate
when it is defined and has $E'$ as a result.
The key part of defining this relation is that of articulating
its meaning when $E$ is a canonical term.
This is done in Figure~\ref{fig:hsub} via rules for deriving this
relation. 
These rules use an auxiliary definition of substitution into an
atomic term which accounts for any normalization that is necessitated
by the replacement of a variable by a term.
The different categories of rules in this figure are distinguished by
being preceded by a box containing the judgement form they relate to.
The extension of this definition to the case where $E$ is a kind or a type
corresponds essentially to the application of the substitution to the
terms that appear within $E$. 
This idea is made explicit for types in Figure~\ref{fig:hsubtypes} and its
elaboration for kinds is similar.
A substitution is meaningfully applied to a context only when it does
not replace variables to which the context assigns types and when a
replacement does not lead to inadvertent capture.
When these conditions are satisfied, the substitution distributes to
the types that are assigned to the variables as the rules in
Figure~\ref{fig:hsubctx} make clear. 

\begin{figure}[tbhp]

\fbox{$\hsub{\theta}{M}{M'}$}

\begin{center}
\begin{tabular}{ccc}    
  \infer{\hsub{\theta}
              {R}
              {R'}}
        {\hsubr{\theta}
              {R}
              {R'}} \qquad
  & \qquad
  \infer{\hsub{\theta}
              {R}
              {M'}}
        {\hsubr{\theta}
              {R}
              {M':\alpha'}}

  & \qquad
\infer{\hsub{\theta}
            {(\lflam{x}{M})}
            {\lflam{x}{M'}}}
      {x\ \mbox{\rm not free in}\ \domain{\theta} \cup \range{\theta}
        \qquad
        \hsub{\theta}{M}{M'}}
\end{tabular}
\end{center}

\vspace{7pt}

\fbox{$\hsubr{\theta}{R}{M':\alpha'}$}

\begin{center}
\begin{tabular}{c}
\infer{\hsubr{\theta}{x}{M:\alpha}}
      {\langle x,M,\alpha \rangle \in \theta}

\\[10pt]

\infer{\hsubr{\theta}{(R\app M)}{M''':\alpha''}}
      {\hsubr{\theta}{R}{\lflam{x}{M'}:\arr{\alpha'}{\alpha''}} 
       \qquad
       \hsub{\theta}{M}{M''}
       \qquad
       \hsub{\{\langle x, M'',\alpha'\rangle \}}{M'}{M'''}}
\end{tabular}
\end{center}

\vspace{7pt}

\fbox{$\hsubr{\theta}{R}{R'}$}

\begin{center}
\begin{tabular}{ccc}
\infer{\hsubr{\theta}{c}{c}}
      { }
& \qquad
\infer{\hsubr{\theta}
            {x}
            {x}}
      {x\not\in\domain{\theta}}

& \qquad 
\infer{\hsubr{\theta}{(R\app M)}{R'\app M'}}
       {\hsubr{\theta}{R}{R'} \qquad\qquad
        \hsub{\theta}{M}{M'}}
\end{tabular}
\end{center}

\caption{Applying Substitutions to Terms}
\label{fig:hsub}
\end{figure}

\begin{figure}[tbhp]
\begin{center}
\begin{tabular}{c}
\infer{\hsub{\theta}{a}{a}}
      { }

\qquad \qquad
  
\infer{\hsub{\theta}{(P\app M)}{(P'\app M')}}
      {\hsub{\theta}{P}{P'} &
       \hsub{\theta}{M}{M'}}

\\[10pt]

\infer{\hsub{\theta}{(\typedpi{x}{A_1}{A_2})}{\typedpi{x}{A_1'}{A_2'}}}
      {x\ \mbox{\rm not free in}\ \domain{\theta} \cup \range{\theta}
        \qquad
        \hsub{\theta}{A_1}{A_1'} \qquad
       \hsub{\theta}{A_2}{A_2'}}
\end{tabular}
\end{center}

\caption{Applying Substitutions to Types}
\label{fig:hsubtypes}
\end{figure}

\begin{figure}[tbhp]
\begin{center}
\begin{tabular}{c}
  \infer
    {\hsub{\theta}
          {\emptyctx}
          {\emptyctx}}
    {}

    \qquad
    
  \infer
    {\hsub{\theta}{(\Gamma, x:A)}{\Gamma', x:A'}}
    {x\ \mbox{\rm not free in}\ \domain{\theta} \cup \range{\theta} \qquad
     \hsub{\theta}{\Gamma}{\Gamma'} \qquad
     \hsub{\theta}{A}{A'}} 
\end{tabular}
\end{center}  
\caption{Applying Substitutions to Contexts}
\label{fig:hsubctx}
\end{figure}

We define a measure on substitutions that is useful in showing that
their application terminates. 

\begin{definition}[Size]\label{def:typesubsize}
The size of an arity type is the number of occurrences of $\atyarr$ in
it. The size of a substitution is the largest of the sizes of the arity types
in each of its triples.
\end{definition}

The following theorem shows that simultaneous hereditary substitution
is terminating and the result will be unique if it exists.

\begin{theorem}[Uniqueness]\label{th:uniqueness}
For any context, kind, type or canonical term $E$ and any substitution
$\theta$, it is decidable whether there is an $E'$ such that
$\hsub{\theta}{E}{E'}$ is derivable.
Moreover, there is at most one $E'$ for which it is derivable.
Similarly, for any atomic term $R$ and substitution $\theta$, it is
decidable whether there is an $R'$ or an $M'$ and $\alpha'$ such
that $\hsubr{\theta}{R}{R'}$ or $\hsubr{\theta}{R}{M' : \alpha'}$ is
derivable.
At most one of these judgements is derivable and for at most one
$R'$, respectively, $M'$ and $\alpha'$. 
\end{theorem}
\begin{proof}
This theorem is proved by induction first on
the size of substitutions and then on the structure of expressions.
We first prove it simultaneously for canonical and atomic terms, and then
extended to atomic types, canonical types, kinds, and finally contexts.
\end{proof}

The following theorem shows that the application of a vacuous
hereditary substitution always exists. 

\begin{theorem}[Vacuous Substitutions]\label{th:vacuoussubs}
If $E$ is a kind, a type or a canonical term none of whose free
variables is a member of $\domain{\theta}$, then $\hsub{\theta}{E}{E}$ has
as derivation. If $R$ is an atomic term none of whose free variables
is a member of $\domain{\theta}$ then $\hsubr{\theta}{R}{R}$ has a
derivation.
\end{theorem}
\begin{proof}
The proof is by induction on the structure of the expression.
\end{proof}

Simultaneous hereditary substitution enjoys a permutation 
property that is similar to the one described in \cite{harper07jfp}
for unitary substitution. This is the content of the theorem below.
\begin{theorem}[Permutation of Substitutions]\label{th:subspermute}
Let $\theta_1$ be an arbitrary substitution of the form
$\{\langle x_1, M_1,\alpha_1 \rangle, \ldots, 
   \langle x_n,M_n,\alpha_n \rangle \}$.
Further, let
$\theta_2$ be an arbitrary substitution of he form
$\{\langle y_1, N_1,\beta_1\rangle, \ldots, 
   \langle y_m,N_m,\beta_m \rangle \}$ 
where
$y_1,\ldots,y_m$ are variables that are distinct from $x_1,\ldots,x_n$ 
and that do not appear free in $M_1,\ldots,M_n$.  
Finally, suppose that for each $i$, $1 \leq i \leq m$, there is some
term $N'_i$ such that $\hsub{\theta_1}{N_i}{N'_i}$ has a derivation and let
$\theta_3$ be the substitution defined by
$\{ \langle y_1, N'_1,\beta_1 \rangle, \ldots, 
    \langle y_m,N'_m,\beta_m \rangle \}$.
For every kind, type, or canonical term $E$, $E_1$, and
$E_2$ such that $\hsub{\theta_1}{E}{E_1}$ and $\hsub{\theta_2}{E}{E_2}$ have
derivations, there must be an $E'$ such 
that $\hsub{\theta_1}{E_2}{E'}$ and $\hsub{\theta_3}{E_1}{E'}$ have
derivations.
\end{theorem}
\begin{proof}
The proof proceeds by a primary induction on the sum of the sizes
of $\theta_1$ and $\theta_2$ and a secondary induction on the
derivation of $\hsub{\theta_2}{E}{E_2}$.
We omit the details which are similar to those for Lemma 2.10 in
\cite{harper07jfp}.
\end{proof}

\begin{figure}[tbhp]
\begin{center}

  \begin{tabular}{c}

  \infer
      {\stlctyjudgr{\STLCGamma}{c}{\alpha}}
      {c: \alpha \in \STLCGamma}

  \qquad
  \infer
      {\stlctyjudgr{\STLCGamma}{x}{\alpha}}
      {x: \alpha \in \STLCGamma}

  \qquad
  \infer
      {\stlctyjudgr{\STLCGamma}{R \app M}{\alpha}}
      {\stlctyjudgr{\STLCGamma}{R}{\alpha' \atyarr \alpha} \qquad
       \stlctyjudg{\STLCGamma}{M}{\alpha'}}

  \\[10pt]
  \infer
      {\stlctyjudg{\STLCGamma}{\lflam{x}{M}}{\alpha_1 \atyarr \alpha_2}}
      {\stlctyjudg{\aritysum{\{x:\alpha_1\}}{\STLCGamma}}{M}{\alpha_2}}

  \qquad
  \infer
      {\stlctyjudg{\STLCGamma}{R}{\oty}}
      {\stlctyjudgr{\STLCGamma}{R}{\oty}}
  \end{tabular}
\end{center}
\caption{Arity Typing for Canonical Terms}\label{fig:aritytyping}
\end{figure}

While the application of a substitution to an LF expression may not
always exist, this is guaranteed to be the case when certain arity
typing constraints are satisfied as we describe below.

\begin{definition}[Arity Typing]\label{def:aritytyping}
An \emph{arity context} $\STLCGamma$ is a set of unique assignments of
arity types to (term) constants and variables; these assignments are
written as $x:\alpha$ or $c:\alpha$.
Given two arity contexts $\STLCGamma_1$ and $\STLCGamma_2$, we write
$\aritysum{\STLCGamma_1}{\STLCGamma_2}$ to denote the collection of all the 
assignments in $\STLCGamma_1$ and the assignments in $\STLCGamma_2$ to
the constants or variables not already assigned a type in $\STLCGamma_1$. 
The rules in Figure~\ref{fig:aritytyping} define the arity typing
relation denoted by $\stlctyjudg{\STLCGamma}{M}{\alpha}$ between
a term $M$ and an arity type $\alpha$ relative to an arity context
$\STLCGamma$. 
A kind or type $E$ is said to respect an arity context $\STLCGamma$
under the following conditions: if $E$ is $\type$; if $E$ is an atomic
type and for each canonical term $M$ appearing in $E$ there is an
arity type $\alpha$ such that $\stlctyjudg{\STLCGamma}{M}{\alpha}$ is 
derivable; and if $E$ has the form $\typedpi{x}{A}{E'}$ and $A$
respects $\STLCGamma$ and $E'$ respects
$\aritysum{\{x:\erase{A}\}}{\STLCGamma}$. 
A context $\Gamma$ is said to respect $\STLCGamma$ if
for every $x:A$ appearing in $\Gamma$ it is the case that $A$ respects
$\STLCGamma$.
A substitution $\theta$ is {\it arity type preserving}
with respect to $\Theta$ if for every $\langle x,M,\alpha \rangle \in
\theta$ it is  the case that $\stlctyjudg{\STLCGamma}{M}{\alpha}$ is
derivable. 
Associated with a substitution $\theta$ is the arity context $\{ x :
\alpha\ \vert\ \langle x, M, \alpha \rangle \in \theta \}$ that is
denoted by $\context{\theta}$.
\end{definition}

\begin{theorem}[Arity Type Preserving Substitution Always Defined]\label{th:aritysubs}
Let $\theta$ be a substitution that is arity type preserving with
respect to $\STLCGamma$ and let $\STLCGamma'$ denote the arity context 
$\aritysum{\context{\theta}}{\STLCGamma}$. 
\begin{enumerate}
\item If $E$ is a canonical type or kind that respects the
  arity context $\STLCGamma'$, then there must be an $E'$ that
  respects $\STLCGamma$ and that is such that $\hsub{\theta}{E}{E'}$
  is derivable. 

\item If $M$ is a canonical term such that
$\stlctyjudg{\STLCGamma'}{M}{\alpha}$ is
derivable, then there must be an $M'$ such that
$\hsub{\theta}{M}{M'}$ and $\stlctyjudg{\STLCGamma}{M'}{\alpha}$ are
derivable.

\item If $R$ is an atomic term such that
  $\stlctyjudgr{\STLCGamma'}{R}{\alpha}$
  is derivable, then either there is an atomic term $R'$ such that
  $\hsub{\theta}{R}{R'}$ and $\stlctyjudgr{\STLCGamma}{R}{\alpha}$ are
  derivable or there is a canonical term $M$ such that
  $\hsub{\theta}{R}{M : \alpha}$ and
  $\stlctyjudg{\STLCGamma}{M}{\alpha}$ are derivable. 
\end{enumerate}
\end{theorem}
\begin{proof}
The first clause in the theorem is an easy consequence of the
second. We prove clauses (2) and (3) simultaneously by induction first
on the sizes of substitutions and then on the structure of terms.
The argument proceeds by considering the cases for the term structure,
first proving (3) and then using this in proving (2). 
\end{proof}

We will often consider expressions and substitutions that satisfy the
arity typing requirements of the theorem above, which then guarantees
that the applications of the substitutions have results.  
We introduce a notation that is convenient in this situation: we will
write $\hsubst{\theta}{E}$ to denote the unique $E'$ such that
$\hsub{\theta}{E}{E'}$ has a derivation whenever such a derivation is
known to exist.

\begin{definition}[Composition of Substitutions]\label{def:composition}
Two substitutions $\theta_1$ and $\theta_2$ are said to be \emph{arity type
compatible} relative to the arity context $\STLCGamma$ if 
$\theta_2$ is type preserving with respect to $\STLCGamma$ and
$\theta_1$ is type preserving with respect to
$\aritysum{\context{\theta_2}}{\STLCGamma}$. The composition of two
such substitutions, written as $\theta_2 \circ \theta_1$, is 
the substitution
\begin{tabbing}
  \qquad\qquad\qquad\=\qquad\qquad\qquad\qquad\qquad\=\kill
  \> $\{ \langle x,M',\alpha \rangle\ \vert\ \langle x,M,\alpha
  \rangle \in \theta_1\ \mbox{\rm and}\ \hsub{\theta_2}{M}{M'}\ \mbox{\rm has a
    derivation}\}\ \cup$\\
  \>\>$\{ \langle y,N,\beta \rangle\ \vert\ \langle y,N,\beta \rangle
  \in \theta_2\ \mbox{\rm and}\ y \not\in \domain{\theta_1} \}$
\end{tabbing}
By Theorem~\ref{th:aritysubs} there must be an $M'$ for
which $\hsub{\theta_2}{M}{M'}$ has a derivation for each $\langle
x,M,\alpha \rangle \in \theta_1$. Moreover such an $M'$ must be 
unique. Thus, the composition described herein is well-defined. Note
also that the composition must also be arity type preserving with
respect to $\STLCGamma$.
\end{definition}

\begin{theorem}[Composition]\label{th:composition}
Let $\theta_1$ and $\theta_2$ be substitutions that are arity type
compatible relative to $\STLCGamma$ and let $\STLCGamma'$ denote the
arity context $\aritysum{\context{\theta_2 \circ \theta_1}}{\STLCGamma}$.
\begin{enumerate}
\item If $E$ is a canonical kind, type or context that respects $\STLCGamma'$ and 
$E'$ and $E''$ are, respectively, canonical types or kinds such that 
$\hsub{\theta_1}{E}{E'}$ and $\hsub{\theta_2}{E'}{E''}$ have
  derivations, then $\hsub{\theta_2 \circ \theta_1}{E}{E''}$ has a
  derivation.
  
\item If $M$ is a canonical term such that, for some arity type
  $\alpha$, $\stlctyjudg{\STLCGamma'}{M}{\alpha}$ is derivable and $M'$
  and $M''$ are canonical terms such that 
$\hsub{\theta_1}{M}{M'}$ and $\hsub{\theta_2}{M'}{M''}$ have
derivations, then $\hsub{\theta_2 \circ \theta_1}{M}{M''}$ has a
derivation.

\item If $R$ is a canonical term such that, for some arity type
  $\alpha$, $\stlctyjudg{\STLCGamma'}{R}{\alpha}$ is derivable and
  \begin{enumerate}
    \item $M'$ and $M''$ are canonical terms such that
      $\hsubr{\theta_1}{R}{M' : \alpha}$ and $\hsub{\theta_2}{M'}{M''}$
      have derivations, then $\hsubr{\theta_2 \circ \theta_1}{R}{M'':\alpha}$
      has a derivation;

    \item $R'$ and $M''$ are, respectively, an atomic and a
      canonical term such that both $\hsubr{\theta_1}{R}{R'}$ and
      $\hsubr{\theta_2}{R'}{M'' : \alpha}$ have derivations then
      $\hsubr{\theta_2 \circ \theta_1}{R}{M'' : \alpha}$ has a
      derivation; 

    \item $R'$ and $R''$ are atomic terms such that
      $\hsubr{\theta_1}{R}{R'}$ and $\hsubr{\theta_2}{R'}{R''}$ have
      derivations, then $\hsubr{\theta_2 \circ \theta_2}{R}{R''}$ has
      a derivation.
  \end{enumerate}
\end{enumerate}
\end{theorem}

\begin{proof}
Clause 1 of the theorem follows easily from an induction on the
structure of the canonical type or kind, assuming the property stated
in clause 2.

\smallskip
We prove clauses 2 and 3 together. These clauses are premised on the 
existence of a derivation corresponding to the application of the
substitution $\theta_1$ to either $M$ or $R$.
The argument is by induction on the size of this derivation and it
proceeds by considering the cases for the last rule in the
derivation. 

\smallskip
We consider first the cases where the derivation is for
$\hsub{\theta_1}{M}{M'}$; the clause in the theorem relevant to these
cases is 2. 
An easy argument using the induction hypothesis yields the desired
conclusion when $M$ is of the form $\lflam{x}{M_1}$.
In the case that $M$ is an atomic term, there is a shorter derivation
for $\hsubr{\theta_1}{M}{M' : \alpha}$ or $\hsubr{\theta_1}{M}{M'}$.
In the first case, the induction hypothesis, specifically clause 3(a),
allows us to conclude that $\hsub{\theta_2 \circ \theta_1}{M}{M''}$
has a derivation. 
In the second case, $M'$ must be an atomic term and there must
therefore be a derivation for $\hsubr{\theta_2}{M'}{M'':\alpha}$ or
$\hsubr{\theta_2}{M'}{M''}$.
Using the induction hypothesis, specifically clause 3(b) or 3(c), we
can again conclude that there must be a derivation for $\hsub{\theta_2
  \circ \theta_1}{M}{M''}$.

\smallskip
We consider next the cases for the last rule when the derivation is for
$\hsubr{\theta_1}{R}{M' : \alpha}$.
\begin{itemize}
\item If $M$ is a variable $x$ such that $\langle x,M',\alpha \rangle
  \in \theta_1$, then it must be the case that there is some 
  $\langle x,M'',\alpha\rangle \in \theta_2 \circ \theta_1$.
Hence there must be a derivation for $\hsubr{\theta_2 \circ
  \theta_1}{M}{M'' : \alpha}$. 

\item Otherwise $M$ must be of the form $(R_1\app M_2)$ where there are 
  derivations of judgements
  $\hsub{\theta_1}{R_1}{\lflam{x}{M_3} : \alpha' \atyarr \alpha}$, 
  $\hsub{\theta_1}{M_2}{M_4}$, and
  $\hsub{\{\langle x, M_4,\alpha'\rangle \}}{M_3}{M'}$ 
  for suitable choices for $M_3$, $\alpha'$ and
  $M_4$.
  We note first that the arity context
  $\aritysum{(\context{\theta_2 \circ \theta_1})}{\STLCGamma}$
  is equal to 
  $\aritysum{\context{\theta_1}}{(\aritysum{\context{\theta_2}}{\STLCGamma})}$. 
  Then, by the assumptions of the theorem and
  Theorem~\ref{th:aritysubs}, it follows that there must be terms
  $M'_3$ and $M'_4$ such that $\hsub{\theta_2}{M_3}{M'_3}$ and
  $\hsub{\theta_2}{M_4}{M'_4}$ have derivations.
  We see by using the induction hypothesis with respect to the derivation for
  $\hsub{\theta_1}{R_1}{\lflam{x}{M_3} : \alpha' \atyarr \alpha}$ that
  there must be a derivation for $\hsub{\theta_2 \circ 
    \theta_1}{R_1}{\lflam{x}{M'_3} : \alpha' \atyarr \alpha}$. 
  Using the induction hypothesis again with respect
  to the derivation for $\hsub{\theta_1}{M_2}{M_4}$, we see that
  there must be a derivation for $\hsub{\theta_2 \circ
    \theta_1}{M_2}{M'_4}$. By Theorem~\ref{th:subspermute} it follows
  that $\hsub{\{\langle  
      x, M'_4,\alpha' \rangle \}}{M'_3}{M''}$ has a derivation and,
  hence that $\hsubr{\theta_2 \circ \theta_1}{(R_1 \app M_2)}{M'' :
    \alpha}$ has one too. 
\end{itemize}

Finally we consider the cases for the last rule when the derivation is
for $\hsubr{\theta_1}{R}{R'}$.
The argument when $R$ is a constant is trivial.
The case when $R$ is a variable follows almost as immediately using
the definition of $\theta_2\circ\theta_1$.
The only remaining case is when $R$ is of the form $(R_1 \app M_2)$
and $R'$ is $(R'_1\app M'_2)$ where $\hsubr{\theta_1}{R_1}{R'_1}$ and
$\hsub{\theta_1}{M_2}{M'_2}$ have shorter derivations for suitable
terms $R'_1$ and $M'_2$.
We then have two subcases to consider with respect to the application
of $\theta_2$ to $(R'_1 \app M'_2)$:
\begin{itemize}
  \item There is a derivation for $\hsubr{\theta_2}{(R'_1
    \app M'_2)}{(R''_1 \app M''_2)}$ where $R''_1$ and $M''_2$ are terms
    such that $\hsubr{\theta_2}{R'_1}{R''_1}$ and
    $\hsubr{\theta_2}{M'_2}{M''_2}$ have derivations; note that the
    relevant clause in this case is 3(c) and $R''$ is $(R''_1 \app
    M''_2)$. 
    The induction hypothesis lets us conclude 
    that $\hsubr{\theta_2 \circ \theta_1}{R_1}{R''_1}$ and
    $\hsubr{\theta_2 \circ \theta_1}{M_2}{M''_2}$ have derivations.
    Hence, $\hsubr{\theta_2 \circ \theta_1}{(R_1\app M_2)}{(R''_1\app
      M''_2)}$ must have a derivation.

   \item There is a derivation for $\hsubr{\theta_2}{(R'_1\app
     M'_2)}{M'' : \alpha}$.
     In this case, for suitable choices for $M_3$,
     $\alpha'$ and $M_4$, there must be derivations for
     $\hsubr{\theta_2}{R'_1}{\lflam{x}{M_3} : \alpha' \atyarr \alpha}$,
     $\hsub{\theta_2}{M'_2}{M_4}$ and $\hsub{\{\langle x, M_4,\alpha'
     \rangle \}}{M_3}{M''}$.
     The induction hypothesis now lets us conclude that there are
     derivations for judgments
     $\hsubr{\theta_2\circ\theta_1}{R_1}{\lflam{x}{M_3} : \alpha' 
       \atyarr \alpha}$ and $\hsub{\theta_2 \circ \theta_1}{M_2}{M_4}$.
     It then follows easily that there must be a derivation for
     $\hsubr{\theta_2 \circ \theta_1}{(R_1 \app M_2)}{M'' : \alpha}$.
    \end{itemize}
\end{proof}

The erased form of a type is invariant under substitution. This is the
content of the theorem below whose proof is straightforward.

\begin{theorem}[Erasure is Invariant Under Substitution]\label{th:erasure}
For any type $A$ and substitution $\theta$, if $\hsub{\theta}{A}{A'}$
has a derivation, then $\erase{A} = \erase{A'}$.
\end{theorem}
\begin{proof}
The proof is by induction on the height of the derivation for 
$\hsub{\theta}{A}{A'}$ and using the definition of the erasure.
\end{proof}

\subsection{Wellformedness Judgements}\label{lf-typing}

\begin{figure}
\fbox{$\lfsig{\Sigma}$}
  
\begin{center}
\begin{tabular}{c}

\infer[\sigempty]{\lfsig{\emptysig}}{}

\\[10pt]
  
\infer[\sigterm]
      {\lfsig{\Sigma,c:A}}
      {\lfsig{\Sigma} \qquad \lftype{\emptyctx}{A} \qquad c\ \mbox{\rm does not
          appear in}\ \Sigma}

\\[10pt]

\infer[\sigfam]
      {\lfsig{\Sigma,a:K}}
      {\lfsig{\Sigma} \qquad \lfkind{\emptyctx}{K} \qquad a\ \mbox{\rm does not
          appear in}\ \Sigma}
\end{tabular}
\end{center}
      
\vspace{4pt}
\fbox{$\lfctx{\Gamma}$}

\begin{center}
\begin{tabular}{c}
    
\infer[\ctxempty]
      {\lfctx{\emptyctx}}{}

\\[10pt]
      
\infer[\ctxterm]
      {\lfctx{\Gamma,x:A}}
      {\lfctx{\Gamma} \qquad \lftype{\Gamma}{A} \qquad x\ \mbox{\rm does not
          appear free in}\ \Gamma}
\end{tabular}
\end{center}
\caption{The Formation Rules for LF Signatures and Contexts}
\label{fig:lf-judgements-a}
\end{figure}

\begin{figure}
\fbox{$\lfkind{\Gamma}{K}$}

\begin{center}
\begin{tabular}{c}
\infer[\canonkindtype]
      {\lfkind{\Gamma}{\type}}{}

\\[10pt]

\infer[\canonkindpi]
      {\lfkind{\Gamma}{\typedpi{x}{A}{K}}}
      {\lftype{\Gamma}{A} \qquad
       \lfkind{\Gamma,x:A}{K}}
\end{tabular}
\end{center}

 \fbox{$\lftype{\Gamma}{A}$}

\begin{center}
\begin{tabular}{c}       
\infer[\canonfamatom]
      {\lftype{\Gamma}{P}}
      {\lfsynthkind{\Gamma}{P}{\type}}
      
\\[10pt]

\infer[\canonfampi]
      {\lftype{\Gamma}{\typedpi{x}{A_1}{A_2}}}
      {\lftype{\Gamma}{A_1} \qquad \lftype{\Gamma, x:A_1}{A_2}}
\end{tabular}
\end{center}

\vspace{4pt}
\fbox{$\lfsynthkind{\Gamma}{P}{K}$}

\begin{center}
\begin{tabular}{c}
\infer[\atomfamconst]
      {\lfsynthkind{\Gamma}{a}{K}}
      {a:K\in\Sigma}
\\[10pt]
\infer[\atomfamapp]
      {\lfsynthkind{\Gamma}{P\app M}{K}}
      {\lfsynthkind{\Gamma}{P}{\typedpi{x}{A}{K_1}} \qquad
       \lfchecktype{\Gamma}{M}{A} \qquad
      \hsub{\{\langle x, M, \erase{A}\rangle\}}{K_1}{K}}
\end{tabular}
\end{center}

\vspace{4pt}
\fbox{$\lfchecktype{\Gamma}{M}{A}$}

\begin{center}
\begin{tabular}{cc}
\infer[\canontermatom]
      {\lfchecktype{\Gamma}{R}{P}}
      {\lfsynthtype{\Gamma}{R}{P}} &
\infer[\canontermlam]
      {\lfchecktype{\Gamma}{\lflam{x}{M}}{\typedpi{x}{A_1}{A_2}}}
      {\lfchecktype{\Gamma,x:A_1}{M}{A_2}}
\end{tabular}
\end{center}

\vspace{4pt}
\fbox{$\lfsynthtype{\Gamma}{R}{A}$}
    
\begin{center}
  \begin{tabular}{c}
\infer[\atomtermvar]
      {\lfsynthtype{\Gamma}{x}{A}}
      {x:A\in\Gamma}
\qquad
\infer[\atomtermconst]
      {\lfsynthtype{\Gamma}{c}{A}}
      {c:A\in\Sigma}
      
\\[10pt]

\infer[\atomtermapp]
      {\lfsynthtype{\Gamma}{R\app M}{A}}
      {\lfsynthtype{\Gamma}{R}{\typedpi{x}{A_1}{A_2}} \quad\ 
       \lfchecktype{\Gamma}{M}{A_1} \quad\ 
       \hsub{\{\langle x, M, \erase{A_1}\rangle\}}{A_2}{A}} 
 \end{tabular}
\end{center}

\caption{The Formation Rules for LF Kinds, Types, and Terms}
\label{fig:lf-judgements-b}
\end{figure}

Canonical LF includes seven judgements: $\lfsig{\Sigma}$ that ensures
that the constants declared in a signature are distinct and their type
or kind classifiers are well-formed; $\lfctx{\Gamma}$ that ensure that
the variables declared in a signature are distinct and their type
classifiers are well-formed in the preceding declarations and
well-formed signature $\Sigma$; $\lfkind{\Gamma}{K}$ that determines
that a kind $K$ is well-formed with respect to a well-formed signature
and context pair; $\lftype{\Gamma}{A}$ and
$\lfsynthkind{\Gamma}{P}{K}$ that check, respectively, the formation
of a canonical and atomic type relative to a well-formed 
signature, context and kind triple; and $\lfchecktype{\Gamma}{M}{A}$ and
$\lfsynthtype{\Gamma}{R}{A}$ that ensure, respectively, that a
canonical and atomic term are well-formed with respect to a
well-formed signature, context and canonical type triple. 
Figure~\ref{fig:lf-judgements-a} presents the rules for deriving
the first two of these judgements, and the remaining judgments are
presented in Figure~\ref{fig:lf-judgements-b}.
In the rules \canonkindpi\ and \canontermlam\ we assume $x$ to be a
variable that does not appear free in $\Gamma$.
The formation rule for type and term level application,
i.e. $\atomfamapp$ and $\atomtermapp$, require the substitution of a
term into a kind or a type.
Use is made towards this end of hereditary substitution. The index for
such a substitution is obtained by erasure from the type established
for the term.

The judgement forms other than $\lfsig{\Sigma}$ that are described
above are parameterized by a signature that remains unchanged in the
course of their derivation.
In the rest of this thesis we will assume a fixed signature that has in
fact been verified to be well-formed at the outset. 
The judgement forms require some of their other components to satisfy
additional restrictions. 
For example, judgements of the form $\lfchecktype{\Gamma}{M}{A}$
require that $\Sigma$, $\Gamma$ and $A$
be well-formed as an ensemble.
Judgements of the form $\lfsynthtype{\Gamma}{R}{A}$ instead
require that $\Sigma$ and $\Gamma$ be well-formed and ensure
the well-formedness of both $R$ and $A$.
To be coherent, the rules in Figure~\ref{fig:lf-judgements-b} must
ensure that in deriving a judgement that satisfies these requirements,
it is necessary only to consider the derivation of judgements that
also accord with these requirements.
The fact that they possess this property can be verified by an
inspection of their structure, using the observation that will be made
in Theorem~\ref{th:transitivity} that hereditary substitution
preserves the property of being well-formed for kinds and types. 

Arity typing judgements for terms approximate LF typing judgements as
made precise below. 

\begin{definition}[Induced Arity Context]
  The arity context induced by the signature $\Sigma$  and context
  $\Gamma$ is the collection of assignments that includes $x :
  \erase{A}$ for each $x : A \in \Gamma$ and $c : \erase{A}$ for each
  $c : A \in \Sigma$.
  When the context $\Gamma$ is irrelevant or empty, we shall refer to
  the arity context as the one induced by just $\Sigma$.
\end{definition}

\begin{theorem}[Arity Typing Approximates LF Typing]\label{th:arityapprox}
  Let $\STLCGamma$ be the arity context induced by the signature
  $\Sigma$ and context $\Gamma$.
  If $\lfctx{\Gamma}$ then $\Gamma$ respects $\Theta$.
  If $\lfkind{\Gamma}{K}$ or $\lftype{\Gamma}{A}$ then, respectively,
  $K$ or $A$ respect $\STLCGamma$.
  If $\lfchecktype{\Gamma}{M}{A}$ is derivable, then
  $\stlctyjudg{\STLCGamma}{M}{\erase{A}}$ must also be derivable.
  If $\lfsynthtype{\Gamma}{R}{A}$ is derivable, then
  $\stlctyjudg{\STLCGamma}{R}{\erase{A}}$ must also be derivable.
\end{theorem}
\begin{proof}
The last two parts of the theorem are proved simultaneously by
induction on the size of the derivation of
$\lfchecktype{\Gamma}{M}{A}$ and $\lfsynthtype{\Gamma}{R}{A}$. 
The first two parts follows from them, again by induction on the
derivation size.  
\end{proof}

%% file: LF/lf-ex.tex
\section{Formalizing Object Systems in LF}
\label{sec:lf-ex}

\begin{figure}
\[
\begin{array}{lcl}
  \of{\tpty}{\type}
  & \quad\quad & 
  \of{\ofemptytm}{\ofty\app\emptytm\app\unittm}

  \\
  
  \of{\unittm}{\tpty}
  & &

  \\

  \of{\arrtm}{\arr{\tpty}{\tpty}}
  & &
  \of{\ofapptm}
     {\typedpi{E_1}{\tmty}{\typedpi{E_2}{\tmty}
     {\typedpi{T_1}{\tpty}{\typedpi{T_2}{\tpty}{}}}}}
  \\
  & &
    \qquad
    \typedpi{D_1}{\ofty\app E_1\app (\arrtm\app T_1\app T_2)}{
      \typedpi{D_2}{\ofty\app E_2\app T_1}{}}

  \\

  \of{\tmty}{\type}
  & &
  \qquad \ofty\app(\apptm\app E_1\app E_2)\app T_2

  \\

  \of{\emptytm}{\tmty}
  & &
  \\

  \of{\apptm}{\arr{\tmty}{\arr{\tmty}{\tmty}}} & &
    \of{\oflamtm}
       {\typedpi{R}
       {\arr{\tmty}{\tmty}}
       {\typedpi{T_1}{\tpty}{\typedpi{T_2}{\tpty}{}}}}
  \\

  \of{\lamtm}{\arr{\tpty}{\arr{(\arr{\tmty}{\tmty})}{\tmty}}}
  &  & 
  \qquad \typedpi{D}{(\typedpi{x}{\tmty}
          {\typedpi{y}{\ofty\app x\app T_1}
          {\ofty\app(R\app x)\app T_2}})}{}
  \\
  
  & & \qquad \ofty\app(\lamtm\app T_1\app(\lflam{x}{R\app x}))
                             \app(\arrtm\app T_1\app T_2)
  \\

  \of{\ofty}{\arr{\tmty}{\arr{\tpty}{\type}}} & & \\  

  \of{\eqty}{\arr{\tpty}{\arr{\tpty}{\type}}} & &
    \of{\refltm}{\typedpi{T}{\tpty}{\eqty\app T\app T}}
\end{array}
\]
\caption{An LF Specification for the Simply-Typed Lambda Calculus}
\label{fig:stlc-term-spec}
\end{figure}

A key use of LF is in formalizing systems that are described 
through relations between objects that are specified through a
collection of inference rules. 
In the paradigmatic approach, each such relation is represented by a
dependent type whose term arguments are encodings of objects that
might be in the relation in question.
The inference rules translate in this context into term constructors
for the type representing the relation.
We illustrate these ideas through an encoding of the typing relation
for the simply-typed $\lambda$-calculus, a running example for this
thesis. 

We assume the reader to be familiar with the types and terms in the
simply typed $\lambda$-calculus  and also with the rules that define
its typing relation.
Figure~\ref{fig:stlc-term-spec} presents an LF signature that serves as
an encoding of this system. 
This encoding uses the higher-order abstract syntax approach to
treating binding. 
The specification introduces two type families, $\tpty$ and $\tmty$ to 
represent the simple types and $\lambda$-terms.
Additionally, for each expression form in the object system, it
includes a constant that produces a term of type $\tpty$ or $\tmty$;
as should be apparent from the declarations, we have assumed an object
language whose terms are constructed from a single constant of
atomic type that is represented by the LF constant \emptytm\ and whose
type is represented by the LF constant \unittm.
This signature also provides a representation of two relations over
object language expressions: typing between terms and types and
equality between types. 
Specifically, the type-level constants $\ofty$ and $\eqty$ are
included towards this end. 
The rules defining the relations of interest in the object system are
encoded by constants in the signature.
The types associated with these constants ensure that well-formed
terms of atomic type that are formed using the constants correspond to
derivations of the relation in the object language that is represented
by the type. 

One of the purposes for constructing a specification is to use them to
prove properties about the object system.
For example, we may want to show that when a type can be associated
with a term in the simply typed $\lambda$-calculus, it must be
unique.
Based on our encoding, this property can be stated as the following
about typing derivations in LF:
\begin{quotation}
\noindent For any terms $M_1,M_2,E,T_1,T_2$, if there are LF
derivations for 
$\lfchecktype{}{M_1}{\ofty\app E\app T_1}$ and
$\lfchecktype{}{M_2}{\ofty\app E\app T_2}$, then there must be a term
$M_3$ such than there is a derivation for
$\lfchecktype{}{M_3}{\eqty\app T_1\app T_2}$. 
\end{quotation}
To prove this property, we would obviously need to unpack its logical
structure.
We would also need to utilize an understanding of LF in
analyzing the hypothesized typing derivations. 
Considering the case where $E$ is an abstraction will lead us to actually
wanting to prove a more general property:
\begin{quotation}
\noindent For any terms $M_1,M_2,E,T_1,T_2$ and contexts $\Gamma$, if
there are LF derivations for the judgements
$\lfchecktype{\Gamma}{M_1}{\ofty\app E\app T_1}$ and
$\lfchecktype{\Gamma}{M_2}{\ofty\app E\app T_2}$, 
then there must be a term $M_3$ such than there is a derivation
for $\lfchecktype{}{M_3}{\eqty\app T_1\app T_2}$. 
\end{quotation}
Now, this property is not provable without some constraints on the form
of contexts.
In this example, it suffices to prove it when $\Gamma$ is restricted
to being of the form
\begin{center}
 $(x_1:\tmty,y_1:\ofty\app x_1\app{Ty}_1,\ldots,x_n:\tmty,y_n:\ofty\app x_n\app{Ty}_n)$.
\end{center}
\noindent In completing the argument, we would need to use properties of LF
derivability.
A property that would be essential in this case is the finiteness of
LF derivations, which enables us to use an inductive argument.

The objective in this thesis is to provide a formal mechanism for
carrying out such analysis.
We do this by describing a logic that is suitable for this purpose.
One of the requirements of this logic is that it should permit the
expression of the kinds of properties that arise in the process of
reasoning.
Beyond this, it should further be possible to complement the statement
of properties with inference rules that permit the encoding of
interesting and sound forms of reasoning. 

%% file: LF/lf-metathm.tex
\section{Meta-Theoretic Properties of LF}
\label{sec:lf-metathm}

Our reasoning system will need to embody an understanding of
derivability in LF.  
We describe some properties related to this notion here that will be 
useful in this context. 
The first three theorems, which express structural properties
about derivations, have easy proofs.
The fourth theorem states a subsitutivity property for wellformedness
judgements.
This theorem is proved in \cite{harper07jfp}.

\begin{theorem}\label{th:weakening}
If $\mathcal{D}$ is a derivation for $\lfkind{\Gamma}{K}$,
$\lftype{\Gamma}{A}$ or $\lfchecktype{\Gamma}{M}{A}$, then, for any
variable $x$ that is fresh to the judgement and for any $A'$ such that
$\lftype{\Gamma}{A'}$ is derivable, there is a derivation,
respectively, for $\lfkind{\Gamma, x:A'}{K}$,
$\lftype{\Gamma, x : A'}{A}$ or $\lfchecktype{\Gamma, x: A'}{M}{A}$
that has the same structure as $\mathcal{D}$. 
\end{theorem}

\begin{theorem}\label{th:strengthening}
  If $\mathcal{D}$ is a derivation for $\lfkind{\Gamma, x:A'}{K}$,
$\lftype{\Gamma, x : A'}{A}$ or $\lfchecktype{\Gamma, x: A'}{M}{A}$
  and $x$ is a variable that does not appear free in $K$, $A$, or $M$
  and $A$ respectively, then there must be a derivation that has the
  same structure as $\mathcal{D}$ for judgment $\lfkind{\Gamma}{K}$,
  $\lftype{\Gamma}{A}$ or $\lfchecktype{\Gamma}{M}{A}$, respectively.
\end{theorem}

\begin{theorem}\label{th:exchange}
If $x$ does not appear in $A_2$ then
$\Gamma_1,y:A_2,x:A_1,\Gamma_3$ is a well-formed context with respect
to a signature $\Sigma$ whenever $\Gamma_1,x:A_1,y:A_2,\Gamma_3$ is. 
Further, if there is a derivation $\mathcal{D}$ for
$\lfkind{\Gamma, x:A_1, y:A_2,\Gamma_2}{K}$,
$\lftype{\Gamma, x:A_1, y:A_2,\Gamma_2}{A}$ or $\lfchecktype{\Gamma,
  x:A_1, y:A_2,\Gamma_2}{M}{A}$, then there must be a derivation that
has the same structure as $\mathcal{D}$ for
$\lfkind{\Gamma, y:A_2, x:A_1,\Gamma_2}{K}$,
$\lftype{\Gamma, y:A_2, x:A_1,\Gamma_2}{A}$ or $\lfchecktype{\Gamma,
  y:A_2, x:A_1,\Gamma_2}{M}{A}$, respectively.
\end{theorem}

\begin{theorem}\label{th:transitivity}
Assume that $\lfctx{\Gamma_1,x_0:A_0,\Gamma_2}$ and
$\lfchecktype{\Gamma_1}{M_0}{A_0}$ have derivations, and let $\theta$ be
the substitution $\{ \langle x_0, M_0,\erase{A_0} \rangle \}$. 
Then there is a $\Gamma'_2$ such that
$\hsub{\theta}{\Gamma_2}{\Gamma'_2}$ and $\lfctx{\Gamma_1,\Gamma'_2}$
have derivations. 
Further,
\begin{enumerate}
  \item if $\lfkind{\Gamma_1,x_0:A_0,\Gamma_2}{K}$ has a derivation,
    then there is a $K'$ such that $\hsub{\theta}{K}{K'}$ and
    $\lfkind{\Gamma_1,\Gamma'_2}{K'}$ have derivations;

  \item if $\lftype{\Gamma_1,x_0:A_0,\Gamma_2}{A}$ has a derivation,
    then there is an $A'$ such that $\hsub{\theta}{A}{A'}$ and
    $\lftype{\Gamma_1,\Gamma'_2}{A'}$ have derivations; and

  \item if $\lfchecktype{\Gamma_1,x_0:A_0,\Gamma_2}{M}{A}$ has a
    derivation (for some well-formed type $A$), there is an $A'$ and
    an $M'$ such that $\hsub{\theta}{A}{A'}$, $\hsub{\theta}{M}{M'}$,
    and $\lfchecktype{\Gamma_1,\Gamma_2'}{M'}{A'}$ have derivations.
\end{enumerate}
\end{theorem}

The reasoning system will need to build in a means for analyzing
typing derivations of the form $\lfchecktype{\Gamma}{M}{A}$.
This analysis will be driven by the structure of the type $A$.
The decomposition when $A$ is of the form $\typedpi{x_1}{A_1}{A_2}$
has an obvious form.
The development below, culminating in Theorem~\ref{th:atomictype},
provides the basis for the analysis when $A$ is an atomic type. 

\begin{lemma}\label{lem:arityrespecting}
Let $\Gamma$ be a context such that $\lfctx{\Gamma}$ has a derivation
and let $\STLCGamma$ be the arity context induced by $\Sigma$ and
$\Gamma$. 
Suppose that $\typedpi{y_1}{A_1}{\ldots\typedpi{y_n}{A_n}{A}}$ is a
type associated with a (term) constant or variable by $\Sigma$ or
$\Gamma$, or that $\typedpi{y_1}{A_1}{\ldots\typedpi{y_n}{A_n}{K}}$ is
a kind associated with a (type) constant by $\Sigma$, where the $y_i$s
are distinct variables.  
Then, for $1 \leq i \leq n$, $A_i$ and
$\typedpi{y_{i}}{A_{i}}{\ldots\typedpi{y_n}{A_n}{A}}$ or,
respectively, $\typedpi{y_{i}}{A_{i}}{\ldots\typedpi{y_n}{A_n}{K}}$  
respect the arity context $\aritysum{\{y_1 : \erase{A_1}, \ldots, y_{i-1} :
  \erase{A_{i-1}}\}}{\STLCGamma}$.
Further, $A$ or, respectively, $K$ respects the arity context 
$\aritysum{\{y_1 : \erase{A_1}, \ldots, y_{n} : \erase{A_{n}}\}}{\STLCGamma}$.
\end{lemma}

\begin{proof}
  Since $\Sigma$ and $\Gamma$ are well-formed by assumption, depending
  on the case under consideration, either
  $\lftype{\Gamma}{\typedpi{y_1}{A_1}{\ldots\typedpi{y_n}{A_n}{A}}}$  
  or $\lfkind{\emptyctx}{\typedpi{y_1}{A_1}{\ldots\typedpi{y_n}{A_n}{K}}}$
  must have a derivation.
  The desired conclusions now follow from Theorem~\ref{th:arityapprox}
  and Definition~\ref{def:aritytyping}.
\end{proof} 

\begin{lemma}\label{lem:lfcomp}
  Let $\Gamma_1$ be a context such that $\lfctx{\Gamma_1}$ has a
  derivation, let $\STLCGamma$ be the arity context induced by $\Sigma$
  and $\Gamma_1$, and let $\theta$ be a substitution that is arity type
  preserving with respect to $\STLCGamma$. Further, let $x_0$ be a
  variable that is neither bound in $\Gamma_1$ nor a member of 
  $\domain{\theta}$,  let $A_0$ and $M_0$ be such that
  $\lftype{\Gamma_1}{A_0}$ and $\lfchecktype{\Gamma_1}{M_0}{A_0}$ are
  derivable and let $\theta' = \theta \cup \{\langle
  x_0,M_0,\erase{A_0}\rangle \}$.
  \begin{enumerate}
    \item $\theta'$ is arity type preserving with respect to
      $\STLCGamma$.

    \item Let $\Gamma_2$ be a context that respects an arity context
      $\STLCGamma'$ such that 
      $\aritysum{\context{\theta'}}{\STLCGamma}\subseteq \STLCGamma'$ and let
      $\Gamma'_2$ be a context such that
      $\hsub{\theta}{\Gamma_2}{\Gamma'_2}$, and
      $\lfctx{\Gamma_1, x_0 : A_0, \Gamma'_2}$ have derivations. Then
      there is a context $\Gamma_2''$ such that the following hold:

      \begin{enumerate}
      \item $\hsub{\{\langle x_0,M_0,\erase{A_0}\rangle \}}
                  {\Gamma'_2}
                  {\Gamma''_2}$,
            $\hsub{\theta'}{\Gamma_2}{\Gamma''_2}$ and
            $\lfctx{\Gamma, \Gamma''_2}$ have derivations;

      \item if $K$ is a kind that also respects $\STLCGamma'$ and $K'$
        is a kind such that there are derivations for $\hsub{\theta}{K}{K'}$ and
        $\lfkind{\Gamma_1,x_0:A_0,\Gamma'_2}{K'}$, 
        then there is a kind $K''$ such that
        $\hsub{\{\langle x_0,M_0,\erase{A_0}\rangle \}}{K'}{K''}$,
        $\hsub{\theta'}{K}{K''}$ and
        $\lfkind{\Gamma_1, \Gamma''_2}{K''}$ are derivable; and

      \item if $A$ is a type that also respects $\STLCGamma'$ and $A'$
        is a type such that there are derivations for
        $\hsub{\theta}{A}{A'}$ and
        $\lftype{\Gamma_1,x_0:A_0,\Gamma'_2}{A'}$,
        then there is a type $A''$ such that
        $\hsub{\{\langle x_0,M_0,\erase{A_0}\rangle \}}{A'}{A''}$,
        $\hsub{\theta'}{A}{A''}$ and $\lftype{\Gamma_1,
        \Gamma''_2}{A''}$ have derivations.
      \end{enumerate}
   \end{enumerate}
\end{lemma}    

\begin{proof}
Since $\lfchecktype{\Gamma_1}{M_0}{A_0}$ has a derivation, it follows
from Theorem~\ref{th:arityapprox} that the substitution
$\{\langle x_0,M_0,\erase{A_0}\rangle\}$ is type preserving with respect to 
$\STLCGamma$.
It then follows from the assumptions in the lemma that $\theta'$ is in fact
$\{\langle x_0, M_0, \erase{A_0} \rangle \} \circ \theta$ and type
preserving with respect to $\STLCGamma$. The various observations in clause 2 
now follow from Theorems~\ref{th:composition} and \ref{th:transitivity}.
\end{proof}

\begin{theorem}\label{th:atomictype}
Let $\Gamma$ be a context such that $\lfctx{\Gamma}$ has a derivation.
\begin{enumerate}
\item $\lfsynthtype{\Gamma}{R}{A'}$  has a derivation 
  if
  \begin{enumerate}
  \item $R$ is of the form $(c \app M_1 \app \ldots\app M_n)$ for some
    $c:\typedpi{y_1}{A_1}{\ldots \typedpi{y_n}{A_n}{A}} \in \Sigma$ or
    of the form $(x \app M_1 \app \ldots\app M_n)$ for some
    $x:\typedpi{y_1}{A_1}{\ldots \typedpi{y_n}{A_n}{A}} \in \Gamma$,
    
  \item there is a sequence of types $A'_1,\ldots,A'_n$ such that, for
    $1\leq i\leq n$, there are derivations for both
    $\hsub{\{\langle y_1, M_1,\erase{A_1}\rangle, \ldots,
             \langle y_{i-1}, M_{i-1}, \erase{A_{i-1}} \rangle\}}
          {A_i}
          {A'_i}$
    and $\lfchecktype{\Gamma}{M_i}{A'_i}$, and
         
  \item $\hsub{\{\langle y_1, M_1,\erase{A_1}\rangle, \ldots,
                 \langle y_n, M_n, \erase{A_n} \rangle\}}
              {A}
              {A'}$
    and $\lftype{\Gamma}{A'}$ have derivations.
  \end{enumerate}
  
\item $\lfsynthtype{\Gamma}{R}{A'}$  has a derivation of height $h$
  only if
  \begin{enumerate}
  \item $R$ is of the form $(c \app M_1 \app \ldots\app M_n)$
    for some $c:\typedpi{y_1}{A_1}{\ldots \typedpi{y_n}{A_n}{A}} \in \Sigma$
    or of the form $(x \app M_1 \app \ldots\app M_n)$ for some
    $x:\typedpi{y_1}{A_1}{\ldots \typedpi{y_n}{A_n}{A}} \in \Gamma$,

  \item there is a sequence of types $A'_1,\ldots,A'_n$ such that, for
    $1 \leq i \leq n$, there is a derivation for 
    $\hsub{\{\langle y_1, M_1,\erase{A_1}\rangle, \ldots,
             \langle y_{i-1}, M_{i-1}, \erase{A_{i-1}} \rangle\}}
          {A_i}
          {A'_i}$
    and a derivation of height less than $h$ for
    $\lfchecktype{\Gamma}{M_i}{A'_i}$, and  

  \item $\hsub{\{\langle y_1, M_1,\erase{A_1}\rangle, \ldots,
                 \langle y_n, M_n, \erase{A_n} \rangle\}}
              {A}
              {A'}$
    and $\lftype{\Gamma}{A'}$ have derivations.
  \end{enumerate}      
\end{enumerate}
\end{theorem}

\begin{proof} At the outset, we should check the coherence of clauses
  1(b) and 2(b) 
  in the theorem statement by verifying that, for $1 \leq i \leq n$, it
  is the case that $\lftype{\Gamma}{A'_i}$ has a derivation.
  Towards this end, we first note that there must be a derivation for
  the type formation judgment
  $\lftype{\Gamma, y_1 : A_1, \ldots, y_{i-1} : A_{i-1}}{A_i}$ since
  $\Sigma$ and $\Gamma$ are well-formed. 
  The desired conclusion then follows from using
  Lemma~\ref{lem:lfcomp} repeatedly and observing, via
  Theorem~\ref{th:erasure}, that erasure is preserved under
  substitution.

\smallskip
\noindent We now introduce some notation that will be useful in the
arguments that follow.
We will use $\STLCGamma$ to denote the arity context induced by
$\Sigma$ and $\Gamma$.
Further, for $1 \leq i \leq n+1$, we will write $\theta_i$ for the  
substitution $\{ \langle y_1,M_1,\erase{A_1} \rangle, \ldots, \langle
y_{i-1}, M_{i-1}, \erase{A_{i-1}}\rangle \}$.
An observation that we will make use of below is that if for $1 \leq j
< i$ it is the case that $\lfchecktype{\Gamma}{M_i}{A'_i}$ has a
derivation, then $\theta_i$ is type preserving with respect to
$\STLCGamma$.
This is an easy consequence of Theorems~\ref{th:arityapprox} and
\ref{th:erasure}.

\smallskip
\noindent {\it Proof of (1).}
  We will consider explicitly only the case where $R$ is $(c\app M_1 \app
  \ldots \app M_n)$; the argument for the case when $R$ is $(x\app M_1
  \app \app \ldots \app M_n)$ is similar.
  For $1 \leq i \leq n+1$ we will show that, under the conditions
  assumed for $M_1,\ldots, M_{i-1}$, there is a type $A''_{i}$ such
  that $\hsub{\theta_{i}}
             {(\typedpi{y_{i}}{A_{i}}{\ldots\typedpi{x_n}{A_n}{A}})}
             {A''_{i}}$,
  $\lftype{\Gamma}{A''_{i}}$ and
  $\lfsynthtype{\Gamma}{(c\app M_1 \app \ldots \app M_{i-1})}{A''_{i}}$ have
  derivations. 
  The desired conclusion follows from noting that $A'$ must be
  $A''_{n+1}$ because the result of substitution application is
  unique.  

  The claim is proved by induction on $i$.
  Consider first the case when $i$ is $1$.
  Since $\theta_1 = \emptyset$, $A''_{1}$ is
  $\typedpi{y_1}{A_1}{\ldots \typedpi{y_n}{A_n}{A}}$. 
  The wellformedness of $\Sigma$ ensures that $\lftype{\Gamma}{A''_1}$
  has a derivation and we get a derivation for
  $\lfsynthtype{\Gamma}{c}{A'}$ by using an \atomtermconst\ rule.

  Let us then assume the claim for $i$ and show that it
  must also hold for $i+1$.
  By the hypothesis, there is an $A''_i$ of the form 
  $\typedpi{y_i}{A'_i}{A''}$ where $A''$ is a type such that
  $\hsub{\theta_i}{(\typedpi{y_{i+1}}{A_{i+1}}{\ldots\typedpi{x_n}{A_n}{A}})}{A''}$
  has a derivation.
  Since $\lftype{\Gamma}{A''_i}$ has a derivation, so must
  $\lftype{\Gamma, y_i : A'_i}{A''}$.
  By Lemma~\ref{lem:arityrespecting},
  $\typedpi{y_{i+1}}{A_{i+1}}{\ldots\typedpi{x_n}{A_n}{A}}$
  respects the arity context
  $\aritysum{\{y_1 : \erase{A_1},\ldots, y_i : \erase{A_i}\}}{\STLCGamma}$.
  Since there are derivations for $\lfchecktype{\Gamma}{M_j}{A'_j}$ for $1 \leq
  j < i$, $\theta_i$ is type preserving over $\STLCGamma$.
  We now invoke Lemma~\ref{lem:lfcomp} to conclude that there is a
  term $A'''$ such that there are derivations for 
  $\hsub{\{ \langle y_i, M_i,\erase{A'_i} \rangle \}}{A''}{A'''}$,
  $\hsub{\theta_{i+1}}{\typedpi{y_{i+1}}{A_{i+1}}{\ldots\typedpi{x_n}{A_n}{A}}}{A'''}$,
  and $\lftype{\Gamma}{A'''}$.
  By the hypothesis, there is a derivation for
  $\lfsynthtype{\Gamma}{c\app M_1 \app \ldots \app
    M_{i-1}}{\typedpi{y_i}{A'_i}{A''}}$. 
  Using an \atomtermapp\ rule together with this derivation and the
  ones for  
  $\lfchecktype{\Gamma}{M_i}{A'_i}$, and 
  $\hsub{\{ \langle y_i, M_i,\erase{A'_i} \rangle \}}{A''}{A'''}$, we
  get a derivation for $\lfsynthtype{\Gamma}{(c\app M_1 \app \ldots \app
    M_i)}{A'''}$.
  Letting $A''_{i+1}$ be $A'''$ we see that all the requirements are satisfied.

\smallskip
\noindent {\it Proof of (2).} We prove the claim by induction on
the height of the derivation of the type synthesis judgment 
$\lfsynthtype{\Gamma}{R}{A'}$. We 
consider the cases for the last rule used in the derivation.
If this rule is \atomtermvar\ or \atomtermconst, the argument is
straightforward. 
The only case to be considered further, then, is that when the rule is
\atomtermapp.

In this case, we know that $R$ must be of the form $(R'\app M')$
where there is a shorter derivation for $\lfsynthtype{\Gamma}{R'}{B'}$
for some type $B'$.
From the induction hypothesis, it follows that $R'$ has the form
$(c\app M_1\app \ldots\app M_n)$ or $(x\app M_1\app \ldots\app M_n)$ 
for some $c : \typedpi{y_1}{A_1}{\ldots \typedpi{y_n}{A_n}{B}} \in
  \Sigma$ or $x : \typedpi{y_1}{A_1}{\ldots \typedpi{y_n}{A_n}{B}} \in
  \Gamma$ and that there must be a sequence of types
  $A'_1,\ldots,A'_n$ that, together with the terms $M_1,\ldots,M_n$
  satisfy the requirements stated in clause 2(b).
Moreover, $B'$ must be such that $\hsub{\theta_{n+1}}{B}{B'}$ and
$\lftype{\Gamma}{B'}$ have derivations.
Since the rule is an \atomtermapp, $B'$ must have the structure of an
abstracted type.
From this it follows that $B$ must be of the form
$\typedpi{y_{n+1}}{A_{n+1}}{A}$ and, correspondingly, $B'$ must be of
the form $\typedpi{y_{n+1}}{A'_{n+1}}{A''}$ where
$\hsub{\theta_{n+1}}{A_{n+1}}{A'_{n+1}}$ and $\hsub{\theta_{n+1}}{A}{A''}$ have
derivations.
Noting that the type of $c$ or $x$ is really of the form
$\typedpi{y_1}{A_1}{\ldots \typedpi{y_{n+1}}{A_{n+1}}{A}}$ it follows
  from Lemma~\ref{lem:arityrespecting} that $A$
respects the arity context $\aritysum{\{y_1 : \erase{A_1}, \ldots,
  y_{n+1} : \erase{A_{n+1}}\}}{\STLCGamma}$.
Also, since $\lftype{\Gamma}{B'}$ has a derivation, it must be the
case that $\lftype{\Gamma, y_{n+1} : A'_{n+1}}{A''}$ has one.
Since the derivation concludes with a \atomtermapp\ rule, it must be
the case that $\lfchecktype{\Gamma}{M'}{A'_{n+1}}$ and $\hsub{\{
  \langle y_{n+1}, M', \erase{A'_{n+1}} \rangle \}}{A''}{A'}$ have
shorter derivations than the one for $\lfsynthtype{\Gamma}{R}{A'}$.
Since $\theta_{n+1}$ is type preserving with respect to $\STLCGamma$, we
may now use Lemma~\ref{lem:lfcomp} and Theorem~\ref{th:erasure} to
conclude that $\hsub{\theta_{n+1} \cup \{\langle y_{n+1}, M_2, \erase{A_{n+1}} 
  \rangle \}}{A}{A'}$ and $\lftype{\Gamma}{A'}$ have derivations. 
Renaming $M'$ to $M_{n+1}$ we see that all the requirements of clause
2 are satisfied.
\end{proof}

Theorem~\ref{th:atomictype} gives us an alternative means for deriving
judgements of the analysis form $\lfchecktype{\Gamma}{R}{P}$, in the process
dispensing with judgements of the synthesis form $\lfsynthtype{\Gamma}{R}{A}$.
Note also that in \emph{analyzing} judgements of the form
$\lfchecktype{\Gamma}{R}{P}$, it is necessary to consider only
\emph{shorter} derivations for subterms of $R$.
This observation will be used in developing a means for arguing
inductively on the heights of LF derivations. 

A property similar to that in Theorem~\ref{th:atomictype} can be
observed for wellformedness judgements for atomic types.
Theorem~\ref{th:atomickind} presents a version that suffices for this
thesis. 
A proof of this theorem can be constructed based essentially on the
one for Theorem~\ref{th:atomictype}.
  
\begin{theorem}\label{th:atomickind}
Let $\Gamma$ be a context such that $\lfctx{\Gamma}$ is derivable.
Then $\lfsynthkind{\Gamma}{P}{K'}$ has a derivation if and only if
there is an
$a : \typedpi{y_1}{A_1}{\ldots \typedpi{y_n}{A_n}{K}} \in \Sigma$
such that 
\begin{enumerate}
\item $P$ is of the form $(a \app M_1 \app \ldots\app M_n)$;
  
\item there is a sequence of types $A'_1,\ldots,A'_n$ such that, for $1
  \leq i \leq n$, there are derivations for
  $\hsub{\{\langle y_1,  M_1,\erase{A_1}\rangle, \ldots,
           \langle y_{i-1}, M_{i-1},  \erase{A_{i-1}} \rangle\}}
        {A_i}
        {A'_i}$ and
  $\lfchecktype{\Gamma}{M_i}{A'_i}$; and 

\item $\hsub{\{\langle y_1, M_1,\erase{A_1}\rangle, \ldots, \langle
  y_n, M_n, \erase{A_n} \rangle\}}{K}{K'}$ and $\lfkind{\Gamma}{K'}$
  have derivations. 
\end{enumerate}
\end{theorem}

%% file: logic/logic.tex
\chapter{A Logic for Expressing Properties of LF Specifications}
\label{ch:logic}

Our objective in this chapter is to describe a logic in which we can express
properties of an object system that has been specified in LF.
The discussions in Section~\ref{sec:lf-ex} suggest a 
possible structure for such a logic.
The logic would be parameterized by an LF signature that has been
determined to be well-formed at the outset.
The basic building blocks for the properties that are to be described
would be typing judgements.
More specifically, the logic would use such judgements as its atomic
formulas and would interpret them using LF derivability.
More complex formulas would then be constructed using 
logical connectives and quantifiers over LF terms.
As the example in Section~\ref{sec:lf-ex} illustrates,
it would be necessary to also permit a quantification over LF contexts.

To develop an actual logic based on these ideas, we need to describe a
more precise correspondence between LF typing judgements and atomic
formulas. 
The judgement forms that need to be considered in this context are
those for typing canonical and atomic terms, \ie, the
$\lfchecktype{\Gamma}{M}{A}$ and $\lfsynthtype{\Gamma}{R}{A}$ forms. 
The main judgement form is in fact the first one: the second form
serves mainly to explicate judgements of the first kind when the type
is atomic and, as we have noted already, Theorem~\ref{th:atomictype}
provides the basis for circumventing such an explicit treatment
through a special ``focused'' typing rule. 
In light of this, it suffices to describe an encoding of only the
first judgement form.
The judgement in the LF setting assumes the wellformedness of the
context $\Gamma$ and the type $A$.
In the logic, the context and, therefore, also the type can be
dynamically determined by instantiations for context variables.
To deal with this situation, we will build the wellformedness of
$\Gamma$ and $A$ into the interpretation of the encoding of the
judgement. 
There is, however, an aspect of the wellformedness checking that we
would like to extract into a static pre-processing phase.
The LF typing rules combine the checking of canonicity of terms with
the determination of inhabitation that relies on the semantically more
meaningful aspect of dependencies in types.
To allow the focus in the logic to be on the latter aspect, we will
build the former into a wellformedness criterion for formulas using
arity types.

Another aspect that needs further consideration is the treatment of
contexts in atomic formulas.
To support typing derivations that use the
\canontermlam\ rule, such contexts must allow for the explicit
association of types with variables.
These variables may appear free in the terms and types in the atomic
formula.
However, their interpretation in this context must be different from
the variables that are bound by quantifiers: in particular, these
variables cannot be instantiated and each of them must be treated as
being distinct within the atomic formula. 
The necessary treatment of these variables can be realized by
representing them by \emph{nominal constants} in the style of
\cite{gacek11ic,tiu06lfmtp}.
Contexts must, in addition, allow for an unspecified part whose exact
extent is to be determined by instantiation of an external context
quantifier.
To support this ability, we will allow context variables to appear in
contexts.
However, as observed in Section~\ref{sec:lf-ex}, we
would like to be able to restrict the instantiation of such variables
to blocks of declarations adhering to specified forms.
To impose such constraints, the logic will permit context variables to
be typed by \emph{context schemas} that are motivated by regular world
descriptions used in the Twelf
system~\cite{Pfenning02guide,schurmann00phd}. 

In the rest of this chapter, we describe the logic in detail, thereby
substantiating the ideas outlined above.
The first two sections present the well-formed formulas and
identify their intended meaning. 
The end result of this discussion is a means for describing properties
of a specification comprised of an LF signature and for assessing the
validity of such properties.
The third section illuminates this capability through a
collection of examples.
The last section observes a property that will
be useful in later chapters, namely, the irrelevance of the
particular names that are chosen for the variables bound by the
context in an LF judgement.
Noting that these variables are represented by nominal constants in
the logic, the statement of this property takes the form of invariance
of validity of atomic formulas under permutations of nominal
constants. 

\input{logic/formulas}
\input{logic/semantics}

\input{logic/logic-ex}
\input{logic/permutations}

%% file: logic/formulas.tex
\section{The Formulas of the Logic}
\label{sec:formulas}

\begin{figure}[tbhp]
\[
\begin{array}{r r c l}
  \mbox{\bf Terms} & M,N & ::= & R\ |\ \lflam{x}{M}\\
  \mbox{\bf Atomic Terms} & R & ::= & c\ |\ x\ |\ n\ |\ R\app M\\[5pt]
  \mbox{\bf Types} & A & ::= &
           P\ |\ \typedpi{x}{A_1}{A_2}\\
  \mbox{\bf Atomic Types} & P & ::= & a\ |\ P\app M\\[5pt]
\end{array}
\]
\caption{Terms and Types in the Logic}
\label{fig:logic-terms-and-types}
\end{figure}

We begin by considering the representation of LF terms and types in
the logic. 
Figure~\ref{fig:logic-terms-and-types} presents the syntax of the
corresponding expressions.
As with LF syntax, we use $c$ and $d$ to represent term level
constants, $a$ and $b$ to represent type level constants and $x$ and
$y$ to represent term-level variables.
We also use $n$ to represent a special category of symbols called the
nominal constants.
LF terms and types are obviously a subset of the expressions presented
here.
Going the other way, there are two main additions to the LF
counterparts in the collection of expressions described here.
First, nominal constants may be used in constructing terms.
Second, as we shall soon see, variables may be bound not only by
term and type level abstractions but also by formula level
quantifiers.

We assume as given a set $\noms$ of nominal constants, each specified
with an arity type.
Elements of $\noms$ are written in the form $n : \alpha$.
We assume that there is a countably infinite supply of nominal
constants in $\noms$ for each arity type $\alpha$.
The logic is parameterized by an LF style signature $\Sigma$ that assigns
kinds to type-level constants and types to term-level ones.
This signature is assumed to be well-formed in the sense described in
Section~\ref{sec:canon-lf}.

\begin{figure}[tbhp]

\begin{center}
\begin{tabular}{c}

\infer{\akindingp{\STLCGamma}{a}{K}}
      {a:K \in \Sigma}

\qquad\qquad

\infer{\akindingp{\STLCGamma}{P\app M}{K}}
      {\akindingp{\STLCGamma}{P}{\typedpi{x}{A}{K}} \qquad
       \stlctyjudg{\STLCGamma}{M}{\erase{A}}}

\\[10pt]

\infer{\wftype{\STLCGamma}{P}}
      {\akindingp{\STLCGamma}{P}{\type}}

\qquad\qquad

\infer{\wftype{\STLCGamma}{\typedpi{x}{A_1}{A_2}}} 
      {\wftype{\STLCGamma}{A_1} \qquad \wftype{\aritysum{\{x :
            \erase{A_1} \}}{\STLCGamma}}{A_2}}
\end{tabular}
\end{center}

\caption{Arity Kinding for Canonical Types}
\label{fig:arity-kinding}
\end{figure}

As explained earlier, expressions in the logic will be expected to
satisfy typing constraints that check for canonicity.
At the term level, these constraints will be realized through arity
typing relative to a suitable arity context.
At the type level, we must additionally ensure that (type) constants
have been supplied with an adequate number of arguments.
We make these notions precise below; we assume the obvious extension
of erasure to types in the logic here and elsewhere.
\begin{definition}
The typing relation between an arity context, a term and an arity type
that is described in Definition~\ref{def:aritytyping} is extended to
the present context by permitting terms to contain nominal constants
and by allowing arity contexts to contain assignments to such
constants.
The rules in Figure~\ref{fig:arity-kinding} define an arity kinding
property denoted by $\wftype{\STLCGamma}{A}$ for a type $A$ relative to
an arity context $\STLCGamma$.
In these rules, $\Sigma$ is the signature parameterizing the logic.
We will often need to refer to the arity context induced by $\Sigma$.
We call this the \emph{initial constant context} and we reserve the symbol
$\STLCGamma_0$ to denote it.
\end{definition}

Hereditary substitution extends naturally to the terms and types in the
logic by treating nominal constants like other constants.
The following theorem relating to such substitutions has an obvious
proof.  

\begin{theorem}\label{th:aritysubs-ty}
If $\theta$ is type preserving with respect to $\STLCGamma$ and
$\wftype{\aritysum{\context{\theta}}{\STLCGamma}}{A}$ and
$\hsub{\theta}{A}{A'}$ have derivations, then
$\wftype{\STLCGamma}{A'}$ has a derivation.
\end{theorem}

\begin{figure}[tbhp]
  \[\begin{array}{rrcl}
\mbox{\bf Block Declarations} & \Delta & ::= & \emptybb\ \vert\ \Delta, y : A \\
\mbox{\bf Block Schema}   & \mathcal{B} & ::= & \{x_1:\alpha_1,\ldots, x_n:\alpha_n\}\Delta\\
\mbox{\bf Context Schema} & \mathcal{C} & ::= & \emptycs\ \vert\ \mathcal{C}, \mathcal{B}
\end{array}\]
\caption{Block Schemas and Context Schemas}
\label{fig:context-schemas}
\end{figure}

The logic allows for quantifiers over contexts.
In the intended interpretation, such quantifiers are meant to be
instantiated with actual contexts that will correspond to assignments
of LF types to nominal constants.
However, it will be necessary to be able to constrain the possible
instantiations in real applications.
This ability is supported by typing
context quantifiers using \emph{context schemas} whose structure is
presented in Figure~\ref{fig:context-schemas}.
In essence, a context schema comprises a collection of \emph{block
  schemas}.
A block schema consists of a header of variables annotated with
arity types and a body of declarations associating types with
variables. 
Each variable in the header and that is assigned a type in the body of
a block schema is required to be distinct. 
A block is intended to serve as a template for generating a sequence
of bindings for nominal constants through an instantiation process
that will be made clear in the next section.
An actual context corresponding to a context schema is to be obtained
by some number of instantiations of its block schemas.
Block and context schemas are required to satisfy typing constraints
towards ensuring that the contexts generated from them will be
well-formed in the manner required by the logic.
These constraints are represented by the typing judgements
$\abstyping{\mathcal{B}}$ and $\acstyping{\mathcal{C}}$, respectively,
that are defined by the rules in Figure~\ref{fig:schematyping}.

\begin{figure}[tbhp]

\begin{center}
\begin{tabular}{c}

\infer{\wfdecls{\STLCGamma}{\emptybb}{\STLCGamma}}{}

\qquad\qquad

\infer{\wfdecls{\STLCGamma}{\Delta, y:A}{\STLCGamma' \cup \{y:\erase{A}\}}} 
      {\wfdecls{\STLCGamma}{\Delta}{\STLCGamma'} \qquad 
       y\ \mbox{\rm is not assigned by}\ \STLCGamma' \qquad
       \wftype{\STLCGamma'}{A}}

\\[10pt]

\infer{\abstyping{\{x_1:\alpha_1,\ldots, x_n:\alpha_n\}\Delta}}
      {x_1,\ldots,x_n\ \mbox{\rm are distinct variables}
       \qquad
       \wfdecls{\STLCGamma_0 \cup \{x_1 : \alpha_1, \ldots,
                                    x_n : \alpha_n\}}
               {\Delta}
               {\STLCGamma'}}

\\[10pt]

\infer{\acstyping{\emptycs}}{}

\qquad

\infer{\acstyping{\mathcal{C},\mathcal{B}}}
      {\acstyping{\mathcal{C}} \qquad \abstyping{\mathcal{B}}}

\end{tabular}
\end{center}
\caption{Wellformedness Judgements for Block and Context Schemas}
\label{fig:schematyping}
\end{figure}

\begin{figure}[tbhp]
\[\begin{array}{lrcl}
\mbox{\bf Context Expressions} & G & ::= &
    \emptyce\ |\ \Gamma\ |\ G,n:A\\
\mbox{\bf Formulas} & F & ::= & 
    \fatm{G}{\of{M}{A}}\ |\ \ftrue\ |\ \ffalse\ |\ \fimp{F_1}{F_2}\ |\ \fand{F_1}{F_2}\ |\\
& & &\for{F_1}{F_2}\ |\ \fctx{\Gamma}{\mathcal{C}}{F}\ |\ \fall{x:\alpha}{F}\ |\ \fexists{x:\alpha}{F}
\end{array}\]
\caption{The Formulas of the Logic}
\label{fig:formula-syntax}
\end{figure}

We are finally in a position to describe the formulas in the logic.
The syntax of these formulas is presented in
Figure~\ref{fig:formula-syntax}.
The symbol $\Gamma$ is used in these formulas to represent context
variables.
Atomic formulas, which represent LF typing judgements, have the form
$\fatm{G}{M:A}$.
The context in these formulas is constituted by a sequence of type
associations with nominal constants, possibly preceded by a context
variable.
Included in the collection are the logical constants $\ftrue$ and
$\ffalse$ and the familiar connectives for constructing more complex
formulas.
Universal and existential quantification over term variables is also
permitted and these are written as $\fall{x:\alpha}{F}$ and
$\fexists{x:\alpha}{F}$, respectively.
Such quantification is indexed, as might be expected, by arity types.
The collection also includes universal quantification over context
variables that is typed by context schemas, written as
$\fctx{\Gamma}{\mathcal{C}}{F}$.
We assume the usual principle of equivalence under renaming with
respect to the term and context quantifiers and apply them as needed. 

\begin{figure}[tbhp]

\begin{center}
\begin{tabular}{c}

\infer{\wfctx{\STLCGamma}{\Xi}{\emptyce}}
      {} 
\qquad

\infer{\wfctx{\STLCGamma}{\Xi}{\Gamma}}
      {\Gamma \in\Xi}

\\[10pt]

\infer{\wfctx{\STLCGamma}{\Xi}{G,n:A}}
      {\wfctx{\STLCGamma}{\Xi}{G} \qquad
       n:\erase{A}\in \STLCGamma \qquad
       \wftype{\STLCGamma}{A}}
\\[10pt]

\infer{\wfform{\STLCGamma}{\Xi}{\fatm{G}{M:A}}}
      {\wfctx{\STLCGamma}{\Xi}{G} \qquad
       \wftype{\STLCGamma}{A} \qquad
       \stlctyjudg{\STLCGamma}{M}{\erase{A}}}

\\[10pt]

\infer{\wfform{\STLCGamma}{\Xi}{\ftrue}}{} 

\qquad
      
\infer{\wfform{\STLCGamma}{\Xi}{\ffalse}}{}

\qquad
      
\infer[\bullet \in \{\supset,\land,\lor\}]
      {\wfform{\STLCGamma}{\Xi}{F_1 \bullet F_2}}
      {\wfform{\STLCGamma}{\Xi}{F_1} \qquad 
       \wfform{\STLCGamma}{\Xi}{F_2}}

\\[10pt]

\infer{\wfform{\STLCGamma}{\Xi}{\fctx{\Gamma}{\mathcal{C}}{F}}}
      {\acstyping{\mathcal{C}} \qquad
       \wfform{\STLCGamma}{\Xi \cup \{ \Gamma \}}{F}}

\qquad\qquad
      
\infer[\genericq \in \{\forall, \exists \}]
      {\wfform{\STLCGamma}{\Xi}{\fgeneric{x:\alpha}{F}}}
      {\wfform{\aritysum{\{x:\alpha\}}{\STLCGamma}}{\Xi}{F}}
 
\end{tabular}
\end{center}
 
\caption{The Wellformedness Judgement for Formulas}
\label{fig:wfform}
\end{figure}

A formula $F$ is determined to be well-formed or not relative to an arity
context $\STLCGamma$ and a collection of context variables $\Xi$.
This judgement is written concretely as $\wfform{\STLCGamma}{\Xi}{F}$
and the rules defining it are presented in Figure~\ref{fig:wfform}.
At the top-level, formulas are expected to be closed, \ie, to not have
any free term or context variables.
More specifically, we expect $\wfform{\noms \cup \STLCGamma_0}{\emptyset}{F}$ to be
derivable for such formulas. 
The analysis within the scope of term and context
quantifiers augments these sets in the expected way.
For context quantifiers, this analysis must also check that the
annotating context schema is well-formed.
An atomic formula $\fatm{G}{M:A}$ is deemed well-formed if its
components $G$, $M$ and $A$ are well-formed and if $M$ can be assigned
the erased form of $A$ as its arity type.
The context expression $G$ is well-formed if any context variable used
in it is bound in the overall formula and if the types assigned to
nominal constants in the explicit part of $G$ are well-formed and
such that their erased forms match the arity types of the nominal
constants they are assigned to.
Note that these types may use nominal constants arbitrarily; assessing
whether they are used in a manner that respects dependencies is a part
of the meaning of the atomic formula.

An obvious result about well-formed formulas is that the derivability of
the judgement is not changed by the addition of unused bindings in either
parametrizing context.
This is the content of the following theorem, which is used in later proofs 
to ensure that we can match the context under which well-formedness has been
determined with a possibly larger arity (resp. context variable) context.

\begin{theorem}\label{th:wf-form-wk}
For any $F$, $\Theta\subseteq\Theta'$, and $\Xi\subseteq\Xi'$, if
$\wfform{\Theta}{\Xi}{F}$ is derivable then $\wfform{\Theta'}{\Xi'}{F}$
is derivable.
\end{theorem}
\begin{proof}
This proof is completed by a straightforward induction on the well-formedness
derivation, and the resulting derivation has the same structure as the given
derivation.
\end{proof}

%% file: logic/semantics.tex
\section{The Interpretation of Formulas}
\label{sec:semantics}

A key component to understanding the meanings of formulas is
understanding the interpretation of the quantifiers over term and
context variables.
These quantifiers are intended to range over closed expressions of the
relevant categories.
For a quantifier over a term variable, this translates concretely into
closed terms of the relevant arity type.
For a quantifier over a context variable, we must first explain when
an LF context in which variables are represented by nominal constants
satisfies a context schema. 

\begin{figure}[tbhp]

\begin{center}
\begin{tabular}{c}

\infer{\declinst{\mathbb{N}}{\emptybb}{\emptyce}{\emptyset}}{}

\qquad

\infer{\declinst{\mathbb{N}}{\Delta,y:A}{G, n : A'}
                 {\theta \cup \{\langle y,n,\erase{A}\rangle\}}}
      {\declinst{\mathbb{N}}{\Delta}{G}{\theta} \qquad
        n : \erase{A} \in {\mathbb{N}} \qquad
        \hsub{\theta}{A}{A'}}

\\[10pt]

\infer{\bsinst{\mathbb{N}}{\Psi}{\{x_1 : \alpha_1,\ldots, x_n : \alpha_n\}\Delta}{G}}
      {\begin{array}{c}
          \declinst{\mathbb{N}}{\Delta}{G'}{\theta}
          \\
          \hsub{\{\langle x_i,t_i,\alpha_i\rangle \ \vert\ 1 \leq i \leq n \}} 
                 {G'}
                 {G}
        \end{array}
       \qquad
       \{ \stlctyjudg{{\mathbb{N}} \cup \Psi \cup \STLCGamma_0}{t_i}{\alpha_i}\ \vert\ 1 \leq i \leq n \}
       }

\\[10pt]

\infer{\csinstone{\mathbb{N}}{\Psi}{\mathcal{C},\mathcal{B}}{G}}
      {\bsinst{\mathbb{N}}{\Psi}{\mathcal{B}}{G}}

\qquad

\infer{\csinstone{\mathbb{N}}{\Psi}{\mathcal{C},\mathcal{B}}{G}}
      {\csinstone{\mathbb{N}}{\Psi}{\mathcal{C}}{G}}

\\[10pt]

\infer{\csinst{\mathbb{N}}{\Psi}{\mathcal{C}}{\emptyce}}
      {}

\qquad 

\infer{\csinst{\mathbb{N}}{\Psi}{\mathcal{C}}{G, G'}}
      {\csinst{\mathbb{N}}{\Psi}{\mathcal{C}}{G} \qquad
       \csinstone{\mathbb{N}}{\Psi}{\mathcal{C}}{G'}}

\end{tabular}
\end{center}

\caption{Instantiating a Context Schema}
\label{fig:ctx-schema}
\end{figure}

We do this by describing the relation of ``being an instance of''
between a closed context expression $G$ and a context schema $\mathcal{C}$.
This relation is indexed by a nominal constant context $\mathbb{N}$
that is a subset of $\noms$ and a variable context $\Psi$ that
identifies variables together with their arity types: in combination with 
the constants in $\STLCGamma_0$, these collections, circumscribe the
symbols that can be used in the declarations in the context
expressions.\footnote{In 
  determining closed instances of context schemas, $\mathbb{N}$ will
  be $\noms$ and $\Psi$ will be the empty set. The more general form
  for this relation, which includes a parameterization by these sets,
  will be useful in later sections.} 
The relation is written as $\csinst{\mathbb{N}}{\Psi}{\mathcal{C}}{G}$
and it is defined by the rules in Figure~\ref{fig:ctx-schema}.
This relation is defined via the repeated use of a ``one-step''
instantiation relation written as
$\csinstone{\mathbb{N}}{\Psi}{\mathcal{C}}{G}$; note that by $G, G'$
we mean a context expression that is obtained by adding the bindings
corresponding to $G'$ in front of those in $G$.
The definition of the one-step instantiation relation for context
schemas uses an auxiliary judgement
$\bsinst{\mathbb{N}}{\Psi}{\mathcal{B}}{G}$ that denotes the relation
of ``being an instance of'' between a block 
schema and a context expression fragment.
This relation holds when the context expression is obtained by
generating a sequence of bindings for nominal constants from
$\mathbb{N}$ using the body of the block schema and then instantiating
the variables in the header of the block schema with terms of the
right arity types.
The former task is realized through the relation
$\declinst{\mathbb{N}}{\Delta}{G}{\theta}$ that holds between a block of
declarations $\Delta$, a context expression $G$ that is obtained
by replacing the variables assigned in $\Delta$ with
suitable nominal constants, and a substitution $\theta$ that
corresponds to this replacement.
We assume here and elsewhere that the application of a hereditary
substitution to a sequence of declarations  corresponds to its
application to the type in each assignment. 

\begin{theorem}\label{th:schemainst}
Let $\mathcal{C}$ and $G$ be a context schema and a context expression
such that both $\acstyping{\mathcal{C}}$ and
$\csinst{\mathbb{N}}{\Psi}{\mathcal{C}}{G}$ are derivable. Then for
any arity context $\STLCGamma$ such that
$\mathbb{N} \cup \Psi \cup \STLCGamma_0 \subseteq \STLCGamma$, it is the case that 
$\wfctx{\STLCGamma}{\emptyset}{G}$ has a derivation. 
\end{theorem}

\begin{proof}
We first show that for any block declaration
$\Delta$ and any arity context $\STLCGamma$ such that $\mathbb{N}
\subseteq \STLCGamma$, if $\wfdecls{\STLCGamma}{\Delta}{\STLCGamma'}$ and
$\declinst{\mathbb{N}}{\Delta}{G'}{\theta'}$ are derivable for some $\STLCGamma'$
and $\theta$, then (a)~$\theta'$ is type preserving with respect to
$\Theta$, (b)~$\Theta'$ is $\aritysum{\context{\theta}}{\Theta}$, and
(c)~each binding in $G'$ is of the form $n : A$ where $n:\erase{A}\in
\Theta$ and $\wftype{\Theta}{A}$ has a derivation.
This claim is proved by induction on the derivation of
$\wfdecls{\STLCGamma}{\Delta}{\STLCGamma'}$; properties (a) and (b)
are included in the claim because they are useful together with
Theorem~\ref{th:aritysubs-ty} in showing property (c) in the
induction step. 
Next we show, through an easy inductive argument, that if
$\wfdecls{\STLCGamma_0 \cup \{x_1 : \alpha_1, \ldots,
                              x_n : \alpha_n\}}
         {\Delta}
         {\STLCGamma'}$
has a derivation and $\STLCGamma$ is such that
$\mathbb{N} \cup \Psi \cup \STLCGamma_0 \subseteq \STLCGamma$, then, for
some $\STLCGamma''$, it is the case that 
$\wfdecls{\aritysum{\{x_1 : \alpha_1, \ldots,x_n : \alpha_n\}}
                   {\STLCGamma}}           
         {\Delta}
         {\STLCGamma''}$
has a derivation.
Using Theorem~\ref{th:aritysubs-ty} with these two
observations, we can show easily that if
$\abstyping{\{x_1 : \alpha_1,\ldots, x_n : \alpha_n\}\Delta}$ and 
$\bsinst{\mathbb{N}}{\Psi} 
        {\{x_1 : \alpha_1,\ldots, x_n : \alpha_n\}\Delta}
        {G}$ 
have derivations then for each binding of the form $n:A$ in $G$ it is
the case that $n:\erase{A} \in \Theta$ and $\wftype{\Theta}{A}$.
The theorem follows easily from this observation.
\end{proof}

In defining validity for formulas, we need to consider
substitutions for context and term variables.
For context variables this will correspond to the naive replacement
of the free occurrences of the variables by given context expressions.
We write $\subst{G_1/\Gamma_1,\ldots,G_1/\Gamma_n}{F}$ to denote
the result of the replacement, for $1 \leq i \leq n$, of $\Gamma_i$ by
$G_i$.
For term variables, the replacement must also ensure the
transformation of the resulting expression to normal form.
Towards this end, we adapt hereditary substitution to formulas.
The application of this substitution simply distributes over
quantifiers and logical symbols, respecting the scopes of quantifiers
through the necessary renaming.
The application to the atomic formula $\fatm{G}{M:A}$ also distributes
to the component parts.
We have already discussed the application to terms and types.
The application to context expressions leaves context variables
unaffected and simply distributes to the types in the explicit
bindings.
In particular, no check is made of the possibility of inadvertent
capture.
In this respect, this application is unlike that to LF contexts
that is defined in Figure~\ref{fig:hsubctx}.

\begin{theorem}\label{th:subst-formula}
Let $\STLCGamma$ be an arity context and let $\Xi$ be a collection of
context variables.
\begin{enumerate}
\item If $\theta$ is a term variable substitution that is arity type preserving with
respect to $\STLCGamma$ and $F$ is a formula such that there is a
derivation for $\wfform{\aritysum{\context{\theta}}{\STLCGamma}}{\Xi}{F}$,
then there is a unique formula $F'$ such that
$\hsub{\theta}{F}{F'}$ has a derivation.
Moreover, for this $F'$ it is the case that
$\wfform{\STLCGamma}{\Xi}{F'}$ is derivable.  

\item If $\sigma=\{G_1/\Gamma_1,\ldots,G_n/\Gamma_n\}$ is a context variable
substitution which is such that all judgements in the collection
$\left\{\wfctx{\STLCGamma}{\Xi\setminus\{\Gamma_1,\ldots,\Gamma_n\}}{G_i}\ |\ 
        1\leq i\leq n\right\}$
are derivable and $F$ is a formula such that there is a derivation for
$\wfform{\STLCGamma}{\Xi}{F}$, then there is a derivation for
$\wfform{\STLCGamma}{\Xi\setminus\{\Gamma_1,\ldots,\Gamma_n\}}{\subst{\sigma}{F}}$.
\end{enumerate}
\end{theorem}

\begin{proof}
The first clause is easily provable by induction on 
$\wfform{\aritysum{\context{\theta}}{\STLCGamma}}{\Xi}{F}$, using 
Theorems~\ref{th:uniqueness} and \ref{th:aritysubs} in the atomic case
to ensure the appropriate arity typing judgements will be derivable under 
the substitution $\theta$.

The second clause is by induction on the structure of 
$\wfform{\STLCGamma}{\Xi}{F}$, using the assumption derivations that
$\wfctx{\STLCGamma}{\Xi\setminus\{\Gamma_1,\ldots,\Gamma_n\}}{G_i}$
is derivable to ensure the result of the substitution on a context variable
is well-formed in the atomic case.
\end{proof}

\noindent Following the notation introduced after Theorem~\ref{th:aritysubs}, 
if $F$ and $\theta$ are a formula and a substitution that together satisfy
the requirements of the first part of the theorem, we will write
$\hsubst{\theta}{F}$ to denote the $F'$ for which 
$\hsub{\theta}{F}{F'}$ is derivable.
Note that a term variable substitution may introduce new nominal
constants. 
We will write $\supportof{\theta}$ to denote the collection of such
constants that appear in $\range{\theta}$.
A similar observation holds for context variable substitutions: if
$\sigma$ is the substitution $\{G_1/\Gamma_1,\ldots,G_n/\Gamma_n\}$,
we will write $\supportof{\sigma}$ to denote the collection of nominal
constants that appear in $G_1,\ldots,G_n$.

As discussed previously, a closed atomic formula of the form
$\fatm{G}{M:A}$ is intended to encode an LF judgement of the
form $\lfchecktype{\Gamma}{M}{A}$.
In this encoding, nominal constants that appear in terms represent
free variables for which bindings appear in the context in LF
judgements.
To substantiate this interpretation, the rules \canonkindpi, \canonfampi\ and
\canontermlam\ must introduce fresh nominal constants into contexts in
typing derivations and they must replace bound variables appearing in
terms and types with these constants. 
We use this interpretation to define validity for closed atomic
formulas with one further qualification: unlike in the LF judgement,
for the atomic formula we must also ascertain the wellformedness of
the context and the type.
This notion of validity is then extended to all closed formulas by
recursion on formula structure as we describe below.

\begin{definition}\label{def:semantics}
Let $F$ be a formula such that $\wfform{\noms \cup
  \STLCGamma_0}{\emptyset}{F}$ is derivable. 
\begin{itemize}
\item If $F$ is $\ftrue$ it is valid and if it is $\ffalse$ it is not valid.

\item If $F$ is $\fatm{G}{M:A}$, it is valid exactly when all of $\lfctx{G}$,
  $\lftype{G}{A}$, and $\lfchecktype{G}{M}{A}$ are derivable in LF,
  under the interpretation of nominal constants as variables bound in
  a context and with the modification of the rules \canonkindpi, \canonfampi\ and
  \canontermlam\ to introduce fresh nominal constants into contexts and to
  instantiate the relevant bound variables in kinds, types and terms with
  these constants.

\item If $F$ is $\fimp{F_1}{F_2}$, it is valid if $F_2$ is valid in
  the case that $F_1$ is valid.

\item If $F$ is $\fand{F_1}{F_2}$, it is valid if both $F_1$ and
  $F_2$ are valid.

\item If $F$ is $\for{F_1}{F_2}$, it is valid if either $F_1$ or $F_2$ is valid.

\item If $F$ is $\fctx{\Gamma}{\mathcal{C}}{F}$, it is valid if
  $\subst{G/\Gamma}{F}$ is valid for every $G$ such that
  $\csinst{\noms}{\emptyset}{\mathcal{C}}{G}$ is derivable.

\item If $F$ is $\fall{x:\alpha}{F}$, it is valid if
  $\hsubst{\{\langle x, M,\alpha\rangle \}}{F}$ is valid for every term $M$ such that
  $\stlctyjudg{\noms \cup \STLCGamma_0}{M}{\alpha}$ is derivable.

\item If $F$ is $\fexists{x:\alpha}{F}$, it is valid if
  $\hsubst{\{\langle x, M,\alpha\rangle \}}{F}$ is valid for some term $M$ such that 
  $\stlctyjudg{\noms \cup \STLCGamma_0}{M}{\alpha}$ is derivable.
\end{itemize}
Theorem~\ref{th:subst-formula} guarantees the coherence of this definition.
\end{definition}

%% file: logic/logic-ex.tex
\section{Understanding the Notion of Validity}
\label{sec:logic-ex}

In the examples we consider below, we assume an instantiation of the logic
based on the signature presented in
Section~\ref{sec:lf-ex}.
Obviously, any LF typing judgement based on that signature is expressable
in the logic.
Moreover, the corresponding formula will be valid exactly when the
typing judgement is derivable in LF.
Thus, the formulas
$\fatm{\emptyce}{\emptytm : \tmty}$,
$\fatm{\emptyce}{(\lamtm\app \unittm\app (\lflam{x}{x})) : \tmty}$ and
$\fatm{n : \tmty}{n:\tmty}$ are all valid. 
Similarly, the formulas
$\fexists{d:\oty}{\fatm{\emptyce}{d : (\ofty \app \emptytm\app
    \unittm)}}$ and
\begin{tabbing}
\qquad\=\kill
\>$\fexists{d:\oty}{\fatm{\emptyce}{d :\ofty \app (\lamtm \app \unittm\app
    (\lflam{x}{x}))\app (\arrtm\app \unittm\app \unittm)}}$
\end{tabbing}
are valid but the formula
$\fexists{d:\oty}{\fatm{\emptyce}{d:\ofty \app
    (\lamtm \app \unittm\app (\lflam{x}{x}))\app \unittm}}$
is not.
Note that the arity type associated with the quantified variable in
each of these formulas provides only a rough constraint on the
instantiation needed to verify the validity of the formula; to do
this, the instance must also satisfy LF typeability requirements
represented by formula that appears within the scope of the
quantifier.

Wellformedness conditions for formulas ensure only that the terms
appearing within formulas satisfy canonicity requirements, i.e. that
these terms are in $\beta$-normal form and that variables and
constants are applied to as many arguments as they can take.
Arity typing does not distinguish between terms in different
expression categories.
For example, the formula
\begin{tabbing}
\qquad\=\kill
\> $\fexists{d:\oty}{\fatm{\emptyce}{d:\ofty \app (\lamtm \app \emptytm \app
    (\lflam{x}{x}))\app (\arrtm\app \unittm\app \unittm)}}$
\end{tabbing}
is well-formed but not valid.
An alternative design choice, with equivalent consequences from the
perspective of the valid properties that can be expressed in the
logic, might have been to let the fact that $\lamtm$ is ill-applied to
$\emptytm$ impact on the wellformedness of the formula.
The wellformedness conditions do not also enforce a distinctness
requirement for bindings in a context.
Thus, the formula $\fatm{n : \tmty, n : \tpty}{\emptytm : \tmty}$ is
well-formed.
However, it is not valid because $\lfctx{n : \tmty, n : \tpty}$ is not
derivable in LF under the described interpretation for nominal
constants.
An implication of these observations is that a naive form of weakening 
does not hold with respect to the encoding of LF derivability in the
logic; additional conditions similar to this described in
Theorem~\ref{th:weakening} must be verified for this principle to
apply.

To provide a more substantive example of the kinds of properties that
can be expressed in the logic, let us consider the formal statement of
the uniqueness of type assignment for simply typed $\lambda$-calculus
terms.
As noted in Section~\ref{sec:lf-ex}, this property is
best described in a form that considers typing expressions in
contexts that have a particular kind of structure. 
That structure can be formalized in the logic by a context
schema comprising the single block
\begin{tabbing}
  \qquad\=\kill
  \> $\{t : o\}x:tm,y:\ofty\app x\app t$.
\end{tabbing}
Let us denote this context schema by $ctx$.
Observe that any context that instantiates this schema will not
provide a variable that can be used to construct an atomic term of
type $\tpty$.
Thus, the strengthening property for expressions representing types
that is expressed by the formula 
\begin{tabbing}
  \qquad\=\kill
  \> $\fctx{\Gamma}{ctx}
           {\fall{t :\oty}
                 {\fimp{\fatm{\Gamma}{t : \tpty}}
                   {\fatm{\emptyce}{t:\tpty}}}}$.
\end{tabbing}
should hold.
We can in fact easily show this formula to be valid by
using Theorem~\ref{th:atomictype} and an induction on the height of
the derivation for $\fatm{G}{t : \tpty}$ for a closed term $t$ and a
closed instance $G$ of $ctx$.
Using the validity of this formula, we can also easily argue that the
following formula that expresses a strengthening property pertaining to
the equality of types is also valid:
\begin{tabbing}
  \qquad\=\kill
  \> $\fctx{\Gamma}{ctx}
           {\fall{d :\oty}
           {\fall{t_1 :\oty}
           {\fall{t_2 : \oty}
                 {\fimp{\fatm{\Gamma}{d : \eqty\app t_1\app t_2}}
                 {\fatm{\emptyce}{d:\eqty\app t_1\app t_2}}}}}}$.
\end{tabbing}

Then the uniqueness of typing property can be expressed through the
following formula:
\begin{tabbing}
\qquad\=\qquad\=\kill
\>$\fctx{\Gamma}{ctx}{\fall{e:\oty}{\fall{t_1:\oty}{\fall{t_2:\oty}{\fall{d_1:\oty}{\fall{d_2:\oty}{}}}}}}$\\
\>\>$\fimp{\fatm{\Gamma}{d_1:\ofty\app e\app t_1}}
          {\fimp{\fatm{\Gamma}{d_2:\ofty\app e\app t_2}} 
                {\fexists{d_3:\oty}{\fatm{\emptyctx}{d_3:\eqty\app t_1\app t_2}}}}$
\end{tabbing}
This formula can be seen to be valid using the strengthening property
just described if we can establish the validity of the formula
\begin{tabbing}
\qquad\=\qquad\=\kill
\>$\fctx{\Gamma}{ctx}{\fall{e:\oty}{\fall{t_1:\oty}{\fall{t_2:\oty}{\fall{d_1:\oty}{\fall{d_2:\oty}{}}}}}}$\\
\>\>$\fimp{\fatm{\Gamma}{d_1:\ofty\app e\app t_1}}
          {\fimp{\fatm{\Gamma}{d_2:\ofty\app e\app t_2}} 
                {\fexists{d_3:\oty}{\fatm{\Gamma}{d_3:\eqty\app t_1\app t_2}}}}$.
\end{tabbing}
To show this, it suffices to argue that, for a closed context
expression $G$ that instantiates the schema $ctx$ and for closed
expressions $d_1$, $d_2$, $e$, $t_1$, and $t_2$, if
$\fatm{G}{d_1 : \ofty\app e\app t_1}$ and
$\fatm{G}{d_2 : \ofty\app e\app t_2}$ are valid, then there must be a
closed expression $d_3$ such that
$\fatm{G}{d_3 : \eqty\app t_1\app t_2}$ is also valid.
Such an argument can be constructed by induction on the height
of the LF derivation of $\lfchecktype{G}{d_1}{\ofty\app e\app t_1}$,
which we analyze using Theorem~\ref{th:atomictype} in the manner
discussed earlier. 
There are essentially four cases to consider, corresponding to
whether the head symbol of $d_1$ is \ofemptytm, \ofapptm, 
\oflamtm, or a nominal constant that is assigned the type
$(\ofty\app n\app t_1)$ in $G$ where $n$ is also a nominal constant that
is bound in $G$.
In the last case, we use the fact that the validity of
$\fatm{G}{d_1 : \ofty\app e\app t_1}$ implies that $\lfctx{G}$ is
derivable to conclude the uniqueness of $n$ and, hence, of the typing.
The argument when $d_1$ is \ofemptytm\ has an obvious form.
The argument when $d_1$ has \ofapptm\ or \oflamtm\ as its head symbol
will invoke the induction hypothesis. 
In the case where the head symbol is \oflamtm, we will need
to consider a shorter derivation of a typing judgement in which the
context has been enhanced.
However, we will be able to use the induction hypothesis by observing
that the enhancements to the context conform to the constraints
imposed by the context schema. 
Note that the form of $d_1$ also constrains the form of $e$ in all the
cases, a fact that is used implicitly in the analysis outlined.

%% file: logic/permutations.tex
\section{Nominal Constants and Invariance Under Permutations}
\label{sec:permutations}

As noted in Chapter~\ref{ch:lf}, the particular choices for bound
variable names in the kinds, types and terms that comprise LF
expressions are considered irrelevant.
This understanding is built in concretely through the notion of
$\alpha$-conversion that renders equivalent expressions that differ
only in the names used for such variables.
Typing derivations transform expressions with bound variables into
ones where variables are ostensibly free but in fact bound implicitly
in the associated contexts.
The lack of importance of name choices is reflected in this case in an
invariance in the validity of typing judgements under a suitable
renaming of variables appearing in the judgements.
In a situation where context variables are represented by nominal
constants, this property has a simple expression in the form of an
invariance of formula validity under permutations of nominal constants.
We will need this property in later chapters and so we present
it formally below.

We begin with a definition of the notions of permutations of
nominal constants and their applications to expressions. 
\begin{definition}[Permutation]

A permutation of the nominal constants is an arity type preserving
bijection from $\noms$ to $\noms$ that differs from the identity
map at only a finite number of constants. 
The permutation that maps $n_1,\ldots, n_m$ to $n_1',\ldots,n_m'$,
respectively, and is the identity everywhere else is written as 
$\{n_1'/n_1,\ldots,n_m'/n_m\}$.
The support of a permutation $\pi=\{n_1'/n_1,\ldots,n_m'/n_m\}$, 
denoted by $\supp{\pi}$, is the collection of nominal constants
$\{n_1,\ldots, n_m\}\cup\{n_1',\ldots,n_m'\}$.
Every permutation $\pi$ has an obvious inverse that is written as
$\inv{\pi}$. 
\end{definition}

\begin{definition}[Permutation Application]
The application of a permutation $\pi$ to an expression $E$ of a
variety of kinds is described below and is denoted in all cases by
$\permute{\pi}{E}$.
If $E$ is a term, type, or kind then the
application consists of replacing each nominal constant $n$ that
appears in $E$ with $\pi(n)$. 
If $E$ is a context then the application of $\pi$ to $E$
replaces each explicit binding $n:A$ in $E$ with 
$\pi(n):\permute{\pi}{A}$.
If $E$ is an LF judgement $\mathcal{J}$ then the permutation is applied
to each component of the judgement in the way described above.
If $E$ is a formula then the permutation is applied to its component parts.
The application of $\pi$ to a term variable substitution
$\{\langle x_1,M_1,\alpha_1\rangle,\ldots,
   \langle x_n,M_n,\alpha_n\rangle\}$
yields the substitution 
$\{\langle x_1,\permute{\pi}{M_1},\alpha_1\rangle,\ldots,
   \langle x_n,\permute{\pi}{M_n},\alpha_n\rangle\}$.
The application of $\pi$ to a context variable substitution
$\{G_1/\Gamma_1,\ldots, G_n/\Gamma_n\}$ yields
$\{\permute{\pi}{G_1}/\Gamma_1,\ldots, \permute{\pi}{G_n}/\Gamma_n\}$.
\end{definition}

The following theorem expresses the property of interest
concerning LF judgements cast in the form relevant to the
logic. 
\begin{theorem}\label{th:perm-lf}
Let LF judgements and derivations be recast in the form discussed
earlier in this section: variables that are bound in a context are
represented by nominal constants and the rules \canonkindpi,
\canonfampi\ and \canontermlam\ introduce fresh nominal constants into
contexts and replace variables in kinds, types and terms with these
constants.
In this context, let $\mathcal{J}$ be an LF judgement which has a
derivation. 
Then for any permutation $\pi$, $\permute{\pi}{\mathcal{J}}$ is derivable.
Moreover, the structure of this derivation is the same as that for
$\mathcal{J}$.
\end{theorem} 
\begin{proof}
This proof is by induction on the derivation for $\mathcal{J}$.
Perhaps the only observation worthy of note is that the freshness of
nominal constants used in \canonkindpi, \canonfampi, and
\canontermlam\ rules is preserved under permutations of nominal
constants. 
\end{proof}
The above observation underlies the main theorem of this section.
\begin{theorem}\label{th:perm-form}
Let $F$ be a closed formula and let $\pi$ be a permutation.
Then $F$ is valid if and only if $\permute{\pi}{F}$ is valid.
\end{theorem}
\begin{proof}
Noting that $\inv{\pi}$ is also a permutation and that
$\permute{\inv{\pi}}{\permute{\pi}{F}}$ is $F$, it suffices to prove
the claim in only one direction.
We do this by induction on the structure of $F$.

The desired result follows easily from
Theorem~\ref{th:perm-lf} and the relationship of validity to LF
derivability when $F$ is atomic.
The cases where $F$ is $\ftrue$ or $\ffalse$ are trivial and the
ones in which $F$ is $\fimp{F_1}{F_2}$, $\fand{F_1}{F_2}$ or
$\for{F_1}{F_2}$ are easily argued with recourse to the induction
hypothesis and by noting that the permutation distributes to the
component formulas.

In the case where $F$ is $\fctx{\Gamma}{\mathcal{C}}{F'}$, we first
note that if $\csinst{\noms}{\emptyset}{\mathcal{C}}{G}$ has a
derivation then
$\csinst{\noms}{\emptyset}{\mathcal{C}}{\permute{\inv{\pi}}{G}}$ must also
have one.
From this and the validity of $F$ it follows that
$\subst{\permute{\inv{\pi}}{G}/\Gamma}{F'}$ must be valid.
Moreover, $\subst{\permute{\inv{\pi}}{G}/\Gamma}{F'}$ has the same
structural complexity as $F'$.
Hence, by the induction hypothesis,
$\permute{\pi}{(\subst{\permute{\inv{\pi}}{G}/\Gamma}{F'})}$ is valid.
Noting that this formula is the same as
$\subst{G/\Gamma}{(\permute{\pi}{F'})}$ and that
$\fctx{\Gamma}{\mathcal{C}}{\permute{\pi}{F'}}$ is identical to 
$\permute{\pi}{(\fctx{\Gamma}{\mathcal{C}}{F'})}$, the validity of
$\permute{\pi}{F}$ easily follows.

Suppose that $F$ has the form $\fall{x:\alpha}{F'}$.
We observe here that if
$\stlctyjudg{\noms \cup \STLCGamma_0}{M}{\alpha}$ has a derivation
then
$\stlctyjudg{\noms \cup \STLCGamma_0}{\permute{\inv{\pi}}{M}}{\alpha}$
has one too and that
$\permute{\pi}{(\hsubst{\{\langle x,\permute{\inv{\pi}}{M},\alpha\rangle\}}{F'})}$
is the same formula as
$\hsubst{\{\langle x,M,\alpha\rangle\}}{(\permute{\pi}{F'})}$.
Using the definition of validity, the induction hypothesis and the
fact that permutation distributes to the component formula together
with the above observations, we may easily conclude that
$\permute{\pi}{F}$ is valid.

Finally, suppose that $F$ is of the form $\fexists{x:\alpha}{F'}$.
Here we note that if
$\stlctyjudg{\noms \cup \STLCGamma_0}{M}{\alpha}$ has a derivation
then
$\stlctyjudg{\noms \cup \STLCGamma_0}{\permute{\pi}{M}}{\alpha}$
has one too and that
$\permute{\pi}{(\hsubst{\{\langle x, M,\alpha\rangle\}}{F'})}$
is the same formula as
$\hsubst{\{\langle
  x,\permute{\pi}{M},\alpha\rangle\}}{(\permute{\pi}{F'})}$.
Using the definition of validity and the induction hypothesis, it is
now easy to conclude that $\permute{\pi}{F}$ must be valid.
\end{proof}

%% file: proof-system/proof-system.tex
\chapter{A Proof System for Constructing Arguments of Validity}
\label{ch:proof-system}

We have presented a logic in Chapter~\ref{ch:logic} that can be used
to describe properties of LF specifications.
We have also shown there how we can argue informally about
the validity of formulas that encapsulate such properties.
Our goal now is to develop a formal mechanism for constructing
arguments of validity.
Towards this end, we describe in this chapter a proof system that
complements our logic.
This proof system is oriented around sequents that represent
assumption and conclusion formulas augmented with devices that capture
additional aspects of states that arise in the process of reasoning. 
The syntax for sequents is more liberal than is meaningful at the
outset, and this is rectified by imposing wellformedness requirements
on them.
We associate a semantics with sequents that is consistent with their
intended use in enabling reasoning about the validity of formulas.
We then present a collection of proof rules that can be used to derive
sequents.
These rules belong to two broad categories.
The first category comprises rules that embody logical aspects such as
the meanings of sequents and of the logical symbols that appear in formulas.
A key aspect of our logic is that its atomic formulas represent the
notion of derivability in LF that is also open to analysis.
The second category of proof rules builds in capabilities for such
analysis.
To be coherent, our proof rules must preserve the wellformedness
property of sequents and our first endeavor concerning their
presentation is to show that they indeed satisfy this property.
At a more substantive level, the proof rules must support
a reasoning process that is both sound and effective.
The focus in this chapter is on ensuring soundness.
We do this by
demonstrating for each proposed rule that the conclusion sequent must
be valid if its premise sequents are.
The demonstration of effectiveness for the proof system will be the
subject of later chapters. 

The first section below presents the sequents underlying the proof
system and identifies a semantics that enables their use in 
validity arguments and that also undergirds the demonstrations of
soundness of proof rules. 
The remaining sections in the chapter develop the collection of proof
rules. 
In Section~\ref{sec:core-logic}, we present the ``core'' rules, \ie,
the rules that encapsulate the meanings of the logical symbols and
also certain aspects of sequents.
We then turn to the rules that internalize aspects of LF derivability
that permeate the logic through the interpretation of the atomic
formulas.
Section~\ref{sec:atomic-rules} develops rules for analyzing atomic formulas.
An important component of these rules is the interpretation of typing
judgements involving atomic types via the particular LF specification
that parameterizes the logic: this interpretation leads, in
particular, to a case analysis rule for such atomic formulas that
appear as assumptions in a sequent.
Another important component of informal reasoning that needs to be
supported by the proof system is that which is based on an induction
on the height of an LF derivation.
To support this ability, we introduce an induction rule in
Section~\ref{sec:ind} that is inspired by the annotation based
scheme used in Abella~\cite{baelde14jfr,gacek09phd}.
In the final section of the chapter, we introduce proof rules that
encode meta-theorems concerning LF derivability that often find use in
reasoning about LF specifications. 
   
\input{proof-system/sequents}
\input{proof-system/core-logic}

\input{proof-system/atomic}
\input{proof-system/induction}
\input{proof-system/meta-theory}

%% file: proof-system/sequents.tex
\section{The Structure of Sequents}
\label{sec:sequents}

A sequent in our proof system is characterized by a collection of
assumption formulas and a conclusion or goal formula.
The formulas may contain free term and context variables that are to
be interpreted as being implicitly universally quantified over the  
sequent and, therefore, its proof.
We find it useful also to identify with the sequent a collection of
nominal constants that circumscribes the ones that appear in its
formulas. 

The nominal constants and term variables that appear in the
sequent have arity types associated with them.
Context variables are also typed and their types are, in spirit, based
on context schemas.
However, subproofs may require a partial elaboration of a context
variable and the types associated with such variables accommodates
this possibility.
More specifically, these types have the form
$\ctxty{\mathcal{C}}{G_1;\ldots; G_n}$ where $\mathcal{C}$ is a context 
schema and $G_1,\ldots, G_n$ are context expressions.
Such a type is intended to represent the collection of context
expressions obtained by interspersing $G_1,\ldots,G_n$ with
instantiations of the context schema $\mathcal{C}$ and possibly
prefixed by a context variable of suitable type that represents a yet
to be elaborated sequence of declarations.
Additionally, context variables are annotated with a collection of
nominal constants that express the constraint that the elaborations of
these variables must not use names in these collections; the ability
to express such constraints is an essential part of the mechanism for
analysing typing judgements involving abstractions as we will see
later in this section.

The ideas pertaining to context variable typing are made precise
through the following definition. 

\begin{figure}
\begin{center}
\begin{tabular}{c}
\infer{\wfctxvarty{\mathbb{N}}{\Psi}{\ctxty{\mathcal{C}}{\cdot}}}
      {}

\qquad\qquad      
\infer{\wfctxvarty{\mathbb{N}}{\Psi}{\ctxty{\mathcal{C}}{\mathcal{G};G}}}
      {\csinstone{\mathbb{N}}{\Psi}{\mathcal{C}}{G}
       \qquad
       \wfctxvarty{\mathbb{N}}{\Psi}{\ctxty{\mathcal{C}}{\mathcal{G}}}} 

\\[15pt]

\infer{\ctxtyinst{\mathbb{N}}{\Psi}{\Xi}{\ctxty{\mathcal{C}}{\emptycb}}{\emptyce}}
      {}

\qquad
\infer{\ctxtyinst{\mathbb{N}}{\Psi}{\Xi}{\ctxty{\mathcal{C}}{\cdot}}{\Gamma}}
      {\ctxvarty{\Gamma}
                {\mathbb{N}_{\Gamma}}
                {\ctxty{\mathcal{C}}{\mathcal{G}}} \in\Xi
         &
       (\noms\setminus\mathbb{N}_{\Gamma})\subseteq\mathbb{N}}

\\[15pt]
      
\infer{\ctxtyinst{\mathbb{N}}{\Psi}{\Xi}{\ctxty{\mathcal{C}}{\mathcal{G}}}{G,G'}}
      {\ctxtyinst{\mathbb{N}}{\Psi}{\Xi}{\ctxty{\mathcal{C}}{\mathcal{G}}}{G}
      \qquad
      \csinstone{\mathbb{N}}{\Psi}{\mathcal{C}}{G'}}

\qquad\qquad

\infer{\ctxtyinst{\mathbb{N}}{\Psi}{\Xi}{\ctxty{\mathcal{C}}{\mathcal{G};G'}}{G,G'}}
      {\ctxtyinst{\mathbb{N}}{\Psi}{\Xi}{\ctxty{\mathcal{C}}{\mathcal{G}}}{G}}

\end{tabular}
\end{center}
\caption{Well-Formed Context Variable Types and their Instantiations}
\label{fig:wfctxvar}
\end{figure}

\begin{definition}[Context Variable Types and their Instances]\label{def:cvar-types}
A \emph{context variable type} is a expression of the form
$\ctxty{\mathcal{C}}{\mathcal{G}}$ where $\mathcal{C}$ is a context
schema such that $\acstyping{\mathcal{C}}$ is derivable and
$\mathcal{G}$ represents a sequence of \emph{context blocks} given as
follows: 
\[ \mathcal{G} ::= \emptycb \ | \ \mathcal{G}; n_1 : A_1, \ldots, n_k : A_k. \]
Such a type is said to be well-formed with respect to a nominal
constant set $\mathbb{N} \subseteq \noms$ and a term variable context
$\Psi$ that associates arity types with term variables if it is the
case that the relation
$\wfctxvarty{\mathbb{N}}{\Psi}{\ctxty{\mathcal{C}}{\mathcal{G}}}$
that is defined by the rules in Figure~\ref{fig:wfctxvar} holds.
A \emph{context variable context} is a collection of associations of
sets of nominal constants and context variable types with context
variables, each written in the
form $\ctxvarty{\Gamma}{\mathbb{N}_\Gamma}{\ctxty{\mathcal{C}}{\mathcal{G}}}$.
Given a context variable context $\Xi$, we write $\ctxsanstype{\Xi}$ for
the set
$\{ \Gamma\ |\ \ctxvarty{\Gamma}
                        {\mathbb{N}_\Gamma}
                        {\ctxty{\mathcal{C}}{\mathcal{G}}}
               \in \Xi\}$, \ie, the collection of context variables
assigned types by $\Xi$.
A context expression $G$ is said to be an instance of a context type
$\ctxty{\mathcal{C}}{\mathcal{G}}$ relative to $\mathbb{N}$, $\Psi$
and the context variable context $\Xi$ if the relation
$\ctxtyinst{\mathbb{N}}{\Psi}{\Xi}{\ctxty{\mathcal{C}}{\mathcal{G}}}{G}$,
that is also defined in Figure~\ref{fig:wfctxvar}, holds.
\end{definition}

The following theorem, whose proof is based on an obvious induction,
shows that an instance of a well-formed context variable type is a
well-formed context relative to the relevant arity context and context
variable collection. 

\begin{theorem}\label{th:ctx-var-type-instance}
If $\wfctxvarty{\mathbb{N}}{\Psi}{\ctxty{\mathcal{C}}{\mathcal{G}}}$
and
$\ctxtyinst{\mathbb{N}}{\Psi}{\Xi}{\ctxty{\mathcal{C}}{\mathcal{G}}}{G}$
have derivations then so does
$\wfctx{\mathbb{N}\cup \Psi}{\Xi}{G}$.
\end{theorem}

In using the theorem above and in other similar situations, we will
often need to adjust the contexts that parameterize the relevant
wellformedness judgements.
The theorem below, whose proof is also based on a straightforward
induction, provides the basis for such adjustments.

\begin{theorem}\label{th:ctx-ty-wk}
Let $\mathbb{N} \subseteq \mathbb{N}'$, $\Psi \subseteq \Psi'$, and
$\Xi \subseteq \Xi'$.  
\begin{enumerate}
\item If $\wfctxvarty{\mathbb{N}}{\Psi}{\ctxty{\mathcal{C}}{\mathcal{G}}}$
is derivable, then 
$\wfctxvarty{\mathbb{N}'}{\Psi'}{\ctxty{\mathcal{C}}{\mathcal{G}}}$
is derivable.
\item If 
$\ctxtyinst{\mathbb{N}}{\Psi}{\Xi}{\ctxty{\mathcal{C}}{\mathcal{G}}}{G}$
is derivable, then
$\ctxtyinst{\mathbb{N}'}{\Psi'}{\Xi'}{\ctxty{\mathcal{C}}{\mathcal{G}}}{G}$
is derivable.
\end{enumerate}
\end{theorem}

The wellformedness judgements in Figure~\ref{fig:wfctxvar} are
preserved under meaningful substitutions as the theorem below
explicates.

\begin{theorem}\label{th:ctxtyinst-hsubst}
Let $\theta$ be an arity type preserving substitution with respect to
$\noms \cup \STLCGamma_0 \cup \Psi$ and let
$\wfctxvarty{\mathbb{N}}
            {\aritysum{\context{\theta}}{\Psi}}
            {\ctxty{\mathcal{C}}{\mathcal{G}}}$
have a derivation.
Then
\begin{enumerate}
\item there must be a derivation for
$\wfctxvarty{\mathbb{N} \cup \supportof{\theta}}{\Psi}{\ctxty{\mathcal{C}}{\hsubst{\theta}{\mathcal{G}}}}$,
and 
\item if
$\ctxtyinst{\mathbb{N}}{\aritysum{\context{\theta}}{\Psi}}{\Xi}{\ctxty{\mathcal{C}}{\mathcal{G}}}{G}$
also has a derivation, there must be a derivation for
$\ctxtyinst{\mathbb{N}\cup \supportof{\theta}}{\Psi}{\Xi}{\ctxty{\mathcal{C}}{\hsubst{\theta}{\mathcal{G}}}}{\hsubst{\theta}{G}}$.
\end{enumerate}
\end{theorem}
\begin{proof}
Since $\theta$ is arity type preserving with respect to 
$\noms \cup \STLCGamma_0 \cup \Psi$, using Theorem~\ref{th:aritysubs} we see that 
if there is a derivation for 
$\stlctyjudg{\mathbb{N} \cup (\aritysum{\context{\theta}}{\Psi}) \cup \STLCGamma_0}
            {t}
            {\alpha}$ 
then $\hsubst{\theta}{t}$ is well-defined and
$\stlctyjudg{(\mathbb{N} \cup \supportof{\theta}
               \cup \Psi
               \cup \STLCGamma_0}
            {\hsubst{\theta}{t}}
            {\alpha}$
has a derivation.
An induction on the derivation of 
$\wfctxvarty{\mathbb{N}}
            {\aritysum{\context{\theta}}{\Psi}}
            {\ctxty{\mathcal{C}}{\mathcal{G}}}$
using these observations allows us to confirm the first part of the
theorem.
Further, let there be a derivation for
$\csinstone{\mathbb{N}}{\aritysum{\context{\theta}}{\Psi}}{\mathcal{C}}{G'}$.
By an induction on this derivation using the facts observed earlier,
it can be concluded that there must be a derivation for
$\csinstone{\mathbb{N} \cup \supportof{\theta}}
           {\Psi}
           {\mathcal{C}}
           {\hsubst{\theta}{G'}}$.
The second part of the theorem follows by another obvious induction from this.
\end{proof}

\begin{definition}[Sequents]\label{def:sequent}
A sequent, written as $\seq[\mathbb{N}]{\Psi}{\Xi}{\Omega}{F}$, is a
judgement that relates a finite subset $\mathbb{N}$ of $\noms$, a
finite set $\Psi$ of arity type assignments to term variables, a
context variable context $\Xi$, a finite set $\Omega$ of
\emph{assumption formulas} and a conclusion or goal formula $F$.
The sequent is well-formed if (a)~for each type association
$\ctxvarty{\Gamma}{\mathbb{N}_\Gamma}{\ctxty{\mathcal{C}}{\mathcal{G}}}$ in
$\Xi$ it is the case that 
$\wfctxvarty{\mathbb{N} \setminus \mathbb{N}_\Gamma}{\Psi}{\ctxty{\mathcal{C}}{\mathcal{G}}}$
is derivable and (b)~for each formula $F'$ in $\{F\} \cup \Omega$ the judgement $\wfform{\mathbb{N} \cup \Psi \cup \Theta_0}{\ctxsanstype{\Xi}}{F'}$
is derivable. 
Given a well-formed sequent $\seq[\mathbb{N}]{\Psi}{\Xi}{\Omega}{F}$,
we refer to $\mathbb{N}$ as its support set, to $\Psi$ as its
eigenvariable context and $\Xi$ as its context variable context.
We use a comma to denote set union in representing sequents, writing
$\setand{\Omega}{F_1,\ldots,F_n}$ to denote the set $\Omega\cup\{F_1,\ldots,F_n\}$.
\end{definition}

We will need to consider substitutions for term variables in
sequents. 
We will require legitimate substitutions to not use the nominal
constants in the support set of the sequent; this restriction will be
part of a mechanism for controlling dependencies in context
declarations.  
We will further require substitutions to satisfy arity typing
constraints for their applications to be well-defined.
These considerations are formalized below in a notion of
compatibility between substitutions and sequents. 

\begin{definition}[Term Substitutions Compatible with Sequents]\label{def:seq-term-subst}
A pair $\langle \theta, \Psi' \rangle$ consisting of a term variable 
substitution and an arity context assigning types to variables is said
to be substitution compatible with a well-formed sequent
$\mathcal{S} = \seq[\mathbb{N}]{\Psi}{\Xi}{\Omega}{F}$ if
\begin{enumerate}
\item $\theta$ is arity type preserving with respect to the context
  $\noms \cup \Theta_0 \cup \Psi'$, 
  
\item $\supportof{\theta} \cap \mathbb{N} = \emptyset$, and 

\item for any variable $x$, if $x : \alpha \in \Psi$ and
  $x : \alpha' \in \aritysum{\context{\theta}}{\Psi'}$, then $\alpha =
  \alpha'$. 
\end{enumerate}
\end{definition}

The application of a substitution may introduce new nominal constants
into a sequent.
When this happens, substitutions for the eigenvariables in the
resulting sequent must be permitted to contain these constants. 
We use the technique of raising to realize this
requirement~\cite{miller92jsc}. 
The following definition is useful in formalizing this idea.
\begin{definition}[Raising a Context over Nominal Constants]
Let $\Psi$ be a set of the form $\{x_1:\alpha_1,\ldots,x_m:\alpha_m\}$ that
associates arity types with a finite collection of variables, let
$n_1,\ldots,n_k$ be a listing of the elements of a finite 
collection of the nominal constants $\mathbb{N}$,
and let $\beta_1,\ldots,\beta_k$ be the arity types
associated by $\noms$ with these constants.  
Then a version of $\Psi$ raised over $\mathbb{N}$ is a set 
$\{y_1 : \gamma_1,\ldots, y_m : \gamma_m\}$ where, for $1 \leq i \leq
m$, $y_i$ is a distinct variable that is also different from the
variables in $\{x_1,\ldots,x_m\}$ and $\gamma_i$ is
$\beta_1 \atyarr \cdots \atyarr \beta_k \atyarr \alpha_i$.
Further, the raising substitution associated with this version is the set
$\{\langle x_i, (y_i \app n_1 \app \ldots \app n_k),
\alpha_i\rangle\ |\ 1 \leq i \leq m\}$.
\end{definition}

The basis for using raising in the manner described is the content of
the following theorem.
We say here and elsewhere that an arity context $\STLCGamma$ is
compatible with $\noms$ if the types that $\STLCGamma$ assigns to
nominal constants are identical to their assignments in $\noms$.

\begin{theorem}\label{th:raised-subs}
Let $\theta$ be a substitution that is arity type preserving with respect to
an arity context $\STLCGamma$ that is compatible with $\noms$.
Further, let $\Psi$ be a version of $\context{\theta}$ raised over
some listing of a collection
$\mathbb{N}$ of nominal constants and let $\theta_r$ be the associated 
raising substitution.
Then there is a substitution $\theta'$ with
$\supportof{\theta'} = \supportof{\theta}\setminus \mathbb{N}$
and $\context{\theta'} = \Psi$ that is
arity type preserving with respect to $\STLCGamma$
and such that for any $E$ for which
$\wftype{\aritysum{\context{\theta}}{\STLCGamma}}{E}$ or, for some 
arity type $\alpha$,
$\stlctyjudg{\aritysum{\context{\theta}}{\STLCGamma}}{E}{\alpha}$ has a  
derivation, it is the case that 
 $\hsubst{\theta'}{\hsubst{\theta_r}{E}} =
\hsubst{\theta}{E}$.  
\end{theorem}
\begin{proof}
Each of the substitutions involved in the
expression
$\hsubst{\theta'}{\hsubst{\theta_r}{E}} = \hsubst{\theta}{E}$
will have a result under the conditions described, thereby justifying
the use of the notation introduced after Theorem~\ref{th:aritysubs}.
Now, let $n_1,\ldots,n_k$ be the listing of the constants in
$\mathbb{N}$ in the raising substitution, let
$\alpha_1,\ldots,\alpha_k$ be the respective types of these constants,
let $\langle x, t, \alpha \rangle$ be a tuple in $\theta$ and let
$\langle x, (y \app n_1 \app \cdots \app n_k), \alpha \rangle$ be the
tuple corresponding to $x$ in $\theta_r$. 
Let $x_1,\ldots,x_k$ be a listing of distinct variables that do not
appear in $t$ and let $t'$ be the result of replacing $n_i$ by $x_i$
in $t$, for $1 \leq i \leq k$.
We construct $\theta'$ by including in it the substitution
$\langle y, \lflam{x_1}{\ldots\lflam{x_k}{t'}},
         \alpha_1 \atyarr \cdots \atyarr \alpha_k \atyarr \alpha 
\rangle$
for each case of the kind considered.
It is easy to see that $\supportof{\theta'} = \supportof{\theta}\setminus
\mathbb{N}$ and that $\theta_r$ and $\theta'$ are arity type
compatible with respect to $\STLCGamma \cup \noms$.
The remaining part of the theorem follows from noting that
$\theta = \theta' \circ \theta_r$ and using
Theorem~\ref{th:composition}. 
\end{proof} 

The following definition formalizes the application of a term
substitution to a well-formed sequent when the conditions of
substitution compatibility are met.
We assume here and elsewhere that the application of a 
substitution  to a set of formulas distributes to each member of the
set, 
its application to a context variable context
distributes to each context variable type in the context and its
application to a context variable type
$\ctxty{\mathcal{C}}{\mathcal{G}}$ distributes to each context block
in $\mathcal{G}$. 
\begin{definition}[Applying a Term Substitution to a
    Sequent]\label{def:seq-term-subst-app} 
Let $\mathcal{S}$ be the well-formed sequent
$\seq[\mathbb{N}]{\Psi}{\Xi}{\Omega}{F}$
and $\langle \theta, \Psi'\rangle$
be substitution compatible with $\mathcal{S}$.
Further, let $\Psi''$ be a version of $(\Psi \setminus
\context{\theta}) \cup \Psi'$  raised over $\supportof{\theta}$ and
let $\theta_r$ be the corresponding raising substitution. 
Then the application of $\theta$ to $\mathcal{S}$ relative to $\Psi'$,
denoted by $\hsubstseq{\Psi'}{\theta}{\mathcal{S}}$, is the sequent
$\seq[\mathbb{N} \cup \supportof{\theta}]
     {\Psi''} 
     {\hsubst{\theta_r}{\hsubst{\theta}{\Xi}}}
     {\hsubst{\theta_r}{\hsubst{\theta}{\Omega}}}
     {\hsubst{\theta_r}{\hsubst{\theta}{F}}}$.
\end{definition}

The definition and notation above are obviously ambiguous since they 
depend on the particular choices of $\Psi''$ and $\theta_r$.
We shall mean $\hsubstseq{\Psi'}{\theta}{\mathcal{S}}$ to denote any
one of the sequents so determined, referring to $\Psi''$ and
$\theta_r$ as the raised context and the raising substitution
associated with the application of the substitution where
disambiguation is needed.
Note also that the definition assumes that the application of the
substitutions $\theta$ and $\theta_r$ to the relevant context variable
types and formulas is well-defined.
We show this to be the case in the theorem below. 

\begin{theorem}\label{th:seq-term-subs-ok}
Let $\seqsub{\theta}{\Psi'}$ be substitution compatible with
a well-formed sequent $\mathcal{S}$. 
Then
$\hsubstseq{\Psi'}{\theta}{\mathcal{S}}$ is well-defined and is a well-formed sequent.
\end{theorem}
\begin{proof}
Let the sequent $\mathcal{S}$ be $\seq[\mathbb{N}]{\Psi}{\Xi}{\Omega}{F}$.
Following Definition~\ref{def:seq-term-subst-app}, let $\Psi''$ be the context
$(\Psi \setminus \context{\theta}) \cup \Psi'$  raised over $\supportof{\theta}$
and $\theta_r$ the corresponding raising substitution.
We must then show the following:
\begin{enumerate}
\item for each 
  $\ctxvarty{\Gamma}
            {\mathbb{N}_\Gamma}
            {\ctxty{\mathcal{C}}{\hsubst{\theta_r}{\hsubst{\theta}{\mathcal{G}}}}}
   \in\hsubst{\theta_r}{\hsubst{\theta}{\Xi}}$, 
  it is the case that
  $\mathbb{N}_{\Gamma}\subseteq\left(\mathbb{N}\cup\supportof{\theta}\right)$ and
  $\wfctxvarty{\left(\mathbb{N}\cup\supportof{\theta}\right)\setminus\mathbb{N}_{\Gamma}}
              {\Psi''}
              {\ctxty{\mathcal{C}}{\hsubst{\theta_r}{\hsubst{\theta}{\mathcal{G}}}}}$
  has a derivation, and
\item for each formula
      $\hsubst{\theta_r}{\hsubst{\theta}{F'}}\in\{\hsubst{\theta_r}{\hsubst{\theta}{F}}\}\cup\hsubst{\theta_r}{\hsubst{\theta}{\Omega}}$, 
      it is the case that the judgement
      $\wfform{\mathbb{N}\cup\supportof{\theta}\cup\STLCGamma_0\cup\Psi''}
              {\ctxsanstype{\hsubst{\theta_r}{\hsubst{\theta}{\Xi}}}}
              {\hsubst{\theta_r}{\hsubst{\theta}{F'}}}$
      has a derivation.
\end{enumerate}
In showing these two requirements we will also ensure that the relevant 
substitutions are well-defined.

We first show that requirement (1) holds.
For each
$\ctxvarty{\Gamma}
          {\mathbb{N}_{\Gamma}}
          {\ctxty{\mathcal{C}}{\hsubst{\theta_r}{\hsubst{\theta}{\mathcal{G}}}}}
   \in\hsubst{\theta_r}{\hsubst{\theta}{\Xi}}$
there must be 
$\ctxvarty{\Gamma}
          {\mathbb{N}_{\Gamma}}
          {\ctxty{\mathcal{C}}{\mathcal{G}}}\in\Xi$
which by the well-formedness of $\mathcal{S}$ is such that
$\mathbb{N}_{\Gamma}\subseteq\mathbb{N}$ and
$\wfctxvarty{\mathbb{N}\setminus\mathbb{N}_{\Gamma}}
            {\Psi}
            {\ctxty{\mathcal{C}}{\mathcal{G}}}$
has a derivation.
From the former it is obvious that the requirement 
$\mathbb{N}_{\Gamma}\subseteq\left(\mathbb{N}\cup\supportof{\theta}\right)$
is satisfied.
The substitution compatibility of $\seqsub{\theta}{\Psi'}$ with $\mathcal{S}$
entails that $\theta$ is type 
preserving with respect to $\supportof{\theta}\cup\STLCGamma_0\cup\Psi'$,
and that $\context{\theta}$ and $\Psi'\setminus\context{\theta}$ agree 
with $\Psi$ on the type assignments to the variables that are common to them.
Through an application of Theorem~\ref{th:ctx-ty-wk} there must be a derivation 
then for 
$\wfctxvarty{\mathbb{N}\setminus\mathbb{N}_{\Gamma}}
            {\aritysum{\context{\theta}}{((\Psi\setminus\context{\theta})\cup\Psi')}}
            {\ctxty{\mathcal{C}}{\mathcal{G}}}$.
Using Theorem~\ref{th:ctxtyinst-hsubst} we determine that
$\wfctxvarty{(\mathbb{N}\setminus\mathbb{N}_{\Gamma})\cup\supportof{\theta}}
            {(\Psi\setminus\context{\theta})\cup\Psi')}
            {\ctxty{\mathcal{C}}{\hsubst{\theta}{\mathcal{G}}}}$
has a derivation.
The substitution compatibility of $\seqsub{\theta}{\Psi'}$ with $\mathcal{S}$
further entails that $\supportof{\theta}$ is disjoint from $\mathbb{N}_{\Gamma}$, 
and therefore
$(\mathbb{N}\cup\supportof{\theta})\setminus\mathbb{N}_{\Gamma}=
    (\mathbb{N}\setminus\mathbb{N}_{\Gamma})\cup\supportof{\theta}$.
Since $\theta_r$ is clearly type preserving with respect to $\supportof{\theta}\cup\Psi''$,
by definition, $\context{\theta_r}$ and $\Psi''$ are disjoint, and
$\context{\theta_r}=(\Psi\setminus\context{\theta})\cup\Psi''$,
there is a derivation for
$\wfctxvarty{(\mathbb{N}\cup\supportof{\theta})\setminus\mathbb{N}_{\Gamma}}
            {\aritysum{\context{\theta_r}}{\Psi''}}
            {\ctxty{\mathcal{C}}{\hsubst{\theta}{\mathcal{G}}}}$
by Theorem~\ref{th:ctx-ty-wk}.
Noting that $\supportof{\theta_r}\subseteq{(\mathbb{N}\cup\supportof{\theta})\setminus\mathbb{N}_{\Gamma}}$,
an application of Theorem~\ref{th:ctxtyinst-hsubst} will determine that
$\wfctxvarty{(\mathbb{N}\cup\supportof{\theta})\setminus\mathbb{N}_{\Gamma}}
            {\Psi''}
            {\ctxty{\mathcal{C}}{\hsubst{\theta_r}{\hsubst{\theta}{\mathcal{G}}}}}$
must have a derivation.

We now show that requirement (2) holds.
It should be obvious that $\ctxsanstype{\Xi}$  is the same as 
$\ctxsanstype{\hsubst{\theta_r}{\hsubst{\theta}{\Xi}}}$ as these substitutions
do not change the context variables of the sequent.
For any formula 
$\hsubst{\theta_r}{\hsubst{\theta}{F'}}\in
    \{\hsubst{\theta_r}{\hsubst{\theta}{F}}\}
      \cup
    \hsubst{\theta_r}{\hsubst{\theta}{\Omega}}$,
the well-formedness of $\mathcal{S}$ ensure there must
exist a derivation for $\wfform{\mathbb{N}\cup\STLCGamma_0\cup\Psi}{\ctxsanstype{\Xi}}{F'}$.
It follows from the substitution compatibility of $\seqsub{\theta}{\Psi'}$ for
$\mathcal{S}$ that $\theta$ is arity type preserving with respect to the arity typing context
$\mathbb{N}\cup\supportof{\theta}\cup\STLCGamma_0\cup(\Psi\setminus\context{\theta})\cup\Psi'$.
Since $\context{\theta}$ and $\Psi$ agree on the type assignments
to common variables, we can conclude, using Theorem~\ref{th:wf-form-wk}, 
that there must be a derivation for the judgement
$\wfform{\aritysum{\context{\theta}}
          {\left(\mathbb{N}\cup\supportof{\theta}\cup\STLCGamma_0\cup(\Psi\setminus\context{\theta})\cup\Psi'\right)}}
         {\ctxsanstype{\Xi}}
         {F'}$.
It follows then, from Theorem~\ref{th:subst-formula}, that
$\hsubst{\theta}{F'}$ is well-defined and
there is a derivation for
$\wfform{\mathbb{N}\cup\supportof{\theta}\cup\STLCGamma_0\cup(\Psi\setminus\context{\theta})\cup\Psi'}
        {\Xi}
        {\hsubst{\theta}{F'}}$.
The raising substitution $\theta_r$ is type preserving with respect to 
$\supportof{\theta}\cup\Psi''$ by its construction and so is also type preserving 
with respect to $\mathbb{N}\cup\supportof{\theta}\cup\STLCGamma_0\cup\Psi''$.
Again using Theorem~\ref{th:wf-form-wk} we adjust the arity typing context 
from the formation judgement for $\hsubst{\theta}{F'}$ to 
the form $\aritysum{(\Psi\setminus\context{\theta}\cup\Psi')}
                   {(\mathbb{N}\cup\supportof{\theta}\cup\STLCGamma_0\cup\Psi'')}$
and noting that $\context{\theta_r}$ will be equal to
$\Psi\setminus\context{\theta}\cup\Psi'$, 
a second application of Theorem~\ref{th:subst-formula} will let us conclude 
$\hsubst{\theta_r}{\hsubst{\theta}{F'}}$ is well-defined and
$\wfform{\mathbb{N}\cup\supportof{\theta}\cup\STLCGamma_0\cup\Psi''}
              {\ctxsanstype{\Xi}}
              {\hsubst{\theta_r}{\hsubst{\theta}{F'}}}$
must have a derivation.
\end{proof}

We will also need to consider the application of substitutions for
context variables to sequents.
To be meaningfully applied, the context expressions being substituted
for the variables must be well-formed with respect to the types
associated with the variables.
In contrast to term variable substitutions, context variable
substitutions are not permitted to introduce new term variables into
the sequent and they may use nominal constants that are already present.
These notions are formalized below in the notion of appropriate substitutions.
\begin{definition}[Appropriate Context Variable Substitutions]\label{def:seq-ctxt-subst}
Let $\sigma$ be the context variable substitution
$\{G_1/\Gamma_1,\ldots,G_n/\Gamma_n\}$.
We write $\ctxvarminus{\Xi}{\sigma}$ to denote the context $\Xi$ trimmed so
as not to include context variables which are in the domain of $\sigma$.
We say that $\sigma$ is appropriate for a context variable context $\Xi$
with respect to an arity context $\Psi$ if, for $1 \leq i \leq n$, it is the case that
$\ctxvarty{\Gamma_i}{\mathbb{N}_i}{\ctxty{\mathcal{C}_i}{\mathcal{G}_i}} \in \Xi$ 
and 
$\ctxtyinst{\supportof{\sigma} \setminus \mathbb{N}_i}{\Psi}{\ctxvarminus{\Xi}{\sigma}}{\ctxty{\mathcal{C}_i}{\mathcal{G}_i}}{G_i}$
has a derivation.
The substitutions $\sigma$ is additionally said to cover $\Xi$ if $\ctxvarminus{\Xi}{\sigma}= \emptyset$.

Lifting this definition to sequents we say that $\sigma$ is appropriate for a well-formed
sequent $\mathcal{S}=\seq[\mathbb{N}]{\Psi}{\Xi}{\Omega}{F}$ if
it is appropriate for $\Xi$ with respect to $\Psi$.
\end{definition}

Context types are evidently unaffected by context variable substitutions.
Context expressions are impacted by such substitutions but, for the
right kind of substitution, they continue to be instances of relevant
context types.
This is made precise in the theorem below.

\begin{theorem}\label{th:ctxtyinst-subst}
Let $\sigma$ be a context variable substitution which is appropriate for 
$\Xi$ with respect to $\Psi$ and such that $\supportof{\sigma}\subseteq\mathbb{N}$.
If 
$\ctxtyinst{\mathbb{N}}{\Psi}{\Xi}{\ctxty{\mathcal{C}}{\mathcal{G}}}{G}$
has a derivation, then 
$\ctxtyinst{\mathbb{N}}{\Psi}{\ctxvarminus{\Xi}{\sigma}}{\ctxty{\mathcal{C}}{\mathcal{G}}}{\subst{\sigma}{G}}$
also has a derivation.
\end{theorem}
\begin{proof}
This proof is by induction on the derivation of 
$\ctxtyinst{\mathbb{N}}{\Psi}{\Xi}{\ctxty{\mathcal{C}}{\mathcal{G}}}{G}$.
When $G$ is empty the result is obvious, and when it is of the form $(G_1,G_2)$ it
follows from application of the inductive hypothesis.
When $G$ is a context variable $\Gamma$ then there must exist some
$\ctxvarty{\Gamma}{\mathbb{N}_\Gamma}{\ctxty{\mathcal{C}_\Gamma}{\mathcal{G}_\Gamma}}\in\Xi$.
Further, it must be that either
$\subst{\sigma}{\Gamma}=\Gamma$ and
$\ctxvarty{\Gamma}{\mathbb{N}_\Gamma}{\ctxty{\mathcal{C}_\Gamma}{\mathcal{G}_\Gamma}}\in\ctxvarminus{\Xi}{\sigma}$, or
$\subst{\sigma}{\Gamma}=G_i$ and,
by the appropriateness of $\sigma$,
$\ctxtyinst{\supportof{\sigma}\setminus\mathbb{N}_\Gamma}
           {\Psi}
           {\ctxvarminus{\Xi}{\sigma}}
           {\ctxty{\mathcal{C}_\Gamma}{\mathcal{G}_\Gamma}}
           {G_i}$
has a derivation.
Since $\supportof{\sigma}\subseteq\mathbb{N}$, we can infer that
$\supportof{\sigma}\setminus\mathbb{N}_\Gamma\subseteq\mathbb{N}\setminus\mathbb{N}_\Gamma$
and so by Theorem~\ref{th:ctx-ty-wk} there must be a derivation of
$\ctxtyinst{\mathbb{N}}{\Psi}{\ctxvarminus{\Xi}{\sigma}}{\ctxty{\mathcal{C}}{\mathcal{G}}}{\subst{\sigma}{\Gamma}}$
for any such $\Gamma$.
\end{proof}

As with the term substitutions the application of context substitutions 
may introduce new nominal constants, and we use the technique of raising to
permit these constants in substitutions for eigenvariables in the resulting sequent.
We formalize the application of a context variable substitution to a well-formed 
sequent when the conditions of appropriateness are met in the following definition.
We assume here and elsewhere that the application of a substitution to a set of 
formulas distributes to each member of the set.

\begin{definition}[Applying a Context Substitution to a Sequent]\label{def:seq-ctxt-subst-app}
Let $\mathcal{S}$ be a well-formed sequent $\seq[\mathbb{N}]{\Psi}{\Xi}{\Omega}{F}$
and $\sigma$ be appropriate for $\mathcal{S}$.
Further let $\Psi'$ be a version of $\Psi$ raised over $\supportof{\sigma}$
and let $\theta_r$ be the corresponding raising substitution.
Then the application of $\sigma$ to $\mathcal{S}$, 
denoted by $\subst{\sigma}{\mathcal{S}}$, is
the sequent
$\seq[\mathbb{N}\cup\supportof{\sigma}]
     {\Psi'}
     {\hsubst{\theta_r}{\ctxvarminus{\Xi}{\sigma}}}
     {\hsubst{\theta_r}{\subst{\sigma}{\Omega}}}
     {\hsubst{\theta_r}{\subst{\sigma}{F}}}$.
\end{definition}

The following theorem is the counterpart of
Theorem~\ref{th:seq-term-subs-ok} for context variable substitutions. 

\begin{theorem}\label{th:seq-ctx-var-subst-okay}
Let $\mathcal{S}$ be a well-formed sequent and let $\sigma$ be a
context variable substitution that is appropriate for
$\mathcal{S}$. Then $\subst{\sigma}{\mathcal{S}}$ is well-defined and
is a well-formed sequent.
\end{theorem}
\begin{proof}
Suppose that $\mathcal{S}=\seq[\mathbb{N}]{\Psi}{\Xi}{\Omega}{F}$ is a 
well-formed sequent and $\sigma$ an appropriate context variable substitution
for $\mathcal{S}$.
Following Definition~\ref{def:seq-ctxt-subst-app}, let $\Xi'$ be the
context variable context $\Xi$ without the context variables substituted
for by $\sigma$, $\Psi'$ be a version of 
$\Psi$ raised over $\supportof{\sigma}\setminus\mathbb{N}$,
and $\theta_r$ the corresponding raising substitution.
We must then show the following:
\begin{enumerate}
\item for each 
  $\ctxvarty{\Gamma}
            {\mathbb{N}_\Gamma}
            {\ctxty{\mathcal{C}}{\hsubst{\theta_r}{\mathcal{G}}}}
   \in\hsubst{\theta_r}{\Xi'}$, 
  it is the case that
  $\mathbb{N}_{\Gamma}\subseteq\left(\mathbb{N}\cup\supportof{\sigma}\right)$ and
  $\wfctxvarty{\left(\mathbb{N}\cup\supportof{\sigma}\right)\setminus\mathbb{N}_{\Gamma}}
              {\Psi'}
              {\ctxty{\mathcal{C}}{\hsubst{\theta_r}{\mathcal{G}}}}$
  has a derivation, and
\item for each formula
      $\hsubst{\theta_r}{\subst{\sigma}{F'}}\in\{\hsubst{\theta_r}{\subst{\sigma}{F}}\}\cup\hsubst{\theta_r}{\subst{\sigma}{\Omega}}$, 
      it is the case that the judgement
      $\wfform{\mathbb{N}\cup\supportof{\sigma}\cup\STLCGamma_0\cup\Psi'}
              {\ctxsanstype{\hsubst{\theta_r}{\Xi'}}}
              {\hsubst{\theta_r}{\subst{\sigma}{F'}}}$
      has a derivation.
\end{enumerate}
In showing these two requirements we will also ensure that the relevant 
substitutions are well-defined.

Using an argument similar to that in the proof of Theorem~\ref{th:seq-term-subs-ok},
we can show that 
$\wfctxvarty{\left(\mathbb{N}\cup\supportof{\sigma}\right)\setminus\mathbb{N}_{\Gamma}}
              {\Psi'}
              {\ctxty{\mathcal{C}}{\hsubst{\theta_r}{\mathcal{G}}}}$
has a derivation for each
$\ctxvarty{\Gamma}
            {\mathbb{N}_\Gamma}
            {\ctxty{\mathcal{C}}{\hsubst{\theta_r}{\mathcal{G}}}}
   \in\hsubst{\theta_r}{\Xi'}$.

We now show that requirement (2) holds.
Note that $\ctxsanstype{\hsubst{\theta_r}{\Xi'}}$ is the same
collection as $\ctxsanstype{\Xi'}$.
For any formula 
$\hsubst{\theta_r}{\subst{\sigma}{F'}}\in\{\hsubst{\theta_r}{\subst{\sigma}{F}}\}\cup\hsubst{\theta_r}{\subst{\sigma}{\Omega}}$
it must be that $\wfform{\mathbb{N}\cup\STLCGamma_0\cup\Psi}{\ctxsanstype{\Xi}}{F'}$
has a derivation by the well-formedness of $\mathcal{S}$.
Thus by Theorem~\ref{th:wf-form-wk} we can conclude that 
$\wfform{\mathbb{N}\cup\supportof{\sigma}\cup\STLCGamma_0\cup\Psi}{\ctxsanstype{\Xi}}{F'}$
must be derivable.
The appropriateness of $\sigma$ for $\mathcal{S}$ and an application of
Theorem~\ref{th:ctx-var-type-instance} ensure that for all $G/\Gamma\in\sigma$,
there is a derivation of
$\wfctx{\mathbb{N}\cup\supportof{\sigma}\cup\STLCGamma_0\cup\Psi}{\ctxsanstype{\Xi'}}{G}$.
We can then conclude from an application of Theorem~\ref{th:subst-formula} that
there must be a derivation of
$\wfform{\mathbb{N}\cup\supportof{\sigma}\cup\STLCGamma_0\cup\Psi}{\Xi'^-}{\subst{\sigma}{F'}}$.
Again extending the arity typing context, using Theorem~\ref{th:wf-form-wk} 
there must be a derivation for 
$\wfform{\aritysum{\Psi}{(\mathbb{N}\cup\supportof{\sigma}\cup\STLCGamma_0\cup\Psi')}}{\Xi'^-}{\subst{\sigma}{F'}}$.
Recalling that $\theta_r$ is arity type preserving with respect to 
$(\supportof{\sigma}\setminus\mathbb{N})\cup\Psi'$, and thus with respect to 
$\mathbb{N}\cup\supportof{\sigma}\cup\STLCGamma_0\cup\Psi'$, and that
$\context{\theta_r}=\Psi$,
Theorem~\ref{th:subst-formula} allows us conclude that 
$\hsubst{\theta_r}{\subst{\sigma}{F'}}$ is well-defined and that
$\wfform{\mathbb{N}\cup\supportof{\sigma}\cup\STLCGamma_0\cup\Psi'}
        {\Xi'^-}
        {\hsubst{\theta_r}{\subst{\sigma}{F'}}}$
has a derivation
\end{proof}

We are now in a position to define validity for sequents.
For a sequent containing term and context variables, this is done by
considering all their relevant substitution instances.
For a sequent devoid of variables, we base the definition on the
validity of closed formulas.

\begin{definition}\label{def:seq-validity}
A well-formed sequent of the form
$\seq[\mathbb{N}]{\emptyset}{\emptyset}{\Omega}{F}$
is valid if whenever all the formulas in $\Omega$ are valid, $F$ is a
valid formula; well-formed sequents of this form are referred to as
\emph{closed} sequents.
A well-formed sequent $\mathcal{S}$ of the form
$\seq[\mathbb{N}]{\Psi}{\Xi}{\Omega}{F}$ is valid if
for every term substitution $\theta$ that is type preserving with
respect to $\noms \cup \Theta_0$ and such that
$\Psi = \context{\theta}$ and
$\langle \theta, \emptyset \rangle$ is 
substitution compatible with $\mathcal{S}$, and for every context
substitution $\sigma$ that is appropriate and closed for 
$\hsubstseq{\emptyset}{\theta}{\mathcal{S}}$, it is the case that
$\subst{\sigma}{\hsubstseq{\emptyset}{\theta}{\mathcal{S}}}$ is valid.
Note that each such
$\subst{\sigma}{\hsubstseq{\emptyset}{\theta}{\mathcal{S}}}$
will be a well-formed and closed sequent in these circumstances and we 
shall refer to it as the closed instance of $\mathcal{S}$ identified
by $\theta$ and $\sigma$. 
\end{definition}

The following theorem, whose proof is obvious, provides the basis for
using our proof system for determining the validity of formulas.

\begin{theorem}\label{seq-fmla-validity}
Let $F$ be a formula such that
$\wfform{\noms \cup \Theta_0}{\emptyset}{F}$ is derivable and let
$\mathbb{N}$ be the set of nominal constants that appear in $F$.
Then the sequent
$\seq[\mathbb{N}]{\emptyset}{\emptyset}{\emptyset}{F}$ is
well-formed.
Moreover, $F$ is valid if and only if
$\seq[\mathbb{N}]{\emptyset}{\emptyset}{\emptyset}{F}$ is.
\end{theorem} 

In Section~\ref{sec:permutations}, we had noted an invariance of
validity for formulas under permutations of nominal constants.
We observe an analogous property concerning sequents.
We first explain what it means to apply a permutation to a sequent.

\begin{definition}[Applying Permutations of Nominal Constants to Sequents]\label{def:perm-seq}
The application of a permutation $\pi$ to a context variable type distributes to the
constituent block instances and the application to a collection of
formulas or a context variable context distributes to the members of
the collection. 
The application to a sequent
$\seq[\mathbb{N}]{\Psi}{\Xi}{\Omega}{F}$ 
yields the sequent
$\seq[\mathbb{N}']{\Psi}{\permute{\pi}{\Xi'}}{\permute{\pi}{\Omega'}}{\permute{\pi}{F'}}$
where $\mathbb{N}'=\{n' |\ n\in\mathbb{N}\ \mbox{\rm and}\ \pi(n)=n'\}$.
\end{definition}

Some useful results about the interaction of permutations with sequents are
given below.
Theorems~\ref{th:perm-hsubst} and \ref{th:perm-subst}
consider the composition of permutations and substitutions.
Their proofs are straightforward, relying on the invariance of the 
derivability of judgements under permutation which can be verified
easily by induction on such derivations.
Theorem~\ref{th:perm-valid} is the analogous property of Theorem~\ref{th:perm-form}, 
and it's proof is straightforward given that result and the definition of validity.

\begin{theorem}\label{th:perm-hsubst}
Let $\pi$ be a permutation of the nominal constants, let $\mathcal{S}$ be
a well-formed sequent and let $\langle \theta, \Psi \rangle$ be
substitution compatible with $\permute{\pi}{\mathcal{S}}$.
Then $\langle \permute{\inv{\pi}}{\theta}, \Psi \rangle$ is substitution compatible with
$\mathcal{S}$ and
$\hsubstseq{\Psi}{\theta}{(\permute{\pi}{\mathcal{S}})} =
\permute{\pi}{(\hsubstseq{\Psi}{\permute{\inv{\pi}}{\theta}}{\mathcal{S}})}$. 
\end{theorem}

\begin{theorem}\label{th:perm-subst}
Let $\pi$ be a permutation of the nominal constants, let $\mathcal{S}$
be a well-formed sequent and let $\sigma$ be a context variable
substitution that is appropriate for $\permute{\pi}{\mathcal{S}}$.
Then $\inv{\pi}{\sigma}$ is appropriate for $\mathcal{S}$ and
$\subst{\sigma}{(\permute{\pi}{\mathcal{S}})} =
\permute{\pi}{(\subst{\permute{\inv{\pi}}{\sigma}}{\mathcal{S}})}$.
\end{theorem}

\begin{theorem}\label{th:perm-valid}
If $\pi$ is a permutation of the nominal constants and $\mathcal{S}$
is a well-formed closed sequent that is valid, then
$\permute{\pi}{\mathcal{S}}$ is also a closed, valid sequent.
\end{theorem}

%% file: proof-system/core-logic.tex
\section{The Core Proof Rules}
\label{sec:core-logic}
We consider in this section a collection of proof rules that
internalize the interpretation of the logical symbols that may appear
in formulas and also some properties that flow from the meanings of
sequents.
More specifically, the first subsection below presents some structural
rules, the second subsections identifies axioms and the cut rule, and
the third subsection introduces rules for the logical symbols.
Beyond presenting the rules, we are interested in showing that the
rules are sound and that in their context we may limit our attention
to only well-formed sequents in trying to construct a derivation for a
well-formed sequent. 
The latter property is verified by showing that the premise sequents
of each rule are well-formed if the conclusion sequent is.

\input{proof-system/structural-rules}
\input{proof-system/other-rules}
\input{proof-system/logical-rules}

%% file: proof-system/structural-rules.tex
\subsection{Structural Rules}

This subcollection of rules is presented in
Figure~\ref{fig:rules-structural}.
These rules can be subcategorized into those
that allow for weakening and contracting the assumption set in a
sequent and those that permit the weakening and strengthening of the
support set, the eigenvariable context, and the context variable context.
Rules of the second subcategory encode the fact that vacuous
(well-formed) extensions to the bindings manifest in a sequent will
not impact its validity.
The strengthening and weakening rules for contexts include premises
that force modifications to context variable types and the
satisfaction of typing judgements that are necessary to ensure the 
well-formedness of the sequents in any 
application of the rule. 

\begin{figure}
\[
\begin{array}{cc}
\infer[\weakening]
      {\seq[\mathbb{N}]{\Psi}{\Xi}{\setand{\Omega}{F_2}}{F_1}}
      {\seq[\mathbb{N}]{\Psi}{\Xi}{\Omega}{F_1}}
&
\infer[\contraction]
      {\seq[\mathbb{N}]{\Psi}{\Xi}{\setand{\Omega}{F_2}}{F_1}}
      {\seq[\mathbb{N}]{\Psi}{\Xi}{\setand{\Omega}{F_2, F_2}}{F_1}}
\end{array}
\]
\smallskip
\[
\infer[\sstr]
      {\seq[\mathbb{N}]{\Psi}{\Xi}{\Omega}{F}}
      {\begin{array}{c}
         \Xi=\{
         \ctxvarty{\Gamma_i}
                  {(\mathbb{N}_i\setminus\mathbb{N}')}
                  {\ctxty{\mathcal{C}_i}{\mathcal{G}_i}}
               \ |\ 
               \ctxvarty{\Gamma_i}
                        {\mathbb{N}_i}
                        {\ctxty{\mathcal{C}_i}{\mathcal{G}_i}}\in\overline{\Xi}
             \}
         \\
         \left\{\wfctxvarty{(\mathbb{N},\mathbb{N}')\setminus\mathbb{N}_i}
                           {(\Psi,\Psi')}
                           {\ctxty{\mathcal{C}_i}{\mathcal{G}_i}}\ \middle|\ 
                 \ctxvarty{\Gamma_i}
                          {\mathbb{N}_i} 
                          {\ctxty{\mathcal{C}_i}{\mathcal{G}_i}}\in\Xi'\right\}
         \\
         \seq[\mathbb{N},\mathbb{N}']{\Psi,\Psi'}{\overline{\Xi},\Xi'}{\Omega}{F}
       \end{array}} 
\]
\smallskip
\[
\infer[\sweak]
      {\seq[\mathbb{N},\mathbb{N}']{\Psi,\Psi'}{\overline{\Xi},\Xi'}{\Omega}{F}}
      {\begin{array}{c}
         \Xi=\{
               \ctxvarty{\Gamma_i}{(\mathbb{N}_i\setminus\mathbb{N}')}{\ctxty{\mathcal{C}_i}{\mathcal{G}_i}}
               \ |\ 
               \ctxvarty{\Gamma_i}
                        {\mathbb{N}_i}
                        {\ctxty{\mathcal{C}_i}{\mathcal{G}_i}}\in\overline{\Xi}
             \}
         \\
         \left\{ \wfctxvarty{(\mathbb{N}\setminus\mathbb{N}_i)}
                            {\Psi}
                            {\ctxty{\mathcal{C}_i}{\mathcal{G}_i}}
                      \ \middle|\ 
                     \ctxvarty{\Gamma_i}
                              {\mathbb{N}_i}
                              {\ctxty{\mathcal{C}_i}{\mathcal{G}_i}}\in\Xi \right\}                            
         \\
         \left\{\wfform{\mathbb{N}\cup\STLCGamma_0\cup\Psi}
                       {\{\Gamma\ |\ \ctxvarty{\Gamma}
                                       {\mathbb{N}}
                                       {\ctxty{\mathcal{C}}
                                              {\mathcal{G}}}\in\Xi\}}
                       {F'}
                  \ \middle|\ 
                  F'\in\Omega\cup\{F\} \right\}
         \\
         \seq[\mathbb{N}]{\Psi}{\Xi}{\Omega}{F}
       \end{array}}
\]
\caption{The Structural Rules}\label{fig:rules-structural}
\end{figure}

The following theorem shows that these rules require the proof of only
well-formed sequents in constructing a proof of a well-formed sequent.

\begin{theorem}\label{th:structural-wf}
The following property holds for each rule in
Figure~\ref{fig:rules-structural}: if the conclusion sequent is
well-formed, the premises expressing typing conditions have
derivations and the conditions expressed by the other, non-sequent
premises are satisfied, then all the sequent premises must
be well-formed.
\end{theorem}

\begin{proof}
We consider each rule described in Figure~\ref{fig:rules-structural}.

\case{\weakening}
For a well-formed conclusion sequent 
$\seq[\mathbb{N}]{\Psi}{\Xi}{\setand{\Omega}{F_2}}{F_1}$
we must show that the premise sequent $\seq[\mathbb{N}]{\Psi}{\Xi}{\Omega}{F_1}$ 
is also well-formed.
Since the context variable contexts are the same in both sequents, the 
goal formulas are the same, and 
$\Omega\subseteq\Omega\cup\{F_2\}$ we can easy determine the well-formedness of
the premise by applying the definition of well-formedness to the conclusion sequent.

\case{\contraction}
Suppose that the sequent
$\seq[\mathbb{N}]{\Psi}{\Xi}{\setand{\Omega}{F_2}}{F_1}$
is well-formed.
Noting that the only addition to the premise sequent is a copy of the
well-formed formula $F_2$ to the assumption set, and otherwise the 
two sequents are identical, it is clear that the 
premise sequent must also be valid.

\case{\sstr}
Suppose that a sequent $\seq[\mathbb{N}]{\Psi}{\Xi}{\Omega}{F}$ is well-formed,
the context variable context $\Xi$ is of the form described in the rule,
and for each
$\ctxvarty{\Gamma_i}{\mathbb{N}_i}{\ctxty{\mathcal{C}_i}{\mathcal{G}_i}}\in\Xi'$ 
the judgement
$\wfctxvarty{\left((\mathbb{N},\mathbb{N}')\setminus\mathbb{N}_i\right)}
            {(\Psi,\Psi')}
            {\ctxty{\mathcal{C}_i}{\mathcal{G}_i}}$ 
has a derivation.
By Theorem~\ref{th:wf-form-wk} it is clear that the well-formedness of
the formulas in $\Omega\cup\{F\}$ is preserved by the extensions to the
support set, arity typing context, and context variable context.
The premises of this rule ensure that the context variable types in $\Xi'$
are well-formed, and so it only remains to conclude those in
$\overline{\Xi}$ are as well.
For each 
$\ctxvarty{\Gamma_i}
          {\mathbb{N}_i}
          {\ctxty{\mathcal{C}_i}{\mathcal{G}_i}}\in\overline{\Xi}$
there is a 
$\ctxvarty{\Gamma_i}
          {(\mathbb{N}_i\setminus\mathbb{N}')}
          {\ctxty{\mathcal{C}_i}{\mathcal{G}_i}}\in\Xi$
for which there is a derivation of
$\wfctxvarty{\left(\mathbb{N}\setminus(\mathbb{N}_i\setminus\mathbb{N}')\right)}
            {\Psi}
            {\ctxty{\mathcal{C}_i}{\mathcal{G}_i}}$
by the well-formedness of the conclusion sequent.
So by Theorem~\ref{th:ctx-ty-wk}, 
$\wfctxvarty{\left(\mathbb{N}\setminus\mathbb{N}_i\right)}
            {\Psi,\Psi'}
            {\ctxty{\mathcal{C}_i}{\mathcal{G}_i}}$
is derivable for each entry in $\overline{\Xi}$.
Therefore the premise sequent must be well-formed.

\case{\sweak}
The well-formedness of $\seq[\mathbb{N}]{\Psi}{\Xi}{\Omega}{F}$ 
is obvious given the sets of well-formedness derivations in the premises.
\end{proof}

The argument for soundness has an intuitively obvious structure in all
the cases other than when the support set is affected.
In the case when the support set is expanded, the reasoning is still
straightforward and is based on observing that the instances of the
weakened form of the sequent will be a subset of the instances of the 
premise sequent.
When the support set is smaller in the conclusion sequent, the
argument is a little more subtle: we must use the fact that
permutations of nominal constants that do not appear in the context
types or the formulas in a sequent do not impact on validity.
This observation is embedded in the following lemma.

\begin{lemma}\label{lem:seq-equiv}
Let $\seq[\mathbb{N}]{\Psi}{\Xi}{\Omega}{F}$ and 
$\seq[\mathbb{N}']{\Psi'}{\Xi'}{\Omega}{F}$ be well-formed sequents such that
$\mathbb{N}'\subseteq\mathbb{N}$, $\Psi'\subseteq\Psi$, and there is some
subset $\overline{\Xi}$ for the context variable context $\Xi$ where
$\Xi'=
  \left\{
    \ctxvarty{\Gamma_i}
             {\left(\mathbb{N}_i\setminus(\mathbb{N}\setminus\mathbb{N}')\right)}
             {\ctxty{\mathcal{C}_i}{\mathcal{G}_i}}
      \ \middle|\ 
    \ctxvarty{\Gamma_i}{\mathbb{N}_i}{\ctxty{\mathcal{C}_i}{\mathcal{G}_i}}\in\overline{\Xi}
  \right\}$.
Then the sequent $\seq[\mathbb{N}]{\Psi}{\Xi}{\Omega}{F}$ will be valid if and only if
$\seq[\mathbb{N}']{\Psi'}{\Xi'}{\Omega}{F}$ is valid.
\end{lemma}
\begin{proof}
Let $\mathcal{S}$ denote the sequent $\seq[\mathbb{N}]{\Psi}{\Xi}{\Omega}{F}$ and
$\mathcal{S}'$ the sequent $\seq[\mathbb{N}']{\Psi'}{\Xi'}{\Omega}{F}$.
Given that both $\mathcal{S}$ and $\mathcal{S}'$ are well-formed sequents,
no context variables in $(\Xi\setminus\Xi')$ appear in the formula $F$ or the
formulas in $\Omega$, no variables in $(\Psi\setminus\Psi')$ appear in the
context variable context $\Xi'$, the formula $F$, or the formulas in $\Omega$, 
and no nominal constants in $(\mathbb{N}\setminus\mathbb{N}')$ appear in the
context variable context $\Xi'$, the formula $F$, or the formulas in $\Omega$.

We consider each direction of the implication.

\case{($\Rightarrow$)}
Let $\theta$ and $\sigma$ identify an arbitrary closed instance of the sequent
$\mathcal{S}'$.
Then $\seqsub{\theta}{\emptyset}$ is substitution compatible for $\mathcal{S}'$ and
$\sigma$ is appropriate for $\hsubstseq{\emptyset}{\theta}{\mathcal{S}'}$.
So $\supp{\theta}\cap\mathbb{N}'=\emptyset$, but it is possible that 
$\supp{\theta}\cap\mathbb{N}$ is not empty.
Let $\pi$ be a permutation which maps nominal constants in 
$(\mathbb{N}\setminus\mathbb{N}')$ to some new names not appearing in 
$\mathbb{N}$, $\theta$, or $\sigma$.
Then $\seqsub{\permute{\pi}{\theta}}{(\Psi\setminus\Psi')}$ will be
substitution compatible for the sequent $\mathcal{S}$.
And since the restricted names for context variables in $\overline{\Xi}$ 
only differ from the restricted sets of names from $\Xi'$ by nominal constants
appearing in $\mathbb{N}\setminus\mathbb{N}'$ the substitution 
$\permute{\pi}{\sigma}$ will be appropriate for 
$\hsubstseq{(\Psi\setminus\Psi')}{\permute{\pi}{\theta}}{\mathcal{S}}$.

So consider the sequent 
$\subst{\permute{\pi}{\sigma}}{\hsubstseq{(\Psi\setminus\Psi')}{\permute{\pi}{\theta}}{\mathcal{S}}}$.
Let $\theta'$ and $\sigma'$ identify an arbitrary closed instance of this sequent.
This closed sequent must be the same as the closed instance of $\mathcal{S}$ identified by 
$(\theta'\circ\permute{\pi}{\theta})$ and $(\sigma'\circ\permute{\pi}{\sigma})$
given that $\theta'$, $\theta$, $\sigma'$, and $\sigma$ are all closed substitutions.
Since $\mathcal{S}$ is valid by assumption, this particular closed instance identified by 
$(\theta'\circ\permute{\pi}{\theta})$ and $(\sigma'\circ\permute{\pi}{\sigma})$
will thus be valid.
But we know that no context variables in $\Xi\setminus\Xi'$ or variables in
$\Psi\setminus\Psi'$ appear in $F$ or any formula in $\Omega$.
Further, both $\theta$ and $\sigma$ are closed substitutions and so these variables
also cannot appear in $\subst{\permute{\pi}{\sigma}}{\hsubst{\permute{\pi}{\theta}}{F}}$
or any formula in $\subst{\permute{\pi}{\sigma}}{\hsubst{\permute{\pi}{\theta}}{\Omega}}$.
Thus the closed sequent
$\subst{\sigma'}{\hsubstseq{\emptyset}{\theta'}{(\subst{\permute{\pi}{\sigma}}{\hsubstseq{(\Psi\setminus\Psi')}{\permute{\pi}{\theta}}{\mathcal{S}}})}}$
is in fact equivalent to the closed sequent
$\subst{\permute{\pi}{\sigma}}{\hsubstseq{(\Psi\setminus\Psi')}{\permute{\pi}{\theta}}{\mathcal{S}'}}$
which must therefore be a valid closed sequent.
Since $\pi$ was constructed such that it does not permute any of the nominal constants
in $\mathbb{N}'$ this is the same as the closed sequent
$\permute{\pi}{(\subst{\sigma}{\hsubst{\theta}{\mathcal{S}'}})}$.
By Theorem~\ref{th:perm-valid}, validity of closed sequents is preserved by 
permutations and thus we can conclude that $\subst{\sigma}{\hsubst{\theta}{\mathcal{S}'}}$ 
is valid.

Since all closed instances of $\mathcal{S}'$ must therefore be valid, this sequent
is valid.

\case{($\Leftarrow$)}
Let $\theta$ and $\sigma$ identify an arbitrary closed instance of the sequent
$\mathcal{S}$.
Then $\seqsub{\theta}{\emptyset}$ is substitution compatible with $\mathcal{S}$ and
$\sigma$ is appropriate for $\hsubstseq{\emptyset}{\theta}{\mathcal{S}}$.
If some formula in $\subst{\sigma}{\hsubst{\theta}{\Omega}}$ is not valid
then $\subst{\sigma}{\hsubstseq{\emptyset}{\theta}{\mathcal{S}}}$ is vacuously valid.
So suppose instead that all the formulas in 
$\subst{\sigma}{\hsubst{\theta}{\Omega}}$ are valid.

Given that $\mathbb{N}'\subseteq\mathbb{N}$ and $\Psi'\subseteq\Psi$ this
$\seqsub{\theta}{\emptyset}$ is also substitution compatible with $\mathcal{S}'$.
Further, because $\dom{\Xi'}\subseteq\dom{\Xi}$ and the annotation sets on
variables in $\Xi'$ will all be subsets of the annotation set for that 
context variable in $\Xi$ the substitution $\sigma$ will be appropriate
for $\hsubstseq{\emptyset}{\theta}{\mathcal{S}'}$.
Given the validity of $\mathcal{S}'$, the closed instance
$\subst{\sigma}{\hsubstseq{\emptyset}{\theta}{\mathcal{S}'}}$ must be valid.
Given the assumption that all formulas in 
$\subst{\sigma}{\hsubst{\theta}{\Omega}}$
are valid, the validity of this closed sequent means
$\subst{\sigma}{\hsubst{\theta}{F}}$ must be valid.
Thus $\subst{\sigma}{\hsubstseq{\emptyset}{\theta}{\mathcal{S}}}$ is valid in this case as well.

Since all closed instances of the sequent $\mathcal{S}$ are valid, it is a valid sequent.
\end{proof}

We may now establish the soundness of the structural rules.

\begin{theorem}\label{th:structural-sound}
The following property holds for every instance of each of the rules
in Figure~\ref{fig:rules-structural}: if the premises expressing
typing judgements are derivable, the conditions described in the other
non-sequent premises are satisfied and the premise sequent is valid,
then the conclusion sequent must also be valid. 
\end{theorem}

\begin{proof}
We consider each rule described in Figure~\ref{fig:rules-structural}.

\case{\weakening}
Given the structure of the sequents, it is clear that any substitutions
identifying a closed instance of the conclusion sequent will also identify
a closed instance of the premise sequent.
Since $\Omega\subseteq\Omega\cup\{F_2\}$, it is also clear that
for any substitutions $\theta$ and $\sigma$ identifying a closed 
instance of the conclusion sequent, the formula 
$\subst{\sigma}{\hsubst{\theta}{F_1}}$ must be valid whenever
every formula in 
$\subst{\sigma}{\hsubst{\theta}{\Omega}}\cup\{\subst{\sigma}{\hsubst{\theta}{F_2}}\}$
is valid.
But then every closed instance of the sequent
$\seq[\mathbb{N}]{\Psi}{\Xi}{\setand{\Omega}{F_2}}{F_1}$
is valid, and thus the sequent itself is valid.

\case{\contraction}
It is clear from the structure of the sequents that any $\theta$ and $\sigma$ identifying
a closed instance of the conclusion sequent will also identify a closed 
instance of the premise sequent.
Furthermore, the collection of formulas 
$\setand{\subst{\sigma}{\hsubst{\theta}{\Omega}}}
        {\subst{\sigma}{\hsubst{\theta}{F_2}},
           \subst{\sigma}{\hsubst{\theta}{F_2}}}$
are clearly all valid whenever the collection of formulas
$\setand{\subst{\sigma}{\hsubst{\theta}{\Omega}}}{\subst{\sigma}{\hsubst{\theta}{F_2}}}$
are all valid.
Therefore from the validity of the premise sequent we can conclude that the
conclusion sequent is also valid as every closed instance of this sequent
must be valid.

\case{\sstr\ or \sweak}
Both cases are resolved through an application of Lemma~\ref{lem:seq-equiv}
given the well-formedness of the conclusion and premise sequents.
\end{proof}

%% file: proof-system/other-rules.tex
\subsection{The Axiom and the Cut Rule}

The two rules of interest here are also related to the interpretation
of sequents but focus more specifically on the logical relationship
between the formula collections.
The \cut\ rule codifies the notion of lemmas: if we can show the
validity of a formula relative to a given assumption set, then this
formula can be included in the assumptions to simplify the
reasoning process.
The \id\ rule recognizes the validity of a sequent in which the
conclusion formula appears in the assumption set.

In its simplest form, the \id\ rule would require the conclusion
formula to be included as is in the assumption set, possibly with a
renaming of the bound variables that appear in it.
It is possible, and also pragmatically useful, to generalize this form
to allow also for a permutation of nominal constants in the formulas
in the process of matching.
However, this has to be done with care to ensure that identity under
the considered permutations continues to hold even after later
instantiations of term and context variables appearing in the
formulas.
The specific form of equivalence for formulas under permutations that
we will use is the content of the following definition.

\begin{definition}[Formula Equivalence]\label{def:form-eq}
The equivalence of two context expressions $G_1$ and $G_2$ with respect to a context
variable context $\Xi$ and a permutation $\pi$, written $\ctxeq{\Xi}{\pi}{G_2}{G_1}$,
is a relation defined by the following three clauses:
\begin{enumerate}
\item $\ctxeq{\Xi}{\pi}{\emptyctx}{\emptyctx}$ holds for any $\Xi$ and $\pi$.

\item $\ctxeq{\Xi}{\pi}{\Gamma}{\Gamma}$ holds if 
$\ctxvarty{\Gamma_i}{\mathbb{N}_i}{\ctxty{\mathcal{C}}{\mathcal{G}}}\in\Xi$ 
and $\supp{\pi}\subseteq\mathbb{N}_i$.

\item If $G_1=(G_1',n_1:A_1)$ and $G_2=(G_2',n_2:A_2)$ then 
$\ctxeq{\Xi}{\pi}{G_2}{G_1}$ holds if
$\permute{\pi}{n_2}$ is identical to $n_1$, $\permute{\pi}{A_2}$ is
  identical to $A_1$ (up to a renaming of bound variables), and
  $\ctxeq{\Xi}{\pi}{G_2}{G_1}$ holds.
\end{enumerate}
Two atomic formulas $\fatm{G'}{M':A'}$ and $\fatm{G}{M:A}$ are
considered equivalent with respect to $\Xi$ and $\pi$ if  $G'$ and $G$
are equivalent with respect to this $\Xi$  and $\pi$ and
$\permute{\pi}{M'}$ and $M$ and $\permute{\pi}{A'}$ and $A$ are
respectively identical up to a renaming of bound variables. 
Two arbitrary formulas are considered equivalent with respect to $\Xi$
and $\pi$ if their component parts are so equivalent, allowing, of
course, for a renaming of variables bound by quantifiers.
Equivalence of formulas is represented by the judgement
$\formeq{\Xi}{\pi}{F'}{F}$.
\end{definition}

\begin{figure}
\[
\infer[\id]
      {\seq[\mathbb{N}]
           {\Psi}
           {\Xi}
           {\Omega}
           {F}}
      {F'\in\Omega 
         & 
       \supp{\pi}\subseteq\mathbb{N}
         &
       \formeq{\Xi}
              {\pi}
              {F'}
              {F}}
\]
\smallskip
\[
\infer[\cut]{\seq[\mathbb{N}]{\Psi}{\Xi}{\Omega}{F_1}}
            {\seq[\mathbb{N}]{\Psi}{\Xi}{\Omega}{F_2} &
             \seq[\mathbb{N}]{\Psi}{\Xi}{\setand{\Omega}{F_2}}{F_1} &
             \wfform{\mathbb{N}\cup\STLCGamma_0\cup\Psi}{\dom{\Xi}}{F_2}}
\]
\caption{The Axiom and the Cut Rule}
\label{fig:rules-other}
\end{figure}

The \id\ and the \cut\ rules are presented in
Figure~\ref{fig:rules-other}.
The \id\ rule limits the permutations that can be considered to be
ones that rename only nominal constants appearing in the support set
of the sequent.
The \cut\ rule includes a premise that ensures the wellformedness of
the cut formula.
The following theorem shows that these require the proofs of only
well-formed sequents in constructing a proof of a well-formed
sequent. 
\begin{theorem}\label{th:other-wf}
The following property holds of the \id\ and the \cut\ rule: if the
conclusion sequent is well-formed, the premises expressing typing
conditions have derivations and the conditions expressed by the other,
non-sequent premises are satisfied, then the premise sequents must be
well-formed.
\end{theorem}
\begin{proof}
The requirement is vacuously true for the \id\ rule and it has an
obvious proof for the \cut\ rule.
\end{proof}

In showing the soundness of the \id\ rule, we will need the
observation that the equivalence of formulas modulo permutations is
preserved under the kinds of substitutions that have to be considered
in determining the validity of sequents.
This observation is the content of the two lemmas below. 

\begin{lemma}\label{lem:equiv-hsub}
Suppose the for some formulas $F_1$ and $F_2$,
$\formeq{\Xi}{\pi}{F_2}{F_1}$
is holds.
If $\theta$ is a hereditary substitution such that
$\supp{\theta}\cap\supp{\pi}=\emptyset$ and
both $\hsub{\theta}{F_2}{F_2'}$ and $\hsub{\theta}{F_1}{F_1'}$ 
have derivations for some $F_1'$ and $F_2'$, then
$\formeq{\hsubst{\theta}{\Xi}}{\pi}{\hsubst{\theta}{F_2}}{\hsubst{\theta}{F_1}}$
holds.
\end{lemma}
\begin{proof}
This is proved by induction on the structure of $\formeq{\Xi}{\pi}{F_2}{F_1}$.
Consider the possible cases for the structure of $F_1$ (or equivalently, $F_2$).

The case where the formula is $\ftrue$ or $\ffalse$, this 
observation is obvious.
The cases when $F_1$ is $\fand{F}{F'}$, $\for{F}{F'}$, $\fimp{F}{F'}$,
$\fall{x:\alpha}{F}$, $\fexists{x:\alpha}{F}$, or $\fctx{\Gamma}{\mathcal{C}}{F}$
are easily argued with recourse
to the induction hypothesis and by noting that the definition of 
substitution distributes to the component parts.

When $F_1$ and $F_2$ are atomic formulas 
$\fatm{G_1}{M_1:A_1}$ and $\fatm{G_2}{M_2:A_2}$ respectively,
$\permute{\pi}{M_2}=_{\alpha}M_1$, $\permute{\pi}{A_2}=_{\alpha}A_1$, 
and $\ctxeq{\Xi}{\pi}{G_2}{G_1}$ will be derivable by definition.
We observe that for any LF term $M$ (resp. type $A$), if $\hsubst{\theta}{M}$ 
(resp. $\hsubst{\theta}{A}$) is defined then $\hsubst{\permute{\pi}{\theta}}{M}$
(resp. $\hsubst{\permute{\pi}{\theta}}{A}$) is defined and
$\permute{\pi}{(\hsubst{\theta}{M})}=\hsubst{\permute{\pi}{\theta}}{(\permute{\pi}{M})}$
(resp. $\permute{\pi}{(\hsubst{\theta}{A})}=\hsubst{\permute{\pi}{\theta}}{(\permute{\pi}{A})}$).
This observation can be proved easily by induction on the structure of $M$, and using
this result the observation for $A$ proved by induction on $A$.
The support sets of $\theta$ and $\pi$ must be disjoint, thus by this observation
$\hsubst{\theta}{(\permute{\pi}{A_2})}=_{\alpha}\permute{\pi}{(\hsubst{\theta}{A_2})}$
and
$\hsubst{\theta}{(\permute{\pi}{M_2})}=_{\alpha}\permute{\pi}{(\hsubst{\theta}{M_2})}$,
and therefore that
$\permute{\pi}{(\hsubst{\theta}{A_2})}=_{\alpha}\hsubst{\theta}{A_1}$ and 
$\permute{\pi}{(\hsubst{\theta}{M_2})}=_{\alpha}\hsubst{\theta}{M_1}$.
What remains is to show that 
$\ctxeq{\hsubst{\theta}{\Xi}}{\pi}{\hsubst{\theta}{G_2}}{\hsubst{\theta}{G_1}}$, which
we argue by an induction on the context expression $G_1$.

If $G_1$ is $\emptyctx$ or some $\Gamma_i$, then clearly
$G_2=G_1$, and further, $\hsubst{\theta}{G_1}=G_1$.
Thus the equivalence
$\ctxeq{\hsubst{\theta}{\Xi}}{\pi}{\hsubst{\theta}{G_2}}{\hsubst{\theta}{G_1}}$
has an obvious derivation.
If the context expression $G_1$ is of the form $(G_1',n_1':A_1')$, then 
$G_2$ is of the form $(G_2',n_2':A_2')$ and
$\permute{\pi}{n_2'}$ is the same nominal constant as $n_1'$, 
$\permute{\pi}{A_2'}$ is equal to $A_1'$ up to renaming of bound variables, and
$\ctxeq{\Xi}{\pi}{G_2'}{G_1'}$ is derivable.
By induction then,
$\ctxeq{\hsubst{\theta}{\Xi}}{\pi}{\hsubst{\theta}{G_2'}}{\hsubst{\theta}{G_1'}}$
will be derivable.
From the assumption that $\supp{\theta}\cap\supp{\pi}=\emptyset$ we can infer that
$\permute{\pi}{\theta}=\theta$.
Thus
$\permute{\pi}{(\hsubst{\theta}{A_2'})}=_{\alpha}\hsubst{\theta}{(\permute{\pi}{A_2'})}$
by our earlier observation about permutations on LF types, and so
$\permute{\pi}{(\hsubst{\theta}{A_2'})}=_{\alpha}\hsubst{\theta}{A_1'}$.
From this we can construct a derivation for
$\ctxeq{\hsubst{\theta}{\Xi}}{\pi}{\hsubst{\theta}{G_2}}{\hsubst{\theta}{G_1}}$.

Thus we have shown that 
$\ctxeq{\hsubst{\theta}{\Xi}}{\pi}{\hsubst{\theta}{G_2}}{\hsubst{\theta}{G_1}}$
holds and therefore can conclude that
$\formeq{\hsubst{\theta}{\Xi}}{\pi}{\hsubst{\theta}{F_2}}{\hsubst{\theta}{F_1}}$
must have a derivation.
\end{proof}

\begin{lemma}\label{lem:equiv-sub}
Suppose $\sigma$ is an appropriate substitution for $\Xi$ with respect to $\Psi$
and that $\formeq{\Xi}{\pi}{F'}{F}$ has a derivation.
Then 
$\formeq{\ctxvarminus{\Xi}{\sigma}}{\pi}{\subst{\sigma}{F'}}{\subst{\sigma}{F}}$
will have a derivation.
\end{lemma}
\begin{proof}
This is proved by induction on the structure of $\formeq{\Xi}{\pi}{F_2}{F_1}$.
Consider the possible cases for the structure of $F_1$ (or equivalently, $F_2$).

The case where the formula is $\ftrue$ or $\ffalse$, this 
observation is obvious.
The cases when $F_1$ is $\fand{F}{F'}$, $\for{F}{F'}$, $\fimp{F}{F'}$,
$\fall{x:\alpha}{F}$, $\fexists{x:\alpha}{F}$, or $\fctx{\Gamma}{\mathcal{C}}{F}$
are easily argued with recourse
to the induction hypothesis and by noting that the definition of 
substitution distributes to the component parts.

In the case that $F_1$ is atomic, then $F_1$ and $F_2$ are of the form
$\fatm{G_1}{M_1:A_1}$ and $\fatm{G_2}{M_2:A_2}$ respectively, and
$\permute{\pi}{M_2}$ is the same as $M_1$ up to renaming, 
$\permute{\pi}{A_2}$ and $A_1$ are equal up to renaming, 
and $\ctxeq{\Xi}{\pi}{G_2}{G_1}$ is derivable.
Given that $\subst{\sigma}{\fatm{G}{M:A}}=\fatm{\subst{\sigma}{G}}{M:A}$,
it only remains to show that
$\ctxeq{\ctxvarminus{\Xi}{\sigma}}{\pi}{\subst{\sigma}{G_2}}{\subst{\sigma}{G_1}}$
has a derivation to conclude
$\formeq{\ctxvarminus{\Xi}{\sigma}}{\pi}{\subst{\sigma}{F_2}}{\subst{\sigma}{F_1}}$ 
is derivable.
We prove that 
$\ctxeq{\ctxvarminus{\Xi}{\sigma}}{\pi}{\subst{\sigma}{G_2}}{\subst{\sigma}{G_1}}$
is derivable by induction on the structure of $\ctxeq{\Xi}{\pi}{G_2}{G_1}$.

When $G_1=G_2=\emptyctx$ this result is obvious.
When $G_1=(G_1',n_1':A_1')$ this is easily argued with recourse to the
inductive hypothesis.
When $G_1$ is some context variable $\Gamma_i$, then $G_2=\Gamma_i$,
$\ctxvarty{\Gamma_i}{\mathbb{N}_i}{\ctxty{\mathcal{C}_i}{\mathcal{G}_i}}\in\Xi$,
and $\supp{\pi}\subseteq\mathbb{N}_i$.
If $\Gamma_i$ is not in the domain of $\sigma$, then clearly 
$\ctxvarty{\Gamma_i}{\mathbb{N}_i}{\ctxty{\mathcal{C}_i}{\mathcal{G}_i}}\in\ctxvarminus{\Xi}{\sigma}$
and 
$\ctxeq{\ctxvarminus{\Xi}{\sigma}}{\pi}{\subst{\sigma}{G_2}}{\subst{\sigma}{G_1}}$ 
has an obvious derivation.
If instead $\Gamma_i$ is in the domain of $\sigma$ then there is
some $G_i$ such that 
$\subst{\sigma}{G_2}=\subst{\sigma}{G_1}=G_i$ and by the appropriateness of $\sigma$,
$\ctxtyinst{\supportof{\sigma} \setminus \mathbb{N}_i}{\Psi}{\ctxvarminus{\Xi}{\sigma}}{\ctxty{\mathcal{C}_i}{\mathcal{G}_i}}{G_i}$ 
must have a derivation.
By an obvious inductive argument we observe that
because $\supp{\pi}\subseteq\mathbb{N}_i$, $\permute{\pi}{G_i}=G_i$
and there is a derivation of
$\ctxeq{\ctxvarminus{\Xi}{\sigma}}{\pi}{\subst{\sigma}{G_2}}{\subst{\sigma}{G_1}}$,
as needed.

From this we can conclude that there will be a derivation for
$\formeq{\ctxvarminus{\Xi}{\sigma}}{\pi}{\subst{\sigma}{F'}}{\subst{\sigma}{F}}$.
\end{proof}

We can now show the soundness of the \id\ and \cut\ rules.

\begin{theorem}\label{th:other-sound}
The following property holds for every instance of the \id\ and \cut\ rules:
if the premises expressing typing judgements are derivable, the
conditions described in the other non-sequent premises are satisfied
and all the premise sequents are valid, 
then the conclusion sequent must also be valid. 
\end{theorem}
\begin{proof}
Consider the case for each rule.

\case{\id}
By assumption, for goal formula $F$, assumption formula $F'\in\Omega$, 
and permutation $\pi$, $\pi$ is a permutation 
of some subset of nominal constants in $\mathbb{N}$ and
$\formeq{\Xi}{\pi}{F'}{F}$ has a derivation.
Consider an arbitrary closed instance of the conclusion sequent identified by
$\theta$ and $\sigma$.
If any formula in $\subst{\sigma}{\hsubst{\theta}{\Omega}}$ were not valid, this
closed instance would be vacuously valid, so assume that all such formulae are valid.
In particular then, $\subst{\sigma}{\hsubst{\theta}{F'}}$ will be valid.
The substitution compatibility of $\seqsub{\theta}{\emptyset}$ and the appropriateness
of $\sigma$ are sufficient to satisfy the requirements of Lemmas~\ref{lem:equiv-sub}
and~\ref{lem:equiv-hsub}.
Thus there must be a derivation of
$\formeq{\emptyset}{\pi}{\subst{\sigma}{\hsubst{\theta}{F'}}}{\subst{\sigma}{\hsubst{\theta}{F}}}$.
But then by Theorem~\ref{th:perm-form},
$\subst{\sigma}{\hsubst{\theta}{F}}$ is also valid, and therefore the conclusion
sequent will be valid.

\case{\cut}
By assumption, the formula $F'$ is such that the premise sequents
$\mathcal{S}_1=\seq[\mathbb{N}]{\Psi}{\Xi}{\Omega}{F'}$ and
$\mathcal{S}_2=\seq[\mathbb{N}]{\Psi}{\Xi}{\setand{\Omega}{F'}}{F}$
are both valid.
Consider an arbitrary closed instance of the conclusion sequent identified by
$\theta$ and $\sigma$; these substitutions clearly also identify closed 
instances of the premise sequents $\mathcal{S}_1$ and $\mathcal{S}_2$.
If any formula in $\subst{\sigma}{\hsubst{\theta}{\Omega}}$ were not valid, then
this closed instance would be vacuously valid.
If all formulas in $\subst{\sigma}{\hsubst{\theta}{\Omega}}$ are valid, then
by the validity of $\mathcal{S}_1$ the formula $\subst{\sigma}{\hsubst{\theta}{F'}}$
must be valid.
But then all formulas in $\subst{\sigma}{\hsubst{\theta}{(\setand{\Omega}{F'})}}$
are valid, and by the validity of $\mathcal{S}_2$ the goal formula
$\subst{\sigma}{\hsubst{\theta}{F}}$ must be valid.
Therefore the conclusion sequent of the $\cut$ rule must be valid.
\end{proof}

%% file: proof-system/logical-rules.tex
\subsection{Rules for the Logical Symbols}

\begin{figure}
\[\begin{array}{cc}
\infer[\topR]{\seq[\mathbb{N}]{\Psi}{\Xi}{\Omega}{\ftrue}}{} \qquad&
\infer[\botL]{\seq[\mathbb{N}]{\Psi}{\Xi}{\setand{\Omega}{\ffalse}}{F}}{}
\end{array}\]
\[\begin{array}{c}
\infer[\ctxR]{\seq[\mathbb{N}]{\Psi}{\Xi}{\Omega}{\fctx{\Gamma}{\mathcal{C}}{F}}}
             {\seq[\mathbb{N}]
                  {\Psi}
                  {\Xi,\ctxvarty{\Gamma'}
                                  {\emptyset}
                                  {\ctxty{\mathcal{C}}{\cdot}}}
                  {\Omega}
                  {\subst{\Gamma'/\Gamma}{F}} 
                &
              \Gamma'\not\in\dom{\Xi}}
\medskip\\
\infer[\ctxL]{\seq[\mathbb{N}]
                  {\Psi}
                  {\Xi}
                  {\setand{\Omega}{\fctx{\Gamma}{\mathcal{C}}{F_1}}}
                  {F_2}}
             {\seq[\mathbb{N}]
                  {\Psi}
                  {\Xi}
                  {\setand{\Omega}{\subst{G/\Gamma}{F_1}}}
                  {F_2}
              &   
               \ctxtyinst{\mathbb{N}}{\Psi}{\Xi}{\ctxty{\mathcal{C}}{\emptycb}}{G}}
\end{array}\]
\[\begin{array}{c}
\infer[\allR]
      {\seq[\mathbb{N}]{\Psi}{\Xi}{\Omega}{\fall{x:\alpha}{F}}}
      {\begin{array}{cc}
           \mathbb{N}=
                  \{n_1:\alpha_1,\ldots,n_m:\alpha_m\}
              &
           y\not\in\dom{\Psi}
           \\
           \seq[\mathbb{N}]
               {\Psi\cup\{y:(\arr{\arr{\alpha_1}{\arr{\ldots}{\alpha_m}}}{\alpha})\}}
               {\Xi}
               {\Omega}
               {F'}
             &
           \hsub{\{\langle x, y\app n_1\ldots n_m,\alpha\rangle\}}{F}{F'}
       \end{array}} \medskip\\
\infer[\allL]{\seq[\mathbb{N}]{\Psi}{\Xi}{\setand{\Omega}{\fall{x:\alpha}{F_1}}}{F_2}}
             {\seq[\mathbb{N}]{\Psi}{\Xi}{\setand{\Omega}{F_1'}}{F_2} &
              \hsub{\{\langle x, t, \alpha\rangle\}}{F_1}{F_1'} &
              \stlctyjudg{\mathbb{N}\cup\STLCGamma_0\cup\Psi}{t}{\alpha}} \medskip\\
\infer[\existsR]{\seq[\mathbb{N}]{\Psi}{\Xi}{\Omega}{\fexists{x:\alpha}{F}}}
                {\seq[\mathbb{N}]{\Psi}{\Xi}{\Omega}{F'} &
                 \hsub{\{\langle x, t, \alpha\rangle\}}{F}{F'} &
                 \stlctyjudg{\mathbb{N}\cup\STLCGamma_0\cup\Psi}{t}{\alpha}} \medskip\\
\infer[\existsL]
      {\seq[\mathbb{N}]{\Psi}{\Xi}{\setand{\Omega}{\fexists{x:\alpha}{F_1}}}{F_2}}
      {\begin{array}{c}
         \begin{array}{ccc}
            \mathbb{N}=
                     \{n_1:\alpha_1,\ldots,n_m:\alpha_m\}
           & 
           y\not\in\dom{\Psi}
           &
           \hsub{\{\langle x, y\app n_1\ldots n_m, \alpha\rangle\}}{F_1}{F_1'}
         \end{array}
         \\
         \seq[\mathbb{N}]
             {\Psi\cup\{y:(\arr{\arr{\alpha_1}{\arr{\ldots}{\alpha_m}}}{\alpha})\}}
             {\Xi}
             {\setand{\Omega}{F_1'}}
             {F_2} 
      \end{array}}
\end{array}\]
\[\begin{array}{cc}
\infer[\andR]{\seq[\mathbb{N}]{\Psi}{\Xi}{\Omega}{\fand{F_1}{F_2}}}
             {\seq[\mathbb{N}]{\Psi}{\Xi}{\Omega}{F_1} &
              \seq[\mathbb{N}]{\Psi}{\Xi}{\Omega}{F_2}} &
\infer[\andL]{\seq[\mathbb{N}]{\Psi}{\Xi}{\setand{\Omega}{\fand{F_1}{F_2}}}{F}}
             {\seq[\mathbb{N}]{\Psi}{\Xi}{\setand{\Omega}{F_i}}{F} &
              i\in\{1,2\}}
\end{array}\]
\[\begin{array}{cc}
\infer[\orR]{\seq[\mathbb{N}]{\Psi}{\Xi}{\Omega}{\for{F_1}{F_2}}}
            {\seq[\mathbb{N}]{\Psi}{\Xi}{\Omega}{F_i} &
             i\in\{1,2\}} &
\infer[\orL]{\seq[\mathbb{N}]{\Psi}{\Xi}{\setand{\Omega}{\for{F_1}{F_2}}}{F}}
            {\seq[\mathbb{N}]{\Psi}{\Xi}{\setand{\Omega}{F_1}}{F} &
             \seq[\mathbb{N}]{\Psi}{\Xi}{\setand{\Omega}{F_2}}{F}} \medskip\\
\infer[\impR]{\seq[\mathbb{N}]{\Psi}{\Xi}{\Omega}{\fimp{F_1}{F_2}}}
             {\seq[\mathbb{N}]{\Psi}{\Xi}{\setand{\Omega}{F_1}}{F_2}} &
\infer[\impL]{\seq[\mathbb{N}]{\Psi}{\Xi}{\setand{\Omega}{\fimp{F_1}{F_2}}}{F}}
             {\seq[\mathbb{N}]{\Psi}{\Xi}{\Omega}{F_1} &
              \seq[\mathbb{N}]{\Psi}{\Xi}{\setand{\Omega}{F_2}}{F}}
\end{array}\]
\caption{The Logical Rules}
\label{fig:rules-base}
\end{figure}

Figure~\ref{fig:rules-base} presents derivation rules that are based
on the meanings of the logical symbols that can appear in formulas.
In the application of these rules, we assume the equivalence of
formulas under the renaming of quantified variables.

As in the previous cases, we show that these rules also allow for a
limitation of attention to well-formed sequents.

\begin{theorem}\label{th:core-wf}
The following property holds of the rules in Figure~\ref{fig:rules-base}: if the
conclusion sequent is well-formed, the premises expressing typing
conditions have derivations and the conditions expressed by the other,
non-sequent premises are satisfied, then the premise sequents must be
well-formed.
\end{theorem}
\begin{proof}
Consider as cases each rule in Figure~\ref{fig:rules-base}.

\case{\topR\ and \botL.}
There are no sequents in the premises of these rules so the property will
clearly hold.

\case{\ctxR}
By the well-formedness of the conclusion sequent
\begin{enumerate}
\item for each $\ctxvarty{\Gamma_i}{\mathbb{N}_i}{\ctxty{\mathcal{C}_i}{\mathcal{G}_i}}\in\Xi$,
$\wfctxvarty{\mathbb{N}\setminus\mathbb{N}_i}{\Psi}{\ctxty{\mathcal{C}_i}{\mathcal{G}_i}}$
has a derivation,
\item $\wfform{\mathbb{N}\cup\STLCGamma_0\cup\Psi}{\Xi^-}{\fctx{\Gamma}{\mathcal{C}}{F}}$ 
has a derivation, and
\item for each $F'\in\Omega$,
$\wfform{\mathbb{N}\cup\STLCGamma_0\cup\Psi}{\Xi^-}{F'}$ 
has a derivation.
\end{enumerate}
From (1) and the obvious derivation for 
$\wfctxvarty{\mathbb{N}}{\Psi}{\ctxty{\mathcal{C}}{\cdot}}$, we can conclude 
that the context variable context
$\Xi,\ctxvarty{\Gamma'}{\emptyset}{\ctxty{\mathcal{C}}{\cdot}}$ is well-formed.
From (2) we can infer that there is a derivation for
$\wfform{\mathbb{N}\cup\STLCGamma_0\cup\Psi}{\Xi^-\cup\{\Gamma'\}}{\subst{\Gamma'/\Gamma}{F}}$, 
and from (3) it is obvious $\Gamma'$ cannot appear in any formula from $\Omega$.
Therefore
$\wfform{\mathbb{N}\cup\STLCGamma_0\cup\Psi}{\Xi^-\cup\{\Gamma'\}}{F'}$
has a derivation of the same structure as that for
$\wfform{\mathbb{N}\cup\STLCGamma_0\cup\Psi}{\Xi^-}{F'}$ 
for all $F'\in\Omega$.
But then the premise sequent
$\seq[\mathbb{N}]{\Psi}{\Xi,\ctxvarty{\Gamma'}{\emptyset}{\ctxty{\mathcal{C}}{\cdot}}}{\Omega}{\subst{\Gamma'/\Gamma}{F}}$
must be well-formed, as needed.

\case{\ctxL}
In this case there must be a derivation of
$\ctxtyinst{\mathbb{N}}{\Psi}{\Xi}{\ctxty{\mathcal{C}}{\cdot}}{G}$.
Given the structure of the conclusion and premise sequents, we need only show the
well-formedness of the formula $\subst{G/\Gamma}{F_1}$ in the premise to conclude
it is a well-formed sequent as the well-formedness of the context variable context
and the set of assumption formulas is assured by the well-formedness of
the conclusion sequent.
By the well-formedness of the conclusion sequent, 
$\wfform{\mathbb{N}\cup\STLCGamma\cup\Psi}{\Xi^-}{\fctx{\Gamma}{\mathcal{C}}{F_1}}$
will have a derivation and thus 
$\wfform{\mathbb{N}\cup\STLCGamma\cup\Psi}{\Xi^-\cup\{\Gamma\}}{F_1}$
is also derivable.
The judgement $\wfctxvarty{\mathbb{N}}{\Psi}{\ctxty{\mathcal{C}}{\cdot}}$ has 
an obvious derivation, and so by Theorem~\ref{th:ctx-var-type-instance} we 
conclude that $\wfctx{\mathbb{N}}{\Psi}{G}$ must be derivable.
Therefore by an application of Theorem~\ref{th:subst-formula} there will be a derivation of
$\wfform{\mathbb{N}\cup\STLCGamma\cup\Psi}{\Xi^-}{\subst{G/\Gamma}{F_1}}$.
Thus we can conclude that the premise sequent
$\seq[\mathbb{N}]{\Psi}{\Xi}{\setand{\Omega'}{\subst{G/\Gamma}{F_1}}}{F_2}$
is well-formed.

\case{\allR}
In this case the support set $\mathbb{N}$ is $\{n_1:\alpha_1,\ldots,n_m:\alpha_m\}$ and
for a new variable $y\not\in\dom{\Psi}$ of type 
$(\arr{\alpha_1}{\arr{\ldots}{\arr{\alpha_m}{\alpha}}})$, the judgement
$\hsub{\{\langle x, y\app n_1\ldots n_m, \alpha\rangle\}}{F}{F'}$
has a derivation.
By the well-formedness of the conclusion sequent we know that
\begin{enumerate}
\item for each $\ctxvarty{\Gamma_i}{\mathbb{N}_i}{\ctxty{\mathcal{C}_i}{\mathcal{G}_i}}\in\Xi$,
$\wfctxvarty{\mathbb{N}\setminus\mathbb{N}_i}{\Psi}{\ctxty{\mathcal{C}_i}{\mathcal{G}_i}}$
has a derivation,
\item 
$\wfform{\mathbb{N}\cup\STLCGamma_0\cup\Psi}{\Xi^-}{\fall{x:\alpha}{F}}$
has a derivation, and
\item for each $F''\in\Omega$,
$\wfform{\mathbb{N}\cup\STLCGamma_0\cup\Psi}{\Xi^-}{F''}$
has a derivation.
\end{enumerate}
The well-formedness of the context variable context of the premise sequent is ensured 
by (1) and the use of Theorem~\ref{th:ctx-ty-wk}, and the well-formedness of 
the assumption set is ensured by (3) and the use of Theorem~\ref{th:wf-form-wk}.
What remains to be shown is that the goal formula $F'$ is well-formed.
From (2) we can infer that
$\wfform{\mathbb{N}\cup\STLCGamma_0\cup\Psi,x:\alpha}{\Xi^-}{F}$ 
also has a derivation.
It is clear that the substitution $\{\langle x,y\app n_1\ldots n_m,\alpha\rangle\}$
will be type preserving with respect to the arity typing context
$\mathbb{N}\cup\STLCGamma_0\cup(\Psi,y:\arr{\alpha_1}{\arr{\ldots}{\arr{\alpha_m}{\alpha}}})$, and so the judgement
$\wfform{\mathbb{N}\cup\STLCGamma_0\cup\Psi,y:\arr{\alpha_1}{\arr{\ldots}{\arr{\alpha_m}{\alpha}}}}
        {\Xi^-}
        {F'}$
has a derivation by Theorem~\ref{th:subst-formula}.
Thus the sequent
$\seq[\mathbb{N}]{\Psi,y:\arr{\alpha_1}{\arr{\ldots}{\arr{\alpha_m}{\alpha}}}}{\Xi}{\Omega}{F'}$
must be well-formed.

\case{\allL}
In this case there is a term $t$
such that both $\stlctyjudg{\mathbb{N}\cup\STLCGamma_0\cup\Psi}{t}{\alpha}$
and $\hsub{\{\langle x, t, \alpha\rangle\}}{F_1}{F_1'}$ have derivations.
We need only show that $F_1'$ is well-formed in the premise sequent to conclude it is 
a well-formed sequent, as the other formulas and the context variable types 
will be well-formed by the assumption that the conclusion sequent is well-formed.
By the well-formedness of the conclusion sequent, there is a derivation for 
$\wfform{\mathbb{N}\cup\STLCGamma_0\cup\Psi}{\Xi^-}{\fall{x:\alpha}{F_1}}$
and thus for
$\wfform{\mathbb{N}\cup\STLCGamma_0\cup\Psi,x:\alpha}{\Xi^-}{F_1}$ 
also.
Clearly $\{\langle x, t, \alpha\rangle\}$ is arity type preserving
with respect to $\mathbb{N}\cup\STLCGamma_0\cup\Psi$ because 
$\stlctyjudg{\mathbb{N}\cup\STLCGamma_0\cup\Psi}{t}{\alpha}$
has a derivation, and so by Theorem~\ref{th:subst-formula}
$\wfform{\mathbb{N}\cup\STLCGamma_0\cup\Psi}{\Xi^-}{F_1'}$ 
has a derivation.
Therefore the sequent 
$\seq[\mathbb{N}]{\Psi}{\Xi}{\setand{\Omega}{F_1'}}{F_2}$
must be well-formed.

\case{\existsR}
In this case there is some term $t$
such that both $\stlctyjudg{\mathbb{N}\cup\STLCGamma_0\cup\Psi}{t}{\alpha}$
and $\hsub{\{\langle x, t, \alpha\rangle\}}{F}{F'}$ have derivations.
We need only show that the goal formula of the premise sequent is well-formed 
to conclude it is a well-formed sequent, as the other formulas and the context 
variable types will be well-formed by the assumption that the conclusion sequent 
is well-formed.
By the well-formedness of the conclusion sequent, there exists a derivation of
$\wfform{\mathbb{N}\cup\STLCGamma_0\cup\Psi}{\Xi^-}{\fexists{x:\alpha}{F}}$
and thus 
$\wfform{\mathbb{N}\cup\STLCGamma_0\cup\Psi,x:\alpha}{\Xi^-}{F}$ 
must be derivable.
Clearly $\{\langle x, t, \alpha\rangle\}$ is arity type preserving
with respect to $\mathbb{N}\cup\STLCGamma_0\cup\Psi$ because 
$\stlctyjudg{\mathbb{N}\cup\STLCGamma_0\cup\Psi}{t}{\alpha}$
has a derivation, and so by Theorem~\ref{th:subst-formula}
$\wfform{\mathbb{N}\cup\STLCGamma_0\cup\Psi}{\Xi^-}{F'}$ 
has a derivation.
Therefore the sequent 
$\seq[\mathbb{N}]{\Psi}{\Xi}{\Omega'}{F'}$
will be well-formed, as needed.

\case{\existsL}
In this case the support set $\mathbb{N}$ is
$\{n_1:\alpha_1,\ldots,n_m:\alpha_m\}$ and for a new variable $y\not\in\dom{\Psi}$ 
of type $\arr{\alpha_1}{\arr{\ldots}{\arr{\alpha_m}{\alpha}}}$, the judgement
$\hsub{\{\langle x, y\app n_1\ldots n_m, \alpha\rangle\}}{F_1}{F'_1}$ 
has a derivation.
By the well-formedness of the conclusion sequent we know that
\begin{enumerate}
\item for each $\ctxvarty{\Gamma_i}{\mathbb{N}_i}{\ctxty{\mathcal{C}_i}{\mathcal{G}_i}}\in\Xi$,
$\wfctxvarty{\mathbb{N}\setminus\mathbb{N}_i}{\Psi}{\ctxty{\mathcal{C}_i}{\mathcal{G}_i}}$
has a derivation,
\item 
$\wfform{\mathbb{N}\cup\STLCGamma_0\cup\Psi}{\Xi^-}{\fexists{x:\alpha}{F_1}}$
has a derivation, and
\item for each $F''\in\Omega$,
$\wfform{\mathbb{N}\cup\STLCGamma_0\cup\Psi}{\Xi^-}{F''}$
has a derivation.
\end{enumerate}
From (1) and the use of Theorem~\ref{th:ctx-ty-wk} we can conclude that the 
context variable types in $\Xi$ will be well-formed under a context extended with $y$.
From (2) and Theorem~\ref{th:wf-form-wk} we conclude the assumption formulas 
from $\Omega$ are well-formed under a context extended with $y$.
Similarly, from (3) we conclude the goal formula well-formed under the extended 
context.
It only remains to show that the formula $F_1'$ is well-formed.
From (2) we can conclude that
$\wfform{\mathbb{N}\cup\STLCGamma_0\cup\Psi,x:\alpha}{\Xi^-}{F_1}$ 
also has a derivation, and 
since the substitution $\{\langle x,y\app n_1\ldots n_m,\alpha\rangle\}$
is type preserving with respect to 
$(\Psi,y:\arr{\alpha_1}{\arr{\ldots}{\arr{\alpha_m}{\alpha}}})$,
$\wfform{\mathbb{N}\cup\STLCGamma_0\cup\Psi,y:\arr{\alpha_1}{\arr{\ldots}{\arr{\alpha_m}{\alpha}}}}
        {\Xi^-}
        {F'_1}$
has a derivation by Theorem~\ref{th:subst-formula}.
Thus it is clear that
$\seq[\mathbb{N}]{\Psi,y:\arr{\alpha_1}{\arr{\ldots}{\arr{\alpha_m}{\alpha}}}}{\Xi}{\setand{\Omega}{F_1'}}{F_2}$
is a well-formed sequent.

\case{\andR, \orR, \impR, \andL, \orL, and \impL}
In all of these cases the well-formedness of the context variable types in 
$\Xi$ and the formulas in $\Omega$ are ensured by the well-formedness of the
conclusion sequent.
In the case of the left rules we further can infer the well-formedness of the goal
formulas from the well-formedness of the conclusion sequent.
What remains to be shown is that the the two sub-formulas $F_1$ and $F_2$
are well-formed.
By the well-formedness of the conclusion sequent there must be a derivation of
$\wfform{\mathbb{N}\cup\STLCGamma_0\cup\Psi}{\Xi^-}{F_1\bullet F_2}$.
In all cases, there must then be derivations for both
$\wfform{\mathbb{N}\cup\STLCGamma_0\cup\Psi}{\Xi^-}{F_1}$
and
$\wfform{\mathbb{N}\cup\STLCGamma_0\cup\Psi}{\Xi^-}{F_2}$ by definition.
Thus it is clear that all sequents appearing in the premises of these rules
must be well-formed.
\end{proof}

We recall that the validity of sequents is based on the validity of
their substitution instances.
In this context, the soundness of the rules $\ctxL$, $\allL$ and
$\existsR$ depends on the ability to invert the order of application
of substitutions.
The three lemmas below show that this is in fact possible.
Lemma~\ref{lem:form-hsubperm} follows from an easy induction and a use
of Theorem~\ref{th:subspermute}.
The two following lemmas have even simpler inductive proofs.

\begin{lemma}\label{lem:form-hsubperm}
Let $\theta_1$ and $\theta_2$ be arity type preserving substitutions with 
respect to $\Theta$ and $\aritysum{\context{\theta_1}}{\Theta}$ respectively and such that the variables
in the domain of $\theta_2$ are (1) distinct from the domain of
$\theta_1$ and (2) do not appear free in the range of $\theta_1$.
Then, letting
\[\theta_2'=\{\langle x,M',\alpha\rangle\ |\ 
                    \langle x,M,\alpha\rangle\in\theta_2\mbox{ and }
                           \hsub{\theta_1}{M}{M'}\},\]
for any formula $F$ well-formed with respect to 
$\aritysum{\context{\theta_2}}{\aritysum{\context{\theta_1}}{\Theta}}$ and some context
variable context $\Xi$,
$\hsubst{\theta_1}{\hsubst{\theta_2}{F}}=\hsubst{\theta_2'}{\hsubst{\theta_1}{F}}$
has a derivation.
\end{lemma}

\begin{lemma}\label{lem:form-subperm}
Let $\sigma_1$ and $\sigma_2$ be context variable substitutions
where the context variables in the domain of $\sigma_2$ are (1)
distinct from the context variables in the domain of $\sigma_1$ and
(2) do not appear free in the range of $\sigma_1$.
Then, letting 
$\sigma_2'=\{G'/\Gamma\ |\ G/\Gamma\in\sigma_2\mbox{ and }\subst{\sigma_1}{G}=G'\}$,
for any formula $F$,
$\subst{\sigma_1}{\subst{\sigma_2}{F}}=\subst{\sigma_2'}{\subst{\sigma_1}{F}}$.
\end{lemma}

\begin{lemma}\label{lem:form-perm}
Let $\theta$ be an arity type preserving substitution with respect to 
$\mathbb{N}\cup\STLCGamma_0\cup\Psi$, and let $\sigma$ be a context variable substitution 
appropriate for $\Xi$ with respect to $\context{\theta}\cup\Psi$.
Then, letting 
$\sigma'=\{G'/\Gamma\ |\ G/\Gamma\in\sigma\mbox{ and }\hsub{\theta}{G}{G'}\}$,
for any formula $F$ well-formed with respect to $\aritysum{\context{\theta}}{\Psi}$
and $\ctxsanstype{\Xi}$, 
$\hsubst{\theta}{\subst{\sigma}{F}}=\subst{\sigma'}{\hsubst{\theta}{F}}$
is drivable.
\end{lemma}

The soundness of the logical rules is the content of the following
theorem.

\begin{theorem}[Soundness]\label{th:core-sound}
The following property holds for every instance of each of the rules
in Figure~\ref{fig:rules-base}:
if the premises expressing typing judgements are derivable, the
conditions described in the other non-sequent premises are satisfied
and all the premise sequents are valid, 
then the conclusion sequent must also be valid. 
\end{theorem}
\begin{proof}
Consider each of the possible rules from Figure~\ref{fig:rules-base}.

\case{\topR}
Such a sequent is always valid since $\ftrue$ is always valid.

\case{\botL}
Since $\ffalse$ is never valid, a sequent in which this formula
appears as an assumption must obviously be valid. 

\case{\ctxR}
Consider an arbitrary closed instance of the conclusion sequent identified by $\theta$ 
and $\sigma$.
If any formula in $\subst{\sigma}{\hsubst{\theta}{\Omega}}$ were not valid this
instance would be vacuously valid, so assume they are all valid.
By the definition of substitution application, 
$\subst{\sigma}{\hsubst{\theta}{(\fctx{\Gamma}{\mathcal{C}}{F})}}=
    \fctx{\Gamma}{\mathcal{C}}{(\subst{\sigma}{\hsubst{\theta}{F}})}$,
and such a formula is valid if for every context expression $G$ such that 
$\csinst{\noms}{\emptyset}{\mathcal{C}}{G}$
has a derivation the formula
$\subst{G/\Gamma}{\subst{\sigma}{\hsubst{\theta}{F}}}$
is valid.
So consider an arbitrary context expression $G$ such that
$\csinst{\noms}{\emptyset}{\mathcal{C}}{G}$ has a derivation.
By assumption the premise sequent is
$\seq[\mathbb{N}]{\Psi}{\Xi,\ctxvarty{\Gamma'}{\emptyset}{\ctxty{\mathcal{C}}{\cdot}}}{\Omega}{\subst{\Gamma'/\Gamma}{F}}$
for some $\Gamma'\not\in\dom{\Xi}$, and is valid.
Clearly $\seqsub{\theta}{\emptyset}$ and $\{G/\Gamma'\}\circ\sigma$ 
identify a closed instance of this sequent, and since
$\Gamma'$ cannot appear in any formulas in $\Omega$ all the assumption
formulas of this sequent must be valid.
Therefore 
$\subst{\{G/\Gamma'\}\circ\sigma}{\hsubst{\theta}{\subst{\Gamma'/\Gamma}{F'}}}$
must be a valid formula.
By Lemmas~\ref{lem:form-subperm} and~\ref{lem:form-perm} this formula is 
equivalent to
$\subst{G/\Gamma}{\subst{\{G/\Gamma'\}\circ\sigma}{\hsubst{\theta}{F}}}$
which is further equivalent to
$\subst{G/\Gamma}{\subst{\sigma}{\hsubst{\theta}{F}}}$
since $\Gamma'$ cannot appear in $F$ or $\hsubst{\theta}{F}$.
From this we infer that the formula 
$\subst{\sigma}{\hsubst{\theta}{(\fctx{\Gamma}{\mathcal{C}}{F})}}$
is valid, and thus conclude that the conclusion sequent
is valid.

\case{\ctxL}
In this case there is a derivation for
$\ctxtyinst{\mathbb{N}}{\Psi}{\Xi}{\ctxty{\mathcal{C}}{\emptycb}}{G}$
for a context expression $G$.
Consider an arbitrary closed instance of the conclusion sequent identified by $\theta$
and $\sigma$.
If any formula in 
$\subst{\sigma}{\hsubst{\theta}{(\setand{\Omega}{\fctx{\Gamma}{\mathcal{C}}{F_1}}}}$
were not valid then this instance would be vacuously valid, so suppose they are all valid.
Then in particular,
$\subst{\sigma}{\hsubst{\theta}{(\fctx{\Gamma}{\mathcal{C}}{F_1})}}$
is valid.
So for any context expression $G$ such that 
$\csinst{\noms}{\emptyset}{\mathcal{C}}{G}$
is derivable, the formula
$\subst{G/\Gamma}{\subst{\sigma}{\hsubst{\theta}{F_1}}}$
will be valid.
But $\subst{\sigma}{\hsubst{\theta}{G}}$ will be such a context expression.
Using Theorem~\ref{th:ctx-ty-wk} it must be that 
$\ctxtyinst{\noms}{\Psi}{\Xi}{\ctxty{\mathcal{C}}{\emptycb}}{G}$
and so by Theorems~\ref{th:ctxtyinst-hsubst} and~\ref{th:ctxtyinst-subst}
$\ctxtyinst{\noms}{\emptyset}{\emptyset}{\ctxty{\mathcal{C}}{\emptycb}}{\subst{\sigma}{\hsubst{\theta}{G}}}$
will be derivable.
By Theorem~\ref{th:ctx-var-type-instance} we can conclude that
$\csinst{\noms}{\emptyset}{\mathcal{C}}{\subst{\sigma}{\hsubst{\theta}{G}}}$,
and therefore
$\subst{\subst{\sigma}{\hsubst{\theta}{G}}/\Gamma}{\subst{\sigma}{\hsubst{\theta}{F_1}}}$
is valid.
By Lemmas~\ref{lem:form-perm} and~\ref{lem:form-subperm}
$\subst{\sigma}{\hsubst{\theta}{\subst{G/\Gamma}{F_1}}}=
    \subst{\subst{\sigma}{\hsubst{\theta}{G}}/\Gamma}{\subst{\sigma}{\hsubst{\theta}{F_1}}}$
so $\subst{\sigma}{\hsubst{\theta}{\subst{G/\Gamma}{F_1}}}$ must be valid as well.

Thus $\theta$ and $\sigma$ identify a closed instance of the premise sequent
under which all the assumption formulas, 
$\subst{\sigma}{\hsubst{\theta}{(\setand{\Omega}{\subst{G/\Gamma}{F_1}})}}$
are valid.
But this sequent is valid by assumption and so $\subst{\sigma}{\hsubst{\theta}{F_2}}$
must be valid, and from this we can conclude that the conclusion sequent
must also be valid.

\case{\allR}
In this case the support set
$\mathbb{N}$ is $\{n_1:\alpha_1,\ldots,n_m:\alpha_m\}$ and for some
new variable $y\not\in\dom{\Psi}$ of type $\alpha'=\arr{\alpha_1}{\arr{\ldots}{\arr{\alpha_m}{\alpha}}}$,
$\hsub{\{\langle x,y\app n_1\ldots n_m,\alpha\rangle\}}{F}{F'}$ has a derivation.
Consider an arbitrary closed instance of the conclusion sequent identified by 
$\theta$ and $\sigma$.
If any formulas in $\subst{\sigma}{\hsubst{\theta}{\Omega}}$ were not valid
then this closed instance would be vacuously valid, so assume that they are all valid.
For any arbitrary term $t$ such that 
$\stlctyjudg{\noms\cup\STLCGamma_0}{t}{\alpha}$
has a derivation there exists $t'=\lflam{n_1:\alpha_1}{\ldots\lflam{n_m:\alpha_m}{t}}$
such that 
$\stlctyjudg{\noms\cup\STLCGamma_0}{t'}{\alpha'}$ ahs a derivation.
Further, it is clear that
$\theta'=\{\langle y, t',\alpha'\rangle\}\circ\theta$ and $\sigma$ will identify a 
closed instance of the premise sequent in this rule.
Clearly all the formulas in $\Omega$ will still be valid under these substitutions
as $y$ cannot appear in any formula in $\Omega$, and therefore
$\subst{\sigma}{\hsubst{\theta'}{F'}}$ will be valid by the 
validity of the premise sequent.
This formula is equivalent to 
$\subst{\sigma}{\hsubst{\theta'}{\hsubst{\{\langle x, y\app n_1\ldots n_m,\alpha\}}{F'}}}$
which by Lemmas~\ref{lem:form-perm} and~\ref{lem:form-hsubperm} is equivalent to 
$\hsubst{\{\langle x,t,\alpha\rangle\}}{\subst{\sigma}{\hsubst{\theta'}{F}}}$.
Clearly $y$ cannot appear in $F$ and thus $\hsubst{\theta'}{F}$ will be the 
same term as $\hsubst{\theta}{F}$.
Since
$\hsubst{\{\langle x, t, \alpha\rangle\}}
        {\subst{\sigma}{\hsubst{\theta}{F}}}$
is valid for this arbitrary choice of term $t$, then
$\subst{\sigma}{\hsubst{\theta}{F}}$ will be a valid formula by definition.
We can therefore conclude that the conclusion sequent is valid.

\case{\allL}
In this case there is a term $t$ and derivations for
$\stlctyjudg{\mathbb{N}\cup\STLCGamma_0\cup\Psi}{t}{\alpha}$ and
$\hsub{\{\langle x,t,\alpha\rangle\}}{F'}{F''}$.
Consider an arbitrary closed instance of the conclusion sequent identified by 
$\theta$ and $\sigma$.
If any formulas in 
$\subst{\sigma}{\hsubst{\theta}{(\setand{\Omega}{\fall{x:\alpha}{F_1}})}}$ 
were not valid then this closed instance would be vacuously valid, so assume 
that they are all valid.
In particular then,
$\subst{\sigma}{\hsubst{\theta}{(\fall{x:\alpha}{F_1})}}$ is valid.
So for any term $t'$ such that $\stlctyjudg{\noms\cup\STLCGamma_0}{t'}{\alpha}$
has a derivation, 
$\hsubst{\{\langle x,t',\alpha\rangle\}}{\subst{\sigma}{\hsubst{\theta}{F_1}}}$
must be valid.
Clearly $t'=\hsubst{\theta}{t}$ is such a term by Theorem~\ref{th:aritysubs}, and thus
$\hsubst{\{\langle x,t',\alpha\rangle\}}{\subst{\sigma}{\hsubst{\theta}{F_1}}}$
is a valid formula.
But by Lemmas~\ref{lem:form-hsubperm} and~\ref{lem:form-perm} this must be
equivalent to the formula $\subst{\sigma}{\hsubst{\theta}{F'_1}}$, and
so $\theta$ and $\sigma$ also identify a closed instance of the premise sequent
where all of the assumption formulas are valid.
Therefore $\subst{\sigma}{\hsubst{\theta}{F_2}}$ will be valid by the validity of the
premise sequent, and thus the conclusion sequent must be valid.

\case{\existsR}
In this case there are derivations for both
$\stlctyjudg{\mathbb{N}\cup\STLCGamma_0\cup\Psi}{t}{\alpha}$ and
$\hsub{\{\langle x,t,\alpha\rangle\}}{F}{F'}$.
Consider an arbitrary closed instance of the conclusion sequent identified by 
$\theta$ and $\sigma$.
If any formulas in 
$\subst{\sigma}{\hsubst{\theta}{\Omega}}$ 
were not valid then this closed instance would be vacuously valid, so assume 
that they are all valid.
Then clearly this same $\theta$ and $\sigma$ identify a closed instance of the 
premise sequent of this rule, and
since all the formulas in $\Omega$ are valid under $\theta$ and $\sigma$
we can infer from the validity of the premise sequent that
$\subst{\sigma}{\hsubst{\theta}{F'}}$ must be valid.
But by Lemmas~\ref{lem:form-perm} and~\ref{lem:form-hsubperm} this formula
is equal to 
$\hsubst{\{\langle x,\hsubst{\theta}{t},\alpha\rangle\}}{\subst{\sigma}{\hsubst{\theta}{F}}}$
and there is a derivation of
$\stlctyjudg{\noms\cup\STLCGamma_0}{\hsubst{\theta}{t}}{\alpha}$ 
by Theorem~\ref{th:aritysubs}.
Thus $\subst{\sigma}{\hsubst{\theta}{\fexists{x:\alpha}{F}}}$ must be valid.
Therefore the conclusion sequent will be valid.

\case{\existsL}
In this case the support set
$\mathbb{N}$ is $\{n_1:\alpha_1,\ldots,n_m:\alpha_m\}$ and for a
new variable $y\not\in\dom{\Psi}$ of type $\alpha'=(\arr{\alpha_1}{\arr{\ldots}{\arr{\alpha_m}{\alpha}}})$,
the judgement
$\hsub{\{\langle x,y\app n_1\ldots n_m,\alpha\rangle\}}{F_1}{F'_1}$ 
has a derivation.
Consider an arbitrary closed instance of the conclusion sequent identified by 
$\theta$ and $\sigma$.
If any formula in 
$\subst{\sigma}{\hsubst{\theta}{(\setand{\Omega}{\fexists{x:\alpha}{F_1}})}}$
were not valid this closed instance would be vacuously valid, so assume all these formulas are valid.
Then in particular,
$\subst{\sigma}{\hsubst{\theta}{(\fexists{x:\alpha}{F_1})}}$ will be valid.
So there must exist some term $t$ where $\stlctyjudg{\noms\cup\STLCGamma_0}{t}{\alpha}$
has a derivation and
$\hsubst{\{\langle x,t,\alpha\rangle\}}{\subst{\sigma}{\hsubst{\theta}{F_1}}}$
is valid.
From this we construct a term 
$t'=\lflam{n_1:\alpha_1}{\ldots\lflam{n_m:\alpha_m}{t}}$
which is such that $\{\langle y, t',\alpha'\rangle\}\circ\theta$
and $\sigma$ identify a closed instance of the premise sequent.
By Lemmas~\ref{lem:form-hsubperm} and~\ref{lem:form-perm} 
the equality
$\subst{\sigma}{\hsubst{\{\langle y, t',\alpha'\rangle\}\circ\theta}{F'_1}}=
  \hsubst{\{\langle x,t,\alpha\rangle\}}
         {\subst{\sigma}
                {\hsubst{\{\langle y, t',\alpha'\rangle\}\circ\theta}
                {F_1}}}$
holds and as $y$ cannot have appeared in $F_1$ this is also equivalent to 
$\hsubst{\{\langle x,t,\alpha\rangle\}}{\subst{\sigma}{\hsubst{\theta}{F_1}}}$
which we know to be valid.
Similarly all formulas in 
$\subst{\sigma}{\hsubst{\{\langle y, t',\alpha'\rangle\}\circ\theta}{\Omega}}$ 
will be valid as $y$ cannot appear in any formula in $\Omega$ and
all the formulas in $\subst{\sigma}{\hsubst{\theta}{\Omega}}$ are valid.

Thus by the validity of the premise sequent,
$\subst{\sigma}{\hsubst{\{\langle y, t',\alpha'\rangle\}\circ\theta}{F_1}}$
must be valid, and again because $y$ cannot have appeared in $F_1$ the formula
$\subst{\sigma}{\hsubst{\theta}{(\fexists{x:\alpha}{F_1})}}$
is valid.
Therefore we can conclude that the conclusion sequent will be valid.

\case{\andR}
Consider an arbitrary closed instance of the conclusion sequent identified by 
$\theta$ and $\sigma$.
If any formulas in 
$\subst{\sigma}{\hsubst{\theta}{\Omega}}$ 
were not valid then this closed instance would be vacuously valid, so assume 
that they are all valid.
There are two premise sequents in this rule:
$\mathcal{S}_1=\seq[\mathbb{N}]{\Psi}{\Xi}{\Omega}{F_1}$ and
$\mathcal{S}_2=\seq[\mathbb{N}]{\Psi}{\Xi}{\Omega}{F_2}$.
Clearly $\theta$ and $\sigma$ also identify a closed instance for both $\mathcal{S}_1$ and
$\mathcal{S}_2$, and since all formulas in $\subst{\sigma}{\hsubst{\theta}{\Omega}}$
are valid both 
$\subst{\sigma}{\hsubst{\theta}{F_1}}$ and
$\subst{\sigma}{\hsubst{\theta}{F_2}}$ must be valid by the
validity of these premise sequents.
But then clearly $\subst{\sigma}{\hsubst{\theta}{(\fand{F_1}{F_2})}}$ 
is valid by Definition~\ref{def:semantics}.

\case{\andL}
Consider an arbitrary closed instance of the conclusion sequent identified by 
$\theta$ and $\sigma$.
If any formula in 
$\subst{\sigma}{\hsubst{\theta}{(\setand{\Omega}{\fand{F_1}{F_2}})}}$
were not valid this closed instance would be vacuously valid, so assume all these formulas are valid.
Then in particular,
$\subst{\sigma}{\hsubst{\theta}{(\fand{F_1}{F_2})}}$ will be valid and so by definition
$\subst{\sigma}{\hsubst{\theta}{F_1}}$ and $\subst{\sigma}{\hsubst{\theta}{F_2}}$
are both valid.

The premise sequent is of the form 
$\seq[\mathbb{N}]{\Psi}{\Xi}{\setand{\Omega'}{F_i}}{F}$
for some $i\in\{1,2\}$.
Clearly $\theta$ and $\sigma$ identify a closed instance of this sequent, and
as all formulas in 
$\subst{\sigma}{\hsubst{\theta}{(\setand{\Omega}{F_i})}}$ are valid,
the formula $\subst{\sigma}{\hsubst{\theta}{F_i}}$ must be valid.
So by the validity of the premise sequent, 
$\subst{\sigma}{\hsubst{\theta}{F}}$
is valid, and thus we can conclude that the conclusion sequent is valid.

\case{\orR}
Consider an arbitrary closed instance of the conclusion sequent identified by 
$\theta$ and $\sigma$.
If any formulas in 
$\subst{\sigma}{\hsubst{\theta}{\Omega}}$ 
were not valid then this closed instance would be vacuously valid, so assume 
that they are all valid.
There is a single premise sequent in this rule and it is of the form 
$\seq[\mathbb{N}]{\Psi}{\Xi}{\Omega}{F_i}$ for some $i\in\{1,2\}$.
As $\theta$ and $\sigma$ clearly identify a closed instance of this premise sequent,
by the validity of this sequent we can conclude that
$\subst{\sigma}{\hsubst{\theta}{F_i}}$ is valid.
So by Definition~\ref{def:semantics} the formula
$\subst{\sigma}{\hsubst{\theta}{(\for{F_1}{F_2})}}$ must be valid, and therefore
the conclusion sequent is valid.

\case{\orL}
Consider an arbitrary closed instance of the conclusion sequent identified by 
$\theta$ and $\sigma$.
If any formula in 
$\subst{\sigma}{\hsubst{\theta}{(\setand{\Omega}{\for{F_1}{F_2}})}}$
were not valid this closed instance would be vacuously valid, so assume all these formulas are valid.
Then in particular,
$\subst{\sigma}{\hsubst{\theta}{(\for{F_1}{F_2})}}$ will be valid.

There are two premise sequents of this rule: 
$\mathcal{S}_1=\seq[\mathbb{N}]{\Psi}{\Xi}{\setand{\Omega}{F_1}}{F}$ and
$\mathcal{S}_2=\seq[\mathbb{N}]{\Psi}{\Xi}{\setand{\Omega}{F_2}}{F}$.
Note that $\theta$ and $\sigma$ identify closed instances of both these sequents.
Since 
$\subst{\sigma}{\hsubst{\theta}{(\for{F_1}{F_2})}}$
is valid, then either 
$\subst{\sigma}{\hsubst{\theta}{F_1}}$ is valid or
$\subst{\sigma}{\hsubst{\theta}{F_2}}$ is valid.
In the former case, we can infer from the validity of $\mathcal{S}_1$ that
$\subst{\sigma}{\hsubst{\theta}{F}}$ must be valid, and in the later case
we infer the same but through the validity of $\mathcal{S}_2$.
Therefore $\subst{\sigma}{\hsubst{\theta}{F}}$ will always be valid, and we can
conclude that the conclusion sequent is valid.

\case{\impR}
Consider an arbitrary closed instance of the conclusion sequent identified by 
$\theta$ and $\sigma$.
If any formulas in 
$\subst{\sigma}{\hsubst{\theta}{\Omega}}$ 
were not valid then this closed instance would be vacuously valid, so assume 
that they are all valid.
There is a single premise sequent in this rule, 
$\seq[\mathbb{N}]{\Psi}{\Xi}{\setand{\Omega}{F_1}}{F_2}$, and $\theta$ and $\sigma$
clearly  identify a closed instance of this sequent.
Whenever the formula $\subst{\sigma}{\hsubst{\theta}{F_1}}$ is valid
we also know that all the formulas in 
$\subst{\sigma}{\hsubst{\theta}{(\setand{\Omega}{F_1})}}$ are valid
and so by the validity of this premise sequent we can infer that
$\subst{\sigma}{\hsubst{\theta}{F_2}}$ is valid.
But then 
$\fimp{\subst{\sigma}{\hsubst{\theta}{F_1}}}{\subst{\sigma}{\hsubst{\theta}{F_1}}}$ 
must be valid by the semantics of formulas.
This is the same as $\subst{\sigma}{\hsubst{\theta}{(\fimp{F_1}{F_2})}}$
by the definition of substitution applications,
and therefore we conclude that the conclusion sequent is valid.

\case{\impL}
Consider an arbitrary closed instance of the conclusion sequent identified by 
$\theta$ and $\sigma$.
If any formula in 
$\subst{\sigma}{\hsubst{\theta}{(\setand{\Omega}{\fimp{F_1}{F_2}})}}$
were not valid this closed instance would be vacuously valid, so assume all these formulas are valid.
Then in particular,
$\subst{\sigma}{\hsubst{\theta}{(\fimp{F_1}{F_2})}}$ is valid.
There are two premise sequents of this rule which are valid by assumption:
$\mathcal{S}_1=\seq[\mathbb{N}]{\Psi}{\Xi}{\Omega}{F_1}$ and
$\mathcal{S}_2=\seq[\mathbb{N}]{\Psi}{\Xi}{\setand{\Omega}{F_2}}{F}$.
Both $\mathcal{S}_1$ and $\mathcal{S}_2$ have closed instances identified by $\theta$ and
$\sigma$ as they share arity and context variables contexts with the conclusion sequent.
By assumption all formulas in $\subst{\sigma}{\hsubst{\theta}{\Omega}}$
are valid, and so by the validity of $\mathcal{S}_1$ the formula
$\subst{\sigma}{\hsubst{\theta}{F_1}}$ is valid.
Thus, by the validity of $\subst{\sigma}{\hsubst{\theta}{(\fimp{F_1}{F_2})}}$,
the formula $\subst{\sigma}{\hsubst{\theta}{F_2}}$ must be valid.
But then all the formulas in 
$\subst{\sigma}{\hsubst{\theta}{(\setand{\Omega}{F_2})}}$
are valid, and so by the validity of $\mathcal{S}_2$ we can conclude the formula
$\subst{\sigma}{\hsubst{\theta}{F}}$ is valid.
We can therefore conclude that the conclusion sequent is valid.
\end{proof}

%% file: proof-system/atomic.tex
\section{Proof Rules that Interpret Atomic Formulas}
\label{sec:atomic-rules}

Atomic formulas in our logic and sequents represent LF typing judgements.
The validity of such formulas is therefore determined by the LF
derivations rules and this fact can be used in their analysis.
When an atomic formula appears as the conclusion of a sequent, the
analysis takes an obvious form: the derivability of the sequent
can be based on that of a sequent in which the conclusion judgement
has been unfolded using an LF rule.
The treatment of an atomic assumption formula requires more thought.
Such formulas may contain context and eigenvariables in them and we
must consider all the possible instantiations for these variables that
could make the judgement true in determining the validity of the
sequent.
The translation of this general requirement into an analysis that is
local to the atomic formula may be driven by the structure of the type
in the formula.
When the type is of the form $\typedpi{x}{A}{B}$, the term must be an
abstraction and there is exactly one way in which a purported typing
derivation could have concluded.
When the type is atomic, Theorem~\ref{th:atomictype}
provides us information about the different cases that need to be
considered.

We transform the ideas outlined above into a concrete collection of
proof rules in this section.
The development of a ``case analysis rule'' for atomic assumption
formulas is somewhat intricate and is the subject of the first
subsection below.
With this done, the second subsection presents the proof rules and
shows that they obey the important properties of soundness and
preservation of wellformedness for sequents. 

\input{proof-system/case-elab}

\input{proof-system/atomic-rules}

%% file: proof-system/case-elab.tex
\subsection{Analyzing an Atomic Assumption Formula with an Atomic Type}
\label{ssec:cases}

Our goal here is to develop the basis for a proof rule around the
analysis of an assumption of the form $\fatm{G}{R:P}$.
Obviously $P$ in this case is of the form $(a\app M_1\app \ldots\app
M_n)$.
Let us suppose initially that this formula is closed.
In this case, for the formula to be valid, $R$ would need to have as a
head a constant declared in $\Sigma$ or a nominal constant assigned a
type in $G$. 
If the arguments of $R$ do not satisfy the constraints imposed by the
type associated with the head, then the typing judgement will not be
derivable and hence we can conclude that the sequent is in fact
valid.
On the other hand, if the arguments of $R$ do satisfy the required
constraints, then Theorem~\ref{th:atomictype} gives us a means for
decomposing the given typing judgement into ones pertaining to
$M_1,\ldots, M_n$.
The validity of the given sequent can therefore be reduced to the
validity of a sequent that results from replacing the atomic formula
under consideration by ones that represent the mentioned typing
judgements.
In this discussion, we have implicitly assumed that $G$ is
well-formed.
However, the reduction described is easily seen to be sound even when
$G$ is not well-formed. 

In the general case, the formula $\fatm{G}{R:P}$ may not be closed.
There are, in fact, two conceptually different possibilities that need
to be considered from this perspective.
First, the context expression may have a part that is yet to be
determined, i.e., $G$ may be of the form
$\Gamma, n_1 : A_1,\ldots, n_m : A_m$ where $\Gamma$ has a set of
names $\mathbb{N}$ and a context
variable type of the form $\ctxty{\mathcal{C}}{G_1;\ldots;G_\ell}$
associated with it in the sequent.
Second, the expressions in the atomic formula and the context variable
type may contain variables in them that are bound in the eigenvariable
context.
To articulate a proof rule around the atomic formula in this
situation, it is necessary to develop a means for analyzing the
formula in a way that pays attention to the validity of the sequent
under all acceptable instantiations of the context and term variables.

The analysis that we describe proceeds in two steps.
We first describe a finite way to consider elaborations of the context
variable that make explicit all the heads that need to be considered
for the term in an analysis of the closed instances of the atomic
formula. 
Concretely, this process yields a finite collection of pairs
comprising a sequent in which the context variable may have been
partially instantiated, and a specifically identified possibility for
the head that is either drawn from the signature or that appears
explicitly in the context; the intent here, which is verified in
Lemma~\ref{lem:heads-cover}, is that considering just
the second components of these pairs as the heads of the term in the
typing judgement will suffice for a complete analysis based on
Theorem~\ref{th:atomictype}. 
The second step actually carries out the analysis in each of these
cases, using the idea of unification in the application of
Theorem~\ref{th:atomictype} to accommodate all possible closed
instantiations of the term variables in the sequent.
We consider these steps in the two subsections below.

\subsubsection{Elaborating Context Variables and Identifying Head Possibilities}

We first note that context expressions have may have implicit and
explicit parts, the former being subject to elaboration via context
substitutions. 

\begin{definition}[Implicit and Explicit Parts of a Context]
Let $\mathcal{S} = \seq[\mathbb{N}]{\Psi}{\Xi}{\Omega}{F}$ be a
well-formed sequent.
If $G$ is a context expression appearing in $\mathcal{S}$ then
it must be of either the form $n_1 : A_1,\ldots, n_m:A_m$ or of the form
$\Gamma, n_1 : A_1,\ldots, n_m:A_m$ where $\Gamma$ is a context 
variable with an associated declaration 
$\ctxvarty{\Gamma}{\mathbb{N}_{\Gamma}}{\ctxty{\mathcal{C}}{G_1; \ldots; G_n}}$ 
in $\Xi$.
In the latter case, we say that $G$ has an implicit part relative to
$\mathcal{S}$ that is given by
$\ctxvarty{\Gamma}{\mathbb{N}_{\Gamma}}{\ctxty{\mathcal{C}}{G_1; \ldots; G_n}}$. 
Further, we refer to $n_1 : A_1,\ldots, n_m:A_m$ in the former case
and to the sequence formed by listing the bindings in $G_1,\ldots,G_n$
followed by $n_1 : A_1,\ldots, n_m:A_m$ in the latter case as the
explicit bindings in $G$ relative to $\mathcal{S}$.
\end{definition}

Let $\Gamma$ be a context variable that has the type
$\ctxty{\mathcal{C}}{G_1;\ldots;G_\ell}$.
Closed instances of $\Gamma$ are then generated by interspersing
$G_1,\ldots,G_\ell$ with blocks of declarations generated from the
block schema comprising $\mathcal{C}$.
In determining possibilities for the head of $R$ from the implicit
part of $G$ in an atomic formula of the form $\fatm{G}{R:P}$, we need
to consider an elaboration of $G$ with only one such block; of course,
for a complete analysis, we will need to consider all the
possibilities for such an elaboration.
The function $\addblksans$ defined below formalizes such an
elaboration, returning a modified sequent and a potential head for the
term in the typing judgement.
Note that in an elaboration based on a block schema 
of the form
$\{x_1:\alpha_1,\ldots,x_n:\alpha_n\}y_1 : A_1, \ldots, y_k : A_k$,
it would be necessary to consider a choice of nominal constants for
the schematic variables $y_1,\ldots,y_k$.
The function is parameterized by such a choice.
We must also accommodate all possible instantiations for the variables
$x_1,\ldots,x_n$, subject to the proviso that these instantiations do
not use nominal constants that appear in a later part of the context
expression.
This is done by using (implicitly universally quantified) term
variables for $x_1,\ldots,x_n$ and by raising such variables over the
nominal constants that are not prohibited from appearing in the
instantiations; to support the latter requirement, the function is
parameterized by a collection of nominal constants.  
Finally, we observe that the elaboration process may introduce new
nominal constants into the sequent, necessitating a raising of the
eigenvariables over the new constants. 

\begin{definition}[Adding a Block and Picking a Binding From It]
Let $\mathcal{S}$ be the the well-formed sequent 
$\seq[\mathbb{N}]{\Psi}{\Xi}{\Omega}{F}$,
let there be some $\ctxvarty{\Gamma}{\mathbb{N}_{\Gamma}}{\ctxty{\mathcal{C}}{G_1; \ldots; G_n}}\in\Xi$, and let $\mathcal{B} = \{x_1:\alpha_1,\ldots,
x_n:\alpha_n\}y_1 : A_1, \ldots, y_k : A_k$ be one of the block
schemas comprising $\mathcal{C}$.
Further, let $\mathbb{N'} \subseteq (\mathbb{N}\setminus\mathbb{N}_{\Gamma})$
be a collection of nominal
constants, let $ns$ be a list $n_1,\ldots,n_k$ of distinct nominal
constants that are also different from the constants in
$\mathbb{N'}$ and that are such that, for $1 \leq i \leq k$,
$n_i : \erase{A_i} \in (\noms\setminus\mathbb{N}_{\Gamma})$. 
Finally, for $0 \leq j \leq n$, let $\mathbb{N}_j$ be the collection
of nominal constants assigned types in $G_1,\ldots,G_j$.
Then, letting
\begin{enumerate}
  \item $\Psi'_j$ be a version of $\Psi$ raised over
    $\{n_1,\ldots,n_k\} \setminus \mathbb{N}$, $\theta'_j$ be the
    associated raising substitution, 
  \item $A_1',\ldots,A_k'$ be the types $A_1,\ldots,A_k$ with the schematic variables
  $y_1,\ldots,y_k$ replaced with the names $n_1,\ldots,n_k$,
  \item  $\Psi''_j$ be a version of $\{x_1:\alpha_1,\ldots,x_n:\alpha_n\}$
    raised over $\mathbb{N'} \cup \mathbb{N}_j \cup
    (\{n_1,\ldots,n_k\} \setminus \mathbb{N})$ with the new variables 
    chosen to be distinct from those in $\Psi'_j$, $\theta''_j$ be the
    associated raising substitution, $G$ be the context expression
    $n_1 : \hsubst{\theta''_j}{A'_1},\ldots, n_k : \hsubst{\theta''_j}{A'_k}$, and
  \item $\Xi'_j$ be the context variable context
\begin{gather*}
  \hsubst{\theta'_j}
         {\left(\Xi \setminus 
                \left\{\ctxvarty{\Gamma}
                                {\mathbb{N}_{\Gamma}}
                                {\ctxty{\mathcal{C}}{G_1; \ldots; G_n}}
                \right\}
          \right)}
  \\\bigcup\\
  \left\{\ctxvarty{\Gamma}
                  {\mathbb{N}_{\Gamma}}
                  {\ctxty{\mathcal{C}}
                         {\hsubst{\theta'_j}{G_1};\ldots;
                          \hsubst{\theta'_j}{G_j};
                          G;
                          \hsubst{\theta'_j}{G_{j+1}};\ldots;
                          \hsubst{\theta'_j}{G_n}}}
  \right\},
\end{gather*}
\end{enumerate}
for $0 \leq j \leq n$ and $1 \leq i \leq k$,
$\addblk{\mathcal{S}}{\ctxvarty{\Gamma}{\mathbb{N}_{\Gamma}}{\ctxty{\mathcal{C}}{G_1; \ldots; G_n}}}{\mathcal{B}}{ns}{\mathbb{N'}}{j}{i}$ is defined to be the tuple
\[\langle \seq[\mathbb{N} \cup ns]
             {\Psi'_j \cup \Psi''_j}
             {\Xi'_j}
             {\hsubst{\theta'_j}{\Omega}}
             {\hsubst{\theta'_j}{F}},
         n_i : \hsubst{\theta''_j}{A'_i}
\rangle.\]
Note that the conditions in the definition ensure that all the
substitutions involved in it will have a result, thereby permitting us
to use the notation introduced after Theorem~\ref{th:aritysubs}.           
\end{definition}

The elaboration just described is parameterized by the choice of
nominal constants for the variables assigned types in the block
schema.
In identifying the choices that have to be considered, it is useful to
partition the members of $(\noms\setminus\mathbb{N}_{\Gamma})$ 
into two sets: those that appear in
the support set of the sequent whose elaboration is being considered
and those that do not.
It is necessary to consider all possible assignments that satisfy
arity typing constraints from the first category.
From the second category, as we shall soon see, it suffices to
consider exactly one representative assignment.
Note also that we may insist that the nominal constant in each
assignment of the block be disinct; if this is not the case, the
sequent is easily seen to be valid.
The function $\namessans$ defined below embodies these ideas.
The function is parameterized be a sequence of arity types
corresponding to the declarations in the block schema, a collection of
``known'' nominal constants that are available for use in an
elaboration of the block schema and a collection of nominal constants
that are already bound in the context expressions and hence must not
be used again. 

\begin{definition}[Identifying a Choice of Nominal Constants]
Let $tys$ be a sequence of arity types and let $\mathbb{N}_o$ and
$\mathbb{N}_b$ be finite sets of nominal constants. 
Further, let $\emptyseq$ denote an empty sequence and 
$\consseq{x}{xs}$ denote a sequence that starts with $x$ and
continues with the sequence $xs$.
Then the collection of name choices for $tys$ relative to $\mathbb{N}_o$
and away from $\mathbb{N}_b$ is denoted by
$\names{tys}{\mathbb{N}_o}{\mathbb{N}_b}$ and defined by recursion
on $tys$ as follows:
\begin{center}
\begin{minipage}{5in}
\begin{tabbing}
\=\qquad\qquad\=\kill
\>
$\names{tys}{\mathbb{N}_o}{\mathbb{N}_b)} =$\\
\>\>
     $\begin{cases}
        \{\emptyseq\} & \mbox{\rm if}\ tys = \emptyseq\\[10pt]
        \{\consseq{n}{nl}\ |\ 
                  n : \alpha \in \noms,
                  n \in \mathbb{N}_o\setminus \mathbb{N}_b,
                  \ \mbox{\rm and}\\
          \qquad\qquad
               nl \in \names{tys'}{\mathbb{N}_o}{\mathbb{N}_b \cup
                 \{n\}}\}\ \cup\\                     
        \{\consseq{n}{nl}\ |\ 
               n\ \mbox{\rm is the first nominal
                            constant}& \mbox{\rm if}\ tys = \consseq{\alpha}{tys'} \\ 
          \qquad\qquad  \mbox{\rm such that}\ n : \alpha \in \noms\
                    \mbox{\rm and}\
                  n \not\in \mathbb{N}_o \cup \mathbb{N}_b,\\
          \qquad\qquad  \mbox{\rm and}\ nl \in
                    \names{tys'}{\mathbb{N}_o}{\mathbb{N}_b \cup \{n\}} \}
      \end{cases}$
\end{tabbing}
\end{minipage}
\end{center}
We assume in this definition the existence of an ordering on the
nominal constants that allows us to select the first of these
constants that satisfies a criterion of interest.
\end{definition}  

We can now identify a finite collection of elaborations of the
implicit part of a context expression that must be considered in the
analysis of an assumption formula of the form $\fatm{G}{R:P}$ that
appears in a sequent $\mathcal{S}$.
We do this below through the definition of the function
$\implheadssans$. 

\begin{definition}[Head Choices from the Implicit Part of a Context]
Let $\mathcal{S}$ be a well-formed sequent 
$\seq[\mathbb{N}]{\Psi}{\Xi}{\Omega}{F}$, let $G$ be a context expression appearing in a
formula in $\mathcal{S}$ that has an implicit part relative to
$\mathcal{S}$ that is given by 
$\ctxvarty{\Gamma}{\mathbb{N}_{\Gamma}}{\ctxty{\mathcal{C}}{G_1; \ldots; G_n}}$, 
and let
$\mathcal{B} = \{x_1:\alpha_1,\ldots, x_n:\alpha_n\}
                      y_1 : A_1, \ldots, y_k : \alpha_k$
be one of the block schemas comprising $\mathcal{C}$.
Further, let $\mathbb{N}_b$ be the collection of nominal constants
assigned types by the explicit bindings of $G$ relative to
$\mathcal{S}$ and let 
$\mathbb{N}_o = \mathbb{N}\setminus\mathbb{N}_{\Gamma}\setminus\mathbb{N}_b$.
Finally, let $\allblks{\mathcal{S}}
              {\ctxvarty{\Gamma}{\mathbb{N}_{\Gamma}}{\ctxty{\mathcal{C}}{G_1; \ldots; G_n}}}
              {\mathcal{B}}$ denote the set
\begin{tabbing}
\qquad\=\qquad\qquad\=\kill
\> $\{ \addblk{\mathcal{S}}
             {\ctxvarty{\Gamma}{\mathbb{N}_{\Gamma}}{\ctxty{\mathcal{C}}{G_1; \ldots; G_n}}}
             {\mathcal{B}}
             {ns}
             {\mathbb{N}_o}
             {j}
             {i} \ \vert$ \\
\>\>  $0 \leq j \leq n,
       1 \leq i \leq k,
       ns \in \names{(\erase{A_1},\ldots,\erase{A_k})}
                    {\mathbb{N}_o}
                    {\mathbb{N}_{\Gamma}\cup\mathbb{N}_b}\}$.
\end{tabbing}
If $\{\mathcal{B}_1,\ldots,\mathcal{B}_m \}$ is the collection of 
of block schemas comprising $\mathcal{C}$, then the implicit 
heads in $G$ relative to $\mathcal{S}$ is defined to be the set
\[ \bigcup \{ \allblks{\mathcal{S}}
                      {\ctxvarty{\Gamma}{\mathbb{N}_{\Gamma}}{\ctxty{\mathcal{C}}{G_1; \ldots; G_n}}}
                      {\mathcal{B}} \ \vert\
                           \mathcal{B} \in
                           \{\mathcal{B}_1,\ldots,\mathcal{B}_m \}\}. \]
This set is denoted by $\implheads{\mathcal{S}}{G}$.
\end{definition}

The complete set of heads and corresponding (elaborated) sequents that
must be considered in the analysis of an atomic formula of the form
$\fatm{G}{R:P}$ is identified through the function \headssans\ that is
defined below.

\begin{definition}[The Complete Set of Head Choices]
Let $\mathcal{S}$ be a well-formed sequent and let $G$ be a context
expression appearing in a formula in $\mathcal{S}$.
Let $\mbox{\sl NewHds}$ be the set $\implheads{\mathcal{S}}{G}$ if $G$ has an
implicit part relative to $\mathcal{S}$ and the empty set otherwise.
Then the heads in $G$ relative to $\mathcal{S}$ is defined to be the set
\begin{gather*}
\{ \langle \mathcal{S}, c : A\rangle\ |\ c : A \in \Sigma\} \cup
\{ \langle \mathcal{S}, n : A\rangle\ |\ n : A\ 
               \mbox{\rm is an explicit binding in G relative to}\ \mathcal{S} \} 
   \\\bigcup\\ 
   \mbox{\sl NewHds}.
\end{gather*}
This set is denoted by $\hds{\mathcal{S}}{G}$.
\end{definition}

The first property that we observe of the elaboration process
described is that it requires us to consider only well-formed
sequents. 

\begin{lemma}
\label{lem:heads-wf}
Let $\mathcal{S}=\seq[\mathbb{N}]{\Psi}{\Xi}{\Omega}{F}$ be a well-formed sequent
and let $\fatm{G}{R:P}$ be an atomic formula in $\Omega$.
Then for each $(\mathcal{S}',h:A)\in \hds{\mathcal{S}}{G}$ it must be the
case that $\mathcal{S}'$ is a well-formed sequent. 
Further, if $\mathcal{S}'$ is
$\seq[\mathbb{N'}]{\Psi'}{\Xi'}{\Omega'}{F'}$, it must be the
case that $\wftype{(\mathbb{N}'\cup\STLCGamma_0\cup\Psi')}{A}$ is derivable.
\end{lemma}

\begin{proof}
The claim is not immediately obvious only 
when $(\mathcal{S}',h:A)\in \implheads{\mathcal{S}}{G}$.
For these cases, it suffices to show that every pair generated by
\addblksans\ satisfies the requirements of the lemma.
However, this is easily argued.
The main observation--that gets used twice---is that if $\Psi_2$ is a
version of $\Psi_1$ raised over some collection of nominal constants
$\mathbb{N}_2$ with $\theta$ being the associated raising
substitution, and 
$\wftype{\STLCGamma \cup \Psi_1}{A'}$ holds for some arity context
$\STLCGamma$ that is disjoint from $\Psi_1$ and $\Psi_2$, then 
$\wftype{\STLCGamma \cup \Psi_2 \cup
  \mathbb{N}_2}{\hsubst{\theta}{A'}}$ also holds.
\end{proof}

We want next to show the adequacy of the elaboration process, i.e.,
that the collection of pairs of sequents and heads it identifies are
sufficient for the analysis of validity for a sequent with an
assumption formula of the form $\fatm{G}{R:P}$.
One aspect that we must account for in our argument is that we
consider all possible choices for a ``new name'' for a binding in a
block instance through a single representative.
The key property that enables this reduction is that the validity of
closed sequents is invariant under permutations of the nominal
constants as we discussed in Theorem~\ref{th:perm-valid}.

The following lemma in combination with Theorem~\ref{th:perm-valid}
yields the desired result concerning the elaboration process. 

\begin{lemma}
\label{lem:heads-cover}
Let $\mathcal{S}=\seq[\mathbb{N}]{\Psi}{\Xi}{\Omega}{F}$ be a
well-formed sequent and let $\fatm{G}{R:P}$ be a formula in $\Omega$.
Further, let $\theta$ and $\sigma$ be term and context
variable substitutions that identify a closed instance of $\mathcal{S}$ and
that are such that $\subst{\sigma}{\hsubst{\theta}{\fatm{G}{R:P}}}$ is
valid.
If the term $\hsubst{\theta}{R}=(h\app M_1\ldots M_n)$, then 
there is a pair $\langle \mathcal{S}', h':A' \rangle$ in $\hds{\mathcal{S}}{G}$
such that
\begin{enumerate}
\item there is a formula $\fatm{G'}{R':P'}$ amongst the assumption
formulas of $\mathcal{S}'$ with $h':A'$ appearing in either $\Sigma$
or in the explicit bindings in $G'$ relative to $\mathcal{S'}$, and 

\item there is a closed instance of $\mathcal{S}'$ identified by 
  closed term and context variable substitutions $\theta'$ and $\sigma'$ and a
  permutation $\pi$ such that $\permute{\pi}{h'} = h$,
$\permute{\pi}{\subst{\sigma'}{\hsubstseq{\emptyset}{\theta'}{\mathcal{S}'}}}= 
    \subst{\sigma}{\hsubstseq{\emptyset}{\theta}{\mathcal{S}}}$, and
$\permute{\pi}{\subst{\sigma'}{\hsubst{\theta'}{\fatm{G'}{R':P'}}}}=
    \subst{\sigma}{\hsubst{\theta}{\fatm{G}{R:P}}}$.
\end{enumerate}
\end{lemma}

\begin{proof}
Since $\subst{\sigma}{\hsubst{\theta}{\fatm{G}{R:P}}}$ is valid, it
must be the case that there are LF derivations for 
$\lfctx{\subst{\sigma}{\hsubst{\theta}{G}}}$,
$\lftype{\subst{\sigma}{\hsubst{\theta}{G}}}
        {\hsubst{\theta}{P}}$, and
$\lfchecktype{\subst{\sigma}{\hsubst{\theta}{G}}}
             {\hsubst{\theta}{R}}
             {\hsubst{\theta}{P}}$.
Using Theorem~\ref{th:atomictype} together with the fact that
$\hsubst{\theta}{R}=(h\app M_1\ldots M_n)$, we see that, for an
appropriate $A$, $h: A$ must be a member of $\Sigma$ or it must appear
in $\subst{\sigma}{\hsubst{\theta}{G}}$.
Our argument distinguishes two ways that this could happen: it could
be because $h:A$ is a member of $\Sigma$ or it is an instance of a
declaration in the explicit part of $G$ or because it is introduced
into the context $\subst{\sigma}{\hsubst{\theta}{G}}$ by the
substitution $\sigma$.

The first collection of cases is easily dealt with: essentially, we
pick $h'$, $\mathcal{S}'$, $\theta'$ and $\sigma'$ to be identical to
$h$, $\mathcal{S}$, $\theta$ and $\sigma$, respectively, and we let
$\pi$ be the identity permutation.
The requirements of the lemma then follow easily from the definition
of the \headssans\ function.

In the cases that remain, $G$ must have the form $\Gamma, n^G_1
: A^G_1,\ldots, n^G_p : A^G_p$ for some context variable $\Gamma$ that
has the set of names $\mathbb{N}_{\Gamma}$ and the type 
$\ctxty{\mathcal{C}}{G_1;\ldots; G_\ell}$ assigned to it
in $\Xi$ and $h$ must be introduced by the substitution that $\sigma$
makes for $\Gamma$ as the $i^{th}$ binding, for
some $i$, in a block of declarations resulting from instantiating
one of the block schemas constituting $\mathcal{C}$. 
Let us suppose  the relevant block schema is $\mathcal{B}$ and
it has the form $\{x_1:\alpha_1,\ldots,x_n:\alpha_n\}
(y_1:B_1,\ldots, y_k:B_k)$.
Moreover, let us suppose that this block of declarations appears in
$\subst{\sigma}{\hsubst{\theta}{G}}$ somewhere 
between the instances of $G_j$ and $G_{j+1}$, for some $j$ between
$0$ and $\ell$.
We may, without loss of generality, assume $x_1,\ldots,x_n$ to be
distinct from the variables assigned types by $\Psi$.
We can then visualize the block introducing $h$ as
$(n_1:\hsubst{\theta^h}{B_1'},\ldots,n_k:\hsubst{\theta^h}{B_k'})$ for some
$n_1,\ldots,n_k$ of the requisite types, for some types $B_1',\ldots,B_k'$ 
which are the types $B_1,\ldots,B_k$ with the schematic variables of the 
schema replaced by these names, and for a closed substitution
$\theta^h$ whose domain is $x_1,\ldots,x_n$ and, since
$\lfctx{\subst{\sigma}{\hsubst{\theta}{G}}}$ is derivable, whose
support does not contain the nominal constants in $\mathbb{N}_{\Gamma}$, 
$n^G_1,\ldots,n^G_p$, or those that are assigned a type in
$G_{j+1},\ldots,G_\ell$.
It follows from this that if we can associate the type
$\ctxty{\mathcal{C}}
       {\hsubst{\theta}{G_1};\ldots;\hsubst{\theta}{G_j};
        n_1:\hsubst{\theta^h}{B_1'},\ldots,n_k:\hsubst{\theta^h}{B_k'};
        \hsubst{\theta}{G_{j+1}};\ldots; \hsubst{\theta}{G_\ell}}$
with $\Gamma$, then the context expression that $\sigma$ substitutes
for $\Gamma$ can still be generated from the changed type.
The key to our showing that the requirements of the lemma are met in
these cases will be to establish that $\hds{\mathcal{S}}{G}$
contains a sequent and head pair such that the type of $\Gamma$ is
elaborated to a form from which the above type can be obtained, up to
a permutation of nominal constants, by a well-behaved substitution and
the head is identified as the $i^{th}$ item in the introduced block of
declarations.  

Towards this end, let us consider the tuple $\langle \mathcal{S''},
h'': A'' \rangle$ 
that is generated by
\[\addblk{\mathcal{S}}{\ctxvarty{\Gamma}{\mathbb{N}_{\Gamma}}{\ctxty{\mathcal{C}}{G_1; \ldots; G_\ell}}}
  {\mathcal{B}}{(n_1,\ldots,n_k)}{\mathbb{N}_o}{j}{i},\] 
where $\mathbb{N}_o$ is the collection of nominal constants obtained
by leaving out of $\mathbb{N}$ the constants in $\mathbb{N}_{\Gamma}$ and 
the constants that appear amongst the
explicit bindings of $G$ relative to $\mathcal{S}$.
In this case, $\mathcal{S''}$ will have the form
$\seq[\mathbb{N}'']{\Psi''}{\Xi''}{{\Omega''}}{F''}$ with the following
properties.
First, $\mathbb{N}''$ will be identical to
$\mathbb{N} \cup \{n_1,\ldots,n_k \}$. 
Second, $\Psi''$ will comprise two disjoint parts $\Psi^\mathcal{S}_r$
and $\Psi^\mathcal{B}_r$, where $\Psi^\mathcal{S}_r$ is a 
version of $\Psi$ raised over the nominal constants in
$\{n_1,\ldots,n_k\}$ that are not members of $\mathbb{N}$ with a
corresponding raising substitution $\theta^\Psi_r$, and
$\Psi^\mathcal{B}_r$ is a version of $\{x_1,\ldots,x_n\}$ raised over
all the nominal constants in $\mathbb{N} \cup \{n_1,\ldots,n_k \}$
except the ones that are assigned a type in $G_{j+1},\ldots,G_\ell$ or
that appear in $n^G_1,\ldots,n^G_p$ with the corresponding raising
substitution $\theta^{\mathcal{B}}_r$.
Third, $\Xi''$ will be
\[  \hat{\Xi} \cup 
   \{ \ctxvarty{\Gamma}
               {\mathbb{N}_{\Gamma}}
               {\ctxty{\mathcal{C}}
                     {\hsubst{\theta^\Psi_r}{G_1};\ldots,\hsubst{\theta^\Psi_r}{G_j};
                      n_1 : \hsubst{\theta^\mathcal{B}_r}{B_1'},
                      \ldots, n_k :\hsubst{\theta^\mathcal{B}_r}{B_k'};
                     \hsubst{\theta^\Psi_r}{G_{j+1}},\ldots\hsubst{\theta^\Psi_r}{G_\ell}}}
   \}
\]
where $\hat{\Xi} =
       \hsubst{\theta^\Psi_r}{(\Xi \setminus \{\ctxvarty{\Gamma}{\mathbb{N}_{\Gamma}}{\ctxty{\mathcal{C}}{G_1; \ldots; G_\ell}}\})}$.
Finally, each formula in $\Omega'' \cup \{F''\}$ is obtained by
applying the raising substitution $\theta^\Psi_r$ to a corresponding
one in $\Omega \cup \{F\}$. 
Using Theorem~\ref{th:raised-subs} we observe that, because
$\supportof{\theta}$ is disjoint from the set $\mathbb{N}$, there is a 
(closed) raising substitution $\theta_r$  with $\context{\theta_r} =
\Psi^\mathcal{S}_r$ whose support is
disjoint from the set $\mathbb{N} \cup \{n_1,\ldots,n_k\}$ and which is such
that $\hsubst{\theta_r}{\Omega''} = \hsubst{\theta}{\Omega}$,
$\hsubst{\theta_r}{F''} = \hsubst{\theta}{F}$,
$\hsubst{\theta_r}{\hat{\Xi}}$ is equal to
$\hsubst{\theta}{(\Xi \setminus \{\ctxvarty{\Gamma}{\mathbb{N}_{\Gamma}}{\ctxty{\mathcal{C}}{G_1; \ldots; G_\ell}}\})}$,
and, for each $q$, $1 \leq q \leq \ell$,
$\hsubst{\theta_r}{\hsubst{\theta^\Psi_r}{G_q}}=\hsubst{\theta}{G_q}$.
Using Theorem~\ref{th:raised-subs} again, we see that there is a
(closed) substitution $\theta^h_r$ with
$\context{\theta^h_r} = \Psi^\mathcal{B}_r$ whose support is
disjoint from $\mathbb{N} \cup \{n_1,\ldots,n_k\}$ and that is such
that, for $1 \leq q \leq k$, it is the case that
$\hsubst{\theta^h_r}{\hsubst{\theta^\mathcal{B}_q}{B_q'}} =
\hsubst{\theta^h}{B_q'}$.  
Based on all these observations, it is easy to see that if we let
$\theta''= \theta_r \cup \theta^h_r$, then $\langle
\theta'',\emptyset\rangle$ is substitution compatible with
$\mathcal{S}''$ and $\hsubstseq{\emptyset}{\theta''}{\mathcal{S}''}$ is
identical to $\hsubstseq{\emptyset}{\theta}{\mathcal{S}}$ except for
the fact that the type associated with $\Gamma$ in its context
variable context is
$\ctxty{\mathcal{C}}
       {\hsubst{\theta}{G_1};\ldots;\hsubst{\theta}{G_j};
        n_1:\hsubst{\theta^h}{B_1'},\ldots,n_k:\hsubst{\theta^h}{B_k'};
        \hsubst{\theta}{G_{j+1}};\ldots; \hsubst{\theta}{G_\ell}}$.
By the earlier observation, $\sigma$ is appropriate for        
$\hsubstseq{\emptyset}{\theta''}{\mathcal{S}''}$ and, in fact
$\subst{\sigma}{\hsubstseq{\emptyset}{\theta''}{\mathcal{S}''}} =
 \subst{\sigma}{\hsubstseq{\emptyset}{\theta}{\mathcal{S}}}$.
 Noting also that $h'':A''$ must, by the definition of \addblksans, be
 $n_i:\hsubst{\theta^\mathcal{B}_r}{B_i'}$, if 
 $\hds{\mathcal{S}}{G}$ includes in it a pair obtained by this
 particular call to \addblksans, then we can pick $\mathcal{S}'$ to be
 $\mathcal{S}''$, $h'$ to be $h''$, $A'$ to be $A''$, $\theta'$ to be
 $\theta''$, $\sigma'$ to be $\sigma$ and $\pi$ to be the identity
 permutation to satisfy the requirements of the lemma.

We are, of course, not assured that there will be a pair in
$\hds{\mathcal{S}}{G}$ corresponding to the use of \addblksans\ with 
exactly the arguments considered above.
Specifically, the sequences of nominal constants that are considered
for the block instance may not include $n_1,\ldots, n_k$.
However, we know that some sequence $n'_1,\ldots, n'_k$
will be considered that is identical to $n_1,\ldots,n_k$ except for
constants in identical locations in the two sequences that are not
drawn from $\mathbb{N}$.
Since the constants in any sequence must be distinct, it follows
easily that we can describe a permutation $\pi'$ on the nominal
constants that is the identity map on $\mathbb{N}$ and that maps
$n'_1,\ldots,n'_k$ to $n_1,\ldots,n_k$. 
It can also be seen then that $\hds{\mathcal{S}}{G}$ will include a
tuple $\langle \mathcal{S}''', h''' : A''' \rangle$ such that
$\permute{\pi'}{\mathcal{S}'''} = \mathcal{S}''$, $\permute{\pi'}{h'''} = h''$, and
$\permute{\pi'}{A'''} = A''$.
Picking $\mathcal{S}'$ to be $\mathcal{S}'''$, $h'$ to be $h'''$, $A'$ to be $A'''$,
$\theta'$ to be $\permute{\inv{\pi'}}{\theta''}$, $\sigma'$ to be
$\permute{\inv{\pi'}}{\sigma''}$, $\pi$ to be $\pi'$ and using
Theorems~\ref{th:perm-hsubst} and \ref{th:perm-subst}, we can once
again see that the requirements of the lemma are met.
\end{proof}

\subsubsection{Generating a Covering Set of Premise Sequents}

Given a sequent $\mathcal{S}$ and a particular atomic assumption
formula $F$ with context expression $G$ in $\mathcal{S}$, Lemma~\ref{lem:heads-cover} assures us
that $\hds{\mathcal{S}}{G}$ correctly identifies all the context
elaborations and corresponding heads that need to be considered in the
analysis of $F$.
However, we are still left with the task of identifying a systematic
way of considering all the term and context substitutions that yield
closed instances of $\mathcal{S}$ in which the term component of $F$
has the relevant head.
We now turn to this task.
Rather than identifying the closed instances immediately, we will
think of taking a step in this direction that also allows us to reduce
the typing judgement represented by $F$ based on the typing rule for
the LF judgement it represents; this analysis will then be reflected
in a proof rule in our system.
The first step in this direction will be to determine a substitution
that makes the head of $F$ identical to the one it needs to be in its
closed form.
We use the idea of unification, refined to fit our context, towards
this end. 
We describe next the idea of reducing a sequent that encodes the
analysis of an LF typing judgement based on the observations in
Theorem~\ref{th:atomictype}.
The eventual proof rule will then combine the identification of
relevant heads using $\hds{\mathcal{S}}{G}$, the solving of a
unification problem based on each such head, and the 
reduction of the sequent.

We begin this development by first elaborating the notion of
unification problems and their solutions.

\begin{definition}[Unification Problems \& their Solutions]\label{def:unification}
A \emph{unification problem} $\mathcal{U}$ is a tuple
$\unif{\mathbb{N}}{\Psi}{\mathcal{E}}$
in which $\mathbb{N}$ is a collection of nominal constants, $\Psi$ is
a collection of arity type assignments to term variables, and
$\mathcal{E}$ is a set of the form
$\{\eqn{E_1}{E_1'},\ldots,\eqn{E_n}{E_n'}\}$ where, for each
$i$, $1 \leq i \leq n$, either $\wftype{\mathbb{N}\cup\Psi\cup\STLCGamma_0}{E_i}$ and 
$\wftype{\mathbb{N}\cup\Psi\cup\STLCGamma_0}{E'_i}$ have derivations or
there is an arity type $\alpha$ such that 
$\stlctyjudg{\mathbb{N}\cup\Psi\cup\STLCGamma_0}{E_i}{\alpha}$ and 
$\stlctyjudg{\mathbb{N}\cup\Psi\cup\STLCGamma_0}{E_i'}{\alpha}$ 
have derivations.
A \emph{solution} to the unification problem $\mathcal{U}$ is a pair
$\solun{\theta}{\Psi'}$ of a substitution and a collection of type
assignments to term variables such that
\begin{enumerate}
\item $\theta$ is type preserving with respect to $\noms\cup\STLCGamma_0\cup\Psi'$,

\item $\supportof{\theta} \cap \mathbb{N} = \emptyset$,

\item for any $x$ if
$x : \alpha \in \Psi$ and $x : \alpha' \in \aritysum{\context{\theta}}{\Psi'}$ then
  $\alpha = \alpha'$, and

\item for each $i$, $1\leq i\leq n$, expressions $E_i$ and $E_i'$ from 
$\eqn{E_i}{E_i'}$ are such that $\hsubst{\theta}{E_i}=\hsubst{\theta}{E_i'}$.
\end{enumerate}
Note that the typing contraints validate the use of the notation
$\hsubst{\theta}{E_i}$ and $\hsubst{\theta}{E'_i}$.
\end{definition}

Given an atomic term $R$, we can determine its instances that have a
particular head $h$ through the unification of $R$ with $h$ applied to
a sequence of fresh variables.
We will use this idea to narrow down the set of instances of a sequent
that must be considered once we have determined what the head of the
term in an atomic goal of the form $\fatm{G}{R:P}$ must be.
However, we must first build into our notion of a fresh variable the
ability to instantiate it with nominal constants appearing in the
sequent.
We do this below by using the mechanism of raising.

\begin{definition}[Generalized Variables]
Let $\Psi$ be a finite set of arity typing assignments to term
variables and let $\mathbb{N}$ be a finite subset of $\noms$.
Further, let $n_1,\ldots,n_k$ be a listing of the nominal constants in
$\mathbb{N}$ and let $\alpha_1,\ldots,\alpha_k$ be the respective
types of these constants.
Then, for any variable $z$ that does not appear in $\Psi$, $z :
\alpha_1 \atyarr \cdots \atyarr \alpha_k \atyarr \beta$ is 
said to be a variable of arity type $\beta$ away from $\Psi$ and
raised over $\mathbb{N}$.
Moreover, $(z \app n_1\app \ldots \app n_k)$ is said to be the
generalized variable term corresponding to $z$.
\end{definition}

The following lemma now formalizes the described refinement of the sequent.

\begin{lemma}\label{lem:inst-solun}
Let $\mathcal{S} = \seq[\mathbb{N}]{\Psi}{\Xi}{\setand{\Omega}{F}}{F'}$ be a
well-formed sequent with F being the formula $\fatm{G}{R : P}$.
Further, let $\theta$ be a term substitution that together with a
context substitution $\sigma$ identifies a closed instance of
$\mathcal{S}$ and is such that, for the head $h$ that is assigned the
type $\typedpi{x_1}{A_1}{\ldots\typedpi{x_n}{A_n}{P'}}$ by $\Sigma$ or
$\subst{\sigma}{\hsubst{\theta}{G}}$, 
$\hsubst{\theta}{R} = (h\app M_1\app \ldots\app M_n)$ and
\[\hsubst{\theta}{P} = \hsubst{\{\langle x_1,M_1,\erase{A_1}\rangle,
  \ldots, \langle x_n,M_n,\erase{A_n}\rangle \}}{P'}.\]
Finally, for each $i$, $1 \leq i \leq n$, let $z_i:\alpha_i'$ be a distinct
variable of type $\erase{A_i}$ away from $\Psi$ and raised over
$\mathbb{N}$, and let $t_i$ be the generalized variable term
corresponding to $z_i$. 
Then $\solun{\theta}{\emptyset}$ is a
solution to the unification problem
\begin{tabbing}
\qquad\quad\=\quad\qquad\qquad\=\kill
\>
$\langle \mathbb{N}; 
         \Psi \cup \{z_1:\alpha_1',\ldots, z_n:\alpha_n'\};$\\
\>\>
$\left\{P = \hsubst{\{\langle x_1,t_1,\erase{A_1}\rangle,
            \ldots, \langle x_n,t_n,\erase{A_n}\rangle \}}{P'}, 
         R = (h \app t_1 \app \ldots\app t_n)\right\}\rangle$.
\end{tabbing}
\end{lemma}
\begin{proof}
As $\theta$ and $\sigma$ identify a closed instance of $\mathcal{S}$
is must be that $\seqsub{\theta}{\emptyset}$ is substitution compatible with
$\mathcal{S}$ and $\sigma$ is appropriate for $\hsubstseq{\emptyset}{\Psi}{\mathcal{S}}$.
Therefore, the first three clauses of the definition for a solution will be
satisfied by $\seqsub{\theta}{\emptyset}$.
The final clause of the definition is satisfied by the assumptions
$\hsubst{\theta}{P} = \hsubst{\{\langle x_1,M_1,\erase{A_1}\rangle,
  \ldots, \langle x_n,M_n,\erase{A_n}\rangle \}}{P'}$
and
$\hsubst{\theta}{R} = (h\app M_1\app \ldots\app M_n)$.
Therefore $\solun{\theta}{\emptyset}$ is a solution to the given unification problem.
\end{proof}

The reduction of a sequent is based on lifting the observations of 
Theorem~\ref{th:atomictype} to the analysis of atomic formulas.
Since such an analysis must be driven by the structure of the LF type,
reduction is only sensible when the atomic formula in question is
one in which the term is atomic, and the head of the application
is either a constant or a nominal constant which is bound in the 
explicit bindings of the context expression.
It is with this type that we identify typing judgements for each
argument term, and replace the original assumption formula with
a set of formulas determined by this analysis.

\begin{definition}\label{def:decompseq}[Reducing a Sequent]
Let $F = \fatm{G}{h\app M_1\ldots M_n:P}$ be a formula appearing in 
the well-formed sequent
$\mathcal{S}=\seq[\mathbb{N}]{\Psi}{\Xi}{\setand{\Omega}{F}}{F'}$,
where $h$ is assigned LF type 
$A=\typedpi{x_1}{A_1}{\ldots\typedpi{x_n}{A_n}{P'}}$ in $\Sigma$ or
the explicit bindings in $G$ relative to $\mathcal{S}$.
Letting 
$A_i'=\hsubst{\{\langle x_1,M_1,\erase{A_1}\rangle,\ldots,
                \langle x_{n-1}, M_{n-1},\erase{A_{n-1}}\rangle\}}
             {A_i}$ 
for each $i$, $1\leq i\leq n$,
the sequent obtained by decomposing the assumption formula $F$ based on the
type $A$ is defined to be
\[
\seq[\mathbb{N}]{\Psi}{\Xi}{\setand{\Omega}{\fatm{G}{M_1:A_1'},\ldots,\fatm{G}{M_n:A_n'}}}{F'}.
\]
This sequent is denoted by $\decompseq{F}{\mathcal{S}}$.
Note that the well-formedness of $\mathcal{S}$ justifies the use of
the notation
$\hsubst{\{\langle x_1,M_1,\erase{A_1}\rangle,\ldots,
           \langle x_{n-1}, M_{n-1},\erase{A_{n-1}}\rangle\}}
        {A_i}$.
\end{definition}

The following lemma expresses the soundness of the idea of reducing
a sequent. Additionally, it identifies a measure with atomic
formulas that diminishes with the replacements effected by
a reduction step; this property will be useful in formulating an
induction rule in Section~\ref{sec:ind}.

\begin{lemma}\label{lem:decomp_decr}
Let $\mathcal{S} = \seq[\mathbb{N}]{\Psi}{\Xi}{\setand{\Omega}{F}}{F'}$
be a well-formed sequent with $F$ an atomic formula of the form
$\fatm{G}{R:P}$.
Further, let $h$ be assigned type
$\typedpi{x_1}{A_1}{\ldots\typedpi{x_n}{A_n}{P'}}$ in either $\Sigma$
or the explicit bindings in $G$ relative to $\mathcal{S}$ and 
for each $i$, $1 \leq i \leq n$, let $z_i:\alpha_i'$ be a distinct
variable of type $\erase{A_i}$ away from $\Psi$ and raised over
$\mathbb{N}$, and let $t_i$ be the generalized variable term
corresponding to $z_i$. 
Finally, let $\mathcal{U}$ be the unification problem 
\begin{tabbing}
\qquad\quad\=\quad\qquad\qquad\=\kill
\>
$\langle \mathbb{N}; 
         \Psi \cup \{z_1:\alpha_1',\ldots, z_n:\alpha_n'\};$\\
\>\>
$\left\{P = \hsubst{\{\langle x_1,t_1,\erase{A_1}\rangle,
            \ldots, \langle x_n,t_n,\erase{A_n}\rangle \}}{P'}, 
         R = (h \app t_1 \app \ldots\app t_n)\right\}\rangle$.
\end{tabbing}
Then any solution to $\mathcal{U}$ is substitution compatible with
$\mathcal{S}$.
Further, for any $\solun{\theta}{\Psi_\theta}$ that is a solution to
$\mathcal{U}$ and $\theta_r$ that is a raising substitution associated
with the application of $\theta$ to $\mathcal{S}$ relative to
$\Psi_\theta$, there must be terms $M_1, \ldots, M_n$ such that the
following hold:
\begin{enumerate}
\item $\hsubst{\theta_r}{\hsubst{\theta}{R}}$ is $(h\app M_1 \app \ldots \app M_n)$.
\item For any $\theta'$ and $\sigma'$ identifying a closed instance
  of $\hsubstseq{\Psi_\theta}{\theta}{\mathcal{S}}$, if 
$\subst{\sigma'}{\hsubst{\theta'}{\hsubst{\theta_r}{\hsubst{\theta}{\fatm{G}{R:P}}}}}$
is valid and there is a derivation for 
$\subst{\sigma'}{\hsubst{\theta'}{\hsubst{\theta_r}{\hsubst{\theta}{(\lfchecktype{G}{R}{P})}}}}$
of height $k$, then for each $i$, $1 \leq i \leq n$, letting
$A'_i = \hsubst{\{\langle x_1,M_1,\erase{A_1}\rangle,\ldots,
                                       \langle x_{i-1}, M_{i-1},\erase{A_{i-1}}\rangle\}}
               {A_i}$, it must be
the case that 
$\subst{\sigma'}{\hsubst{\theta'}
                       {(\fatm{\hsubst{\theta_r}{\hsubst{\theta}{G}}}
                              {M_i: A'_i})}}$  
is valid and that there is a derivation of height less than $k$ for 
$\subst{\sigma'}{\hsubst{\theta'}{(\lfchecktype{\hsubst{\theta_r}{\hsubst{\theta}{G}}}
                                {M_i}{A'_i})}}.$
\item There is an $\mathcal{S'}$ such that 
$\mathcal{S}' = 
    \decompseq{\hsubst{\theta_r}{\hsubst{\theta}{F}}}
              {\hsubstseq{\Psi_\theta}{\theta}{\mathcal{S}}}$
and $\mathcal{S}'$ is valid only if
$\hsubstseq{\Psi_\theta}{\theta}{\mathcal{S}}$  is. 
\end{enumerate}
\end{lemma}
\begin{proof}
A straightforward examination of Definitions~\ref{def:seq-term-subst}
and~\ref{def:unification} suffices to verify that solutions to
$\mathcal{U}$ must be substitution compatible with $\mathcal{S}$.
Any solution $\solun{\theta}{\Psi_{\theta}}$ to the unification
problem $\mathcal{U}$ must be such that
$\hsubst{\theta}{R}=\hsubst{\theta}{(h\app t_1\ldots t_n)}$.
From this it follows that
$\hsubst{\theta_r}{\hsubst{\theta}{R}}=
\hsubst{\theta_r}{\hsubst{\theta}{(h\app t_1\ldots t_n)}}$.
Since $h$ is unaffected by substitutions, it is easy to see that
$\hsubst{\theta_r}{\hsubst{\theta}{(h\app t_1\ldots t_n)}} = (h \app
(\hsubst{\theta_r}{\hsubst{\theta}{t_1}})\app \ldots \app
  (\hsubst{\theta_r}{\hsubst{\theta}{t_n}}))$.
Picking  $M_i$ to be the term $\hsubst{\theta_r}{\hsubst{\theta}{t_i}}$
for each $i$, $1 \leq i \leq n$, we see that clause (1) in the lemma
is satisfied.

For the second clause we note first that the typing judgements in question must 
be closed and hence the consideration is meaningful.
Consider an arbitrary closed instance of $\hsubstseq{\Psi_{\theta}}{\theta}{\mathcal{S}}$
identified by $\theta'$ and $\sigma$.
If $\subst{\sigma}{\hsubst{\theta'}{\hsubst{\theta_r}{\hsubst{\theta}{F}}}}$
is a valid formula then using the definition of validity as well as
clause (1) in the lemma we can extract a derivation for
$\lfctx{\subst{\sigma}{\hsubst{\theta'}{\hsubst{\theta_r}{\hsubst{\theta}{G}}}}}$
and a derivation of height $k$ for
$\lfchecktype{\subst{\sigma}{\hsubst{\theta'}{\hsubst{\theta_r}{\hsubst{\theta}{G}}}}}
             {\hsubst{\theta'}{(h\app M_1\ldots M_n)}}
             {\hsubst{\theta'}{\hsubst{\theta_r}{\hsubst{\theta}{P}}}}$.
Application of Theorems~\ref{th:atomictype} and~\ref{th:subspermute} 
are then sufficient to conclude
that there is a derivation of height less than $k$ for
$\subst{\sigma}{\hsubst{\theta'}
       {(\lfchecktype{\hsubst{\theta_r}{\hsubst{\theta}{G}}}
                     {M_i}{A'_i})}}.$
For us to be able to conclude that, for each $i$, $1 \leq i \leq n$,
the formula
$\subst{\sigma}{\hsubst{\theta'}
                       {(\fatm{\hsubst{\theta_r}{\hsubst{\theta}{G}}}
                              {M_i: A'_i})}}$  
is valid, it only remains to show that there is a derivation for
$\lftype{\subst{\sigma}{\hsubst{\theta'}{\hsubst{\theta_r}{\hsubst{\theta}{G}}}}}
        {\hsubst{\theta'}{A_i'}}$.
However, this has been done in the proof of Theorem~\ref{th:atomictype}.

We finish by proving the third clause.
From clause (1) we know that 
$\hsubst{\theta_r}{\hsubst{\theta}{R}}$ will be of the form $(h\app M_1\ldots M_n)$,
thus by the definition of \decompseqsans\ there must exist an 
$\mathcal{S'}$ such that 
$\mathcal{S}' = 
    \decompseq{\hsubst{\theta_r}{\hsubst{\theta}{F}}}
              {\hsubstseq{\Psi_\theta}{\theta}{\mathcal{S}}}$.
Suppose $\mathcal{S}'$ is valid.
Consider an arbitrary closed instance of $\hsubstseq{\Psi_{\theta}}{\theta}{\mathcal{S}}$ 
identified by $\theta'$ and $\sigma$.
These same $\theta'$ and $\sigma$ also identify a closed instance of $\mathcal{S}'$
as these sequents only differ in that the assumption formula
$\hsubst{\theta_r}{\hsubst{\theta}{F}}$
has been replaced with the collection of reduced formulas
$\left\{
  \fatm{\hsubst{\theta_r}{\hsubst{\theta}{G}}}{M_i: A'_i}
\ \middle|\ 
  1\leq i\leq n\right\}$.
If any formula in the set of assumption formulas of 
$\subst{\sigma}{\hsubstseq{\emptyset}{\theta'}{\hsubstseq{\Psi_{\theta}}{\theta}{\mathcal{S}}}}$
were not valid then this instance would be vacuously valid, so suppose all
such formulas are valid.
Then in particular, 
$\subst{\sigma}{\hsubst{\theta'}{\hsubst{\theta_r}{\hsubst{\theta}{F}}}}$
must be valid and thus by clause (2), for each $i$, $1\leq i\leq n$, the formula
$\subst{\sigma}{\hsubst{\theta'}{\fatm{\hsubst{\theta_r}{\hsubst{\theta}{G}}}
      {M_i: A'_i}}}$
will be valid.
But then all of the assumption formulas of $\subst{\sigma}{\hsubst{\theta'}{\mathcal{S}'}}$
must be valid and since this is a closed instance of a valid sequent
we can conclude that the goal formula
$\subst{\sigma}{\hsubst{\theta'}{\hsubst{\theta_r}{\hsubst{\theta}{F'}}}}$ 
is valid.
Therefore any closed instance of 
$\hsubstseq{\Psi_{\theta}}{\theta}{\mathcal{S}}$ will be valid, 
and clause (3) in the lemma is satisfied.
\end{proof}

Lemmas~\ref{lem:inst-solun} and \ref{lem:decomp_decr} yield the
following possibility for analyzing the derivability of a sequent with
$\fatm{G}{R:P}$ as an atomic formula:
we use the unification problem identified in
Lemma~\ref{lem:inst-solun} to limit the collection of closed term
substitutions for the sequent and we analyze the derivability of the
reduced sequent under these substitutions.
Unfortunately, this approach would not be very effective in practice.
What we would like to do instead is to use the unification problem
directly to generate the collection of substitutions to be
considered.
Moreover, we would like to be able to limit the substitutions even
from this set that actually need to be considered. 
Towards the latter end, we introduce the idea of a covering set of
solutions to a unification problem.
The next three definitions culminate in a formulation of this notion. 

\begin{definition}[Restricted Substitutions]
The restriction of a substitution $\theta$ to the arity typing context $\Psi$ is
the substitution
$
\left\{
  \langle x,M,\alpha\rangle\ \middle|\ 
  \langle x,M,\alpha\rangle\in\theta\mbox{ and } x:\alpha'\in\Psi
\right\}.
$
We denote this substitution by $\restrict{\theta}{\Psi}$.
\end{definition}

\begin{definition}[Covering Substitutions]\label{def:covering-subst}
Let $\Psi$, $\Psi_1$ and $\Psi_2$ be collections of type assignments to term
variables, and let $\theta_1$ and $\theta_2$ be substitutions that are
arity type preserving with respect to
$\noms \cup \STLCGamma_0 \cup \Psi_1$ and
$\noms \cup \STLCGamma_0 \cup \Psi_2$, respectively.
Then $\solun{\theta_2}{\Psi_2}$ is said to cover
$\solun{\theta_1}{\Psi_1}$ relative to $\Psi$ if there exists a pair 
$\solun{\theta_3}{\Psi_3}$ of a substitution and a collection of arity
type assignments to term variables such that
\begin{enumerate}
\item $\theta_3$ is type preserving with respect to 
$\noms\cup\STLCGamma_0\cup\Psi_3$,
\item for any
$x : \alpha \in \Psi_2$, if $x : \alpha' \in \aritysum{\context{\theta_3}}{\Psi_3}$ then
  $\alpha = \alpha'$, and 
\item The substitutions $\restrict{\theta_1}{\Psi}$ and 
$\restrict{(\comp{\theta_3}{\theta_2})}{\Psi}$ are identical.
\end{enumerate} 
Note that the second condition ensures that $\noms \cup \STLCGamma_0
\cup ((\Psi_2\setminus\context{\theta_3}))\cup \Psi_3)$ determines a
valid arity context and that $\theta_2$ and $\theta_3$ are arity type 
compatible with respect to this context.
Thus, the composition of $\theta_2$ and
$\theta_3$ in the third condition is well-defined.
\end{definition}

\begin{definition}[Covering Set of Solutions]\label{def:covering-solns}
A collection $S$ of solutions to a unification problem
$\mathcal{U} = \unif{\mathbb{N}}{\Psi}{\mathcal{E}}$ 
is said to be \emph{covering set of solutions for $\mathcal{U}$} if every
solution to $\mathcal{U}$ is covered by some solution in $S$ relative
to $\Psi$.
\end{definition}

We now show the soundness of using the reduced forms of a sequent
generated by just a covering set of solutions for the relevant
unification problem in carrying out an analysis of the derivability of
the sequent.

\begin{lemma}\label{lem:covers-seq}
Suppose that $\theta_1$ and $\sigma$ identify a closed instance
$\mathcal{S}'$ of a well-formed sequent $\mathcal{S} =
\seq[\mathbb{N}]{\Psi}{\Xi}{\Omega}{F'}$.  
Suppose further that $\solun{\theta_2}{\Psi_2}$ is substitution
compatible with $\mathcal{S}$ and such that it covers 
$\solun{\theta_1}{\emptyset}$ relative to $\Psi$.
Then there is a term substitution $\theta$ that together with $\sigma$
identifies a closed instance of
$\hsubstseq{\Psi_2}{\theta_2}{\mathcal{S}}$ that is valid if and only
if $\mathcal{S}'$ is.
\end{lemma}
\begin{proof}
We argue below that, under the assumptions of the lemma, there is a
substitution $\theta_3$ and a term 
variable context $\Psi_3$ such that $\solun{\theta_3}{\Psi_3}$ is
substitution compatible with
$\hsubstseq{\Psi_2}{\theta_2}{\mathcal{S}}$ and for any term $M$ 
such that $\stlctyjudg{\mathbb{N}\cup\STLCGamma_0\cup\Psi}{M}{\alpha}$
is derivable it is the case that 
$\hsubst{\theta_3}{\hsubst{\theta_{2r}}{\hsubst{\theta_2}{M}}} =
\hsubst{\theta_1}{M}$, where $\theta_{2r}$ is the raising substitution
associated with the application of $\theta_2$ to $\mathcal{S}$ relative to $\Psi_2$.
It follows from this that the formulas and context types appearing in
$\hsubstseq{\Psi_3}{\theta_3}{\hsubstseq{\Psi_2}{\theta_2}{S}}$ must be
identical to the ones in $\hsubstseq{\emptyset}{\theta_1}{S}$.
It is then easily seen that $\theta_3$ can be extended into a
substitution $\theta$ that together with $\sigma$ identifies a closed
instance of $\hsubstseq{\Psi_2}{\theta_2}{\mathcal{S}}$ whose
formulas are identical to those of $\mathcal{S}'$.
The lemma is an immediate consequence.

Since $\solun{\theta_2}{\Psi_2}$ covers $\solun{\theta_1}{\emptyset}$ relative to $\Psi$
there exists some $\solun{\theta_4}{\Psi_4}$ satisfying the conditions of
Definition~\ref{def:covering-subst}.
We claim that we can further assume of this $\theta_4$ that
(1)~$\context{\theta_4}=(\Psi\setminus\context{\theta_2})\cup\Psi_2$
and
(2)~$\supportof{\theta_4}\cap\mathbb{N}=\emptyset$.
The first of these may be violated because $\theta_4$ may not instantiation some
variables from $\Psi_2$ and it may also include instantiations for variables
which are not contained in $(\Psi\setminus\context{\theta_2})\cup\Psi_2$.
For the former we extend $\theta_4$ with $\langle x,x,\alpha\rangle$ and
the $\Psi_4$ with $x:\alpha$. 
For the latter we simply drop the instantiation from $\theta_4$.
It is straightforward to conclude that such changes to $\theta_4$ will not violate any of the
conditions of Definition~\ref{def:covering-subst}.
To address the second condition, 
consider a permutation $\pi$ which renames nominal constants in 
$\supportof{\theta_4}\cap\mathbb{N}$ to new names chosen away from 
$\mathbb{N}\cup\supportof{\theta_1}\cup\supportof{\theta_2}$.
The first two conditions of Definition~\ref{def:covering-subst} are obviously
satisfied by $\permute{\pi}{\theta_4}$ and $\Psi_4$, so consider an arbitrary 
$\langle x,M,\alpha\rangle$ in $\restrict{\theta_1}{\Psi}$.
By the definition of a composition of substitutions,
since this same $\langle x,M,\alpha\rangle$ must appear in 
$\restrict{(\comp{\theta_4}{\theta_2})}{\Psi}$
either (1) $\langle x, M, \alpha\rangle\in\theta_4$ or 
(2) $\hsub{\theta_4}{M'}{M}$ for $\langle x, M',\alpha\rangle\in\theta_2$.
Observe that since $\solun{\theta_1}{\emptyset}$ and $\solun{\theta_2}{\Psi_2}$ are
both substitution compatible with $\mathcal{S}$, neither substitution will contain
any instances of nominal constants appearing in $\mathbb{N}$.
Thus for the former case, $\permute{\pi}{M}=M$ and
$\langle x,M,\alpha\rangle\in\permute{\pi}{\theta_4}$.
For the latter case, a simple inductive argument on the structure of terms permits 
us to conclude that
$\permute{\pi}{(\hsubst{\theta_4}{M'})}=
    \hsubst{\permute{\pi}{\theta_4}}{(\permute{\pi}{M'})}$.
But neither $M$ nor $M'$ contain any nominal constants from $\mathbb{N}$
 and thus
$\hsub{\permute{\pi}{\theta_4}}{M'}{M}$.
Therefore $\restrict{\theta_1}{\Psi}$
is also identical to $\restrict{(\comp{\permute{\pi}{\theta_4}}{\theta_2})}{\Psi}$.
Let $\theta_3'$ and $\Psi_3$ denote the $\permute{\pi}{\theta_4}$ and $\Psi_4$
which satisfy both (1) and (2).

We now use Theorem~\ref{th:raised-subs} to obtain a ``raised'' version
of $\theta_3'$ that together with $\Psi_3$ will constitute the
pair $\langle \theta_3,\Psi_3 \rangle$ that we desired at the outset.
Specifically, the theorem allows us to conclude that there is a 
substitution $\theta_3$ satisfying the following properties:
\begin{enumerate}
\item $\supportof{\theta_3}$ is disjoint from $\mathbb{N} \cup
  \supportof{\theta_2}$,
\item $\context{\theta_3}$ is identical to the
raised version of $(\Psi\setminus\context{\theta_2})\cup\Psi_2$
corresponding to the raising substitution $\theta_{2r}$,
\item $\theta_3$ is arity type preserving with respect to
  $\noms\cup\STLCGamma\cup\Psi_3$, and
\item for every term $M$ such that
$\stlctyjudg{\mathbb{N}\cup\STLCGamma_0\cup\Psi}{M}{\alpha}$,
$\hsubst{\theta'_3}{\hsubst{\theta_2}{M}} =
\hsubst{\theta_3}{\hsubst{\theta_{2r}}{\hsubst{\theta_2}{M}}}$.
\end{enumerate}
The argument for the first three of these properties is obvious.
For the last property, we observe, using Theorem~\ref{th:aritysubs},
that
$\stlctyjudg{(\mathbb{N}\cup\supportof{\theta_2})
              \cup\STLCGamma_0\cup((\Psi\setminus\context{\theta_2})\cup
\Psi_2)}{\hsubst{\theta_2}{M}}{\alpha}$ has a derivation under the
condition described; 
Theorem~\ref{th:raised-subs} can then be invoked in an obvious way.
It follows immediately from the first three properties that $\langle
\theta_3,\Psi_3 \rangle$ is substitution compatible with
$\hsubstseq{\Psi_2}{\theta_2}{\mathcal{S}}$.
It therefore only remains to show that for every $M$ such that
$\stlctyjudg{\mathbb{N}\cup\STLCGamma_0\cup\Psi}{M}{\alpha}$, it is
the case that $\hsubst{\theta_1}{M} =
\hsubst{\theta_3}{\hsubst{\theta_{2r}}{\hsubst{\theta_2}{M}}}$.
An easy inductive argument shows that for any $M$ of the kind
described and any $\theta$, if $\hsubst{\theta}{M}=M'$ is derivable
exactly when $\hsubst{\restrict{\theta}{\Psi}}{M}=M'$ is derivable.
It follows from this that
$\hsubst{\theta_1}{M} = \hsubst{\theta'_3}{\hsubst{\theta_2}{M}}$.
Property (4) then yields the desired result.
\end{proof}

We now use the observations in Lemmas~\ref{lem:heads-cover},
\ref{lem:inst-solun}, \ref{lem:decomp_decr} and \ref{lem:covers-seq}
to describe a complete analysis of the derivability of a sequent
around an atomic assumption formula.

\begin{definition}[Cases Elaboration]\label{def:def-cases}
Let $\mathcal{S}=\seq[\mathbb{N}]{\Psi}{\Xi}{\Omega}{F'}$ be a
well-formed sequent, let $F = \fatm{G}{R:P}$ be a formula in $\Omega$
and let $h:\typedpi{x_1}{A_1}{\ldots\typedpi{x_n}{A_n}}{P'}$ be a type
assignment in $\Sigma$ or in the explicit bindings in $G$.
Further, for each $i$, $1 \leq i \leq n$, let $z_i :\alpha_i$ be a
distinct variable of type $\erase{A_i}$ away from $\Psi$ and raised
over $\mathbb{N}$, and let $t_i$ be the generalized variable term
corresponding to $z_i$. 
Finally, let
$\mathcal{U}$ be the unification problem 
\begin{tabbing}
\qquad\quad\=\quad\qquad\qquad\=\kill
\>
$\langle \mathbb{N}; 
         \Psi \cup \{z_1:\alpha_1,\ldots, z_n:\alpha_n\};$\\
\>\>
$\left\{P = \hsubst{\{\langle x_1,t_1,\erase{A_1}\rangle,
            \ldots, \langle x_n,t_n,\erase{A_n}\rangle \}}{P'}, 
         R = (h \app t_1 \app \ldots\app h_n)\right\}\rangle$
\end{tabbing}
and let $C$ be a covering set of solutions for $\mathcal{U}$.
Then the \emph{analysis of $\mathcal{S}$ based on $F$ and $h$} is
denoted by 
$\makecases[F]{\mathcal{S}}{h:A}$ and is given by the set of sequents
\[
\left\{
\decompseq{F'}
          {\hsubstseq{\Psi_{\theta}}{\theta}{\mathcal{S}}}
\ \middle| \ 
\begin{array}{p{3in}}
$\solun{\theta}{\Psi_{\theta}}\in C$
 and $F'$ is the formula in
 $\hsubstseq{\Psi_{\theta}}{\theta}{\mathcal{S}}$ resulting from $F$
\end{array}
\right\}
.\]
If $\mathcal{S}$ is a well-formed sequent and $F=\fatm{G}{R:P}$ is an
assumption formula in $\mathcal{S}$, then the \emph{complete analysis
  of $\mathcal{S}$ based on $F$} is the set of sequents
\begin{tabbing}
\qquad\qquad\=$\bigcup \{ \makecases[F]{\mathcal{S}'}{h:A}\ \vert\ $\=\kill
\>$\bigcup \{
\makecases[F']{\mathcal{S}'}{h:A}\ \vert\ (\mathcal{S}';\ h:A)\in\hds{\mathcal{S}}{G}$\\
\>\>$\mbox{and}\ F'\ \mbox{is the formula in}\ \mathcal{S'}\ \mbox{resulting
  from}\ F \}.$
\end{tabbing} 
This collection is denoted by $\casesfn[F]{\mathcal{S}}$.
Note that the notations $\makecases[F]{\mathcal{S}}{h:A}$ and
$\casesfn[F]{\mathcal{S}}$ are both ambiguous---for instance, the
first notation leaves out mention of the covering set of solutions
that plays a role in generating the set it denotes.
We will assume them to denote any of the set of sequents that can be
generated in the respective ways described in this definition.
\end{definition}

We show first that all the sequents in the set yielded by
\casessans\ will be well-formed. 

\begin{theorem}\label{th:cases-seq-ok}
If $\mathcal{S}$ is a well-formed sequent and $F$ is an atomic
assumption formula appearing in $\mathcal{S}$, then every sequent in
$\casesfn[F]{\mathcal{S}}$ is well-formed.
\end{theorem}
\begin{proof}
Let $F=\fatm{G}{R:P}$.
%
By Lemma~\ref{lem:heads-wf} we know that every 
$\langle\mathcal{S}',h:A\rangle\in\hds{\mathcal{S}}{G}$
is such that $\mathcal{S}'$ is well-formed and letting $\mathbb{N}'$ and $\Psi'$
be the support set and arity typing context of $\mathcal{S}'$ respectively,
$\wftype{\mathbb{N}'\cup\STLCGamma_0\cup\Psi'}{A}$ has a derivation.
Letting $F'$ be the formula from $\mathcal{S}'$ corresponding to $F$,
we will denote the unification problem from $\makecases[F']{\mathcal{S}'}{h:A}$
by $\mathcal{U}_{(\mathcal{S'}, h:A)}$.
We can conclude by  Theorem~\ref{th:seq-term-subs-ok} that the application of any solution
for $\mathcal{U}_{(\mathcal{S'},h:A)}$ to $\mathcal{S}'$ must be well-formed.
We observe that for any formula $\fatm{G}{h\app M_1\ldots M_n:P}$
which is well-formed with respect to $\mathbb{N}\cup\STLCGamma_0\cup\Psi$ and 
$\ctxsanstype{\Xi}$, if 
$h:\typedpi{x_1}{A_1}{\ldots\typedpi{x_n}{A_n}{P'}}$ appears in $\Sigma$
or the explicit bindings of $G$ then, for $1\leq i\leq n$, we can extract derivations that
$\fatm{G}
      {M_i:\hsubst{\{\langle x_1,M_1,\erase{A_1}\rangle,\ldots,
                     \langle x_{i-1},M_{i-1},\erase{A_{i-1}}\rangle\}}
                  {A_i}}$
are well-formed formulas with respect to the same 
$\mathbb{N}\cup\STLCGamma_0\cup\Psi$ and $\ctxsanstype{\Xi}$.
From this observation we conclude that the reduced form of
each sequent obtained by the application of a solution to 
$\mathcal{U}_{(\mathcal{S'},h:A)}$ will be well-formed, and therefore that
all sequents in $\casesfn[F]{\mathcal{S}}$ must be well-formed.
\end{proof}

We show next that the validity of every sequent in the result of
\casessans\ ensures the validity of the original sequent.

\begin{theorem}\label{th:cases-cover}
Let $\mathcal{S}$ be a well-formed sequent and let $F$ be an atomic
assumption formula in $\mathcal{S}$.
If all the sequents in $\casesfn[F]{\mathcal{S}}$ are valid then
$\mathcal{S}$ must be valid.
\end{theorem}
\begin{proof}
Let $F=\fatm{G}{R:P}$.
Consider an arbitrary closed instance of the sequent $\mathcal{S}$ identified
by $\theta$ and $\sigma$.
If any of the assumption formulas of 
$\subst{\sigma}{\hsubstseq{\emptyset}{\theta}{\mathcal{S}}}$ 
were not valid then this closed instance is vacuously valid, so it remains only
to consider those instances for which they are all valid.
If all these assumption formulas are valid, then
$\subst{\sigma}{\hsubst{\theta}{F}}$ is a valid formula.
Therefore by Lemma~\ref{lem:heads-cover} there will exist a pair
$\langle\mathcal{S}',h:A\rangle\in\hds{\mathcal{S}}{G}$ such that,
letting $F'$ be the formula in $\mathcal{S}'$ corresponding to $F$,
$h:A$ appears in $\Sigma$ or among the explicit bindings of 
the context expression of $F'$ with respect to $\mathcal{S}'$, 
and there is a closed instance of $\mathcal{S}'$ identified by
$\theta'$ and $\sigma'$ which is equivalent to
$\subst{\sigma}{\hsubstseq{\emptyset}{\theta}{\mathcal{S}}}$ 
under some permutation $\pi$.
Clearly then, if the closed instance 
$\subst{\sigma'}{\hsubstseq{\emptyset}{\theta'}{\mathcal{S}'}}$
is valid, $\subst{\sigma}{\hsubstseq{\emptyset}{\theta}{\mathcal{S}}}$ 
will be valid by Theorem~\ref{th:perm-valid}.
By Lemma~\ref{lem:inst-solun} $\solun{\theta'}{\emptyset}$ must be a solution to the
unification problem corresponding to the head pair $\langle\mathcal{S}',h:A\rangle$,
and so there must be some solution $\solun{\theta_1}{\Psi_1}$ in the
covering set of solutions $C$ used by \makecasessans\ which covers it.
By Lemma~\ref{lem:decomp_decr} then, there is some sequent 
$\mathcal{S}''$ such that, letting $F''$ be the formula in 
$\hsubstseq{\Psi_1}{\theta_1}{\mathcal{S}'}$ corresponding to $F'$,
$\mathcal{S}''=\decompseq{\hsubstseq{\Psi_1}{\theta_1}{\mathcal{S}'}}{F''}$ and
$\hsubstseq{\Psi_1}{\theta_1}{\mathcal{S}'}$ is valid if $\mathcal{S}''$ is.
By assumption every sequent in $\casesfn[F]{\mathcal{S}}$ is valid, and so
$\mathcal{S}''$ will be valid and thus $\hsubstseq{\Psi_1}{\theta_1}{\mathcal{S}'}$
as well.
But by Lemma~\ref{lem:covers-seq}, there exists some $\solun{\theta_2}{\emptyset}$
which together with $\sigma'$ identifies a closed instance of
$\hsubstseq{\Psi_1}{\theta_1}{\mathcal{S}'}$
which is valid if and only if
$\subst{\sigma'}{\hsubstseq{\emptyset}{\theta'}{\mathcal{S}'}}$ is.
Thus by the validity of $\hsubstseq{\Psi_1}{\theta_1}{\mathcal{S}'}$ the closed instance
$\subst{\sigma'}{\hsubstseq{\emptyset}{\theta'}{\mathcal{S}'}}$
will be valid, and we have already seen that this will ensure that
$\subst{\sigma}{\hsubstseq{\emptyset}{\theta}{\mathcal{S}}}$ is valid.
Therefore $\mathcal{S}$ will be valid, since all of its closed instances are valid.
\end{proof}

%% file: proof-system/atomic-rules.tex
\subsection{Proof Rules that Introduce Atomic Formulas}

\begin{figure}
\centering
$\begin{array}{cc}
\infer[\appL]
      {\seq[\mathbb{N}]{\Psi}
           {\Xi}
           {\setand{\Omega}{\fatm{G}{R:P}}}
           {F}}
      {\casesfn{\seq[\mathbb{N}]{\Psi}
                    {\Xi}
                    {\setand{\Omega}{\fatm{G}{R:P}}}
                    {F}}}
\\\ &\ 
\end{array}$
$\begin{array}{c}
\infer[\appR]
      {\seq[\mathbb{N}]{\Psi}
           {\Xi}
           {\Omega}
           {\fatm{G}{h\app M_1\ldots M_n:P'}}}
      {\begin{array}{c}
           h:\typedpi{x_1}{A_1}{\ldots
                      \typedpi{x_n}{A_n}{P}}\in\Sigma\mbox{ or the explicit bindings in }G
             \\
           \fatm{G}{N:B}\in\Omega \qquad
           \hsub{\{\langle x_1, M_1, \erase{A_1}\rangle,\ldots,\langle x_n, M_n, \erase{A_n}\rangle\}}
                {P}
                {P'}
         \\
         \left\{
         \begin{array}{l}
           \seq[\mathbb{N}]
               {\Psi}
               {\Xi}
               {\Omega}
               {}
         \\\quad
               \fatm{G}
                    {M_i:
                       \hsubst{\{\langle x_1, M_1, \erase{A_1}\rangle,\ldots,
                                 \langle x_{i-1}, M_{i-1}, \erase{A_{i-1}}\rangle\}}
                              {A_i}}
          \\\hspace{9cm}
         \ \mid\ 1\leq i\leq n
         \end{array}
         \right\}
       \end{array}}
\\\ 
\end{array}$
\[
\infer[\absL]
      {\seq[\mathbb{N}]{\Psi}{\Xi}{\setand{\Omega}{\fatm{G}{\lflam{x}{M}:\typedpi{x}{A_1}{A_2}}}}{F}}
      {\begin{array}{c}
           n\not\in\dom{\mathbb{N}}
         \\
         \Xi'=
           \left\{\begin{array}{cl}
             \left(\Xi
                 \setminus
                 \left\{\ctxvarty{\Gamma}{\mathbb{N}_{\Gamma}}{\ctxty{\mathcal{C}}{\mathcal{G}}}\right\}\right)
             \cup 
             \left\{\ctxvarty{\Gamma}{(\mathbb{N}_{\Gamma},n:\erase{A_1})}{\ctxty{\mathcal{C}}{\mathcal{G}}}\right\}
               & \mbox{if }\Gamma\mbox{ in }G 
             \\
             \Xi & \mbox{otherwise}
           \end{array}\right.
         \\
         \begin{array}{l}
             \mathbb{N},n:\erase{A_1};\Psi;\Xi';
           \\\qquad
             \setand{\Omega}
                    {\fatm{G,n:A_1}
                          {\hsubst{\{\langle x,n,\erase{A_1}\rangle\}}{M}:
                              \hsubst{\{\langle x,n,\erase{A_1}\rangle\}}{A_2}}}
             \longrightarrow F
         \end{array}
       \end{array}}
\]
\[
\infer[\absR]
      {\seq[\mathbb{N}]{\Psi}{\Xi}{\Omega}{\fatm{G}{\lflam{x}{M}:\typedpi{x}{A_1}{A_2}}}}
      {\begin{array}{c}
           n\not\in\dom{\mathbb{N}}
         \\
         \Xi'=
           \left\{\begin{array}{cl}
             \left(\Xi
                 \setminus
                 \left\{\ctxvarty{\Gamma}{\mathbb{N}_{\Gamma}}{\ctxty{\mathcal{C}}{\mathcal{G}}}\right\}\right)
             \cup 
             \left\{\ctxvarty{\Gamma}{(\mathbb{N}_{\Gamma},n:\erase{A_1})}{\ctxty{\mathcal{C}}{\mathcal{G}}}\right\}
               & \mbox{if }\Gamma\mbox{ in }G 
             \\
             \Xi & \mbox{otherwise}
           \end{array}\right.
         \\
         \begin{array}{l}
             \seq[\mathbb{N},n:\erase{A_1}]{\Psi}{\Xi'}{\Omega}{}
           \\\qquad
             \fatm{G,n:A_1}{\hsubst{\{\langle x,n,\erase{A_1}\rangle\}}{M}:\hsubst{\{\langle x,n,\erase{A_1}\rangle\}}{A_2}}
         \end{array}
       \end{array}}
\]
\caption{Proof Rules Interpreting Atomic Formulas}
\label{fig:rules-atom}
\end{figure}

We are finally in a position to describe rules in our logic that
internalize the analysis of typing derivations in LF.
We do this in Figure~\ref{fig:rules-atom}.
Reasoning about atomic formulas over atomic terms, we base
our reasoning step on the result of Theorem~\ref{th:atomictype}.
When analysing a goal formula we require that the head of the atomic term
is known and thus a single structure is possible; when an assumption formula
we rely on the analysis of $\casessans$ to identify all the ways such 
judgements might have been derived.
When reasoning about an atomic formula for an abstraction term, 
we introduce a new nominal constant
to represent the fresh binding of the LF typing rule for abstractions and
extend the context of the formula with this new binding.
In LF this name is distinct and we capture this requirement in the logic
by ensuring that no instantiations for existing eigenvariables or for
context variables in the formula may use this name directly.

The following theorem shows that these rules require the proof of only
well-formed sequents in constructing a proof of a well-formed sequent.

\begin{theorem}\label{th:atom-wf}
The following property holds for each rule in
Figure~\ref{fig:rules-atom}: if the conclusion sequent is
well-formed, the premises expressing typing conditions have
derivations and the conditions expressed by the other, non-sequent
premises are satisfied, then all the sequent premises must
be well-formed.
\end{theorem}
\begin{proof}
Consider each of the rules in Figure~\ref{fig:rules-atom}.

\case{\appL}
For this rule, the well-formedness of the premise sequents is assured by
Theorem~\ref{th:cases-seq-ok} given that the conclusion sequent is well-formed 
by assumption.

\case{\appR}
Let $\mathcal{S}_1$ $\ldots$, $\mathcal{S}_n$
denote the premise sequents for $A_1$, $\ldots$, $A_n$
respectively, and assume that the conclusion sequent is well-formed.
Since the context variable context and assumption formulas remain the same in
these premise sequents, we can infer they are well-formed from the well-formedness
of the conclusion sequent; what remains to be shown is that the goal formulas
of these premise sequents are well-defined and will be well-formed.
For each $i$, $1\leq i\leq n$, let 
$\theta_i$ denote 
$\{\langle x_1,M_1,\erase{A_1}\rangle,\ldots,\langle x_{i-1},M_{i-1},\erase{A_{i-1}}\rangle\}$.
By the well-formedness of the conclusion sequent, 
$\fatm{G}{h\app M_1\ldots M_n:P'}$ is a 
well-formed formula with respect to $\mathbb{N}\cup\STLCGamma_0\cup\Psi$ and
$\ctxsanstype{\Xi}$.
Thus there must exist a derivation of the judgement
$\stlctyjudg{\mathbb{N}\cup\STLCGamma_0\cup\Psi}{M_i}{\erase{A_i}}$
for each $i$, $1\leq i\leq n$.
Therefore the substitutions $\theta_i$ will all be arity type preserving
with respect to $\mathbb{N}\cup\STLCGamma_0\cup\Psi$ and so the premise sequents
will be well-defined.
For each $i$, $1\leq i\leq n$, let $A_i'$ denote the type
$\hsubst{\theta_i}{A_i}$.
Since $h$ is bound in $\Sigma$ or the explicit bindings in $G$ relative to the
conclusion sequent, we further extract derivations of 
$\wftype{\mathbb{N}\cup\STLCGamma_0\cup\Psi_i}{A_i}$
for each $i$.
By Theorem~\ref{th:erasure} $\erase{A_i}=\erase{A_i'}$, and therefore
$\stlctyjudg{\mathbb{N}\cup\STLCGamma_0\cup\Psi}{M_i}{\erase{A_i'}}$
has a derivation for each $i$.
A straightforward inner induction on $i$ showing $\theta_i$ is
arity type preserving with respect to
$\mathbb{N}\cup\STLCGamma_0\cup\Psi$ permits us to conclude 
through Theorem~\ref{th:aritysubs} that for each $i$,
$\wftype{\mathbb{N}\cup\STLCGamma_0\cup\Psi}{A_i'}$
is derivable.
Thus we can construct a derivation for
$\wfform{\mathbb{N}\cup\STLCGamma_0\cup\Psi}{\ctxsanstype{\Xi}}{\fatm{G}{M_i:A_i'}}$
for each $i$, $1\leq i\leq n$, and therefore each premise sequent
$\mathcal{S}_i$ will clearly be well-formed.

\case{\absL\ and \absR}
We first observe that the substitution $\{\langle x,n,\erase{A_1}\rangle\}$
is obviously arity type preserving by construction and so the premise sequents
will be well-defined.
Let $A_2'$ denote the type $\hsubst{\{\langle x,n,\erase{A_1}\rangle\}}{A_2}$ and
$M'$ denote the term $\hsubst{\{\langle x,n,\erase{A_1}\rangle\}}{M}$.
By the well-formedness of the conclusion sequent we know that
\begin{enumerate}
\item for each $\ctxvarty{\Gamma_i}{\mathbb{N}_i}{\ctxty{\mathcal{C}_i}{\mathcal{G}_i}}\in\Xi$,
$\wfctxvarty{\mathbb{N}\setminus\mathbb{N}_i}{\Psi}{\ctxty{\mathcal{C}_i}{\mathcal{G}_i}}$
has a derivation and
\item for each $F'\in\Omega\cup\{\fatm{G}{\lflam{x}{M}:\typedpi{x}{A_1}{A_2}},F\}$ 
(resp. $\Omega\cup\fatm{G}{\lflam{x}{M}:\typedpi{x}{A_1}{A_2}}$ for \absR),
$\wfform{\mathbb{N}\cup\STLCGamma_0\cup\Psi}{\ctxsanstype{\Xi}}{F'}$
has a derivation.
\end{enumerate}

We first consider the well-formedness of the context variable types.
For every binding
$\ctxvarty{\Gamma_i}{\mathbb{N}_i}{\ctxty{\mathcal{C}_i}{\mathcal{G}_i}}\in\Xi$
which does not appear in $G$, it is obvious by Theorem~\ref{th:ctx-ty-wk}
that this context variable type will remain well-formed under the extended 
support set $\mathbb{N},n:\erase{A_1}$.
For any $\Gamma_i$ which appears in $G$, we observe that 
$\mathbb{N}\setminus\mathbb{N}_i$
is the same set as
$(\mathbb{N},n:\erase{A_1})\setminus(\mathbb{N}_i,n:\erase{A_1})$
and thus the declaration
$\ctxvarty{\Gamma_i}{(\mathbb{N}_i,n:\erase{A_1})}{\ctxty{\mathcal{C}_i}{\mathcal{G}_i}}$
must also be well-formed.

We now consider the well-formedness of formulas in the sequent.
We first observe that $\ctxsanstype{\Xi}$ is the same as
$\ctxsanstype{\Xi'}$, thus these sets can be used interchangeable.
By Theorem~\ref{th:wf-form-wk} we easily determine that the formulas in 
$\Omega\cup\{F\}$ (resp. $\Omega$ for \absR) will all be well-formed relative 
to the extended support set $\mathbb{N},n:\erase{A_1}$ and $\ctxsanstype{\Xi'}$.
It remains only to show the well-formedness of the formula
$\fatm{G,n:A_1}{M':A_2'}$.
By assumption $\fatm{G}{\lflam{x}{M}:\typedpi{x}{A_1}{A_2}}$ is well-formed
with respect to 
$\mathbb{N}\cup\initctx\cup\Psi$ and $\ctxsanstype{\Xi}$, and 
therefore that $G$ is a well-formed context, 
$\typedpi{x}{A_1}{A_2}$ a well-formed type, and 
$\lflam{x}{M}$ a well-formed term of type $\erase{\typedpi{x}{A_1}{A_2}}$
with respect to this same $\mathbb{N}\cup\initctx\cup\Psi$ and $\ctxsanstype{\Xi}$.
For a new $n:\erase{A_1}$ then, we can extract derivations showing
$(G,n:A_1)$ is a well-formed context, $A_2'$ a 
well-formed type, and $M'$ a well-formed term of type $\erase{A_2}$
with respect to $(\mathbb{N},n:\erase{A_1})\cup\initctx\cup\Psi$ and 
$\ctxsanstype{\Xi}$.
But then obviously $\fatm{G,n:A_1}{M':A_2'}$ is well-formed
with respect to $(\mathbb{N},n:\erase{A_1})\cup\initctx\cup\Psi$ and 
$\ctxsanstype{\Xi'}$.
Therefore the premise sequent must be well-formed.
\end{proof}

We now establish the soundness of these rules.

\begin{theorem}\label{th:atom-sound}
The following property holds for every instance of each of the rules
in Figure~\ref{fig:rules-atom}: if the premises expressing
typing judgements are derivable, the conditions described in the other
non-sequent premises are satisfied and the premise sequents are valid,
then the conclusion sequent must also be valid. 
\end{theorem}
\begin{proof}
Consider each rule in Figure~\ref{fig:rules-atom}.

\case{\appL}
The soundness of this rule is an immediate consequence of Theorem~\ref{th:cases-cover}.

\case{\appR}
Let $\mathcal{S}_1,\ldots,\mathcal{S}_n$ denote 
the premise sequents corresponding to $A_1,\ldots,A_n$ respectively
and for each $i$, $1\leq i\leq n$ let $A_i'$ denote 
$\hsubst{\{\langle x_1,\hsubst{\theta}{M_1},\erase{A_1}\rangle,\ldots,
           \langle x_{i-1},\hsubst{\theta}{M_{i-1}},\erase{A_{i-1}}\rangle\}}
        {A_i}$.
By assumption all of the premise sequents are valid.
Consider an arbitrary closed instance of the conclusion sequent
identified by $\theta$ and $\sigma$.
If any formula in $\subst{\sigma}{\hsubst{\theta}{\Omega}}$ were not
valid then this instance would be vacuously valid, so assume they
are all valid.
This $\theta$ and $\sigma$ clearly also identify closed instances
for each $\mathcal{S}_i$.
Since all formulas in $\subst{\sigma}{\hsubst{\theta}{\Omega}}$
are valid, $\subst{\sigma}{\hsubst{\theta}{\fatm{G}{N:B}}}$ in particular must
be valid and therefore there will exist a derivation for 
$\lfctx{\subst{\sigma}{\hsubst{\theta}{G}}}$.
Also using the validity of $\subst{\sigma}{\hsubst{\theta}{\Omega}}$, 
the validity of the premise sequents ensures that for 
each $i$, $1\leq i\leq n$,
$\subst{\sigma}{\hsubst{\theta}{\fatm{G}{M_i:A_i'}}}$ is valid.
Using Theorem~\ref{th:subspermute} to permute the substitution application to $A_i$ 
this formula can be written as
$\fatm{\subst{\sigma}{\hsubst{\theta}{G}}}
      {\hsubst{\theta}{M_i} : 
         \hsubst{\{\langle x_1,\hsubst{\theta}{M_1},\erase{A_1}\rangle,\ldots,
                   \langle x_{i-1},\hsubst{\theta}{M_{i-1}},\erase{A_{i-1}}\rangle\}}
                {\hsubst{\theta}{A_i}}}.$
By assumption, $h$ is bound in either $\Sigma$ or the explicit bindings of $G$ 
with respect to the conclusion sequent.
Thus since $\subst{\sigma}{\hsubst{\theta}{G}}$ must be a well-formed LF context,
there will exist a derivation in LF for
$\lftype{\subst{\sigma}{\hsubst{\theta}{G}}}
        {\hsubst{\theta}{(\typedpi{x_1}{A_1}{\ldots\typedpi{x_n}{A_n}{P}})}}$.
Given the validity of the formulas expressed above, the substitution
$\{\langle x_1,\hsubst{\theta}{M_1},\erase{A_1}\rangle,\ldots,
   \langle x_n,\hsubst{\theta}{M_{i-1}},\erase{A_n}\rangle\}$
is clearly arity type preserving with respect to $\noms\cup\STLCGamma_0$
and therefore by Theorem~\ref{th:aritysubs} and another use of 
Theorem~\ref{th:subspermute} there is a derivation of
$\lftype{\subst{\sigma}{\hsubst{\theta}{G}}}
        {\hsubst{\theta}{P'}}$.
We conclude by Theorem~\ref{th:atomictype} that there is a derivation for
$\lfchecktype{\subst{\sigma}{\hsubst{\theta}{G}}}
             {\hsubst{\theta}{(h\app M_1\ldots M_n)}}
             {\hsubst{\theta}{P'}}$, 
in LF and it is then easy to see that
$\subst{\sigma}{\hsubst{\theta}{\fatm{G}{h\app M_1\ldots M_n:P'}}}$
will be valid.
Thus the conclusion sequent must be valid.

\case{\absL}
Let $\mathcal{S}$ denote the conclusion sequent and $\mathcal{S}'$
the premise sequent.
Since $\mathcal{S}$ is well-formed and $n:\erase{A_1}\not\in\mathbb{N}$
we can conclude that $n$ does not appear anywhere in $\mathcal{S}$.
Suppose that $\mathcal{S}'$ is valid.
Consider an arbitrary closed instance of $\mathcal{S}$ identified by 
$\theta$ and $\sigma$.
It suffices to consider only those $\theta$ and $\sigma$ which do not contain
any uses of the nominal constant $n:\erase{A_1}$ as any substitutions
using this name can be obtained by applying a permutation $\pi$ to some 
other substitution which does not, and by Theorem~\ref{th:perm-valid} 
the validity of the closed instance identified by the substitutions without $n$ 
will then ensure the validity of the instance identified by the substitutions using $n$.
Since $n$ does not appear in either $\theta$ or $\sigma$, these substitutions
must also identify a closed instance of the sequent $\mathcal{S}'$.
If any formula in 
$\subst{\sigma}{\hsubst{\theta}{(\setand{\Omega}{\fatm{G}{\lflam{x}{M}:\typedpi{x}{A_1}{A_2}}})}}$ 
were not valid then this instance would be vacuously valid, 
so assume they are all valid.
In particular then, the formula
$\subst{\sigma}{\hsubst{\theta}{\fatm{G}{\lflam{x}{M}:\typedpi{x}{A_1}{A_2}}}}$
is valid.
By the definition of validity for formulas there are derivations in LF for
$\lfctx{\subst{\sigma}{\hsubst{\theta}{G}}}$,
$\lftype{\subst{\sigma}{\hsubst{\theta}{G}}}
        {\subst{\sigma}{\hsubst{\theta}{(\typedpi{x}{A_1}{A_2})}}}$, and
$\lfchecktype{\subst{\sigma}{\hsubst{\theta}{G}}}
             {\subst{\sigma}{\hsubst{\theta}{(\lflam{x}{M})}}}
             {\subst{\sigma}{\hsubst{\theta}{(\typedpi{x}{A_1}{A_2})}}}$.
Since the nominal constant $n$ cannot appear anywhere in these judgements
we can extract from these derivations that the judgements
$\lfctx{\subst{\sigma}{\hsubst{\theta}{(G,n:A_1)}}}$,
$\lftype{\subst{\sigma}{\hsubst{\theta}{(G,n:A_1)}}}
        {\subst{\sigma}{\hsubst{\theta}{\hsubst{\{\langle x,n,\erase{A_1}\rangle\}}{A_2}}}}$,
and
$\lfchecktype{\subst{\sigma}{\hsubst{\theta}{(G,n:A_1)}}}
             {\subst{\sigma}{\hsubst{\theta}{\hsubst{\{\langle x,n,\erase{A_1}\rangle\}}{M}}}}
             {\subst{\sigma}{\hsubst{\theta}{\hsubst{\{\langle x,n,\erase{A_1}\rangle\}}{A_2}}}}$
have derivations as well.
Thus
$\subst{\sigma}{\hsubst{\theta}
       {\fatm{G,n:A_1}
             {\hsubst{\{\langle x,n,\erase{A_1}\rangle\}}{M}:
                  \hsubst{\{\langle x,n,\erase{A_1}\rangle\}}{A_2}}}}$ 
will be a valid formula.
But then all of the assumption formulas of 
$\subst{\sigma}{\hsubstseq{\emptyset}{\theta}{\mathcal{S}'}}$ are valid
and by the validity of this sequent we can conclude
that $\subst{\sigma}{\hsubst{\theta}{F}}$ must be valid.
Therefore the closed instance $\subst{\sigma}{\hsubstseq{\emptyset}{\theta}{\mathcal{S}}}$
is valid, and thus we can conclude that $\mathcal{S}$ is a valid sequent.

\case{\absR}
Let $\mathcal{S}$ denote the conclusion sequent and $\mathcal{S}'$
the premise sequent.
Since $\mathcal{S}$ is well-formed and $n:\erase{A_1}\not\in\mathbb{N}$
we can conclude that $n$ does not appear anywhere in $\mathcal{S}$.
Suppose that $\mathcal{S}'$ is valid.
Consider an arbitrary closed instance of $\mathcal{S}$ identified by 
$\theta$ and $\sigma$.
As with $\absL$ it will suffice to consider only those closed instances which
do not use $n:\erase{A_1}$ as permutations of valid closed instances must also
be valid.
This $\theta$ and $\sigma$ will therefore also identify a closed instance of the
premise sequent $\mathcal{S}'$.
If any formula in 
$\subst{\sigma}{\hsubst{\theta}{\Omega}}$ 
were not valid then this instance of $\mathcal{S}$ would be vacuously valid, 
so assume they are all valid.
But then by the validity of $\mathcal{S}'$ the formula
$\subst{\sigma}{\hsubst{\theta}{\fatm{G,n:A_1}{\hsubst{\{\langle x, n,\erase{A_1}\rangle\}}{M}:\hsubst{\{\langle x, n,\erase{A_1}\rangle\}}{A_2}}}}$
will be valid.
Thus there must be LF derivations for
\begin{enumerate}
\item $\lfctx{\subst{\sigma}{\hsubst{\theta}{(G,n:A_1)}}}$,
\item$\lftype{\subst{\sigma}{\hsubst{\theta}{(G,n:A_1)}}}
        {\subst{\sigma}{\hsubst{\theta}{\hsubst{\{\langle x, n,\erase{A_1}\rangle\}}{A_2}}}}$, and
\item$\lfchecktype{\subst{\sigma}{\hsubst{\theta}{(G,n:A_1)}}}
             {\subst{\sigma}{\hsubst{\theta}{\hsubst{\{\langle x, n,\erase{A_1}\rangle\}}{M}}}}
             {\subst{\sigma}{\hsubst{\theta}{\hsubst{\{\langle x, n,\erase{A_1}\rangle\}}{A_2}}}}$.
\end{enumerate}
From these we are able to construct LF derivations for the judgements
$\lfctx{\subst{\sigma}{\hsubst{\theta}{G}}}$,
$\lftype{\subst{\sigma}{\hsubst{\theta}{G}}}
        {\subst{\sigma}{\hsubst{\theta}{(\typedpi{x}{A_1}{A_2})}}}$,
and
$\lfchecktype{\subst{\sigma}{\hsubst{\theta}{G}}}
             {\subst{\sigma}{\hsubst{\theta}{(\lflam{x}{M})}}}
             {\subst{\sigma}{\hsubst{\theta}{(\typedpi{x}{A_1}{A_2})}}}$.
But then
$\subst{\sigma}{\hsubst{\theta}{\fatm{G}{\lflam{x}{M}:\typedpi{x}{A_1}{A_2}}}}$
must be valid by definition.
Therefore the closed instance
$\subst{\sigma}{\hsubstseq{\emptyset}{\theta}{\mathcal{S}}}$ is valid, and
we can conclude $\mathcal{S}$ is a valid sequent.
\end{proof}

%% file: proof-system/induction.tex
\section{An Annotation Based Scheme for Induction}
\label{sec:ind}

In this section we build into the proof system a means for reasoning
by induction on the height of LF derivations.
The idea we use is borrowed from the Abella proof system, specialized
to the context where atomic formulas encapsulate derivability in LF.
In particular, we describe an annotation scheme that allows us to
encode when an atomic formula represents an LF derivation that has a
height less than that of the LF derivation represented by 
an atomic formula that appears in a formula being proved and, hence,
when a property in which this atomic formula appears negatively can be
assumed to hold in an inductive argument.
In the first subsection we will introduce the extension of formula syntax to
include annotations and define a new notion of semantics in relation to this
syntax which is equivalent to the other semantics we have seen when no
annotations occur.
The second subsection will introduce an induction rule which uses
annotations to capture strong induction on the height of LF derivations.
To work with annotated formulas in reasoning we introduce alternative
forms for atomic proof rules as well as \id\ which are applicable
when the formulas have annotations.

\input{proof-system/annotations}
\input{proof-system/ind-rule}
\input{proof-system/addnl-rules}

%% file: proof-system/annotations.tex
\subsection{Extending Formula Syntax with Annotations}
As mentioned above, we annotate particular atomic formulas to indicate
relative heights associated with them.
The annotations that we use go in pairs: $@$ and $*$, $@@$ and $**$,
and so on.
For ease we use $@^n$ (resp. $*^n$) to denote a sequence in which the
character $@$ (resp. $*$) is repeated $n$ times.
We use $\eqannaux{i}{F}$ on an atomic formula $F$ to indicate that it
has a certain height and $\ltannaux{i}{F}$ to indicate that it has a 
strictly smaller height; we will explain what a height means shortly.
This height annotation is decreased whenever we decompose a derivation into 
sub-derivations based on its structure, as is done in the $\appL$ rule
of the previous section.

To understand the meaning of the annotations recall first that an 
atomic formula $\fatm{G}{M:A}$, is valid if then there are LF derivations for
$\lfctx{G}$, $\lftype{G}{A}$ and $\lfchecktype{G}{M}{A}$.
Each of these derivations in LF will have a particular height, and it is the
height of the typing judgement $\lfchecktype{G}{M}{A}$ which forms the basis 
for our induction.
Thus, when we talk of an atomic formula being restricted to a
particular height or heights, we mean the height of the derivation of
this typing judgement. 
In particular, the valid closed instances of the annotated atomic
formula $\eqannaux{i}{\fatm{G}{M:A}}$ are the ones for which the
corresponding instances of $\lfchecktype{G}{M}{A}$ have derivations of
height up to some particular size $m$, while the closed instances of the relatedly
annotated formula $\ltannaux{i}{\fatm{G'}{M':A'}}$ will be valid only
if the corresponding instances of $\lfchecktype{G'}{M'}{A'}$ have 
derivations of a height strictly smaller than $m$.
Having available a denumerable collection of pairs of
such annotations allows us to simultaneously relate the heights of
different pairs of atomic formulas in this manner. 
We may of course also want to consider atomic formulas without any
annotations. 
We use the notation $\fatmann{G}{M:A}{Ann}$ to denote a formula which
may be unannotated ($\fatm{G}{M:A}$) or have an annotation 
($\fatmann{G}{M:A}{@^n}$ or $\fatmann{G}{M:A}{*^n}$).
Note that only the syntax of atomic formulas is extended with these annotations.
In the remainder of this subsection we formalize the concepts of well-formedness and 
validity in the context of the formula syntax extended with annotations, and show
that the proof system satisfies the same well-formedness and soundness properties
with this extension.

The well-formedness of formulas containing annotations is determined
essentially by looking at the formula after erasing the annotations.

\begin{definition}[Well-formed Formulas with Annotations]\label{def:ann-form-wf}
A formula $F$ containing annotations is well-formed with respect to
$\Theta$ and $\Xi$ if the formula $F'$
obtained by erasing all annotations in $F$ is such that
$\wfform{\Theta}{\Xi}{F'}$ holds.
\end{definition}

Similar to the formulas, sequents containing annotations are
well-formed if, ignoring the annotations, the sequent is well-formed
as defined in Definition~\ref{def:sequent}.
\begin{definition}[Well-formed Sequents with Annotations]\label{def:ann-seq-wf}
A sequent $\mathcal{S}$ containing annotations is well-formed if the sequent
$\mathcal{S}'$ obtained from $\mathcal{S}$ by erasing all annotations is well-formed.
\end{definition}

We now define precisely the meaning of sequents which contain these annotations.
A key part of this definition is describing an association between
annotations and actual heights. 

\begin{definition}[Height Assignments]
A height assignment $\Upsilon$ will map each annotation $@^i$ to a particular height 
$m_i$, with the height restriction associated with $*^i$ inferred from the mapping
for $@^i$.
For a height assignment $\Upsilon$ the height assignment which is the same as $\Upsilon$
everywhere except that $@^i$ is mapped to $m$ is represented by
$\assn{\Upsilon}{@^i}{m}$.
\end{definition}

Annotated formulas are then interpreted relative to height
assignments.

\begin{definition}[Validity of Annotated Formulas]\label{def:ann_form_valid}
We define validity only for closed annotated formulas that are
well-formed, i.e, for formulas $F$ such that
$\wfform{\noms \cup \STLCGamma_0}{\emptyset}{F'}$ is derivable for the
$F'$ that is obtained by erasing the annotations in $F$.
For such formulas, we first define formula validity with respect to a
height assignment $\Upsilon$, written $\fvalid{\Upsilon}{F}$, as
follows: 
\begin{itemize}
\item $\fvalid{\Upsilon}{\fatmann{G}{M:A}{Ann}}$ holds if $\lfctx{G}$
  and $\lftype{G}{A}$ are derivable, and
  \begin{itemize}
  \item $\lfchecktype{G}{M}{A}$ has a derivation if $Ann$ is the empty annotation,

  \item $\lfchecktype{G}{M}{A}$ has a derivation of  height less than or equal to
  $\Upsilon(@^i)$ if $Ann$ is $@^i$, and
  
  \item $\lfchecktype{G}{M}{A}$ has a derivation of height less than
  $\Upsilon(@^i)$ if $Ann$ is $*^i$. 
  \end{itemize}
  
\item $\fvalid{\Upsilon}{\ftrue}$ holds.

\item $\fvalid{\Upsilon}{\ffalse}$ does not hold.

\item $\fvalid{\Upsilon}{\fimp{F_1}{F_2}}$ holds if 
$\fvalid{\Upsilon}{F_2}$ holds in the case that $\fvalid{\Upsilon}{F_1}$ holds.

\item $\fvalid{\Upsilon}{\fand{F_1}{F_2}}$ holds if both
$\fvalid{\Upsilon}{F_1}$ and $\fvalid{\Upsilon}{F_2}$ hold.

\item $\fvalid{\Upsilon}{\for{F_1}{F_2}}$ holds if either 
$\fvalid{\Upsilon}{F_1}$ or $\fvalid{\Upsilon}{F_2}$ holds.

\item $\fvalid{\Upsilon}{\fctx{\Gamma}{\mathcal{C}}{F}}$ holds if
$\fvalid{\Upsilon}{\subst{G/\Gamma}{F}}$ holds for every context expression
$G$ such that
$\csinst{\noms}{\emptyset}{\mathcal{C}}{G}$ is derivable.

\item $\fvalid{\Upsilon}{\fall{x:\alpha}{F}}$ holds if 
$\fvalid{\Upsilon}{\hsubst{\{\langle x,M,\alpha\rangle\}}{F}}$ holds for every 
$M$ such that \\$\stlctyjudg{\noms\cup\STLCGamma_0}{M}{\alpha}$ is derivable.

\item $\fvalid{\Upsilon}{\fexists{x:\alpha}{F}}$ holds if
$\fvalid{\Upsilon}{\hsubst{\{\langle x,M,\alpha\rangle\}}{F}}$ holds for some
$M$ such that \\$\stlctyjudg{\noms\cup\STLCGamma_0}{M}{\alpha}$ is derivable.
\end{itemize}
As in the case with Definition~\ref{def:semantics}, the coherence of
this definition is assured by Theorem~\ref{th:subst-formula}.
\end{definition}

Finally, the validity of a sequent containing annotations corresponds
to the validity of each of its closed instances relative to every
height assignment. 
In formalizing this idea, we assume the adaptation of the
notions of the compatibility of term substitutions
(Definition~\ref{def:seq-term-subst}), the appropriateness of context substitutions
(Definition~\ref{def:seq-ctxt-subst}) and the applications of these substitutions
(Definitions~\ref{def:seq-term-subst-app} and
\ref{def:seq-ctxt-subst-app}) to annotated sequents that is 
obtained by ignoring the annotations on formulas.
Further, we refer to every instance of an annotated well-formed
sequent that is determined by term and context substitutions as in
Definition~\ref{def:seq-validity} as one of its closed instances; the
wellformedness of these closed instances follows easily from the
results of Section~\ref{sec:sequents}.
\begin{definition}[Validity of Annotated Sequents]\label{def:ann_seq_valid}
A well-formed sequent of the form $\seqsans[\mathbb{N}]{\Omega}{F}$ is
valid with respect to a height 
assignment $\Upsilon$ if $\fvalid{\Upsilon}{F}$ holds whenever
$\fvalid{\Upsilon}{F'}$ holds for every $F'\in\Omega$.
A well-formed sequent $\mathcal{S}$ is valid with respect to
$\Upsilon$ if every closed 
instance of $\mathcal{S}$ is valid with respect to $\Upsilon$.
A well-formed sequent $\mathcal{S}$ is considered valid under the
extended semantics if $\mathcal{S}$ is valid with respect to every
height assignment.
\end{definition}

While height assignments associate heights with every annotation,
the assignments to only a finite subset of annotations matter in
determining the validity of a formula or sequent.
This observation is an immediate consequence of the following lemma.

\begin{lemma}\label{lem:ann_ind}
If $F$ is a formula in which the annotations $@^i$ and $*^i$ do not
occur, then $\fvalid{\Upsilon}{F}$ holds for a height assignment 
$\Upsilon$ if and only if $\fvalid{\assn{\Upsilon}{@^i}{m}}{F}$ holds
for every choice of $m$.
Similarly, if $\mathcal{S}$ is a sequent which the annotations $@^i$
and $*^i$ do not occur, then $\mathcal{S}$ is valid with respect to
$\Upsilon$ if and only if $\mathcal{S}$ is valid with respect to
$\assn{\Upsilon}{@^i}{m}$ for every choice of $m$.
\end{lemma}
\begin{proof}
The first clause is shown by a straightforward induction on the
formation of $F$ where in the base case we know that no atomic formula
is annotated by $@^i$ or $*^i$ and therefore the $m_i$ assigned to
that annotation cannot play a role in determining its validity.
The second clause follows from the first using the definition of validity
for sequents of Definition~\ref{def:ann_seq_valid}.
\end{proof}

Our ultimate interest is in determining that formulas devoid of
annotations are valid in the sense articulated in
Definition~\ref{def:semantics}.
This also means that we are eventually interested in the validity of
sequents devoid of annotations in the sense described in
Definition~\ref{def:seq-validity}.
However, in building in capabilities for inductive reasoning, we will
consider rules that will introduce annotated formulas into sequents.
Establishing the soundness of the resulting proof system requires us
to refine the definition of validity for sequents to the one presented
in Definition~\ref{def:ann_seq_valid}.
That this is acceptable from the perspective of our eventual goal is
the content of the following theorem.

\begin{theorem}\label{th:valid__valid}
A well-formed sequent that is devoid of annotations is valid 
in the sense of Definition~\ref{def:ann_seq_valid} if and only if it 
is valid in the sense of Definition~\ref{def:seq-validity}.
\end{theorem}
\begin{proof}
We claim that if $F$ is a closed, well-formed formula that contains no
annotations, then $F$ is valid in the sense of
Definition~\ref{def:semantics} if and only if $\fvalid{\Upsilon}{F}$
holds for any height assignment $\Upsilon$.
This claim is proved easily by induction on the formation of $F$: the
base case is obvious given the definition of
$\fvalid{\Upsilon}{\fatm{G}{M:A}}$ 
and the other cases have a simple inductive structure.
It can now be easily argued that a closed instance of a sequent devoid
of annotations is
valid in the sense of Definition~\ref{def:seq-validity} if and only if
it is valid by virtue of Definition~\ref{def:ann_form_valid} with
respect to any height assignment. 
The theorem is an obvious consequence of this. 
\end{proof}

Our attention henceforth will be on sequents that potentially contain
annotated formulas.
For this reason, absent qualifications, by
``wellformedness'' for formulas and sequents we shall mean by virtue
of Definitions~\ref{def:ann-form-wf} and \ref{def:ann-seq-wf}.
Similarly, ``validity'' shall be interpreted by virtue of 
Definitions~\ref{def:ann_form_valid} and \ref{def:ann_seq_valid}.
Now, the proof rules that we have discussed up to this point that apply
to non-atomic formulas and that are different from the \id\ rule lift
in an obvious way to the situation where formulas carry annotations.
The \id\ rule still applies as before when the formula that is the
focus of the rule is unannotated.
Similarly, no change is needed to the rules for atomic formulas when
these formulas are unannotated.
We show below that when these rules are interpreted in this manner,
they continue to preserve wellformedness of sequents and to be sound,
albeit with respect to the extended semantics.
The extension of the \id\ rule and the rules for atomic formulas to
the situation when the focus formula is annotated needs some care.
We take up this matter after the presentation of the induction rule.

\begin{theorem}\label{th:nonatom_rules_valid}
The following properties hold for the lifted forms of the proof rules
in Figure~\ref{fig:rules-structural}, Figure~\ref{fig:rules-base},
the \cut\ proof rule from Figure~\ref{fig:rules-other}, and the
\id\ rule and the rules from Figure~\ref{fig:rules-atom} when the formula
they pertain to are unannotated:
\begin{enumerate}
\item If the conclusion sequent is well-formed, the premises expressing typing
  conditions have derivations and the conditions expressed by the
  other, non-sequent premises are satisfied, then the premise sequents
  must be well-formed.

\item If the premises expressing typing judgements are derivable, the
  conditions described in the other non-sequent premises are satisfied
  and all the premise sequents are valid, then the conclusion sequent must 
  also be valid.
\end{enumerate}
\end{theorem}
\begin{proof}
The definition of wellformedness for annotated sequents is based on
the original notion via the erasure of annotations.
Hence, the first claim follows immediately from the earlier results
for unannotated sequents.
The second claim can be proved relative to an arbitrary height
assignment; this then generalizes to all possible height assignments.
For the proof rules not relating to atomic formulas, the heights which may
be assigned to annotations can play no role in the soundness argument as there are no
atomic formulas being interpreted within the proof.
For the atomic proof rules the claim is restricted to the case where the 
formulas being analysed in the rules are not annotated, and thus again
the height assignment has no impact on the previously presented argument for
soundness. 
\end{proof}

%% file: proof-system/ind-rule.tex
\subsection{The Induction Proof Rule}

\begin{figure}
\[
\infer[\ind]
      {\seq[\mathbb{N}]{\Psi}
           {\Xi}
           {\Omega}
           {\mathcal{Q}_1.(
             \fimp{F_1}{\fimp{\ldots}
                       {\mathcal{Q}_{k-1}.(
                         \fimp{F_{k-1}}
                              {\mathcal{Q}_k.
                                  (\fimp{\fatm{G}{M:A}}
                                        {\fimp{\ldots}{F_n}})})}})}}
      {\deduce{\mathcal{Q}_1.( 
                  \fimp{F_1}{\fimp{\ldots}
                                  {\mathcal{Q}_{k-1}}.(
                                    \fimp{F_{k-1}}
                                         {\mathcal{Q}_k.
                                            (\fimp{\fatmann{G}{M:A}{@^i}}
                                                  {\fimp{\ldots}{F_n}})})})}
              {\seq[\mathbb{N}]{\Psi}
                   {\Xi}
                   {\setand{\Omega}
                           {\mathcal{Q}_1.( 
                              \fimp{F_1}{\fimp{\ldots}
                                              {\mathcal{Q}_{k-1}}.(
                                               \fimp{F_{k-1}}
                                                    {\mathcal{Q}_k.
                                                      (\fimp{\fatmann{G}{M:A}{*^i}}
                                                            {\fimp{\ldots}{F_n}})})})}}
                   {}}}
\]
\caption{The Induction Proof Rule}
\label{fig:rule-ind}
\end{figure}

The induction rule is presented in Figure~\ref{fig:rule-ind}.
In this rule $\mathcal{Q}_i$ represents a sequence of
context quantifiers or  universal term quantifiers.
There is also a proviso on the rule: the annotations $@^i$ and $*^i$
must be fresh, i.e., they must not already appear in the sequent that is the
conclusion of the rule.
Induction in this form is based on the induction principle for natural
numbers applied to the heights of atomic formulas, which themselves
encode the heights of LF derivations.
We can view the premise of this rule as providing a proof schema for constructing
an argument of validity for any $m$, and so by an inductive argument
we conclude that it must be valid for all heights.

Given the definition of well-formedness for sequents containing annotations,
the well-formedness of the conclusion sequent will clearly be sufficient to ensure
that the premise sequent is well-formed.

\begin{theorem}\label{th:ind-wf}
If the conclusion sequent of an instance of the \ind\ rule is
well-formed, then the premise sequent must be well-formed.
\end{theorem}
\begin{proof}
Clearly the well-formedness of each declaration in $\Xi$ and each formula in 
$\Omega$ will hold given that these are the same in the premise sequent 
as in the conclusion sequent.
Since both the formula
$\mathcal{Q}_1.(
             \fimp{F_1}{\fimp{\ldots}
                       {\mathcal{Q}_{k-1}.(
                           \fimp{F_{k-1}}
                                {\mathcal{Q}_k.
                                    (\fimp{\fatmann{G}{M:A}{@^i}} 
                                          {\fimp{\ldots}{F_n}})})}})$
and the formula
$\mathcal{Q}_1.(
             \fimp{F_1}{\fimp{\ldots}
                       {\mathcal{Q}_{k-1}.(
                           \fimp{F_{k-1}}
                                {\mathcal{Q}_k.
                                    (\fimp{\fatmann{G}{M:A}{*^i}}
                                          {\fimp{\ldots}{F_n}})})}})$
are the same as the formula
$\mathcal{Q}_1.(
             \fimp{F_1}{\fimp{\ldots}
                       {\mathcal{Q}_{k-1}.(
                           \fimp{F_{k-1}}
                                {\mathcal{Q}_k.
                                    (\fimp{\fatm{G}{M:A}}
                                          {\fimp{\ldots}{F_n}})})}})$
under the erasure of annotations, they are also obviously well-formed
by the well-formedness of the conclusion sequent.
\end{proof}

A derivation of the premise sequent of the induction rule 
provides a schema for constructing a concrete, valid derivation for any
closed instance of the conclusion sequent based on the height $m$ of the
assumption derivation in LF.
Since this proof will be general with respect to the choice of $m$ it ensures
that the property holds for all natural numbers, and so the sequent without an annotated 
atomic formula will be valid.
The soundness of this rule is shown by formalizing these ideas as a
meta-level argument using induction on natural numbers.
The following two lemmas are useful towards this end.

\begin{lemma}\label{lem:all_eq__lt}
Let $@^i$ and $*^i$ be annotations which do not appear in any of the formulas
$F_1,\ldots, F_{k-1},F_{k+1},\ldots F_n$, $\Upsilon$ a height assignment,
and let $m$ be a natural number.
If 
\[
\fvalid{\assn{\Upsilon}{@}{l}}{
\mathcal{Q}_1.(
             \fimp{F_1}{\fimp{\ldots}
                       {\mathcal{Q}_{k-1}.(
                         \fimp{F_{k-1}}
                              {\mathcal{Q}_k.
                                  (\fimp{\eqannaux{i}{\fatm{G}{M:A}}}
                                        {\fimp{\ldots}{F_n}})})}})}
\]
holds for every $l<m$, then 
\[
\fvalid{\assn{\Upsilon}{@^i}{m}}{
\mathcal{Q}_1.(
             \fimp{F_1}{\fimp{\ldots}
                       {\mathcal{Q}_{k-1}.(
                         \fimp{F_{k-1}}
                              {\mathcal{Q}_k.
                                  (\fimp{\ltannaux{i}{\fatm{G}{M:A}}}
                                        {\fimp{\ldots}{F_n}})})}})}.
\]
also holds.
\end{lemma}
\begin{proof}
Let
$F=
\mathcal{Q}_1.(
             \fimp{F_1}{\fimp{\ldots}
                       {\mathcal{Q}_{k-1}.(
                         \fimp{F_{k-1}}
                              {\mathcal{Q}_k.
                                  (\fimp{\eqannaux{i}{\fatm{G}{M:A}}}
                                        {\fimp{\ldots}{F_n}})})}})
$.
We will prove this result by induction on the formation of $F$ measuring
size by the number of quantifiers and implications prior to $\eqannaux{i}{\fatm{G}{M:A}}$.

\case{$F=\fimp{\eqannaux{i}{\fatm{G}{M:A}}}{\fimp{\ldots}{F_n}}$}
If $\fvalid{\assn{\Upsilon}{@^i}{m}}{\ltannaux{i}{\fatm{G}{M:A}}}$
did not hold, then by the definition it is clear that
$\fvalid{\assn{\Upsilon}{@^i}{m}}{\fimp{\ltannaux{i}{\fatm{G}{M:A}}}{\fimp{\ldots}{F_n}}}$
would hold, so suppose that it does hold.
Then there is a derivation for $\lfchecktype{G}{M}{A}$ which has some height $j<m$
and thus
$\fvalid{\assn{\Upsilon}{@^i}{j}}{\fimp{\eqannaux{i}{\fatm{G}{M:A}}}{\fimp{\ldots}{F_n}}}$
holds by definition.
Since $j<m$, we can extract from the assumptions that
$\fvalid{\assn{\Upsilon}{@^i}{j}}{\fimp{F_{k+1}}{\fimp{\ldots}{F_n}}}$ holds.
But $@^i$ and $*^i$ cannot appear in $\fimp{F_{k+1}}{\fimp{\ldots}{F_n}}$,
so by Lemma~\ref{lem:ann_ind} this judgement also holds for the height assignment
$\assn{\Upsilon}{@^i}{m}$.
Therefore 
$\fvalid{\assn{\Upsilon}{@^i}{m}}{\fimp{\eqannaux{i}{\fatm{G}{M:A}}}{\fimp{\ldots}{F_n}}}$
holds, as needed.

\case{$F=
\fall{x:\alpha}
     {\mathcal{Q}'_1.(
             \fimp{F_1}{\fimp{\ldots}
                       {\mathcal{Q}_{k-1}.(
                         \fimp{F_{k-1}}
                              {\mathcal{Q}_k.
                                  (\fimp{\eqannaux{i}{\fatm{G}{M:A}}}
                                        {\fimp{\ldots}{F_n}})})}})}$
}
Let $\eqannaux{i}{F'}$ denote the body of the universal in $F$ and 
$\ltannaux{i}{F'}$ the same formula with $@^i$ replaced by $*^i$.
Consider an arbitrary $t$ such that $\stlctyjudg{\noms\cup\STLCGamma_0}{t}{\alpha}$.
By definition, for any such $t$ it must be that 
$\fvalid{\assn{\Upsilon}{@^i}{l}}{\hsubst{\{\langle x,t,\alpha\rangle\}}{\eqannaux{i}{F'}}}$
holds for any $l<m$.
Thus by induction
$\fvalid{\assn{\Upsilon}{@^i}{m}}{\hsubst{\{\langle x,t,\alpha\rangle\}}{\ltannaux{i}{F'}}}$
will hold.
Therefore 
$\fvalid{\assn{\Upsilon}{@^i}{m}}{\fall{x:\alpha}{\ltannaux{i}{F'}}}$
holds by definition.

\case{$F=
\fctx{\Gamma}{\mathcal{C}}
     {\mathcal{Q}'_1.(
             \fimp{F_1}{\fimp{\ldots}
                       {\mathcal{Q}_{k-1}.(
                         \fimp{F_{k-1}}
                              {\mathcal{Q}_k.
                                  (\fimp{\eqannaux{i}{\fatm{G}{M:A}}}
                                        {\fimp{\ldots}{F_n}})})}})}$
}
Let $\eqannaux{i}{F'}$ denote the body of the context quantification in $F$ and 
$\ltannaux{i}{F'}$ the same formula with $@^i$ replaced by $*^i$.
Consider an arbitrary context expression $G$ such that 
$\csinst{\noms}{\emptyset}{\mathcal{C}}{G}$ is derivable.
By definition, for every $l<m$,
$\fvalid{\assn{\Upsilon}{@^i}{l}}{\subst{G/\Gamma}{\eqannaux{i}{F'}}}$
will hold.
Thus by induction,
$\fvalid{\assn{\Upsilon}{@^i}{m}}{\subst{G/\Gamma}{\ltannaux{i}{F'}}}$
will hold.
Therefore by Definition~\ref{def:ann_form_valid},
$\fvalid{\assn{\Upsilon}{@^i}{m}}{\fctx{\Gamma}{\mathcal{C}}{\ltannaux{i}{F'}}}$
must hold, as needed.

\case{$F=
\fimp{F_0}
     {\mathcal{Q}'_1.(
             \fimp{F_1}{\fimp{\ldots}
                       {\mathcal{Q}_{k-1}.(
                         \fimp{F_{k-1}}
                              {\mathcal{Q}_k.
                                  (\fimp{\eqannaux{i}{\fatm{G}{M:A}}}
                                        {\fimp{\ldots}{F_n}})})}})}$
}
Let $\eqannaux{i}{F'}$ denote the conclusion of the top-level implication in $F$ and 
$\ltannaux{i}{F'}$ the same formula with $@^i$ replaced by $*^i$.
Suppose that 
$\fvalid{\assn{\Upsilon}{@^i}{m}}{F_0}$ holds.
Since $@^i$ and $*^i$ cannot appear in $F_0$, Lemma~\ref{lem:ann_ind} permits
us to conclude that 
$\fvalid{\assn{\Upsilon}{@^i}{l}}{F_0}$ holds for any $l$, in particular it
must hold for any $l<m$.
From this we can conclude using the assumptions that for any $l<m$, 
$\fvalid{\assn{\Upsilon}{@^i}{l}}{\eqannaux{i}{F'}}$ holds, and thus by induction
$\fvalid{\assn{\Upsilon}{@^i}{m}}{\ltannaux{i}{F'}}$ will hold.
Therefore 
$\fvalid{\assn{\Upsilon}{@^i}{m}}{\fimp{F_0}{\ltannaux{i}{F'}}}$ 
will hold, as needed.
\end{proof}

\begin{lemma}\label{lem:ann__unann}
Let $F$ be a formula of the form
\[
\mathcal{Q}_1.(
             \fimp{F_1}{\fimp{\ldots}
                       {\mathcal{Q}_{k-1}.(
                         \fimp{F_{k-1}}
                              {\mathcal{Q}_k.
                                  (\fimp{\fatm{G}{M:A}}
                                        {\fimp{\ldots}{F_n}})})}})
\]
in which the annotations $@^i$ or $*^i$ do not occur.
Then $\fvalid{\Upsilon}{F}$ holds if 
\[
\fvalid{\assn{\Upsilon}{@^i}{m}}{
\mathcal{Q}_1.(
             \fimp{F_1}{\fimp{\ldots}
                       {\mathcal{Q}_{k-1}.(
                         \fimp{F_{k-1}}
                              {\mathcal{Q}_k.
                                  (\fimp{\eqannaux{i}{\fatm{G}{M:A}}}
                                        {\fimp{\ldots}{F_n}})})}})}.
\]
holds for every natural number $m$.
\end{lemma}
\begin{proof}
We will prove this by induction on the formation of $F$, measuring
size by the number of quantifiers and implications prior to $\eqannaux{i}{\fatm{G}{M:A}}$.

\case{$F=\fimp{\fatm{G}{M:A}}{\fimp{\ldots}{F_n}}$}
Suppose that
$\fvalid{\assn{\Upsilon}{@^i}{m}}{\fimp{\eqannaux{i}{\fatm{G}{M:A}}}{\fimp{\ldots}{F_n}}}$ 
holds for any $m$.
If the judgement
$\fvalid{\Upsilon}{\fatm{G}{M:A}}$ did not hold then $\fvalid{\Upsilon}{F}$
holds vacuously so suppose it does hold.
Then there must be a derivation of $\lfchecktype{G}{M}{A}$ and this derivation will
have some height $l$.
Therefore 
$\fvalid{\assn{\Upsilon}{@^i}{l}}{\eqannaux{i}{\fatm{G}{M:A}}}$ will hold.
So by the assumption, 
$\fvalid{\assn{\Upsilon}{@^i}{l}}{\fimp{F_{k+1}}{\fimp{\ldots}{F_n}}}$ must hold as well.
Since $@^i$ does not appear in this formula, Lemma~\ref{lem:ann_ind}
permits us to conclude that 
$\fvalid{\Upsilon}{\fimp{F_{k+1}}{\fimp{\ldots}{F_n}}}$ will hold,
and therefore we can conclude $\fvalid{\Upsilon}{F}$ holds, as needed.

\case{$F=
\fall{x:\alpha}
     {\mathcal{Q}'_1.(
             \fimp{F_1}{\fimp{\ldots}
                       {\mathcal{Q}_{k-1}.(
                         \fimp{F_{k-1}}
                              {\mathcal{Q}_k.
                                  (\fimp{\fatm{G}{M:A}}
                                        {\fimp{\ldots}{F_n}})})}})}$
}
Let $F'$ denote the body of the universal in $F$ and $\eqannaux{i}{F'}$
similarly for the annotated formula.
For an arbitrary term $t$ such that
$\stlctyjudg{\noms\cup\STLCGamma_0}{t}{\alpha}$ is derivable,
we know by the assumptions that
$\fvalid{\assn{\Upsilon}{@^i}{m}}{\hsubst{\{\langle x,t,\alpha\rangle\}}{\eqannaux{i}{F'}}}$
holds for every $m$.
By induction then, we can conclude that
$\fvalid{\Upsilon}{\hsubst{\{\langle x,t,\alpha\rangle\}}{F'}}$ must also hold.
But then clearly $\fvalid{\Upsilon}{F}$ holds by Definition~\ref{def:ann_form_valid}.

\case{$F=
\fctx{\Gamma}{\mathcal{C}}
     {\mathcal{Q}'_1.(
             \fimp{F_1}{\fimp{\ldots}
                       {\mathcal{Q}_{k-1}.(
                         \fimp{F_{k-1}}
                              {\mathcal{Q}_k.
                                  (\fimp{\fatm{G}{M:A}}
                                        {\fimp{\ldots}{F_n}})})}})}$
}
Let $F'$ denote the body of the context quantification in $F$ and $\eqannaux{i}{F'}$
similarly for the annotated formula.
For any context expression $G$ such that $\csinst{\noms}{\emptyset}{\mathcal{C}}{G}$,
we know by the assumptions that
$\fvalid{\assn{\Upsilon}{@^i}{m}}{\hsubst{G/\Gamma}{\eqannaux{i}{F'}}}$
holds for every $m$.
We can conclude from this that $\fvalid{\Upsilon}{\subst{G/\Gamma}{F'}}$ will hold
by an application of the inductive hypothesis.
But then clearly $\fvalid{\Upsilon}{F}$ holds by Definition~\ref{def:ann_form_valid}.

\case{$F=
\fimp{F_0}
     {\mathcal{Q}'_1.(
             \fimp{F_1}{\fimp{\ldots}
                       {\mathcal{Q}_{k-1}.(
                         \fimp{F_{k-1}}
                              {\mathcal{Q}_k.
                                  (\fimp{\fatm{G}{M:A}}
                                        {\fimp{\ldots}{F_n}})})}})}$
}
Let $F'$ denote the conclusion of the top-level implication in $F$ and 
$\eqannaux{i}{F'}$ similarly for the annotated formula.
If $\fvalid{\Upsilon}{F_0}$ were not valid, $\fvalid{\Upsilon}{F}$ would be
vacuously valid so suppose it were valid.
Then by Lemma~\ref{lem:ann_ind}, it must be that 
$\fvalid{\assn{\Upsilon}{@^i}{m}}{F_0}$ holds for every $m$.
By the assumptions then, 
$\fvalid{\assn{\Upsilon}{@^i}{m}}{\eqannaux{i}{F'}}$ holds for every natural number $m$.
An application of the inductive hypothesis permits us to conclude from this that
$\fvalid{\Upsilon}{F'}$ holds, and therefore that $\fvalid{\Upsilon}{F}$ holds, as needed.
\end{proof}

\begin{theorem}\label{th:ind-sound}
If the premise sequent of an instance of the \ind\ rule is valid and
the requirement of non-occurrence of the annotations $@^i$ and $*^i$
is satisfied, then the conclusion sequent of the rule instance must be
valid. 
\end{theorem}
\begin{proof}
In this proof we will use $F$ to denote the formula
\begin{tabbing}
\hspace{2.5cm}\=\kill
\>$\mathcal{Q}_1.(
             \fimp{F_1}{\fimp{\ldots}
                       {\mathcal{Q}_{k-1}.(
                         \fimp{F_{k-1}}
                              {\mathcal{Q}_k.
                                  (\fimp{\fatm{G}{M:A}}
                                        {\fimp{\ldots}{F_n}})})}})$,
\end{tabbing}
$\eqannaux{i}{F}$ to denote
$\mathcal{Q}_1.(
             \fimp{F_1}{\fimp{\ldots}
                       {\mathcal{Q}_{k-1}.(
                         \fimp{F_{k-1}}
                              {\mathcal{Q}_k.
                                  (\fimp{\eqannaux{i}{\fatm{G}{M:A}}}
                                        {\fimp{\ldots}{F_n}})})}})$
and $\ltannaux{i}{F}$ to denote
$\mathcal{Q}_1.(
             \fimp{F_1}{\fimp{\ldots}
                       {\mathcal{Q}_{k-1}.(
                         \fimp{F_{k-1}}
                              {\mathcal{Q}_k.
                                  (\fimp{\ltannaux{i}{\fatm{G}{M:A}}}
                                        {\fimp{\ldots}{F_n}})})}})$.
Let $\mathcal{S}$ denote the conclusion sequent, $\Omega$ the assumption formulas
of $\mathcal{S}$, and $\mathcal{S}'$ the premise sequent.
Consider an arbitrary height assignment $\Upsilon$ and closed instance of $\mathcal{S}$
identified by $\theta$ and $\sigma$.
If any formula in $F'\in\subst{\sigma}{\hsubst{\theta}{\Omega}}$ were such that
$\fvalid{\Upsilon}{F'}$ did not hold then this closed instance of $\mathcal{S}$
would be vacuously valid with respect to $\Upsilon$ so suppose they all hold.
Note that this $\theta$ and $\sigma$ must also identify a closed instance of
$\mathcal{S}'$ given that these sequents share support sets, term 
variable contexts, and context variable contexts.
Since $@^i$ and $*^i$ do not occur in $\mathcal{S}$ by assumption, 
Lemma~\ref{lem:ann_ind} permits us to conclude that for every $m$,
$\fvalid{\assn{\Upsilon}{@^i}{m}}{F'}$ must hold.
We will use this observation and the validity of the premise sequent to argue
by strong induction on $m$ that for all $m$,
$\fvalid{\assn{\Upsilon}{@^i}{m}}{\subst{\sigma}{\hsubst{\theta}{\eqannaux{i}{F}}}}$ 
holds.
Once we have concluded that this holds for every $m$, Lemma~\ref{lem:ann__unann}
permits us to conclude that 
$\fvalid{\Upsilon}{\subst{\sigma}{\hsubst{\theta}{F}}}$ holds 
and from this we conclude
$\subst{\sigma}{\hsubstseq{\emptyset}{\theta}{\mathcal{S}}}$
is valid with respect to $\Upsilon$.

Suppose that for all $l<m$, 
$\fvalid{\assn{\Upsilon}{@^i}{l}}{\subst{\sigma}{\hsubst{\theta}{\eqannaux{i}{F}}}}$
holds.
Then by Lemma~\ref{lem:all_eq__lt} we easily conclude 
$\fvalid{\assn{\Upsilon}{@^i}{m}}{\subst{\sigma}{\hsubst{\theta}{\ltannaux{i}{F}}}}$
holds.
But then every assumption formula $F''$ of the closed sequent 
$\subst{\sigma}{\hsubstseq{\emptyset}{\theta}{\mathcal{S}'}}$
is such that $\fvalid{\assn{\Upsilon}{@^i}{m}}{F''}$ holds and therefore
$\fvalid{\assn{\Upsilon}{@^i}{m}}{\subst{\sigma}{\hsubst{\theta}{\eqannaux{i}{F}}}}$
will hold by the validity of $\mathcal{S}'$.
Thus we can conclude that $\mathcal{S}$ is valid using Lemma~\ref{lem:ann__unann} 
as described above.
\end{proof}

%% file: proof-system/addnl-rules.tex
\subsection{Additional Proof Rules that Interpret Annotations}
We now take up the task of describing additional proof rules that take
the meanings of annotations in formulas into consideration.
These rules are an essential part of our proof system: without them,
it would be impossible to construct proofs for the premises of
instances of the induction rule.

One of the rules that we consider in this context is an enhanced
version of the \id\ rule.
The rule that is included in the proof system currently requires the
conclusion formula to be equi-valid to one of the assumption
formulas.
This requirement can be weakened with the refinement of the semantics
that accommodates annotations in formulas.
For example, if we have the formula $\ltann{\fatm{G}{M:A}}$ as an
assumption, this suffices to ensure the validity of a sequent in which 
$\eqann{\fatm{G}{M:A}}$ is the conclusion formula. 
Similarly, the validity of either annotated formula would imply the 
validity of the unannotated atomic formula $\fatm{G}{M:A}$.
This observation can be expanded to include non-atomic formulas with
the proviso that the polarity of the occurrence of the formula must be
paid attention to.
For example, it is the validity of $\fimp{\eqann{\fatm{G}{M:A}}}{F}$
that implies that $\fimp{\ltann{\fatm{G}{M:A}}}{F}$ is valid, rather
than the other way around, and, further, the validity of both of these
forms is implied by $\fimp{\fatm{G}{M:A}}{F}$. 

We formalize the idea discussed above through a notion of comparative
strengths of formulas.

\begin{definition}[Comparative Strengths of Annotated Formulas]
An atomic formula $\fatmann{G}{M:A}{Ann}$ is stronger than 
$\fatmann{G'}{M':A'}{Ann'}$ with respect to $\Xi$ and $\pi$ if it is
the case that
$\formeq{\Xi}{\pi}{\fatm{G}{M:A}}{\fatm{G'}{M':A'}}$ and either $Ann=*^i$ and 
$Ann'=@^i$ or $Ann'$ is no annotation and $Ann$ is $@^i$, $*^i$, or no annotation.
An implication formula $\fimp{F_2}{F_2'}$ is stronger than $\fimp{F_1}{F_1'}$
with respect to $\Xi$ and $\pi$ if $F_1$ is stronger than $F_2$ with respect to
$\Xi$ and $\inv{\pi}$ and $F_2'$ is stronger then $F_1'$ with respect to
$\Xi$ and $\pi$.
For any other arbitrary formula, $F_2$ is stronger than $F_1$ with respect to $\Xi$ and 
$\pi$ if their components satisfy the same relation, under a possibly extended $\Xi$ in 
the case of context quantification, allowing for renaming of variables bound by quantifiers.
The stronger than relation for formulas is represented by the judgement 
$\streq{\Xi}{\pi}{F_2}{F_1}$.
\end{definition}

The key result we show for this definition is that for any well-formed closed formula,
we can indeed conclude that the validity of the stronger formula will ensure
the validity of the other.
In showing this we will rely on two substitution properties which are
analogous to the Lemmas~\ref{lem:equiv-hsub} and~\ref{lem:equiv-sub}
about formula equivalence.
We first prove the substitution properties and then use this lemma in 
proving the desired result about the stronger than relation.
\begin{lemma}\label{lem:streq-subs}
Suppose that for some formulas $F_1$ and $F_2$, $\streq{\Xi}{\pi}{F_2}{F_1}$ holds.
Then both of the following properties hold.

\begin{itemize}
\item If $\theta$ is a hereditary substitution such that
$\supp{\theta}\cap\supp{\pi}=\emptyset$ and
both $\hsub{\theta}{F_2}{F_2'}$ and $\hsub{\theta}{F_1}{F_1'}$ 
have derivations for some $F_1'$ and $F_2'$, then
$\streq{\hsubst{\theta}{\Xi}}{\pi}{\hsubst{\theta}{F_2}}{\hsubst{\theta}{F_1}}$
holds.

\item If $\sigma$ is an appropriate substitution for $\Xi$ with respect to some 
$\Psi$ then the judgement
$\streq{\ctxvarminus{\Xi}{\sigma}}{\pi}{\subst{\sigma}{F_2}}{\subst{\sigma}{F_1}}$
holds.
\end{itemize}
\end{lemma}
\begin{proof}
We observe that the application of both term and context variable substitutions to
formulas does not impact on either their structure or their annotations.
Therefore we can conclude that both of these clauses hold through an inductive argument
on the formation of both $F_2$ and $F_1$.
\end{proof}
\begin{theorem}\label{th:streq-valid}
For a height assignment $\Upsilon$ and 
well-formed closed formulas $F_1$ and $F_2$, if 
$\streq{\Xi}{\pi}{F_2}{F_1}$ and $F_2$ is valid with respect to $\Upsilon$
then $F_1$ is valid with respect to $\Upsilon$.
\end{theorem}
\begin{proof}
We prove this by induction on the formation of $F_1$ and $F_2$.
Suppose that $F_2$ is valid with respect to $\Upsilon$, and consider the
possible structures for $F_1$ and $F_2$.

\case{$F_1=\fatmann{G_1}{M_1:A_1}{Ann_1}$ and $F_2=\fatmann{G_2}{M_2:A_2}{Ann_2}$}
We first observe that since $\formeq{\Xi}{\pi}{\fatm{G_2}{M_2:A_2}}{\fatm{G_1}{M_1:A_1}}$,
we can extract from this that $\permute{\pi}{(\fatm{G_2}{M_2:A_2})}=\fatm{G_1}{M_1:A_1}$
and thus the unannotated forms of these formulas will be equi-valid with respect to 
$\Upsilon$.
So by the assumption of validity for $F_2$, we can determine from this that 
$\lfctx{G_1}$ and $\lftype{G_1}{A_1}$ have derivations, and also that
$\lfchecktype{G_1}{M_1}{A_1}$ has a derivation which satisfies the annotation $Ann_2$
with respect to $\Upsilon$.
The only remaining piece is then to show that this derivation for 
$\lfchecktype{G_1}{M_1}{A_1}$ will also satisfy the annotation $Ann_1$ with respect 
to $\Upsilon$.
There are two possibilities to consider, either $Ann_2=*^i$ and 
$Ann_1=@^i$ or $Ann_1$ is no annotation and $Ann_2$ is $@^i$, $*^i$, or no annotation.
If $Ann_2=*^i$  it is clear that if a derivation of height strictly less than
$\Upsilon(@^i)$ exists then this same derivation will also be of a height 
less than or equal to $\Upsilon(@^i)$.
In the other case, regardless of the height of the derivation which satisfies $Ann_2$
this same derivation will be sufficient to determine that the unannotated form
of the formula will be valid.
Therefore $F_1$ must be valid with respect to $\Upsilon$.

\case{$F_1=\fimp{F_1'}{F_1''}$ and $F_2=\fimp{F_2'}{F_2''}$}
If $F_1'$ were not valid with respect to $\Upsilon$, then $F_1$ must clearly be
valid with respect to $\Upsilon$.
If $F_1'$ were valid with respect to $\Upsilon$, then by an application of the
inductive hypothesis $F_2'$ must be valid with respect to $\Upsilon$.
Thus by the validity of $F_2$, $F_2''$ must be valid with respect to this same
height assignment, and so by a second application of the inductive hypothesis
$F_1''$ is valid with respect to $\Upsilon$, as needed.

\case{$F_1$ and $F_2$ are of some other structure}
For the remaining cases the desired result is a direct result of 
the definition of validity with respect to a height assignment $\Upsilon$ and
an application of the inductive hypothesis on the components of the formulas,
using Lemma~\ref{lem:streq-subs} to address the instantiation of quantifiers in the
relevant cases.
\end{proof}

\begin{figure}
\[
\infer[\id]
      {\seq[\mathbb{N}]{\Psi}{\Xi}{\Omega}{F_1}}
      {F_2\in\Omega
         &
       \supportof{\pi}\subseteq\mathbb{N}
         &
       \streq{\Xi}{\pi}{F_2}{F_1}}
\]
\caption{Id Rule allowing Annotated Formulas}\label{fig:rules-anna}
\end{figure}

\begin{figure}
\[
\begin{array}{c}
\infer[\appL]
      {\seq[\mathbb{N}]{\Psi}
           {\Xi}
           {\setand{\Omega}{\fatmann{G}{R:P}{Ann}}}
           {F}}
      {\begin{array}{c}
         Ann\in\{*^i,@^i\}
           \\
         \mathcal{CS} = 
               \casesfn[\fatm{G}{R:P}]
                       {\seq[\mathbb{N}]
                            {\Psi}
                            {\Xi}
                            {\setand{\Omega}{\fatm{G}{R:P}}}
                            {F}}
           \\
         \left\{
           \begin{array}{l}
               \seq[\mathbb{N}']
                   {\Psi'}
                   {\Xi'}
                   {\setand{\Omega'}
                           {\fatmann{G_1}{M_1:A_1}{*^i},\ldots,\fatmann{G_k}{M_k:A_k}{*^i}}}
                   {F'}
                \\
                   \ \mid\ 
               \seq[\mathbb{N}']{\Psi'}{\Xi'}{\setand{\Omega'}{\fatm{G_1}{M_1:A_1},\ldots,\fatm{G_k}{M_k:A_k}}}{F'}\in
                   \mathcal{CS}
           \end{array}
         \right\}
       \end{array}}
\medskip \\
\infer[\appR]
      {\seq[\mathbb{N}]{\Psi}
           {\Xi}
           {\Omega}
           {\fatmann{G}{h\app M_1\ldots M_n:P'}{@^i}}}
      {\begin{array}{c}
           h:\typedpi{x_1}{A_1}{\ldots
                      \typedpi{x_n}{A_n}{P}}\in\Sigma\mbox{ or the explicit bindings in }G
             \\
           \fatm{G}{N:B}\in\Omega \qquad
           \hsub{\{\langle x_1, M_1, \erase{A_1}\rangle,\ldots,\langle x_n, M_n, \erase{A_n}\rangle\}}
                {P}
                {P'}
         \\
         \left\{
         \begin{array}{l}
           \seq[\mathbb{N}]
               {\Psi}
               {\Xi}
               {\Omega}
               {}
         \\\quad
               \fatmann{G}
                    {M_i:
                       \hsubst{\{\langle x_1, M_1, \erase{A_1}\rangle,\ldots,
                                 \langle x_{i-1}, M_{i-1}, \erase{A_{i-1}}\rangle\}}
                              {A_i}}
                       {*^i}
          \\\hspace{9cm}
         \ \mid\ 1\leq i\leq n
         \end{array}
         \right\}
       \end{array}}
\medskip\\
\infer[\absL]
      {\seq[\mathbb{N}]{\Psi}{\Xi}{\setand{\Omega}{\fatmann{G}{\lflam{x}{M}:\typedpi{x}{A_1}{A_2}}{Ann}}}{F}}
      {\begin{array}{c}
         \begin{array}{cc}
           n\not\in\dom{\mathbb{N}}
           &
           Ann\in\{*^i,@^i\}
         \end{array}
         \\
         \Xi'=
           \left\{\begin{array}{cl}
             \left(\Xi
                 \setminus
                 \left\{\ctxvarty{\Gamma}{\mathbb{N}_{\Gamma}}{\ctxty{\mathcal{C}}{\mathcal{G}}}\right\}\right)
             \cup 
             \left\{\ctxvarty{\Gamma}{(\mathbb{N}_{\Gamma},n:\erase{A_1})}{\ctxty{\mathcal{C}}{\mathcal{G}}}\right\}
               & \mbox{if }\Gamma\mbox{ in }G 
             \\
             \Xi & \mbox{otherwise}
           \end{array}\right.
         \\
         \begin{array}{l}
             \mathbb{N},n:\erase{A_1};\Psi;\Xi';
           \\\qquad
             \setand{\Omega}
                    {\fatmann{G,n:A_1}
                          {\hsubst{\{\langle x,n,\erase{A_1}\rangle\}}{M}:
                              \hsubst{\{\langle x,n,\erase{A_1}\rangle\}}{A_2}}
                          {*^i}}
             \longrightarrow F
         \end{array}
       \end{array}}
\medskip\\
\infer[\absR]
      {\seq[\mathbb{N}]{\Psi}{\Xi}{\Omega}{\fatmann{G}{\lflam{x}{M}:\typedpi{x}{A_1}{A_2}}{@^i}}}
      {\begin{array}{c}
           n\not\in\dom{\mathbb{N}}
         \\
         \Xi'=
           \left\{\begin{array}{cl}
             \left(\Xi
                 \setminus
                 \left\{\ctxvarty{\Gamma}{\mathbb{N}_{\Gamma}}{\ctxty{\mathcal{C}}{\mathcal{G}}}\right\}\right)
             \cup 
             \left\{\ctxvarty{\Gamma}{(\mathbb{N}_{\Gamma},n:\erase{A_1})}{\ctxty{\mathcal{C}}{\mathcal{G}}}\right\}
               & \mbox{if }\Gamma\mbox{ in }G 
             \\
             \Xi & \mbox{otherwise}
           \end{array}\right.
         \\
         \begin{array}{l}
             \seq[\mathbb{N},n:\erase{A_1}]{\Psi}{\Xi'}{\Omega}{}
           \\\qquad
             \fatmann{G,n:A_1}{\hsubst{\{\langle x,n,\erase{A_1}\rangle\}}{M}:\hsubst{\{\langle x,n,\erase{A_1}\rangle\}}{A_2}}{*^i}
         \end{array}
       \end{array}}
\end{array}\]
\caption{Atomic Proof Rules allowing Annotated Formulas}\label{fig:rules-annb}
\end{figure}

The new version of the identity rule is presented in
Figure~\ref{fig:rules-anna}.
Figure~\ref{fig:rules-annb} presents enhancements to the atomic proof
rules that also take into account the semantics of annotations.
As usual, we show that these rules preserve the well-formedness of
sequents and are also sound. 

\begin{theorem}\label{th:ann-sound}
The following properties holds for every instance of each of the rules
in Figures~\ref{fig:rules-anna} and~\ref{fig:rules-annb}: 
\begin{enumerate}
\item If the conclusion sequent is
well-formed, the premises expressing typing conditions have
derivations and the conditions expressed by the other, non-sequent
premises are satisfied, then all the sequent premises must
be well-formed.
\item If the premises expressing
typing judgements are derivable, the conditions described in the other
non-sequent premises are satisfied and the premise sequent is valid,
then the conclusion sequent must also be valid. 
\end{enumerate}
\end{theorem}
\begin{proof}
Given that the well-formedness of these sequents is defined by
the well-formedness of the sequent with annotations erased, clause (1)
holds trivially by Theorems~\ref{th:other-sound} and~\ref{th:atom-sound}.
For the second clause, consider each of the rules in 
Figures~\ref{fig:rules-anna} and~\ref{fig:rules-annb}.

\case{\id}
Consider an arbitrary height assignment $\Upsilon$ and a
closed instance of the conclusion sequent identified
by $\theta$ and $\sigma$ such that $\supportof{\theta}\cap\supportof{\pi}=\emptyset$.
Note that a proof for such a restricted case will be sufficient to conclude the
validity for every closed instance via Theorem~\ref{th:perm-valid}.
The requirements on $\theta$ and $\sigma$ will clearly be sufficient
to ensure that they satisfy the conditions of
Lemma~\ref{lem:streq-subs} and therefore from the the assumptions we 
can conclude that judgement
$\streq{\emptyset}{\pi}{\subst{\sigma}{\hsubst{\theta}{F_2}}}{\subst{\sigma}{\hsubst{\theta}{F_1}}}$
holds.
If any formula in $\subst{\sigma}{\hsubst{\theta}{\Omega}}$ were
not valid with respect to $\Upsilon$ then this closed instance would 
be vacuously valid so suppose they are all valid.
Then in particular, $\subst{\sigma}{\hsubst{\theta}{F_2}}$
is valid with respect to $\Upsilon$.
Therefore by Theorem~\ref{th:streq-valid} it must be that 
$\subst{\sigma}{\hsubst{\theta}{F_1}}$ is also valid, as needed.

\case{\appL}
This argument follows that given for Theorem~\ref{th:cases-cover} generalized
over an arbitrary height assignment $\Upsilon$ for any annotations appearing
in the conclusion sequent.
The only significant change is to observe in the application of
Lemma~\ref{lem:decomp_decr} that 
whenever there is a derivation for
$\lfchecktype{\subst{\sigma}{\hsubst{\theta}{G}}}
             {\hsubst{\theta}{R}}
             {\hsubst{\theta}{P}}$
of some height $k$, it will be the case that for each $i$, $1\leq i\leq n$, 
\[\lfchecktype{\subst{\sigma}{\hsubst{\theta}{G}}}
             {M_i}
             {\hsubst{\langle x_1,M_1,\erase{A_1}\rangle,\ldots,
                         \langle x_{i-1},M_{i-1},\erase{A_{i-1}}\rangle}
                      {A_i}}
\]
has a derivation of height smaller than $k$.
Thus it is sound to annotate the formulas
$\left\{
\fatm{\subst{\sigma}{\hsubst{\theta}{G}}}{M_i:\left(\hsubst{\langle x_1,M_1,\erase{A_1}\rangle,\ldots,\langle x_{i-1},M_{i-1},\erase{A_{i-1}}\rangle}{A_i}\right)}
\middle|
1\leq i\leq n
\right\}$
with $*^i$ as it is known that they must be derivable with a smaller height
and thus will satisfy the requirements of this annotation for the height assignment.

\case{\appR}
Let $A_i'$ denote the type 
$\hsubst{\{\langle x_1,M_1,\erase{A_1}\rangle,\ldots,\langle x_{i-1},M_{i-1},\erase{A_{i-1}}\rangle\}}{A_i}$.
Consider an arbitrary height annotation $\Upsilon$ and closed instance identified
by $\theta$ and $\sigma$.
Suppose all formulas in $\subst{\sigma}{\hsubst{\theta}{\Omega}}$ are valid 
with respect to $\Upsilon$, since if they were not this instance would be
vacuously valid.
Then since clearly $\Upsilon$ will be a valid height assignment for all of 
the premise sequents, and also $\theta$ and $\sigma$ will identify closed
instances of these sequents, by the validity of the premise sequents all of the formulas
$\subst{\sigma}{\hsubst{\theta}{\fatmann{G}{M_i:A_i'}{*^i}}}$ 
are valid with respect to $\Upsilon$.
But then clearly we can ensure that the formula
$\subst{\sigma}{\hsubst{\theta}{\fatmann{G}{h\app M_1\ldots M_n:P}{@^i}}}$
is valid by applying Theorem~\ref{th:atomictype} with the LF derivations for
$\lfchecktype{\subst{\sigma}{\hsubst{\theta}{G}}}
             {\hsubst{\theta}{M_i}}
             {\hsubst{\theta}{A_i'}}$
which must have heights smaller than the height assigned to $@^i$ by $\Upsilon$.

\case{\absL}
Consider for the conclusion sequent an arbitrary height assignment 
$\Upsilon$ and closed instance identified by $\theta$ and $\sigma$.
Assume that $n$ is a name which does not appear in $\theta$ or $\sigma$;
as Theorem~\ref{th:perm-valid} would permit permuting this name to one
which does not.
Suppose all formulas in 
$\subst{\sigma}{\hsubst{\theta}{(\setand{\Omega}{\fatmann{G}{\lflam{x}{M}:\typedpi{x}{A_1}{A_2}}{Ann}})}}$
are valid with respect to $\Upsilon$, as otherwise this instance would be vacuously valid.
In particular then, the closed formula
$\subst{\sigma}{\hsubst{\theta}{\fatmann{G}{\lflam{x}{M}:\typedpi{x}{A_1}{A_2}}{Ann}}}$
is valid with respect to $\Upsilon$.
Thus letting $k$ be the height assigned to $@^i$ by $\Upsilon$ there are derivations of
$\lfctx{\subst{\sigma}{\hsubst{\theta}{G}}}$,
$\lftype{\subst{\sigma}{\hsubst{\theta}{G}}}
        {\hsubst{\theta}{\typedpi{x}{A_1}{A_2}}}$, and
$\lfchecktype{\subst{\sigma}{\hsubst{\theta}{G}}}
             {\hsubst{\theta}{\lflam{x}{M}}}
             {\hsubst{\theta}{\typedpi{x}{A_1}{A_2}}}$
of height $k$.
But then clearly from subderivations of these and an application of \ctxterm\ 
there would be derivations shorter than $k$ for the LF judgements
$\lfctx{\subst{\sigma}{\hsubst{\theta}{(G,n:A_1)}}}$,
$\lftype{\subst{\sigma}{\hsubst{\theta}{(G,n:A_1)}}}
        {\hsubst{\theta}{\hsubst{\{\langle x,n,\erase{A_1}\rangle\}}{A_2}}}$, as well as for
$\lfchecktype{\subst{\sigma}{\hsubst{\theta}{(G,n:A_1)}}}
             {\hsubst{\theta}{\hsubst{\{\langle x,n,\erase{A_1}\rangle\}}{M}}}
             {\hsubst{\theta}{\hsubst{\{\langle x,n,\erase{A_1}\rangle\}}{A_2}}}$.
Thus all the formulas in 
$\subst{\sigma}{\hsubst{\theta}{(\setand{\Omega}{\fatmann{G,n:A_1}{\hsubst{\{\langle x,n,\erase{A_1}\rangle\}}{M}:\hsubst{\{\langle x,n,\erase{A_1}\rangle\}}{A_2}}{*^i}})}}$
are valid and so by the validity of the premise sequents
$\subst{\sigma}{\hsubst{\theta}{F}}$ will be valid, as needed.

\case{\absR}
Consider for the conclusion sequent an arbitrary height assignment 
$\Upsilon$ and closed instance identified by $\theta$ and $\sigma$.
Assume that $n$ is a name which does not appear in $\theta$ or $\sigma$;
as Theorem~\ref{th:perm-valid} would permit permuting this name to one
which does not.
Suppose all formulas in $\subst{\sigma}{\hsubst{\theta}{\Omega}}$
are valid with respect to $\Upsilon$, as otherwise this instance would be vacuously valid.
Since $\theta$ and $\sigma$ also identify a closed instance of the premise sequent
and all formulas in $\subst{\sigma}{\hsubst{\theta}{\Omega}}$ are valid,
we can infer from the validity of the premise sequent that
$\subst{\sigma}{\hsubst{\theta}{\fatmann{G,n:A_1}{\hsubst{\{\langle x,n,\erase{A_1}\rangle\}}{M}:\hsubst{\{\langle x,n,\erase{A_1}\rangle\}}{A_2}}{*^i}}}$ is valid
with respect to $\Upsilon$.
So by definition there exist LF derivations for
\begin{enumerate}
\item$\lfctx{\subst{\sigma}{\hsubst{\theta}{(G,n:A_1)}}}$,
\item$\lftype{\subst{\sigma}{\hsubst{\theta}{(G,n:A_1)}}}
        {\hsubst{\theta}{\hsubst{\{\langle x,n,\erase{A_1}\rangle\}}{A_2}}}$, and
\item$\lfchecktype{\subst{\sigma}{\hsubst{\theta}{(G,n:A_1)}}}
             {\hsubst{\theta}{\hsubst{\{\langle x,n,\erase{A_1}\rangle\}}{M}}}
             {\hsubst{\theta}{\hsubst{\{\langle x,n,\erase{A_1}\rangle\}}{A_2}}}$
\end{enumerate}
of height less than that assigned to $@^i$ by $\Upsilon$.
But then clearly we can construct derivations also for the LF judgements
$\lfctx{\subst{\sigma}{\hsubst{\theta}{G}}}$,
$\lftype{\subst{\sigma}{\hsubst{\theta}{G}}}
        {\hsubst{\theta}{\typedpi{x}{A_1}{A_2}}}$, and
$\lfchecktype{\subst{\sigma}{\hsubst{\theta}{G}}}
             {\hsubst{\theta}{\lflam{x}{M}}}
             {\hsubst{\theta}{\typedpi{x}{A_1}{A_2}}}$
of a height satisfying the annotation assigned to $@^i$ by $\Upsilon$.
Therefore the formula 
$\subst{\sigma}{\hsubst{\theta}{\fatmann{G}{\lflam{x}{M}:\typedpi{x}{A_1}{A_2}}{@^i}}}$
will be valid with respect to $\Upsilon$, as needed.
\end{proof}

%% file: proof-system/meta-theory.tex
\section{Proof Rules Encoding LF Meta-Theorems}
\label{sec:metathm-rules}

The meta-theorems concerning LF derivability that were discussed in
Section~\ref{sec:lf-metathm} are often useful in informal arguments
about the properties of LF specifications.
In this section, we describe proof rules that provide a means for
using these meta-theorems in formal reasoning based on our logic. 

The weakening meta-theorem has a proviso that the type for the new
binding introduced into the context must be well-formed.
This must be reflected in the proof rule that captures the content of
this meta-theorem by a collection of premises that check that this
property of the type will hold.
We refer to the process that generates the typing judgements that must
be checked towards this end as \emph{type decomposition}.
In addition to the type, this process is parameterized by a collection
of nominal constants, a context variable context and a context
expression.
The result of type decomposition is a collection of triples that
comprise an extended collection of nominal constants, a modified
context variable context and an atomic formula expressing a typing
judgement that must be checked.
This idea is made precise below.

\begin{definition}[Decomposition of Types]
The decomposition of a canonical type $A$ with respect to a
collection of nominal constants $\mathbb{N}$, a context variable
context $\Xi$ and a context expression $G$, notated as
$\typdecomp{\mathbb{N}}{\Xi}{G}{A}$, is defined as follows: 
\begin{enumerate}
\item If $A$ is a type of the form $(a\app M_1\ldots M_n)$
  where $a:\typedpi{x_1}{A_1}{\ldots\typedpi{x_n}{A_n}{\type}}
  \in\Sigma$, then $\typdecomp{\mathbb{N}}{\Xi}{G}{A}$ is the
  collection
\[\bigcup_{i \in 1..n}\left\{
    \left(
      \mathbb{N},
      \Xi,
      \fatm{G}
           {M_i:
             (\hsubst{\{\langle x_1,M_1,\erase{A_1}\rangle,\ldots,
                        \langle x_{i-1},M_{i-1},\erase{A_{i-1}}\rangle\}}
                     {A_i})}
    \right)\right\}.\]

\item If $A=\typedpi{x}{A_1}{A_2}$, then letting 
$G'$ be $G,n:A_1$, $\mathbb{N}'$ be
  $\mathbb{N} \cup \{ n \}$, and $\Xi'$ be the set 
\begin{tabbing}
\qquad\=$\left\{\ctxvarty{\Gamma}{\mathbb{N}'_{\Gamma}}{\ctxty{\mathcal{C}}{\mathcal{G}}}\ \middle|\ \right.$\=\kill
\>$\left\{\ \ctxvarty{\Gamma}{\mathbb{N}'_{\Gamma}}{\ctxty{\mathcal{C}}{\mathcal{G}}}\ \middle|\  
              \ctxvarty{\Gamma}{\mathbb{N}_{\Gamma}}{\ctxty{\mathcal{C}}{\mathcal{G}}}\in\Xi
                          \mbox{ and }\right.$\\
\>\>          $\left. \mathbb{N}'_{\Gamma}\ \mbox{is}\ \mathbb{N}_{\Gamma} \cup \{n\}\
                          \mbox{ if }\Gamma\ \mbox{occurs in}\ G\ \mbox{and}\
                          \mathbb{N}_{\Gamma}\ \mbox{ otherwise }
         \right\}$ 
\end{tabbing}
for some nominal constant $n:\erase{A_1} \in \noms \setminus \mathbb{N}$, 
$\typdecomp{\mathbb{N}}{\Xi}{G}{A}$ is the collection
$\typdecomp{\mathbb{N}}{\Xi}{G}{A_1}\cup
    \typdecomp{\mathbb{N}'}{\Xi'}{G'}{\hsubst{\{\langle x,
        n,\erase{A_1}\rangle\}}{A_2}}.$
\end{enumerate} 
\end{definition}
\noindent Note that the decomposition will not always be defined, but as we will show, 
it will be defined in all the cases we need to use it.
Further, we will wish to be careful in how the name $n$ in the second clause 
of this definition is selected in practice for this decomposition to be 
useful in reasoning.
We will show that the way this definition is used will be sound regardless of
the choice of name, however choices for $n$ which are not fresh will lead
to generating formulas which are never provable and thus we will wish to
avoid these name in an implementation of the proof system.

\begin{figure}
\[\begin{array}{c}
\infer[\lfwk]
      {\seq[\mathbb{N}]{\Psi}{\Xi}{\Omega}{\fimp{\fatmann{G}{M:A}{Ann}}{\fatmann{G,n:B}{M:A}{Ann}}}}
      {\begin{array}{c}
         n\mbox{ does not appear in }\ M,\ A, \mbox{ or the explicit bindings in } G
         \\
         \mbox{ if }\ctxvarty{\Gamma_i}{\mathbb{N}_i}{\ctxty{\mathcal{C}_i}{\mathcal{G}_i}}\in\Xi\mbox{ and }\Gamma_i\mbox{ appears in }G,\mbox{ then }n\in\mathbb{N}_i
         \\
         \left\{
           \seq[\mathbb{N}']{\Psi}{\Xi'}{\Omega}{F'}\ |\ (\mathbb{N}',\Xi',F')\in\typdecomp{\mathbb{N}}{\Xi}{G}{B}
         \right\}
       \end{array}}
\\\ \\
\infer[\lfstr]
      {\seq[\mathbb{N}]{\Psi}{\Xi}{\Omega}{\fimp{\fatmann{G,n:B}{M:A}{Ann}}{\fatmann{G}{M:A}{Ann}}}}
      {\begin{array}{c}
         n\mbox{ does not appear in }M\mbox{, } A\mbox{, or the explicit bindings in }G
       \end{array}}
\\\ \\
\infer[\lfperm]
      {\seq[\mathbb{N}]{\Psi}
           {\Xi}
           {\Omega}
           {\fimp{\fatmann{G}{M:A}{Ann}}
                 {\fatmann{G'}{M:A}{Ann}}}}
      {\begin{array}{c}
          G = G'',n_1:A_1,n_2:A_2,n_3:A_3,\ldots,n_m:A_m \\
          G' = G'',n_2:A_2,n_1:A_1,n_3:A_3,\ldots,n_m:A_m \\
          \begin{array}{cc}
            n_1\mbox{ does not appear in }A_2
          \end{array}
       \end{array}}
\\\ \\
\infer[\lfinst]
      {\seq[\mathbb{N}]{\Psi}
           {\Xi}
           {\Omega}
           {\fimp{\fatm{G,G'}{M:A}}{\fimp{
                  \fatm{G}{M_1:A_1}}{
                  \fatm{G''}{M':A'}}}}}
      {\begin{array}{c}
         \begin{array}{cc}
           G' = n_1:A_1,\ldots,n_m:A_m
             &
           G''=G,n_2:A_2',\ldots,n_m:A_m'
           \\
           \hsub{\{\langle n_1, M_1, \erase{A_1}\rangle\}}{M}{M'}
             &
           \hsub{\{\langle n_1, M_1, \erase{A_1}\rangle\}}{A}{A'}
         \end{array}
         \\
         \left\{\hsub{\{\langle n_1, M_1, \erase{A_1}\rangle\}}{A_i}{A_i'}
             \ \mid\ 2\leq i\leq m\right\}
       \end{array}}
\end{array}\]
\caption{Rules Encoding Meta-Theoretic Properties of LF}\label{fig:rules-meta}
\end{figure}
  
Figure~\ref{fig:rules-meta} presents the proof rules which encode the
content of the LF meta-theorems.
Note in particular that the goal formula of the conclusion sequent in
each of these rules expresses the meta-theorem in terms of the atomic
formulas in the logic.
The symbol $Ann$ in the first three rules, which encode weakening,
strengthening, and context permutation, stands for no annotation,
$@^i$ or $*^i$ for some $i$, used in the same manner throughout the
rule instance.
Permitting annotations in these rules is justified by the fact that
the corresponding meta-theorems guarantee the preservation of the
structure, and thus height, of LF derivations.
Clearly instantiation does not share this property and so we do not consider
annotated formulas for this proof rule.

As before we show that extending the proof system with these rules
will maintain both the well-formedness and soundness properties.
The following lemma, which ensures that whenever the decomposition is
performed on a well-formed type relative to well-formed contexts, the
result is defined and further that the formulas will themselves be
well-formed, will be useful in showing the well-formedness property for 
the weakening proof rule.

\begin{lemma}
\label{lem:tydecomp-wf}
Assume $\Xi$ is a context variable context such that for each
$\ctxvarty{\Gamma_i}{\mathbb{N}_i}{\ctxty{\mathcal{C}_i}{\mathcal{G}_i}}\in\Xi$
there is a derivation of
$\wfctxvarty{\mathbb{N} \setminus \mathbb{N}_i}{\Psi}{\ctxty{\mathcal{C}_i}{\mathcal{G}_i}}$.
Also assume
$\wfctx{\mathbb{N}\cup\STLCGamma_0\cup\Psi}{\Xi}{G}$ and
$\wftype{\mathbb{N}\cup\STLCGamma_0\cup\Psi}{B}$ have derivations. 
Then $\typdecomp{\mathbb{N}}{\Xi}{G}{B}$ is defined and for each 
$(\mathbb{N}',\Xi',F)\in\typdecomp{\mathbb{N}}{\Xi}{G}{B}$ it is the case that
\begin{enumerate}
\item for each
$\ctxvarty{\Gamma_i'}{\mathbb{N}_i'}{\ctxty{\mathcal{C}_i'}{\mathcal{G}_i'}}\in\Xi'$
there is a derivation of
$\wfctxvarty{\mathbb{N}' \setminus \mathbb{N}_i'}{\Psi}{\ctxty{\mathcal{C}_i'}{\mathcal{G}_i'}}$ and
\item $F$ is a well-formed formula with respect to $(\mathbb{N}'\cup\STLCGamma_0\cup\Psi)$ and $\ctxsanstype{\Xi'}$.
\end{enumerate}
\end{lemma}
\begin{proof}
We prove this by induction on the formation of the type $B$.
Consider the cases for the structure of $B$.

\case{$B$ is an atomic type $P$.}
Given that
$\wftype{\mathbb{N}\cup\STLCGamma_0\cup\Psi}{B}$ is derivable,
$B$ is of the form $(a\app M_1\ldots M_n)$ and
there must exist $a:K$ in the LF signature for a kind of the form
$\typedpi{x_1}{A_1}{\ldots\typedpi{x_n}{A_n}{\type}}$ and subderivations of
$\stlctyjudg{\mathbb{N}\cup\STLCGamma_0\cup\Psi}{M_i}{\erase{A_i}}$ for each $M_i$.
From these observations $\typdecomp{\mathbb{N}}{\Xi}{G}{B}$ is clearly
defined. 
Noting that the arity kinding for $B$ will ensure the substitution application
is defined, let
$\hsubst{\{\langle x_1,M_1,\erase{A_1}\rangle,\ldots,
                \langle x_{i-1},M_{i-1},\erase{A_{i-1}}\rangle\}}
             {A_i}$
be denoted by $A_i'$ for each $i$.
Then we can express the collection $\typdecomp{\mathbb{N}}{\Xi}{G}{B}$
as the set of tuples $\bigcup_{i\in1..n}(\mathbb{N},\Xi,\fatm{G}{M_i:A_i'})$.

Clearly, for any tuple in $\typdecomp{\mathbb{N}}{\Xi}{G}{B}$
the context variable context satisfies condition (1).
Since $a:K$ is from the LF signature it is obvious that $\lftype{\emptyctx}{K}$
is derivable in LF, and thus for each $i$, $1\leq i\leq n$,
$\wftype{\aritysum{\{x_1:A_1,\ldots,x_{i-1}:A_{i-1}\}}{(\mathbb{N}\cup\STLCGamma_0\cup\Psi)}}{A_i}$ 
is derivable.
By Theorem~\ref{th:aritysubs-ty} there must exist derivations of
$\wftype{\mathbb{N}\cup\STLCGamma_0\cup\Psi}{A_i'}$ for each $i$.
We can then conclude that there are derivations of 
$\wfctx{\mathbb{N}\cup\STLCGamma_0\cup\Psi}{\Xi}{G}$,
$\wftype{\mathbb{N}\cup\STLCGamma_0\cup\Psi}{A_i'}$
and 
$\stlctyjudg{\mathbb{N}\cup\STLCGamma_0\cup\Psi}{M_i}{\erase{A_i}}$, 
and therefore for each $i$, $1\leq i\leq n$,
$\wfform{\mathbb{N}\cup\STLCGamma_0\cup\Psi}{\ctxsanstype{\Xi}}{\fatm{G}{M_i:A_i'}}$
is derivable, satisfying condition (2).

\case{$B$ is a canonical type $\typedpi{x}{A_1}{A_2}$.}
Then letting
$G'=G,n:\erase{A_1}$, $\mathbb{N}'=\mathbb{N}\cup\{n\}$, and
\begin{tabbing}
\qquad\=$\left\{\ctxvarty{\Gamma}{\mathbb{N}'_{\Gamma}}{\ctxty{\mathcal{C}}{\mathcal{G}}}\ \middle|\ \right.$\=\kill
\>$\Xi'=\left\{\ \ctxvarty{\Gamma}{\mathbb{N}'_{\Gamma}}{\ctxty{\mathcal{C}}{\mathcal{G}}}\ \middle|\  
              \ctxvarty{\Gamma}{\mathbb{N}_{\Gamma}}{\ctxty{\mathcal{C}}{\mathcal{G}}}\in\Xi
                          \mbox{ and }\right.$\\
\>\>          $\left. \mathbb{N}'_{\Gamma}\ \mbox{is}\ \mathbb{N}_{\Gamma} \cup \{n\}\
                          \mbox{ if }\Gamma\ \mbox{occurs in}\ G\ \mbox{and}\
                          \mathbb{N}_{\Gamma}\ \mbox{ otherwise }
         \right\}$ 
\end{tabbing}
for some new nominal constant $n$,
the decomposition $\typdecomp{\mathbb{N}}{\Xi}{G}{B}$
is defined if both $\typdecomp{\mathbb{N}}{\Xi}{G}{A_1}$ and 
$\typdecomp{\mathbb{N}'}{\Xi'}{G'}{\hsubst{\{\langle x, n,\erase{A_1}\rangle\}}{A_2}}$
are defined, and it will be the union of these two sets.
Given that 
$\mathbb{N}\setminus\mathbb{N}_i=(\mathbb{N},n:\erase{A_1})\setminus(\mathbb{N}_i,n:\erase{A_1})$
for any $\mathbb{N}_i$, the context variable context $\Xi'$ will clearly satisfy the
requirements of this lemma.
From the derivation of $\wftype{\mathbb{N}\cup\STLCGamma_0\cup\Psi}{B}$
we can obtain derivations of both
$\wftype{\mathbb{N}\cup\STLCGamma_0\cup\Psi}{A_1}$
and
$\wftype{\aritysum{\{x:\erase{A_1}\}}{\mathbb{N}\cup\STLCGamma_0\cup\Psi}}{A_2}$.
By Theorem~\ref{th:aritysubs-ty} there must then exist a derivation for
$\wftype{{\mathbb{N}'\cup\STLCGamma_0\cup\Psi}}{\hsubst{\{\langle x,n,\erase{A_1}\rangle\}}{A_2}}$.
From the derivations of $\wfctx{\mathbb{N}\cup\STLCGamma_0\cup\Psi}{\Xi}{G}$
and $\wftype{\mathbb{N}\cup\STLCGamma_0\cup\Psi}{A_1}$ we can construct
a derivation for the judgement
$\wfctx{\mathbb{N}'\cup\STLCGamma_0\cup\Psi}{\Xi}{G,n:A_1}$.
The type $A_1$ is clearly smaller than $B$, and it is straightforward
to conclude that
$\hsubst{\{\langle x,n,\erase{A_1}\rangle\}}{A_2}$ must be as well.
Thus by invoking the inductive hypothesis twice we determine that both
$\typdecomp{\mathbb{N}}{\Xi}{G}{A_1}$ and 
$\typdecomp{\mathbb{N}'}{\Xi'}{G'}{\hsubst{\{\langle x, n,\erase{A_1}\rangle\}}{A_2}}$
are defined, and that each tuple in these sets satisfy the conditions (1) \& (2).
Therefore $\typdecomp{\mathbb{N}}{\Xi}{G}{B}$ will be defined and
each tuple in this set will satisfy the necessary conditions.
\end{proof}

We now show the well-formedness property for the rules in Figure~\ref{fig:rules-meta}.
The only interesting case to consider is for \lfwk.
\begin{theorem}\label{th:meta-wf}
The following property holds of the rules in Figure~\ref{fig:rules-meta}: if the
conclusion sequent is well-formed, the premises expressing typing
conditions have derivations and the conditions expressed by the other,
non-sequent premises are satisfied, then the premise sequents must be
well-formed.
\end{theorem}
\begin{proof}
Consider each of the rules defined in Figure~\ref{fig:rules-meta}.

\case{\lfstr, \lfperm, and \lfinst.}
There are no premise sequents in these rules so the property holds vacuously.

\case{\lfwk}
Given the well-formedness of the conclusion sequent, 
$\wfctxvarty{\mathbb{N} \setminus \mathbb{N}_i}{\Psi}{\ctxty{\mathcal{C}_i}{\mathcal{G}_i}}$
is derivable for each 
$\ctxvarty{\Gamma_i}{\mathbb{N}_i}{\ctxty{\mathcal{C}_i}{\mathcal{G}_i}}\in\Xi$
and 
$\wfform{\mathbb{N}\cup\Psi\cup\STLCGamma_0}{\ctxsanstype{\Xi}}{F}$ 
is derivable for each formula
$F\in\Omega\cup\{\fimp{\fatmann{G}{M:A}{Ann}}{\fatmann{G,n:B}{M:A}{Ann}}\}$.
From this it is obvious that $G$ is a well-formed context expression with
respect to $(\mathbb{N}\cup\Psi\cup\STLCGamma_0)$ and $\Xi$, and that
$B$ is a good type with respect to
$(\mathbb{N}\cup\Psi\cup\STLCGamma_0)$.
Thus by Lemma~\ref{lem:tydecomp-wf} the decomposition is defined and 
for each $(\mathbb{N}',\Xi',F)\in\typdecomp{\mathbb{N}}{\Xi}{G}{B}$
the following hold
\begin{enumerate}
\item for each
$\ctxvarty{\Gamma_i'}{\mathbb{N}_i'}{\ctxty{\mathcal{C}_i'}{\mathcal{G}_i'}}\in\Xi'$,
$\wfctxvarty{\mathbb{N}' \setminus \mathbb{N}_i'}{\Psi}{\ctxty{\mathcal{C}_i'}{\mathcal{G}_i'}}$
has a derivation and
\item 
$\wfform{\mathbb{N}'\cup\STLCGamma_0\cup\Psi}{\ctxsanstype{\Xi'}}{F}$ has a derivation.
\end{enumerate}
By Theorem~\ref{th:wf-form-wk} we can determine that every formula in $\Omega$
is well-formed under the extended contexts $(\mathbb{N}'\cup\Psi\cup\STLCGamma_0)$
and $\ctxsanstype{\Xi'}$, and thus all of the premise sequents
will be well-formed.
\end{proof}

The following lemma captures the intended meaning of the type decomposition;
that the typing judgements identified by
type decomposition are sufficient in determining that any instance
of the type is well-formed in LF.
Thus result will be key to proving soundness 
of the weakening proof rule.

\begin{lemma}
\label{lem:tydecomp-defn}
Assume $\Xi$ is a context variable context such that for each
$\ctxvarty{\Gamma_i}{\mathbb{N}_i}{\ctxty{\mathcal{C}_i}{\mathcal{G}_i}}\in\Xi$
there is a derivation of
$\wfctxvarty{\mathbb{N} \setminus \mathbb{N}_i}{\Psi}{\ctxty{\mathcal{C}_i}{\mathcal{G}_i}}$.
Also assume
$\wfctx{\mathbb{N}\cup\STLCGamma_0\cup\Psi}{\Xi}{G}$ and
$\wftype{\mathbb{N}\cup\STLCGamma_0\cup\Psi}{B}$ have derivations. 
Let $\theta$ and $\sigma$ be some substitutions such that
$\dom{\theta}\subseteq\Psi$, $\dom{\sigma}\subseteq\Xi$, and
for every $\mathbb{N}',\Xi',F'\in\typdecomp{\mathbb{N}}{\Xi}{G}{B}$
the formula $\subst{\sigma}{\hsubst{\theta}{F'}}$ is defined and valid,
then $\lftype{\subst{\sigma}{\hsubst{\theta}{G}}}{\hsubst{\theta}{B}}$
is derivable in LF.
\end{lemma}
\begin{proof}
We begin by observing that Lemma~\ref{lem:tydecomp-wf} ensures that the 
decomposition will be defined.
The proof then proceeds by induction on the formation of $B$.
Consider the possible structures for $B$.

\case{$B$ is an atomic type $P$.}
Given that
$\wftype{\mathbb{N}\cup\STLCGamma_0\cup\Psi}{B}$ is derivable,
$B$ is of the form $(a\app M_1\ldots M_n)$ and
there must exist $a:K$ in the LF signature for a kind of the form
$\typedpi{x_1}{A_1}{\ldots\typedpi{x_n}{A_n}{\type}}$.
Since the arity kinding for $B$ will ensure the substitution application
is defined, let $A_i'$ denote the type
$\hsubst{\{\langle x_1,M_1,\erase{A_1}\rangle,\ldots,
           \langle x_{i-1},M_{i-1},\erase{A_{i-1}}\rangle\}}
        {A_i}$
for each $i$, $1\leq i\leq n$.
We can then express the collection $\typdecomp{\mathbb{N}}{\Xi}{G}{B}$
as $\bigcup_{i\in 1..n}(\mathbb{N},\Xi,\fatm{G}{M_i:A_i'})$
given that it must be defined.
If the formula
$\subst{\sigma}{\hsubst{\theta}{\fatm{G}{M_i:A_i'}}}$ is defined and valid
for each 
$(\mathbb{N},\Xi,\fatm{G}{M_i:A_i'})\in\typdecomp{\mathbb{N}}{\Xi}{G}{B}$, then
there must be derivations of
$\lfctx{\subst{\sigma}{\hsubst{\theta}{G}}}$,
$\lftype{\subst{\sigma}{\hsubst{\theta}{G}}}
        {\hsubst{\theta}{A_i'}}$, and
$\lfchecktype{\subst{\sigma}{\hsubst{\theta}{G}}}
             {\hsubst{\theta}{M_i}}
             {\hsubst{\theta}{A_i'}}$
for each $i$.
By Theorem~\ref{th:subspermute} the type 
$\hsubst{\{\langle x_1, \hsubst{\theta}{M_1},\erase{A_1}\rangle,
             \ldots,
           \langle x_{i-1}, \hsubst{\theta}{M_{i-1}},\erase{A_{i-1}}\rangle\}}
        {\hsubst{\theta}
                {A_i}}$ 
is the same as $\hsubst{\theta}{A_i'}$
since $\theta$ does not make substitution for any $x_j$.
From these derivations and the fact that $a:K$ is in the LF signature,
Theorem~\ref{th:atomickind} ensures that there exists a derivation
$\lfsynthkind{\subst{\sigma}{\hsubst{\theta}{G}}}{\hsubst{\theta}{B}}{\type}$.
From this we can easily infer the derivability of
$\lftype{\subst{\sigma}{\hsubst{\theta}{G}}}{\hsubst{\theta}{B}}$, 
as is needed.

\case{$B$ is a canonical type $\typedpi{x}{A_1}{A_2}$.}
Letting
$G'=G,n:\erase{A_1}$, $\mathbb{N}'=\mathbb{N}\cup\{n\}$, and
\begin{tabbing}
\qquad\=$\left\{\ctxvarty{\Gamma}{\mathbb{N}'_{\Gamma}}{\ctxty{\mathcal{C}}{\mathcal{G}}}\ \middle|\ \right.$\=\kill
\>$\Xi'=\left\{\ \ctxvarty{\Gamma}{\mathbb{N}'_{\Gamma}}{\ctxty{\mathcal{C}}{\mathcal{G}}}\ \middle|\  
              \ctxvarty{\Gamma}{\mathbb{N}_{\Gamma}}{\ctxty{\mathcal{C}}{\mathcal{G}}}\in\Xi
                          \mbox{ and }\right.$\\
\>\>          $\left. \mathbb{N}'_{\Gamma}\ \mbox{is}\ \mathbb{N}_{\Gamma} \cup \{n\}\
                          \mbox{ if }\Gamma\ \mbox{occurs in}\ G\ \mbox{and}\
                          \mathbb{N}_{\Gamma}\ \mbox{ otherwise }
         \right\}$ 
\end{tabbing}
for a new nominal constant $n$,
the decomposition $\typdecomp{\mathbb{N}}{\Xi}{G}{B}$
is defined to be the set 
$\typdecomp{\mathbb{N}}{\Xi}{G}{A_1} \cup
    \typdecomp{\mathbb{N}'}{\Xi'}{G'}{\hsubst{\{\langle x, n,\erase{A_1}\rangle\}}{A_2}}$.
Given that 
$\mathbb{N}\setminus\mathbb{N}_i=(\mathbb{N},n:\erase{A_1})\setminus(\mathbb{N}_i,n:\erase{A_1})$
for any $\mathbb{N}_i$, the context variable context $\Xi'$ will then satisfy the
requirements of this lemma.
From the derivation of $\wftype{\mathbb{N}\cup\STLCGamma_0\cup\Psi}{B}$
we can extract derivations for
$\wftype{\mathbb{N}\cup\STLCGamma_0\cup\Psi}{A_1}$
and 
$\wftype{\aritysum{\{x:\erase{A_1}\}}{(\mathbb{N}\cup\STLCGamma_0\cup\Psi)}}{A_2}$.
By Theorem~\ref{th:aritysubs-ty}
$\wftype{{\mathbb{N}'\cup\STLCGamma_0\cup\Psi}}{\hsubst{\{\langle x,n,\erase{A_1}\rangle\}}{A_2}}$
must then be derivable.
We have as assumptions derivations for both
$\wfctx{\mathbb{N}\cup\STLCGamma_0\cup\Psi}{\Xi}{G}$
and $\wftype{\mathbb{N}\cup\STLCGamma_0\cup\Psi}{A_1}$, thus
we are able to construct a derivation for
$\wfctx{\mathbb{N}'\cup\STLCGamma_0\cup\Psi}{\Xi}{G,n:A_1}$.
The type $A_1$ is clearly smaller than $B$, and it is easy to argue that
$\hsubst{\{\langle x,n,\erase{A_1}\rangle\}}{A_2}$ must be as well.
Thus by invoking the inductive hypothesis twice we determine that both
$\lftype{\subst{\sigma}{\hsubst{\theta}{G}}}{\hsubst{\theta}{A_1}}$
and
$\lftype{\subst{\sigma}{\hsubst{\theta}{(G,n:A_1)}}}{\hsubst{\theta}{A_2'}}$
are derivable.
From these we can construct a derivation for
$\lftype{\subst{\sigma}{\hsubst{\theta}{G}}}{\hsubst{\theta}{\typedpi{x}{A_1}{A_2}}}$
using an application of \canonfampi.
\end{proof}

We conclude this section by proving the soundness of these proof
rules with respect to the semantics.

\begin{theorem}\label{th:meta-sound}
The following property holds for every instance of each of the rules
in Figure~\ref{fig:rules-meta}:
if the premises expressing typing judgements are derivable, the
conditions described in the other non-sequent premises are satisfied
and all the premise sequents are valid, 
then the conclusion sequent must also be valid.
\end{theorem}
\begin{proof}
Consider each rule in Figure~\ref{fig:rules-meta}.

\case{\lfwk}
Let $\Upsilon$ be an arbitrary height assignment and
$\theta$ and $\sigma$ identify an arbitrary closed instance of the conclusion
sequent.
Suppose that all the formulas in 
$\subst{\sigma}{\hsubst{\theta}{\Omega}}$
are valid with respect to $\Upsilon$
as if any were not this closed instance would be vacuously valid.
The instance would be similarly valid if 
$\subst{\sigma}{\hsubst{\theta}{\fatmann{G}{M:A}{Ann}}}$ were
not valid, so suppose it is.
Further, assume that any nominal constants in $\mathbb{N'}\setminus\mathbb{N}$ 
do not appear in $\theta$ or $\sigma$ as by Theorem~\ref{th:perm-valid}
the validity of any instance which does use such nominal constants is ensured by
the validity of every instance which does not.
Such $\theta$ and $\sigma$ will also identify closed instances for each of
the premise sequents.
From the validity of these premise sequents,
we can conclude by Lemma~\ref{lem:tydecomp-defn} that
$\lftype{\subst{\sigma}{\hsubst{\theta}{G}}}{\hsubst{\theta}{B}}$ is derivable 
in LF.
By the validity of 
$\subst{\sigma}{\hsubst{\theta}{\fatmann{G}{M:A}{Ann}}}$
there must exist LF derivations for
\begin{enumerate}
\item $\lfctx{\subst{\sigma}{\hsubst{\theta}{G}}}$,
\item $\lftype{\subst{\sigma}{\hsubst{\theta}{G}}}{\hsubst{\theta}{A}}$, and
\item $\lfchecktype{\subst{\sigma}{\hsubst{\theta}{G}}}{\hsubst{\theta}{M}}{\hsubst{\theta}{A}}$
\end{enumerate}
which satisfy the restriction on the derivation height identified by $Ann$
with respect to $\Upsilon$.
Since $n$ cannot appear in $G$, $M$, or $A$ by assumption or in $\theta$ by virtue 
of its being substitution compatible with the conclusion sequent, we can 
construct a derivation for
$\lfctx{\subst{\sigma}{\hsubst{\theta}{(G,n:B)}}}$ through an application of
$\ctxterm$
and also conclude that
$\lftype{\subst{\sigma}{\hsubst{\theta}{(G,n:B)}}}{\hsubst{\theta}{A}}$, and
$\lfchecktype{\subst{\sigma}{\hsubst{\theta}{(G,n:B)}}}{\hsubst{\theta}{M}}{\hsubst{\theta}{A}}$
have derivations of the same height as (2) and (3) respectively through the application 
of Theorem~\ref{th:weakening}.
Thus the restriction identified for the annotation $Ann$ with respect to $\Upsilon$
will clearly also be satisfied by this derivation for the later judgement,
and so
$\subst{\sigma}{\hsubst{\theta}{\fatmann{G,n:B}{M:A}{Ann}}}$ will be valid with respect to $\Upsilon$.
Thus this closed instance of the sequent must be valid.

\case{\lfstr}
Let $\Upsilon$ be an arbitrary height assignment and
$\theta$ and $\sigma$ identify an arbitrary closed instance of the conclusion
sequent.
Suppose that all the formulas in 
$\subst{\sigma}{\hsubst{\theta}{\Omega}}$
are valid with respect to $\Upsilon$ 
as if any were not this closed instance would be vacuously valid.
Consider the goal formula
$\fimp{\fatmann{G,n:B}{M:A}{Ann}}{\fatmann{G}{M:A}{Ann}}$.
If $\subst{\sigma}{\hsubst{\theta}{\fatmann{G,n:B}{M:A}{Ann}}}$ 
were not valid with respect to $\Upsilon$ then this instance of the formula would be vacuously valid, 
so assume it is valid with respect to $\Upsilon$.
Then there must be derivations of
\begin{enumerate}
\item $\lfctx{\subst{\sigma}{\hsubst{\theta}{(G,n:B)}}}$,
\item $\lftype{\subst{\sigma}{\hsubst{\theta}{(G,n:B)}}}
              {\hsubst{\theta}{A}}$, and
\item $\lfchecktype{\subst{\sigma}{\hsubst{\theta}{(G,n:B)}}}
                   {\hsubst{\theta}{M}}
                   {\hsubst{\theta}{A}}$
\end{enumerate}
of some height $m$ which satisfies the height restriction identified by 
$Ann$ with respect to $\Upsilon$.
The derivation of (1) must conclude by \ctxterm\ using a derivation of
$\lfctx{\subst{\sigma}{\hsubst{\theta}{G}}}$ and $n$ cannot appear in
the context $\subst{\sigma}{\hsubst{\theta}{G}}$.
Since $n$ does not appear in $A$ or $M$, and since it cannot be
in the support of the substitution $\theta$, it cannot appear in
$\hsubst{\theta}{M}$ or $\hsubst{\theta}{A}$.
So by Theorem~\ref{th:strengthening} the judgements
$\lftype{\subst{\sigma}{\hsubst{\theta}{G}}}
        {\hsubst{\theta}{A}}$, and
$\lfchecktype{\subst{\sigma}{\hsubst{\theta}{G}}}
             {\hsubst{\theta}{M}}
             {\hsubst{\theta}{A}}$
have derivations and the later judgement has a derivation of height $m$.
Thus it clearly also satisfies the height restriction identified by $Ann$
with respect to $\Upsilon$.
Therefore the formula $\subst{\sigma}{\hsubst{\theta}{\fatmann{G}{M:A}{Ann}}}$
is valid with respect to $\Upsilon$, as needed.

\case{\lfperm}
Let $\Upsilon$ be an arbitrary height assignment and
$\theta$ and $\sigma$ identify an arbitrary closed instance of the conclusion
sequent.
Suppose that all the formulas in 
$\subst{\sigma}{\hsubst{\theta}{\Omega}}$
are valid with respect to $\Upsilon$, as if any were not this closed instance 
would be vacuously valid.
Consider the goal formula under this instance,
$\subst{\sigma}{\hsubst{\theta}{(\fimp{\fatmann{G}{M:A}{Ann}}{\fatmann{G'}{M:A}{Ann}})}}$.
If $\subst{\sigma}{\hsubst{\theta}{\fatmann{G}{M:A}{Ann}}}$ 
were not valid with respect to $\Upsilon$ the implication would be vacuously valid,
so suppose it is valid with respect to $\Upsilon$.
Then there must be derivations of
\begin{enumerate}
\item $\lfctx{\subst{\sigma}{\hsubst{\theta}{G}}}$, 
\item $\lftype{\subst{\sigma}{\hsubst{\theta}{G}}}
              {\hsubst{\theta}{A}}$, and 
\item $\lfchecktype{\subst{\sigma}{\hsubst{\theta}{G}}}
                   {\hsubst{\theta}{M}}
                   {\hsubst{\theta}{A}}$ 
\end{enumerate}
in LF which satisfy the derivation height restriction identified by $Ann$ 
with respect to $\Upsilon$.
Given that $n_1$ cannot appear in $A_2$, it also cannot appear in
$\hsubst{\theta}{A_2}$ as the support of $\theta$ cannot contain $n_1$.
Thus by Theorem~\ref{th:exchange} there are also derivations for
$\lfctx{\subst{\sigma}{\hsubst{\theta}{G'}}}$, 
$\lftype{\subst{\sigma}{\hsubst{\theta}{G'}}}
        {\hsubst{\theta}{A}}$, and 
$\lfchecktype{\subst{\sigma}{\hsubst{\theta}{G'}}}
             {\hsubst{\theta}{M}}
             {\hsubst{\theta}{A}}$.
Furthermore, these derivations are of the same height as (1), (2), and (3) respectively,
and thus must also respect the annotation $Ann$ with respect to $\Upsilon$.
Therefore the formula $\subst{\sigma}{\hsubst{\theta}{\fatmann{G'}{M:A}{Ann}}}$
will be valid with respect to $\Upsilon$, as needed.

\case{\lfinst}
Let $\Upsilon$ be an arbitrary height assignment and
$\theta$ and $\sigma$ identify an arbitrary closed instance of the conclusion
sequent.
Suppose that all the formulas in 
$\subst{\sigma}{\hsubst{\theta}{\Omega}}$
are valid with respect to $\Upsilon$ since if any were not this closed instance would be vacuously valid.
Consider the goal formula
$\subst{\sigma}{\hsubst{\theta}{(
 \fimp{\fatm{G,G'}{M:A}}{\fimp{
       \fatm{G}{M_1:A_1}}{
       \fatm{G''}{M':A'}}})}}$.
If either $\subst{\sigma}{\hsubst{\theta}{\fatm{G}{M_1:A_1}}}$
or $\subst{\sigma}{\hsubst{\theta}{\fatm{G,G'}{M:A}}}$ 
are not valid with respect to $\Upsilon$, then this implication would be 
vacuously valid, so suppose that both of these formulas are valid.
Then there must be derivations in LF for
\begin{enumerate}
\item
$\lfctx{\subst{\sigma}{\hsubst{\theta}{(G,G')}}}$,
\item
$\lfctx{\subst{\sigma}{\hsubst{\theta}{G}}}$,
\item
$\lftype{\subst{\sigma}{\hsubst{\theta}{(G,G')}}}
        {\hsubst{\theta}{A}}$,
\item
$\lftype{\subst{\sigma}{\hsubst{\theta}{G}}}
        {\hsubst{\theta}{A_1}}$,
\item
$\lfchecktype{\subst{\sigma}{\hsubst{\theta}{(G,G')}}}
             {\hsubst{\theta}{M}}
             {\hsubst{\theta}{A}}$, and
\item
$\lfchecktype{\subst{\sigma}{\hsubst{\theta}{G}}}
             {\hsubst{\theta}{M_1}}
             {\hsubst{\theta}{A_1}}$.
\end{enumerate}
By Theorems~\ref{th:transitivity} and~\ref{th:subspermute} there are
derivations of
$\lfctx{\subst{\sigma}{\hsubst{\theta}{G''}}}$,
$\lftype{\subst{\sigma}{\hsubst{\theta}{G''}}}
        {\hsubst{\theta}{A'}}$, and
$\lfchecktype{\subst{\sigma}{\hsubst{\theta}{G''}}}
             {\hsubst{\theta}{M'}}
             {\hsubst{\theta}{A'}}$.
Thus the formula $\subst{\sigma}{\hsubst{\theta}{\fatm{G''}{M':A'}}}$
will be valid with respect to $\Upsilon$, as needed.
\end{proof}

%% file: Adelfa/adelfa.tex
\chapter{Adelfa: An Implementation of the Proof System}
\label{ch:adelfa}
Now that we have described a proof system for constructing validity arguments,
we would like to provide a tool for mechanically constructing these proofs.
We have implemented a proof assistant based on the logic which allows for the
construction of proofs via a collection of tactics which correspond to 
the different reasoning steps that are available in the proof system.
The sequents of the proof system correspond to states in the prover.
The system was built in OCaml and has been used in a collection
of reasoning examples, which are covered in the next chapter.

In the first section we will introduce the Adelfa reasoning
system, it's structure and how it is used to mechanize the construction
of proofs in our logic.
There are a few special considerations in implementing this system which make up
the remaining sections in this chapter.
The $\appL$ rule depends on a function $\casessans$ which generates sequents using
a covering set of arity type preserving substitutions, and we must
provide some realization of this process in Adelfa.
In the second section we show that an implementation of
higher-order pattern unification~\cite{miller91jlc, nadathur05iclp}
can be used to identify these covering sets of solutions such that they
satisfy the necessary restrictions.
The final section will look at the form of formulas used in the proof
assistant.
There are certain proof rules which may be permuted out to the 
end of a derivation, allowing us to work with a focused formula syntax
in reasoning.

\input{Adelfa/sys-intro}
\input{Adelfa/cases}
\input{Adelfa/focus-form}

%% file: Adelfa/sys-intro.tex
\section{An Overview of the System}
The architecture of the Adelfa system is influenced significantly by
that of the Abella proof assistant~\cite{gacek09phd,baelde14jfr}.
Adelfa has two levels of execution: the top level
interaction and the proof level interaction.
The former allows for the definition of an LF signature, context
schemas, and theorems.
Proved theorems are stored at this level as formulas which can be
used as lemmas in later proofs.
LF expressions in Adelfa follow Twelf syntax except that there are no
types in abstraction terms as our logic is based on Canonical LF.
We use \verb'{G |- M : A}' to represent atomic formulas, $\ftrue$ and
$\ffalse$ become \verb'true' and \verb'false', the connectives
$\fimp{}{}$, $\fand{}{}$, and $\for{}{}$ become \verb'=>', \verb'/\',
and \verb'\/' respectively, and the quantifiers $\forall$, $\exists$,
and $\Pi$ are denoted by \verb'forall', \verb'exists', and \verb'ctx'.

The proof level is entered when a theorem is proposed.
In essence, the proof states in Adelfa correspond to
sequents of the proof system and represent that there is an obligation
to provide a proof for that sequent.
Proofs are constructed bottom-up using a defined set of tactics which
correspond to a sequence of valid applications of proof rules.
Some proof rules contain multiple sequents in the premises, thus a
proof state in Adelfa will also keep track of all the remaining
obligations as a stack.

Proof states in Adelfa are displayed in the following form.
\begin{verbatim}
Vars: T:o, E:o
Nominals: n1:o, n:o
Contexts: Gamma:c[(n:tm, n1:of n T)]
H1:{Gamma |- n : tm}
H2:{T : ty}

==================================
exists D3, {Gamma |- D3 : of n T}
\end{verbatim}
Above the line are the components which are to the left of the sequent
arrow in a 
sequent, and below the line is the goal formula which would appear on
the right.
Any further obligations in the stack are included below this state and are
identified by the goal formula.
We can see each of the sequent components identified here with both
nominal constants and variables identified along with their arity
types, and context variables identified along with their context
variable type.
The formulas in the assumption set appear with identifiers which are
used to reference them in the system.

\begin{figure}
\begin{tabular}{|p{4cm}|p{10cm}|}
{\bf Tactic} 
  & 
  {\bf Effect}\\
\hline
\tapply\ {\tt H} to {\tt H1 ... Hn}
with {\tt (G1 = g1), ...,}
{\tt n1 = t1, ...}
  &
  Extends the hypotheses with the result of applying lemma or hypothesis {\tt H}
  with the given arguments.\\
\hline
\tassert\ {\tt f}
  &
  Creates two new subgoals, one with {\tt f} as the goal and 
  another proving the current goal with {\tt f} as an assumption.\\
\hline
\tcase\ {\tt H}
  &
  Applies an appropriate left to the hypothesis {\tt H}.\\
\hline
\texists\ {\tt t}
  &
  Instantiates the outermost existential quantifier with 
  {\tt t}.\\
\hline\tind\ {\tt on i}
  &
  Adds height annotations to the {\tt i}th implication antecedent and introduces
  the inductive hypothesis to assumptions.\\
\hline
\tinst\ {\tt H} with {\tt n = t}
  &
  Instantiates the name {\tt n} with {\tt t} in the atomic hypothesis {\tt H}.\\
\hline
\tintros 
  &
  Introduces variables for outer universal and context quantifiers and 
  hypotheses for outer implications. \\
\hline
\tleft,\ \tright 
  &
  Replaces goal {\tt F1 \verb'\'/ F2} with {\tt F1} (\tleft) or
  {\tt F2} (\tright).\\
\hline
\tpermute\ {\tt H} to {\tt G}
  &
  Permutes the context expression of atomic {\tt H} to {\tt G}.\\
\hline
\tsearch 
  &
  Attempts to conclude the current subgoal for the user automatically.\\
\hline
\tsplit 
  &
  Splits a goal {\tt F1 /\verb'\' F2} into two subgoals for {\tt F1} and 
  {\tt F2}.\\
\hline
\tstr\ {\tt H}
  &
  Strengthens the context of atomic {\tt H}.\\
\hline
\tweak\ {\tt H} with {\tt t}
  &
  Weakens the context of atomic {\tt H} with the type {\tt t}.\\
\hline
\end{tabular}
\caption{Tactics of Adelfa}
\label{fig:tactics}
\end{figure}

Central to the system are the tactics which are used to construct
proofs of the theorems.
Tactics are designed to capture only valid reasoning steps in the logic, and 
are also meant to capture the natural reasoning steps of a development.
Figure~\ref{fig:tactics} lists the tactics and provides a brief explanation
of the result of applying the tactic in an Adelfa development.

The tactics \tassert, \texists, \tleft, \tright, and \tsplit\ encode the
use of a single particular proof rule in reasoning, namely 
$\cut$, $\existsR$, $\orR$, and $\andR$ respectively.
The \tcase\ tactic uses the structure of a hypothesis to determine
an applicable left rule from \appL, \andL, or \orL.
Induction is realized through the \tind\ tactic which captures an application
of $\ind$, which introduce the inductive hypothesis for reasoning inductively
on the height of the identified typing judgement into the assumption formulas.

The tactics \tapply, \tintros, and \tsearch\ all involve following
a procedure for applying multiple proofs rules in the given state.
The application of hypotheses or lemmas is encoded by \tapply.
Based on the structure of the formula being applied, this tactic 
encodes using proof rules \allL, \ctxL, and \impL\ to instantiate
the formula with the given arguments.
Previously proved theorems are given by their identifier and would
encode first using \cut\ to introduce the formula and then following the
above process.
The \tintros\ tactic captures the introduction of variables and hypotheses 
for the outer quantifiers and implications.
This encodes the application of some number of instances of the
\allR, \ctxR, and \impR\ proof rules based on the structure of the
current goal formula.
The \tsearch\ tactic attempts to automatically construct a proof for the 
current subgoal using $\id$, $\appR$, and $\absR$.
The process is one which will always eventually find such a derivation if one
exists given that the application of $\appR$ and $\absR$ will decrease the
size of the term in an atomic formula and exactly one of the two applies 
must apply to any closed term.

The remaining tactics, \tweak, \tstr, \tpermute, and \tinst, 
capture the application of an instance of an LF meta-theorems on the left.
The general structure is one which uses \cut\ to introduce the formula capturing
the meta-theorem to be applied using the corresponding proof rule to determine
any conditions that remain to be shown.
Then the meta-theorem formula is applied using \impL\ to the appropriate
hypotheses.
The application of weakening also must first introduce a new name
for the binding using \sstr.
In Adelfa these tactics are only applied successfully when the premises of the
proof rule encoding the meta-theorem can be verified automatically.

We will demonstrate the use of these tactics in practice through the use of 
example Adelfa developments in Chapter~\ref{ch:examples}.

%% file: Adelfa/cases.tex
\section{Finding Covering Sets of Solutions for Case Analysis}

The $\appL$ rule in the proof system provides a means for reasoning
from atomic assumption formulas by analysing the possible reasons for
their validity. 
The function $\casessans$ is the main process defining the way in which
cases are determined for a given formula.
A key part of its definition is the notion of covering set of solutions
to unification problems.
The substitutions comprising this set must satisfy particular requirements
that are identified in
Definitions~\ref{def:unification} and~\ref{def:covering-subst} respectively.
For the discussion in this section, we assume that the set of equations $\mathcal{E}$
of a well-formed unification problem 
$\mathcal{U}=\langle \mathbb{N}; \Psi; \mathcal{E}\rangle$
comprise only equations between canonical terms and not canonical types.
The reduction operation on a sequent also requires us to consider
equations between atomic types.
However these equations are of the form
$a\app M_1\ldots M_n = a\app M_1'\ldots M_n'$ for some type level
constant $a$ in the signature or the two types in the equation have
two different type constants as their heads. 
In the former case, the equation can be replaced by the set of 
equations $\{M_1 = M'_1, \ldots, M_n = M'_n\}$ and in the latter case
the unification problem is seen immediately to have no solutions. 

In light of the above observation, to implement the analysis embedded
in the $\casessans$ rule, we have to describe a procedure for finding
a covering set of solutions to a set of equations between
(arity-typed) $\lambda$-terms. 
This task is identified with solving a higher-order pattern
unification problem~\cite{miller91jlc} and is addressed in Adelfa by
using the higher-order pattern unification algorithm described by
Nadathur and Linnell~\cite{nadathur05iclp}.
At a high level this procedure determines unifiers for terms by descending through
their structure, ensuring that the fixed, non-variable parts are
identical, eventually simplifying the problem to be solved to
equations of the form $(x\app t_1\app \cdots \app t_n) = s$, where
$x$ is a term variable for which a substitution is to be considered. 
A characteristic of the higher-order pattern unification problem is
that the terms $t_1, \ldots, t_n$ must all be nominal constants or
variables bound by abstractions within whose scope the term $(x\app
t_1\app \cdots \app t_n)$ and, correspondingly, $s$ occurs.
If this equation is solvable, then the substitution that is generated
has the form of the term $s$ enclosed within a sequence of
abstractions binding the occurrences of $t_1,\ldots,t_n$ in $s$.
The circumstances under which the equation is deemed solvable ensure
that the generated substitution term will not contain any nominal
constants.
As a consequence, the support set for substitutions found by the
higher-order pattern unification procedure will be disjoint from any
set $\mathbb{N}$ of nominal constants.
Moreover, the substitution terms will be typeable with respect to the
arity type assignments under which the terms in the unification
problem are well-typed augmented with suitable type assignments for
any new term variables that are introduced in them, and the overall
substitution will itself be arity type preserving with respect to
these assignments.
In short, the requirements 1-3 in Definition~\ref{def:unification}
will be satisfied by the solutions found by the higher-order pattern
unification procedure and, being most-general unifiers, these
solutions will also constitute (singleton) sets of covering solutions.

There is, however, one wrinkle to the use of higher-order pattern
unification in Adelfa.
This form of unification is defined for a more general form of
$\lambda$-terms that includes non-canonical terms and in a setting in
which substitution application does not include the concurrent
reduction of terms to canonical form.
We must therefore verify that the solutions it finds to a unification
problem in which the terms are in canonical form also constitutes a
solution in the sense described in this thesis. 
One requirement, that the substitution terms be in canonical form, is
easily seen to be met: if the input terms are in canonical form, then
the procedure we have outlined above will generate substitutions in
which all the terms are in canonical form under the relevant arity
typing.
Thus it only remains to be seen that if two terms are determined to be
equal under a substitution applied in the sense of the higher-order
pattern unification procedure, then they are also equal when the
substitution is applied in the sense relevant to our logic. 

The remainder of this section will focus on arguing that this last
requirement is also met. 
To do this we essentially argue that the reduction steps which are built
into hereditary substitution can be separated out from the application of
the substitution while ensuring that the resulting terms will be identical.
To provide a framework for such an argument, we will introduce an
extended term syntax and a notion of substitution application which
does not reduce the result to a canonical form. 
We will then identify equivalence classes for terms in this extended
syntax based on the $\lambda$-conversion rules under which terms are
considered to be equal in the context of the higher-order pattern
unification algorithm.
The key result we will show is that the application of a hereditary substitution 
to a term in canonical form will produce a term that is in the same
equivalence class as the term that is obtained by applying the
substitution in the sense of the higher-order pattern unification
algorithm to the term. 

\begin{figure}
\[\begin{array}{lrcl}
\mbox{\bf Terms} & T & ::= & c\ |\ x\ |\ \lflam{x}{T}\ |\ T_1\app T_2
\end{array}\]
\caption{Extended Term Syntax}
\label{fig:ext-terms}
\end{figure}

The extended term syntax is given in Figure~\ref{fig:ext-terms} and it follows
the standard syntax for $\lambda$-terms which allows for terms that
contain $\beta$-redexes.
Two terms are considered to be equal if one can be converted to the
other using the usual rules of $\lambda$-conversion.
We write $t_1\equiv t_2$ to denote the fact that $t_1$ and $t_2$ are
equal in the sense described.
The terms that we consider here are ones that that are well-typed in
an arity typing sense.
In this context, every term will have a canonical form, i.e. a form
in which there are no $\beta$-redexes and the top-level applications
carry the arity type $\oty$.
These canonical forms in fact constitute a subset of the the full set
of terms in the extended syntax that is identical to the collection of
well-formed terms in canonical LF.
We will also view the canonical forms as distinguished representatives
of the equivalence classes they belong to under the operative equality
notion. 

A substitution in this setting can be captured conveniently by
creating a sequence of $\beta$-redexes; the resulting term will be in
the same equivalence class as the term that results from an actual
replacement, paying attention to renaming of bound variables to avoid
accidental capture.
\begin{definition}\label{def:extsubst}
A substitution $\theta$ is a mapping from variables to terms in the
extended term syntax written $\{t_1/x_1,\ldots,t_n/x_n\}$.
The application of substitution 
$\theta=\{t_1/x_1,\ldots,t_n/x_n\}$ to a term $T$ is written as
$\extsubst{\theta}{T}$ and is defined to be  
the term $\lflam{x_1}{\ldots\lflam{x_n}{(T\app t_1\ldots t_n)}}$.
\end{definition}

We now show that the result of hereditary substitution will be
a term which, under the extended syntax for terms, is a member of the same
equivalence class as a term obtained by applying the substitution
following the simpler notion of application given in 
Definition~\ref{def:extsubst}.
Clearly then, normalizing the term obtained using the simpler notion of 
substitution application will result in the same term, up to renaming,
as that obtained using hereditary substitution.

\begin{theorem}
Suppose that $\theta$ is an arity type preserving substitution with
respect to $\Theta$ and
$\stlctyjudg{\aritysum{\context{\theta}}{\Theta}}{t}{\alpha}$ has a derivation.
Letting $\theta'=\left\{t/x\middle| \langle x, t, \alpha\rangle\in\theta\right\}$,
$\extsubst{\theta'}{t}\equiv\hsubst{\theta}{t}$.
\end{theorem}
\begin{proof}
This proof is by induction using a lexicographic ordering of the size of $\theta$, identified by the sum of the size of each type indexing $\theta$, and
the formation of the term $t$.
We will consider each case based on the structure of $t$.

\case{$\mathbf{t=x}$ for some $\mathbf{x:\alpha\in\aritysum{\context{\theta}}{\Theta}}$}
If $x\in\dom{\theta}$ then 
$\extsubst{\theta'}{x}= t_x$ for $\langle x, t_x,\alpha\rangle\in\theta$ and therefore
$\extsubst{\theta'}{x}\equiv\hsubst{\theta}{x}$ will clearly hold.
If $x\not\in\dom{\theta}$ then $\extsubst{\theta'}{x}= x$,
and $\extsubst{\theta'}{x}\equiv\hsubst{\theta}{x}$ in this
case as well.

\case{$\mathbf{t=(t_1\app t_2)}$}
Then there will be derivations for the arity typing judgements
$\stlctyjudg{\aritysum{\context{\theta}}{\Theta}}{t_1}{\arr{\alpha'}{\alpha}}$ and 
$\stlctyjudg{\aritysum{\context{\theta}}{\Theta}}{t_2}{\alpha'}$.
By induction we can determine that
$\extsubst{\theta'}{t_1}\equiv\hsubst{\theta}{t_1}$ and
$\extsubst{\theta'}{t_2}\equiv\hsubst{\theta}{t_2}$.
It is not difficult to see that
$\extsubst{\theta'}{t}\equiv(\extsubst{\theta'}{t_1})\app(\extsubst{\theta'}{t_2})$, and 
thus that
$\extsubst{\theta'}{t}\equiv(\hsubst{\theta}{t_1})\app(\hsubst{\theta}{t_2})$.
If $\hsubst{\theta}{t_1}$ is not an abstraction term then
$\hsubst{\theta}{t}=
    (\hsubst{\theta}{t_1})\app(\hsubst{\theta}{t_2})$, and therefore
$\extsubst{\theta'}{t}\equiv\hsubst{\theta}{t}$ will clearly hold.
If on the other hand
$\hsub{\theta}{t_1}{\lflam{y}{s}:\arr{\alpha'}{\alpha}}$, then
$\hsubst{\theta}{t}=
    \hsubst{\{\langle y,(\hsubst{\theta}{t_2}),\alpha'\rangle\}}{s}$.
In this case we can see that the equivalence
$\extsubst{\theta'}{t}\equiv(\lflam{y}{s})\app(\hsubst{\theta}{t_2})$
will hold.
Using an application of $\beta$-reduction for the
top-level $\beta$-redex in $(\lflam{y}{s})\app(\hsubst{\theta}{t_2})$,
we see that $\extsubst{\theta'}{t}$ is then equivalent to
$\extsubst{(\hsubst{\theta}{t_2})/y}{s}$.
From the derivation of 
$\stlctyjudg{\aritysum{\context{\theta}}{\Theta}}{t_1}{\arr{\alpha'}{\alpha}}$
and that $\theta$ is arity type preserving with respect to $\Theta$ we can
determine that
$\stlctyjudg{\aritysum{y:\alpha'}{\Theta}}{s}{\alpha}$ has a derivation.
The size of the type $\alpha'$ must be strictly smaller than the size of the
types indexing $\theta$ by the definition of hereditary substitution, and
thus we can apply the induction hypothesis to conclude that
$\extsubst{(\hsubst{\theta}{t_2})/y}{s}\equiv
    \hsubst{\{\langle y,(\hsubst{\theta}{t_2}),\alpha'\rangle\}}{s}$.
Therefore, $\extsubst{\theta'}{t}\equiv\hsubst{\theta}{t}$.

\case{$\mathbf{t=\lflam{x}{t'}}$}
Then $\alpha=\arr{\alpha_1}{\alpha_2}$ and there is a derivation of
$\stlctyjudg{\aritysum{\context{\theta}}{\Theta,x:\alpha_1}}{t'}{\alpha_2}$.
By induction then,
$\extsubst{\theta'}{t'}\equiv\hsubst{\theta}{t'}$.
Since
$\extsubst{\theta'}{t}\equiv\lflam{x}{(\extsubst{\theta'}{t'})}$
and
$\hsubst{\theta}{t}=\lflam{x}{(\hsubst{\theta}{t'})}$,
we can thus conclude that
$\extsubst{\theta'}{t}\equiv\hsubst{\theta}{t}$.
\end{proof}

For $\theta$ and
$\theta'=\left\{t/x\middle| \langle x, t, \alpha\rangle\in\theta\right\}$
satisfying the arity typing requirements of the above theorem, we see that
for any canonical terms $M$ and $M'$ if 
$\extsubst{\theta'}{M}=\extsubst{\theta'}{M'}$
we can apply the above lemma to conclude that $\hsubst{\theta}{M}=\hsubst{\theta}{M'}$
also holds.
Therefore since higher-order pattern unification ensures that 
$\extsubst{\theta'}{M}=\extsubst{\theta'}{M'}$ the substitutions found
following this procedure will also be solutions in the sense of 
Definition~\ref{def:unification} and therefore identify a covering set
of solutions, as described at the beginning of this section.

%% file: Adelfa/focus-form.tex
\section{Focusing Formulas}
In this section we consider the construction of proofs in the logic
and note that there are some rules which are permutable such that 
they can be moved to the end of a derivation and thus be applied automatically
in the reasoning system.
This leads to a focusing of formulas in the reasoning system to ones of a
restricted structure.
The rules that we will consider in this discussion are $\absL$ and $\existsL$.

Given the LF typing rules we note that 
$\lfchecktype{\Gamma}{\lflam{x}{M}}{\typedpi{x}{A_1}{A_2}}$
is derivable if and only if
$\lfchecktype{\Gamma,x:A_1}{M}{A_2}$ is derivable.
The following theorem lifts this idea to the derivability of sequents
containing atomic assumption formulas involving abstraction terms.
We argue that for any such sequent which is derivable there must also be a
derivation for the sequent which concludes by an application of $\absL$.
In this way, we see that it is sound to reduce all atomic assumption formulas
of a sequent to be typing judgements over atomic terms.

\begin{theorem}
Suppose the sequent 
$\seq[\mathbb{N}]{\Psi}{\Xi}{\setand{\Omega}{\fatm{G}{\lflam{x}{M}:\typedpi{x}{A_1}{A_2}}}}{F}$ 
is derivable.
Let $n$ be a new nominal constant not appearing in $\mathbb{N}$,
$M'$ be $\hsubst{\{\langle x, n, \erase{A_1}\rangle\}}{M}$,
$A_2'$ be $\hsubst{\{\langle x, n, \erase{A_1}\rangle\}}{A_2}$,
and $\Xi'$ be 
$(\Xi\setminus
   \{\ctxvarty{\Gamma_i}
              {\mathbb{N}_i}
              {\ctxty{\mathcal{C}_i}{\mathcal{G}_i}}\})
           \cup
   \{\ctxvarty{\Gamma_i}
              {(\mathbb{N}_i,n:\erase{A_1})}
              {\ctxty{\mathcal{C}_i}{\mathcal{G}_i}}\}$
if $G$ contains a context variable $\Gamma_i$ or $\Xi$ if $G$ contains
no context variables.
Then the sequent
$\seq[\mathbb{N},n:\erase{A_1}]
     {\Psi}
     {\Xi'}
     {\setand{\Omega}
             {\fatm{G,n:A_1}
                   {M':A_2'}}}
     {F}$
must be derivable.
\end{theorem}
\begin{proof}
Let $\mathcal{D}$ be the derivation for 
$\seq[\mathbb{N}]{\Psi}{\Xi}{\setand{\Omega}{\fatm{G}{\lflam{x}{M}:\typedpi{x}{A_1}{A_2}}}}{F}$.
Applying first $\sweak$ followed by $\weakening$ to the derivation $\mathcal{D}$ 
we construct a derivation for the sequent
$\seq[\mathbb{N},n:\erase{A_1}]
     {\Psi}
     {\Xi'}
     {\setand{\Omega}
             {\fatm{G}{\lflam{x}{M}:\typedpi{x}{A_1}{A_2}},
              \fatm{G,n:A_1}{M':A_2'}}}
     {F}$.
Using an application of $\id$ followed by $\absR$ we can also construct a 
derivation for the sequent
$\seq[\mathbb{N},n:\erase{A_1}]
     {\Psi}
     {\Xi'}
     {\setand{\Omega}{\fatm{G,n:A_1}{M':A_2'}}}
     {\fatm{G}{\lflam{x}{M}:\typedpi{x}{A_1}{A_2}}}$.
Finally, an application of $\cut$ using these two derivations will construct a 
derivation for the sequent
$\seq[\mathbb{N},n:\erase{A_1}]
     {\Psi}
     {\Xi'}
     {\setand{\Omega}{\fatm{G,n:A_1}{M':A_2'}}}
     {F}$
as needed.
\end{proof}

The same basic structure is used to argue for the reduction of existential
formulas in assumptions using $\existsL$.
The essence of the argument is that we can recover the validity of the
existential formula, and thus do not lose any derivations by applying
$\existsL$ eagerly in Adelfa.

\begin{theorem}
Let $\mathbb{N}=n_1:\alpha_1,\ldots,n_m:\alpha_m$ be a collection of arity 
typed nominal constants.
Suppose that the sequent
$\seq[\mathbb{N}]{\Psi}{\Xi}{\setand{\Omega}{\fexists{x:\alpha}{F_1}}}{F_2}$
has a derivation.
Let $y$ be an arbitrary new variable of arity type
$\alpha'=(\arr{\arr{\arr{\alpha_1}{\ldots}}{\alpha_m}}{\alpha})$,
$\Psi'=\Psi,y:\alpha'$, and
$\hsub{\{\langle x,(y\app n_1\ldots n_m),\alpha\rangle\}}{F_1}{F_1'}$.
Then the sequent
$\seq[\mathbb{N}]{\Psi'}{\Xi}{\setand{\Omega}{F_1'}}{F_2}$
has a derivation.
\end{theorem}
\begin{proof}
Let $\mathcal{D}$ be the derivation for 
$\seq[\mathbb{N}]{\Psi}{\Xi}{\setand{\Omega}{\fexists{x:\alpha}{F_1}}}{F_2}$.
From $\mathcal{D}$ using an application of $\sweak$ followed by $\weakening$ 
we construct a derivation for
$\seq[\mathbb{N}]{\Psi'}{\Xi}{\setand{\Omega}{\fexists{x:\alpha}{F_1},F_1'}}{F}$.
We also construct a derivation for the sequent
$\seq[\mathbb{N}]{\Psi'}{\Xi}{\setand{\Omega}{F_1'}}{\fexists{x:\alpha}{F_1}}$
from an application of $\id$ followed by $\existsR$.
From these two derivations we use an application of $\cut$ to construct a 
derivation for the sequent
$\seq[\mathbb{N}]{\Psi'}{\Xi}{\setand{\Omega}{F_1'}}{F}$
as needed.
\end{proof}

Given these results, the system Adelfa will automatically reduce assumption
formulas to have no top-level existential quantifiers and ensure that any
atomic formulas in the assumptions are in the form of an atomic LF term.

%% file: examples/adelfa-ex.tex
\chapter{Constructing Proofs Using Adelfa}
\label{ch:examples}
In this chapter we demonstrate the construction of proofs in
Adelfa using tactics.
Our first example is a proof of the existence of an additive identity for
natural numbers.
This example is used to introduce the basic structure of reasoning in Adelfa and
demonstrates the expressiveness of the logic as the statement of
the theorem contains two quantifier alternations.
We then consider the example of type uniqueness,
which we use to introduce reasoning by induction, and the use of
case analysis, and the role of contexts in Adelfa reasoning.
Next we present an encoding for a simple sequent calculus and 
consider proving that cut is admissible in this system.
Through this example we will bring out how the meta-theorems of LF
are used in reasoning using Adelfa.
We conclude this chapter with a discussion of the transitivity and
narrowing of $F_{<:}$, Problem 1A of the POPLMark
Challenge~\cite{aydemir05tphols}.
The complete development for all of these examples can be found on the
Adelfa website.

\input{examples/plus}
\input{examples/stlc}
\input{examples/cut}

\input{examples/fsub}

%% file: examples/plus.tex
\section{Additive Identity for Natural Numbers}

\begin{figure}
\[
\begin{array}{lcl}
\natty:\type & \qquad&
    \plusty:\arr{\natty}{\arr{\natty}{\arr{\natty}{\type}}}\\
\ztm:\natty & &
    \plusztm:\typedpi{N}{\natty}{\plusty\app\ztm\app N\app N}\\
\stm:\arr{\natty}{\natty} & &
    \plusstm:\typedpi{N_1}{\natty}{\typedpi{N_2}{\natty}{\typedpi{N_3}{\natty}{}}}\\
& & \qquad\qquad
             \typedpi{D}{\plusty\app N_1\app N_2\app N_3}
                     {\plusty\app(\stm\app N_1)\app N_2\app(\stm\app N_3)}
\end{array}
\]
\caption{An LF Specification for Natural Number Addition}
\label{fig:arith-term-spec}
\end{figure}

An encoding of natural numbers and addition over these expressions
is given in Figure~\ref{fig:arith-term-spec}.
There is a single type $\natty$ for representing natural numbers in the system
and the type family $\plusty$ represents the addition relation over these terms.
This signature will be presented to Adelfa as the relevant LF signature to be
used in reasoning.

The existence of an identity for the addition relation is captured by the following formula.
\[\fexists{i:\oty}
          {\fall{x:\oty}
                {\fimp{\fatm{\cdot}{x:\natty}}
                      {\fexists{d:\oty}
                               {\fatm{\cdot}{d:\plusty\app x\app i\app x}}}}}
\]
In particular, this states that there is a right identity for the relation.
Informally, the proof must first introduce the identity term, $\ztm$,
and argue that for this instance the formula will be valid.
Since the encoding of $\plusty$ we have given is recursive in the first argument,
the proof of this theorem will be by induction on the formation of $x$ 
which identifies two structures, $x=\ztm$ or $x=\stm\app x'$.
For the former case, $(\plusztm\app\ztm)$ clearly inhabits the required type
while for the latter case we make use of the inductive hypothesis on the smaller
term $x'$ and construct an appropriate term from this result using $\plusstm$.
The development in Adelfa, as we will see below, follows this structure.

The Adelfa development for this theorem begins by introducing the identity
term through an application of the \texists\ tactic with the expression $\ztm$.
This will result in the following state.
\begin{verbatim}
Vars: 
Nominals: 
Contexts: 
==================================
forall x:o, {x : nat} => exists d:o, {d : plus x z x}
\end{verbatim}
We now want to construct an inductive argument on the formation of $x$,
which is captured in Adelfa through the use of \tind\ with the argument $1$
to identify that the first antecedent is the LF derivation on which the induction 
is to be based.
The result of the application of this tactic is the following state.
\begin{verbatim}
Vars: 
Nominals: 
Contexts: 
IH: forall x:o, {x : nat}* => exists d:o, {d : plus x z x}
==================================
forall x:o, {x : nat}@ => exists d:o, {d : plus x z x}
\end{verbatim}
At this stage we introduce a new eigenvariable for $x$, and add the
formula \verb'{x : nat}@' to the assumptions through an application of the \tintros\ 
tactic.
Following the use of this tactic the new state of the development will be the following.
\begin{verbatim}
Vars: x:o
Nominals: 
Contexts: 
IH: forall x:o, {x : nat}* => exists d:o, {d : plus x z x}
H1: {x : nat}@
==================================
exists d:o, {d : plus x z x}
\end{verbatim}
An application of $\tcase$ on \verb'H1' will identify two structures for 
the derivation of $\fatm{}{x:\natty}$
and thus two subgoals in Adelfa, corresponding to the base and inductive case
of the informal argument.
The subgoal corresponding to the case where $x=\ztm$ is the following.
\begin{verbatim}
Vars: x:o
Nominals: 
Contexts: 
IH: forall x:o, {x : nat}* => exists d:o, {d : plus x z x}
H1: {z : nat}@
==================================
exists d:o, {d : plus z z z}
\end{verbatim}
We can conclude this subgoal through instantiating the existential with an appropriate term
through the tactic \verb$exists (plus_z z)$, and invoking \tsearch\ to
have Adelfa automatically determine that the goal is derivable using the proof rules
for deriving atomic formulas.
Once this subgoal is completed, we must then argue for the inductive case, represented
in the following state.
\begin{verbatim}
Vars: x:o, x':o
Nominals: 
Contexts: 
IH: forall x:o, {x : nat}* => exists d:o, {d : plus x z x}
H1: {s x' : nat}@
H2: {x' : nat}*
==================================
exists d:o, {d : plus (s x') z (s x')}
\end{verbatim}
We now make use of the inductive hypothesis through the application of the 
$\tapply$ tactic, \verb$apply IH to H2$, which will introduce an instance
of the conclusion derivation for some new eigenvariable $d$.
\begin{verbatim}
Vars: x:o, x':o, d:o
Nominals: 
Contexts: 
IH: forall x:o, {x : nat}* => exists d:o, {d : plus x z x}
H1: {s x' : nat}@
H2: {x' : nat}*
H3: {d : plus x' z x'}
==================================
exists d:o, {d : plus (s x') z (s x')}
\end{verbatim}
We complete this proof using this new eigenvariable to instantiate the $d$ in the
goal formula with a term $\plusstm\app x'\app z\app d$ using \texists.
At this point the \tsearch\ tactic can be used and Adelfa will automatically
determine the goal is derivable.
At this point there are no further subgoals representing further proof obligations,
and the derivation is completed.

%% file: examples/stlc.tex
\section{Type Uniqueness for the STLC}
\begin{figure}
\[
\begin{array}{lcl}
  \of{\tpty}{\type}
  & \quad\quad & 
  \of{\ofemptytm}{\ofty\app\emptytm\app\unittm}

  \\
  
  \of{\unittm}{\tpty}
  & &

  \\

  \of{\arrtm}{\arr{\tpty}{\tpty}}
  & &
  \of{\ofapptm}
     {\typedpi{E_1}{\tmty}{\typedpi{E_2}{\tmty}
     {\typedpi{T_1}{\tpty}{\typedpi{T_2}{\tpty}{}}}}}
  \\
  & &
    \qquad
    \typedpi{D_1}{\ofty\app E_1\app (\arrtm\app T_1\app T_2)}{
      \typedpi{D_2}{\ofty\app E_2\app T_1}{}}

  \\

  \of{\tmty}{\type}
  & &
  \qquad \ofty\app(\apptm\app E_1\app E_2)\app T_2

  \\

  \of{\emptytm}{\tmty}
  & &
  \\

  \of{\apptm}{\arr{\tmty}{\arr{\tmty}{\tmty}}} & &
    \of{\oflamtm}
       {\typedpi{R}
       {\arr{\tmty}{\tmty}}
       {\typedpi{T_1}{\tpty}{\typedpi{T_2}{\tpty}{}}}}
  \\

  \of{\lamtm}{\arr{\tpty}{\arr{(\arr{\tmty}{\tmty})}{\tmty}}}
  &  & 
  \qquad \typedpi{D}{(\typedpi{x}{\tmty}
          {\typedpi{y}{\ofty\app x\app T_1}
          {\ofty\app(R\app x)\app T_2}})}{}
  \\
  
  & & \qquad \ofty\app(\lamtm\app T_1\app(\lflam{x}{R\app x}))
                             \app(\arrtm\app T_1\app T_2)
  \\

  \of{\ofty}{\arr{\tmty}{\arr{\tpty}{\type}}} & & \\  

  \of{\eqty}{\arr{\tpty}{\arr{\tpty}{\type}}} & &
    \of{\refltm}{\typedpi{T}{\tpty}{\eqty\app T\app T}}
\end{array}
\]
\caption{An LF Specification for the Simply-Typed Lambda Calculus}
\label{fig:stlc-spec}
\end{figure}

In this section we consider the example of type uniqueness for
the STLC to see how Adelfa is used to reason about systems involving binding.
Proving this property involves reasoning inductively over
typing judgements which requires generalizing over the contexts
of these typing judgements.
Using the same signature presented in Section~\ref{sec:lf-ex}, and the
discussion of the argument structure from Section~\ref{sec:logic-ex}
we look at how this derivation is formalized in Adelfa.
To make it easy to follow the discussion of the reasoning process, 
we have reproduced the specification in Figure~\ref{fig:stlc-spec}.

As in the informal discussion, an Adelfa development of this proof will
introduce a context schema defined by a single block,
$\{t:\oty\} x:\tmty,y:\ofty\app x\app t$.
We name this schema as $ctx$ and use this name to refer to it 
in the discussion that follows.

In this proof we rely on two strengthening properties, one for terms of LF type $\tpty$, 
$\fctx{\Gamma}{ctx}
      {\fall{t :\oty}
            {\fimp{\fatm{\Gamma}{t : \tpty}}
              {\fatm{\emptyce}{t:\tpty}}}}$, 
and building on this one for LF terms of any instance of the $\eqty$ type, 
$\fctx{\Gamma}{ctx}
      {\fall{d :\oty}
      {\fall{t_1 :\oty}
      {\fall{t_2 : \oty}
            {\fimp{\fatm{\Gamma}{d : \eqty\app t_1\app t_2}}
            {\fatm{\emptyce}{d:\eqty\app t_1\app t_2}}}}}}$.
These are clearly valid formulas given the context schema $ctx$ and the strengthening
of Canonical LF judgements, Theorem~\ref{th:strengthening},
but they are also derivable formulas in the proof system by an
induction on the assumption formula.
Case analysis on the assumption will yield only terms constructed from
constants in the signature, and thus nothing from the context can
appear in these judgements.

The formula representing type uniqueness which we prove in Adelfa is the following.
\[
\begin{array}{c}
\fctx{G}{ctx}{\fall{e:\oty}{\fall{t_1:\oty}{\fall{t_2:\oty}{
    \fall{d_1:\oty}{\fall{d_2:\oty}{
    \fimp{\fatm{G}{d_1:\ofty\app e\app t_1}}{}}}}}}}\\
    \fimp{\fatm{G}{d_2:\ofty\app e\app t_2}}{
      \fexists{d_3:\oty}{
        \fatm{\emptyctx}{d_3:\eqty\app t_1\app t_2}}}
\end{array}
\]
The informal argument is based on an induction on the
derivation of $\fatm{G}{d_1:\ofty\app e\app t_1}$, and in Adelfa this is
accomplished using the $\tind$ tactic with argument $1$ to indicate that this
first formula is the one on which we wish to induct.
The resulting state, now containing annotated formulas, is of the following form.
\begin{verbatim}
Vars:
Nominals: 
Contexts:
IH: ctx G:ctx, forall e:o t1:o t2:o d1:o d2:o, {G |- d1 : of e t1}* => 
      {G |- d2 : of e t2} => exists d3:o, {d3 : eq t1 t2}
==================================
ctx G:ctx, forall e:o t1:o t2:o d1:o d2:o, {G |- d1 : of e t1}@ => 
  {G |- d2 : of e t2} => exists d3:o, {d3 : eq t1 t2}
\end{verbatim}
After introducing the assumptions using \tintros\ the proof will proceed 
by determining the possible derivations for instances of
$\eqann{\fatm{G}{d_1:\ofty\app e\app t_1}}$.
We use the \tcase\ tactic on this assumption formula, having
Adelfa analyse the given assumption formula and replace this goal with
a set of new subgoals.
Adelfa will identify the same four cases discussed in the
informal argument: when the head of $d_1$ is $\ofemptytm$, 
$\ofapptm$, $\oflamtm$, or a nominal constant assigned the type 
$(\ofty\app n\app t_1)$ in $G$.
In all of these cases, the form of $d_1$ constrains the form of $e$ which in 
turn will constrain the derivable instances for the other assumption formula
$\fatm{G}{d_2:\ofty\app e\app t_2}$ to those where the head of $d_2$ is the
same as that of $d_1$.
Thus in each of the four cases we identify this constrained structure by
using the \tcase\ tactic also on this second assumption formula to unfold
the definition of this term-level constant.

When the head of $d_1$ is $\ofemptytm$ the proof will conclude using 
an application of the $\texists$ tactic with the obvious $\refltm$ term
followed by $\tsearch$ to complete the derivation.
In the cases where the head of $d_1$ is $\ofapptm$ or $\oflamtm$ we 
invoke the induction hypothesis using the \tapply\ tactic with the appropriate
hypotheses as arguments.
To demonstrate, let us consider the case for typing abstraction terms.
\begin{verbatim}
Vars: a2:o -> o -> o, t2:o, a1:o -> o -> o, r:o -> o, t:o, t1:o
Nominals: n1:o, n:o
Contexts: G:ctx[]
IH:ctx G:ctx, forall e:o t1:o t2:o d1:o d2:o, {G |- d1 : of e t1}* => 
      {G |- d2 : of e t2} => exists d3:o, {d3 : eq t1 t2}
H3:{G, n:tm |- r n : tm}*
H4:{G |- t : ty}*
H5:{G |- t1 : ty}*
H6:{G, n:tm, n1:of n t |- a1 n n1 : of (r n) t1}*
H10:{G |- t : ty}
H11:{G |- t2 : ty}
H12:{G, n:tm, n1:of n t |- a2 n n1 : of (r n) t2}
==================================
exists d3:o, {d3 : eq (arr t t1) (arr t t2)}
\end{verbatim}
Adelfa is able to identify that the extended context expression
$(G,\ n:\tmty,\ n1:\ofty\app n\app t)$ satisfies the context schema 
$ctx$ given that $G$ is a context variable satisfying this same schema.
Thus we are able to apply the formula {\tt IH} to {\tt H6} and {\tt H12}
to conclude that there is some $d3$ such that 
$\fatm{\emptyctx}{d3:\eqty\app t1\app t2}$
is a valid formula.
We know from the definition of the type $\eqty$ that this means $t1=t2$ and $d3$ is 
the term $(\refltm\app t1)$, and using \tcase~ in Adelfa will add this information
resulting in the following state.
\begin{verbatim}
Vars: a2:o -> o -> o, a1:o -> o -> o, r:o -> o, t:o, t1:o
Nominals: n1:o, n:o
Contexts: G:ctx[]
IH:ctx G:ctx, forall e:o t1:o t2:o d1:o d2:o, {G |- d1 : of e t1}* => 
      {G |- d2 : of e t2} => exists d3:o, {d3 : eq t1 t2}
H3:{G, n:tm |- r n : tm}*
H4:{G |- t : ty}*
H5:{G |- t1 : ty}*
H6:{G, n:tm, n1:of n t |- a1 n n1 : of (r n) t1}*
H10:{G |- t : ty}
H11:{G |- t1 : ty}
H12:{G, n:tm, n1:of n t |- a2 n n1 : of (r n) t1}
H13:{refl t1 |- eq t1 t1}
==================================
exists d3:o, {d3 : eq (arr t t1) (arr t t1)}
\end{verbatim}
One may expect that after an application of \texists\ with the term 
$(\refltm\app(\arrtm\app t\app t1))$ this case is completed by an application
of \tsearch.
However, this does not succeed as we still need to use the strengthening
properties to ensure that $(\arrtm\app t\app t1)$ is a good
type expression in an empty context.
The \tapply\ tactic is used for this, to apply the previously proved theorem 
as a lemma.

In the final case, where the head of $d1$ is a nominal constant assigned the type 
$(\ofty\app n\app t1)$ in $G$ the context type for $G$ is extended to include
such a block, and $d2$ must also be this same nominal constant given that there
can be at most one block containing an assignment for the nominal constant $n$.
This case is concluded through an application of the strengthening 
lemma to ensure that the type $t1$ does not rely on anything from $G$.

Once we have a derivation for the formula
\[
\begin{array}{c}
\fctx{G}{ctx}{\fall{e:\oty}{\fall{t_1:\oty}{\fall{t_2:\oty}{
    \fall{d_1:\oty}{\fall{d_2:\oty}{
    \fimp{\fatm{G}{d_1:\ofty\app e\app t_1}}{}}}}}}}\\
    \fimp{\fatm{G}{d_2:\ofty\app e\app t_2}}{
      \fexists{d_3:\oty}{
        \fatm{\emptyctx}{d_3:\eqty\app t_1\app t_2}}}
\end{array}
\]
in our development, the special case
\[
\begin{array}{c}
\fall{e:\oty}{\fall{t_1:\oty}{\fall{t_2:\oty}{
    \fall{d_1:\oty}{\fall{d_2:\oty}{
    \fimp{\fatm{\emptyctx}{d_1:\ofty\app e\app t_1}}{}}}}}}\\
    \fimp{\fatm{\emptyctx}{d_2:\ofty\app e\app t_2}}{
      \fexists{d_3:\oty}{
        \fatm{\emptyctx}{d_3:\eqty\app t_1\app t_2}}}
\end{array}
\]
is derivable directly from an application to the empty context.

%% file: examples/cut.tex
\section{Admissibility of Cut for a Simple Sequent Calculus}
\begin{figure}
\[
\begin{array}{ccc}
\infer[\mbox{init}]
      {\intseq{\Gamma,C}{C}}
      {}
  &
\infer[{\top}R]
      {\intseq{\Gamma}{\top}}
      {}
  &
\infer[{\wedge}R]
      {\intseq{\Gamma}{A\wedge B}}
      {\intseq{\Gamma}{A}
         &
       \intseq{\Gamma}{B}}
  \\\ \\
\infer[{\wedge}L]
      {\intseq{\Gamma,A\wedge B}{C}}
      {\intseq{\Gamma,A\wedge B,A,B}{C}}
  &
\quad
\infer[{\supset}R]
      {\intseq{\Gamma}{A\supset B}}
      {\intseq{\Gamma,A}{B}}
\quad
  &
\infer[{\supset}L]
      {\intseq{\Gamma,A\supset B}{C}}
      {\intseq{\Gamma,A\supset B}{A}
         &
       \intseq{\Gamma,A\supset B,B}{C}}
\end{array}
\]
\caption{A Simple Intuitionistic Sequent Calculus}\label{fig:cut-seq}
\end{figure}

This example deals with the encoding of a very simple sequent calculus
which is defined through the rules given in Figure~\ref{fig:cut-seq}.
A sequent comprises a collection of hypotheses $\Gamma$ and a conclusion 
represented by a proposition $C$.
Our goal with this example is to prove the admissibility of cut for this object system.
This property would is captured by the following: for any $A$, 
if $\intseq{\Gamma}{A}$ and $\intseq{\Gamma,A}{C}$ then $\intseq{\Gamma}{C}$.
The informal structure of this proof, upon which the Adelfa formalization 
will be based, is a standard one; we argue that
the application of cut with $A$ can be permuted up the proof of $C$ 
and uses of the cut formula $A$ are eliminated in the derivation of $C$ essentially by 
replacing uses of the hypothesis with the proof for $A$.
This argument is structured as a nested induction on the structure of the cut 
formula $A$ and the structure of the derivation for $C$ using the cut formula.
The proof then proceeds by examining the structure of the derivation of $C$ using the
cut formula.
For cases which do not make use of the cut formula $A$ in the final
rule application of the derivation we simply apply the inductive hypothesis to each 
branch of the derivation and use the same rule application to
construct a cut-free proof using the cut-free proofs for each branch.
When the cut formula is relevant to the final step of the derivation, which is
those cases where the derivation concludes by using one of the left rules on the
cut formula $A$, we make use of the induction on the cut formula to
replace the use of the hypothesis with an argument structured around the formation
of $A$ which essentially re-create the proof of $A$ in place to produce a cut-free proof.
In this proof, management of the hypotheses $\Gamma$ will play an important role.
In particular, as we descend through the structure of the proof for $C$ the
hypotheses may be extended and the proof for $A$ must correspondingly be extended
for the statement to apply inductively.
Thus this proof relies on weakening in the sequent calculus
which must be realized also in our formalization of the proof.

\begin{figure}
\[
\begin{array}{l}
\propty:\type\\
\toptm:\propty\\
\imptm:\arr{\propty}{\arr{\propty}{\propty}}\\
\andtm:\arr{\propty}{\arr{\propty}{\propty}}\\\ \\
\hypty:\arr{\propty}{\type}\\
\concty:\arr{\propty}{\type}\\\ \\
\inittm:\typedpi{A}{\propty}{\typedpi{D}{\hypty\app A}{\concty\app A}}\\
\toprtm:\concty\app\toptm\\
\andltm:\typedpi{A}{\propty}{\typedpi{B}{\propty}{\typedpi{C}{\propty}{\typedpi{D1}{(\typedpi{x}{\hypty\app A}{\typedpi{y}{\hypty\app B}{\concty\app C}})}{}}}}\\
\qquad\quad \typedpi{D2}{\hypty\app (\andtm\app A\app B)}{\concty\app C}\\
\andrtm:\typedpi{A}{\propty}{\typedpi{B}{\propty}}{\typedpi{D_1}{\concty\app A}{\typedpi{D_2}{\concty\app B}{\concty\app (\andtm\app A\app B)}}}\\
\impltm:\typedpi{A}{\propty}{\typedpi{B}{\propty}{\typedpi{C}{\propty}{\typedpi{D_1}{\concty\app A}{\typedpi{D_2}{(\typedpi{x}{\hypty\app B}{\concty\app C})}{}}}}}\\
\qquad\quad\typedpi{D_3}{\hypty\app (\imptm\app A\app B)}{\concty\app C}\\
\imprtm:\typedpi{A}{\propty}{\typedpi{B}{\propty}{\typedpi{D}{(\typedpi{x}{\hypty\app A}{\concty\app B})}{\concty\app (\imptm\app A\app B)}}}
\end{array}
\]
\caption{An LF Specification for Derivations}
\label{fig:ex-cut-spec}
\end{figure}

The LF specification for the sequent calculus is given in Figure~\ref{fig:ex-cut-spec}.
Key to this encoding is the way in which sequents are represented.
In the encoding we make use of two LF types, $\hypty$ and
$\concty$, to identify propositions which are hypotheses and conclusions 
of a sequent respectively,
thus the sequent $\intseq{A_1,\ldots A_n}{B}$ will be represented by a judgement
$\lfchecktype{x_1:\hypty\app A_1,\ldots,x_n:\hypty\app A_n}{c}{\concty\app B}$.
The contexts we are interested in for this example will thus contain 
bindings only for instances of the $\hypty$ type.
We define a context schema for such LF contexts of the form 
$\{a:\oty\}x:(\hypty\app a)$ and associate this schema with the identifier $c$.
Using this schema definition, cut admissibility can be stated as the following formula.
\begin{gather*}
\fctx{G}{c}{\fall{a:\oty}{\fall{b:\oty}{\fall{d_1:\oty}{\fall{d_2:\arr{\oty}{\oty}}
     {\\
      \fimp{\fatm{\emptyctx}{a:\propty}}
     {\fimp{\fatm{G}{d_1:\concty\app a}}
     {\fimp{\fatm{G,n:\hypty\app a}{d_2\app n:\concty\app b}} 
           {\\
            \fexists{d:\oty}
                    {\fatm{G}{d:\concty\app b}}}}}}}}}}.
\end{gather*}
Note that in the third antecedent the term variable $d_2$ is permitted to depend on
the cut formula $a$ through the explicit dependency on the nominal constant $n$ bound in
the context.

The proof for this formula 
uses a nested induction first on the structure of $a$,
through the first assumption formula $\fatm{\emptyctx}{a:\propty}$, 
and second on the height of the derivation for $b$ which may use $a$, 
through the second assumption formula
$\fatm{G,n:\hypty\app a}{d_2\app n:\concty\app b}$.
One thing to note about this proof is that given the meaning of LF judgements, the context expression
$(G, n:\hypty\app a)$ in this formula requires that the assumption of the 
cut formula $a$ is always at the end of the context for this property to apply.
A property of the object system is that any hypothesis from $\Gamma$ can be 
identified for use in the application of the rules in Figure~\ref{fig:cut-seq}.
The encoding for sequents we have defined is such that this property of
the object system is realized by the permutation meta-theorem for LF; given the
signature defined in Figure~\ref{fig:ex-cut-spec} we can determine that there can be 
no valid dependencies between types in contexts
satisfying the schema $c$.
Therefore any permutation will be valid, and in particular
the permutation which moves the assumption for $a$ to the end will be valid.
Our encoding of the object system also permits weakening in the object system to be 
realized by weakening in LF.
We demonstrate how these meta-theorems are applied in reasoning to manage
the structure of contexts by looking at the formalization of the cut elimination 
proof in Adelfa.
In the discussion below we look in detail at a single case of the argument to demonstrate
how these meta-theorems can be used to realize the informal proof structure.
The full details of this proof can be found on the Adelfa website, and with an 
understanding of the informal structure it is not difficult to see how the rest of
the argument will be formalized.

The case we consider is when the derivation making use of the cut formula
has a conclusion which is an implication, and this derivation
concludes with a use of ${\supset}R$; this of course
corresponds to the case where $(d_2\app n)$ is an instance of $\imprtm$ in
Adelfa.
This is a case where the cut formula is not used in the concluding rule of the derivation, 
so following the informal argument structure we use the inductive hypothesis to
obtain cut-free proofs for the premises and then use these to
construct a cut-free proof for the conclusion using 
an application of $\imprtm$.
The state of Adelfa at the start of this case has the following form.
\pagebreak 
\begin{verbatim}
Vars: d:o -> o -> o, b1:o, b2:o, d1:o, a:o
Nominals: n1:o, n:o
Contexts: G:c[]
IH:ctx G:c. forall a, forall b, forall d1, forall d2,
        {a : prop}* => {G |- d1 : conc a} => 
            {G, n:hyp a |- d2 n : conc b} => exists d, {G |- d : conc b}
IH1:ctx G:c. forall a, forall b, forall d1, forall d2,
        {a : prop}@ => {G |- d1 : conc a} =>
            {G, n:hyp a |- d2 n : conc b}** => exists d, {G |- d : conc b}
H1:{a : prop}@
H2:{G |- d1 : conc a}
H3:{G, n:hyp a |- impR b1 b2 ([x] d n x) : conc (imp b1 b2)}@@
H4:{G, n:hyp a |- b1 : prop}**
H5:{G, n:hyp a |- b2 : prop}**
H6:{G, n:hyp a, n1:hyp b1 |- d n n1 : conc b2}**
==================================
exists d, {G |- d : conc (imp b1 b2)}
\end{verbatim}
Intuitively, we would like at this point to apply the inductive hypothesis {\tt IH1}
using {\tt H1}, {\tt H2}, and {\tt H6} to extend the assumptions with a new
cut-free derivation $d'$ of the type $(\concty\app b_2)$.
However, to successfully apply the inductive hypothesis the 
assignment $n:(\hypty\app a)$ must be moved to the end of the context expression.
Further, the context expression in {\tt H2} will need to be extended with
an assignment to the type $(\hypty\app b_1)$ to ensure that the context 
of this judgement matches the extended context in {\tt H6}.
To permute the cut formula to the end of the context expression, we
will use the \tpermute\ tactic on {\tt H6} with the
permuted context expression $(G, n_1:\hypty\app b_1, n:\hypty\app a)$.
This application will succeed since there can be
no dependency on $n$ in the type $(\hypty\app b_1)$, which aligns with 
the interpretation of hypotheses in the object system.
Adelfa uses the structure of the hypothesis {\tt H6} and the
permuted form of the context, $(G, n_1:\hypty\app b_1, n:\hypty\app a)$, 
supplied by the user to identify a particular
instance of the formula encoding the LF context permutation theorem.
In this particular case that formula would be
\begin{tabbing}
\qquad\qquad\=\qquad\qquad\qquad\qquad\qquad\=\kill
\>$\fimp{\ltann[2]{\fatm{G, n:\hypty\app a, n_1:\hypty\app b_1}
                        {d\app n\app n_1 : \concty\app b_2}}}
        {}$
\\
\>\>    $\ltann[2]{\fatm{G, n_1:\hypty\app b_1, n:\hypty\app a}
                        {d\app n\app n_1 : \concty\app b_2}}$
\end{tabbing}
which, as we have already noted, will be valid given that this permutation
will not violate any dependencies.
The effect of introducing such a formula using \cut\ and then applying it
to the hypothesis {\tt H6} is the behaviour exhibited in Adelfa, and the
resulting state includes a new assumption formula
\begin{center}
{\tt H7:\{G, n1:hyp b1, n:hyp a |- d n n1 : conc b2\}**}.
\end{center}
The application of the \tpermute\ tactic will fail if Adelfa is not able to
determine that the order of bindings in the permuted context expression does not
violate any dependencies, and so this tactic will only succeed when the use of this
proof rule is valid.

We now want to realize the weakening step of the proof to extend the hypotheses
of the proof for $a$ with $b_1$.
Since weakening in the object system is realized by weakening in LF, we achieve
this through an application of the \tweak\ tactic to the 
assumption formula {\tt H2} and the type $(\hypty\app b_1)$.
In this case the instance of the weakening formula identified by Adelfa will be
\begin{center}
$\fimp{\fatm{G}{d_1 : conc\app a}}{\fatm{G, n_2:hyp\app b_1}{d_1 : conc\app a}}$.
\end{center}
However, the application of this tactic to {\tt H2} will fail
because Adelfa is not able to ensure that all the premise sequents
which would be generated by the use of \lfwk\ in proving this formula can
be derived from the current state.
Specifically, Adelfa is unable to automatically derive the subgoal
$\fatm{G}{b_1:\propty}$.
We do, however, have as assumption {\tt H4} the formula 
$\ltann[2]{\fatm{G, n:\hypty\app a}{b_1 : prop}}$, and the formation of
this $b_1$ cannot depend on anything of type $\hypty$.
So in particular, $b_1$ cannot depend on $n$, meaning this binding is vacuous.
An application of the strengthening meta-theorem using the \tstr\ tactic
will result in an assumption of the form necessary for the weakening step
to succeed.
Similar to the behaviour with context permutation and weakening, 
Adelfa introduces an instance of \lfstr\ based on the hypothesis being strengthened
and extends the assumption set with the result of applying this instance to the
identified formula.
In this case, the result would be to extend the state with a formula
{\tt H8:\{G |- b1 : prop\}**}.
At this point weakening can be successfully applied to {\tt H2}, producing the
following state.
\begin{verbatim}
Vars: d:o -> o -> o, b1:o, b2:o, d1:o, a:o
Nominals: n2:o, n1:o, n:o
Contexts: G:c[]
IH:ctx G:c. forall a, forall b, forall d1, forall d2,
        {a : prop}* => {G |- d1 : conc a} => 
            {G, n:hyp a |- d2 n : conc b} => exists d, {G |- d : conc b}
IH1:ctx G:c. forall a, forall b, forall d1, forall d2,
        {a : prop}@ => {G |- d1 : conc a} =>
            {G, n:hyp a |- d2 n : conc b}** => exists d, {G |- d : conc b}
H1:{a : prop}@
H2:{G |- d1 : conc a}
H3:{G, n:hyp a |- impR b1 b2 ([x] d n x) : conc (imp b1 b2)}@@
H4:{G, n:hyp a |- b1 : prop}**
H5:{G, n:hyp a |- b2 : prop}**
H6:{G, n:hyp a, n1:hyp b1 |- d n n1 : conc b2}**
H7:{G, n1:hyp b1, n:hyp a |- d n n1 : conc b2}**
H8:{G |- b1 : prop}**
H9:{G, n2:hyp b1 |- d1 : conc a}
==================================
exists d, {G |- d : conc (imp b1 b2)}
\end{verbatim}

At this stage the application of the inner inductive hypothesis {\tt IH1}
to {\tt H1}, {\tt H9}, and {\tt H7} is possible by instantiating 
{\tt G} with {\tt(G,n1:hyp b1)}, {\tt a} with {\tt a}, {\tt b} with {\tt b2},
{\tt d1} with {\tt d1}, and {\tt d2} with {\tt([x] d x n1)}.
This application results in a new eigenvariable {\tt d'} being introduced
along with a new assumption
{\tt H10:\{G, n1:hyp b1 |- d' n1 : conc b2\}}.
We now construct an inhabitant of the type $(\concty\app(\imptm\app b_1\app b_2))$
using \texists\ to instantiate the existential quantifier in the goal formula
with the term $(\imprtm\app b_1\app b_2\app (\lflam{x}{d'\app x}))$.
The resulting state is the following.
\begin{verbatim}
Vars: d': o -> o, d:o -> o -> o, b1:o, b2:o, d1:o, a:o
Nominals: n2:o, n1:o, n:o
Contexts: G:c[]
IH:ctx G:c. forall a, forall b, forall d1, forall d2,
        {a : prop}* => {G |- d1 : conc a} => 
            {G, n:hyp a |- d2 n : conc b} => exists d, {G |- d : conc b}
IH1:ctx G:c. forall a, forall b, forall d1, forall d2,
        {a : prop}@ => {G |- d1 : conc a} =>
            {G, n:hyp a |- d2 n : conc b}** => exists d, {G |- d : conc b}
H1:{a : prop}@
H2:{G |- d1 : conc a}
H3:{G, n:hyp a |- impR b1 b2 ([x] d n x) : conc (imp b1 b2)}@@
H4:{G, n:hyp a |- b1 : prop}**
H5:{G, n:hyp a |- b2 : prop}**
H6:{G, n:hyp a, n1:hyp b1 |- d n n1 : conc b2}**
H7:{G, n1:hyp b1 |- d n n1 : conc b2}**
H8:{G |- b1 : prop}**
H9:{G, n2:hyp b1 |- d1 : conc a}
H10:{G, n1:hyp b1 |- d' n1 : conc b2}
==================================
{G |- impr b1 b2 ([x] d' x) : conc (imp b1 b2)}
\end{verbatim}
There still remains one step before \tsearch\ can successfully
construct a derivation for this goal.
Clearly \appR\ would be used to derive this subgoal, and would require derivations
then of $\fatm{G}{b_1:\propty}$, $\fatm{G}{b_2:\propty}$, and 
$\fatm{G}{\lflam{x}{d'\app x}:\typedpi{x}{\hypty\app b_1}{\concty\app b_2}}$.
The first and last of these can be ensured in the current state through
applications of \id\ with {\tt H8} and {\tt H10} respectively, but a derivation
for the second cannot be constructed automatically.
However, similar to the earlier application of strengthening, we can apply
\tstr\ to {\tt H5} which will result in extending the assumptions with a formula
of the needed form, and this case of the cut admissibility proof can be completed
using \tsearch.

%% file: examples/fsub.tex
\section{The POPLmark Challenge}
Adelfa is also able to handle larger developments, such as 
Problem 1A of the POPLmark Challenge~\cite{aydemir05tphols}. 
Specifically, for this problem we show the transitivity of System $F_{<:}$.
There are many existing solutions to the challenge problem, which makes it
a good candidate for comparisons across systems.
These also exists a solution to this problem formalized in Twelf~\cite{cmusolution},
and the formalization in Adelfa is based on this 
solution.\footnote{
A Twelf formalization has been presented in~\cite{pientka07poplmark} and it has been claimed 
that it also solves this challenge problem. 
However, this claim is false: the formulation of the typing calculus differs from the 
one presented in~\cite{aydemir05tphols} and that formulation in fact assumes away the
essential aspects of the challenge~\cite{nadathur2019}.
}
This example demonstrates more sophisticated use of contexts in reasoning than
we have seen thus far, and the proof requires a more complex inductive structure.
\begin{figure}
\begin{tabular}{rcl}
$\tyty$ & $:$ & $\type.$\\
$\toptytm$ & $:$ & $\tyty.$\\
$\arrowtm$ & $:$ & $\arr{\tyty}{\arr{\tyty}{\tyty}}$\\
$\alltm$ & $:$ & $\arr{\tyty}{\arr{(\arr{\tyty}{\tyty})}{\tyty}}$
\medskip\\
$\boundty$ & $:$ & $\arr{\tyty}{\arr{\tyty}{\type}}$
\medskip\\
$\subty$ & $:$ & $\arr{\tyty}{\arr{\tyty}{\type}}$\\
$\stoptm$ & $:$ & $\typedpi{S}{\tyty}{\subty\app S\app \toptytm}$\\
$\srefltm$ & $:$ & $\typedpi{X}{\tyty}
             {\typedpi{U}{\tyty}
             {\typedpi{A}{\boundty\app X\app U}
                         {\subty\app X\app X}}}$\\
$\stranstm$ & $:$ & $\typedpi{X}{\tyty}
              {\typedpi{U1}{\tyty}
              {\typedpi{U2}{\tyty}
              {\typedpi{A1}{\boundty\app X\app U1}
              {\typedpi{A2}{\subty\app U1\app U2}{}}}}}$\\
& & \qquad
    $\subty\app X\app U2$\\
$\sarrowtm$ & $:$ & 
    $\typedpi{S_1}{\tyty}
    {\typedpi{S_2}{\tyty}
    {\typedpi{T_1}{\tyty}
    {\typedpi{T_2}{\tyty}
    {\typedpi{a_1}{\subty\app T_1\ app S_1}
    {\typedpi{a_2}{\subty\app S_2\app T_2}{}}}}}}$\\
& &\qquad
    $\subty\app(\arrowtm\app S_1\app S_2)\app(\arrowtm\app T_1\app T_2)$\\
$\salltm$ & $:$ &
    $\typedpi{S_1}{\tyty}
    {\typedpi{S_2}{\arr{\tyty}{\tyty}}
    {\typedpi{T_1}{\tyty}
    {\typedpi{T_2}{\arr{\tyty}{\tyty}}
    {\typedpi{a_1}{\subty\app T_1\app S_1}{}}}}}$\\
& &\qquad
    $\typedpi{a_2}{(\typedpi{w}{\tyty}{\typedpi{y}{\boundty\app w\app T_1}{\subty\app(S_2\app w)\app(T_2\app w))})}}{}$\\
& &\qquad
    $\subty\app(\alltm\app S_1\app(\lflam{x}{S_2\app x}))\app(\alltm\app T_1\app(\lflam{x}{T_2\app x}))$
\end{tabular}
\caption{LF Specification for System $F_{<:}$}
\label{fig:ex-pm-spec}
\end{figure}
Figure~\ref{fig:ex-pm-spec} contains the specification of System $F_{<:}$
and the subtyping relation.
The type $\tyty$ is used to encode the types of the system while $\subty$
encodes the subtyping relation between two types.
Since subtyping assumptions are necessary in the encoding of these types,
we further introduce the type $\boundty$ for this purpose.

The goal in this example is to show that if under a context $\Gamma$, 
$S$ is a subtype of $Q$ and $Q$ is a subtype of $T$, 
then under this same context $\Gamma$, $S$ is a subtype of $T$.
In this proof we will need to make use of narrowing, which states that
if $P$ is a subtype of $Q$ and under a context containing $\subty\app x\app Q$ the
type $N$ is a subtype of $M$, then $N$ is also a subtype of $M$ under the same context
except with assumption $\subty\app x\app P$ replacing $\subty\app x\app Q$.
These are proved simultaneously by mutual induction on the structure of $Q$,
proving first transitivity using narrowing for types smaller than $Q$ and then
proving narrowing using transitivity for the type $Q$.
We have not yet specified the particular form of the contexts in this reasoning,
which has some complexity due to the use of LF for the encoding.
Before we can state precisely the Adelfa theorem we must define a schema
for the contexts we wish to reason about.

At the outset, it is clear that in reasoning about subtyping the contexts
must contain pairs of variable and subtyping assumptions of the form 
$(x:\tmty,y:\subty\app x\app T)$.
However, for narrowing we need to be able to identify a variable from an arbitrary
location within the context and move it to the end of the context so that narrowing
can be applied.
This is a complication because in LF the ordering of bindings within a context
must respect dependencies, and
the type $T$ in a subtyping assignment 
$(\boundty\app x\app T)$ for a variable $x$ might well depend on some
variable declared earlier in the context and such variables could not be
declared within this context at arbitrary locations.
In the formalization of the theorem in the Twelf system~\cite{cmusolution},
this issue is addressed by carefully constructing the contexts
such that variables and their subtyping assumptions can be separated within a
context without becoming dissociated.
We follow this approach in Adelfa and introduce two new types for the purpose,
$\varty:\arr{\tyty}{\type}$ and
$\bvarty:\typedpi{X}{\tyty}{\typedpi{T}{\tyty}{\arr{\boundty\app X\app T}{\arr{\varty\app X}{\type}}}}$ 
to identify the variables independently of their subtyping assumption 
and to ensure that any variable introduced into the context has also 
a single associated subtyping assumption.
We then define the context schema as one which comprises three block
definitions, one which keeps the variable together with its
subtyping assumption, and one for each of the split collection of assumptions.
These block definitions are detailed below.
\begin{center}
$\begin{array}{l}
\{T:\oty\}w:\tyty, x:\varty\app w, y:\boundty\app x\app T, z:\bvarty\app w\app x\app y\\
x:\tyty, y:\varty\app w\\
\{V:\oty,T:\oty,DV:\oty\}x:\boundty\app V\app T,y:\bvarty\app V\app T\app x\app DV
\end{array}$
\end{center}
Let this context schema be given the identifier $c$.

The theorem we prove in Adelfa is the following.
\begin{tabbing}
\quad\=\qquad\=\qquad\qquad\=\qquad\=\kill
\>$\fctx{G}{c}{\fall{q:\oty}
  {\fimp{\fatm{G}{q:\tyty}}{}}}$\\
\>\> $(\fctx{L}{c}
        {\fall{s:\oty}{\fall{t:\oty}{\fall{d_1:\oty}{\fall{d_2:\oty}}}}}$\\
\>\>\>  $\fimp{\fatm{L}{d_1:\subty\app s\app q}}
              {\fimp{\fatm{L}{d_2:\subty\app q\app t}}
              {\fexists{d_3}{\fatm{L}{d_3:\subty\app s\app t}}}})$\\
\>\>$\wedge$\\
\>\> $(\fctx{L}{c}
        {\fall{x:\oty}{\fall{t_1:\oty}{\fall{t_2:\oty}
        {\fall{p:\oty}{\fall{d_1:\oty}
        {\fall{d_2:\arr{\oty}{\oty}}{\fall{dv:\oty}{}}}}}}}}$\\
\>\>\>  $\fimp{\fatm{L}{d_1:\subty\app p\app q}}{}$\\
\>\>\>       $\fimp{\fatm{L,n_1:\boundty\app x\app q,
                             n_2:\bvarty\app x\app q\app n_1\app dv}
                          {d_2\app n_1:\subty\app t_1\app t_2}}{}$\\
\>\>\>       $\fexists{d_4:\arr{\oty}{\oty}}{}$\\
\>\>\>\>              $\fatm{L,n_1:\boundty\app x\app p,
                                n_2:\bvarty\app x\app p\app n_1\app dv}
                             {d_4\app n_1:\subty\app t_1\app t_2})$.
\end{tabbing}
The outer induction of the proof is on $\fatm{G}{q:\tyty}$.
Within this proof, 
transitivity is shown by induction on $\fatm{L}{d_1:\subty\app s\app q}$
and narrowing is shown by induction on the formula
$\fatm{L,n_1:\boundty\app x\app q, n_2:\bvarty\app x\app q\app n_1\app dv}
      {d_2\app n_1:\subty\app t_1\app t_2}$.
As in the admissibility of cut, the management of the context ordering is 
realized in the Adelfa development
through the application of the tactics $\tweak$, $\tstr$, and $\tpermute$.
With this understanding of the encoding and the structure of the proof, it is 
easy to visualize how the formal proof should be constructed. 
The detailed development that makes all the steps explicit is available from the 
Adelfa web site and so we will not present the formalization in detail here.

%% file: comparisons/related-work.tex
\chapter{Comparisions with Related Work}
\label{ch:comp}

The usefulness of LF as a specification language can be seen
through the wide variety of tasks in which it has been used since its 
introduction, for example~\cite{bertot04book, lee07popl, leroy09cacm}.
Reasoning about such LF specifications has thus also been of interest and
there exist various approaches which attempt to provide a means for
effectively reasoning about systems through their encoding in LF.
In this chapter we consider existing systems and approaches to reasoning about
specifications written in LF and contrast them with the logic and proof system 
developed in this thesis.

\section{Twelf}
The Twelf system is a well established system which can be used to reason about
LF specifications~\cite{pfenning99cade,Pfenning02guide}.
It has a wide variety of existing examples, many available from the
Twelf website.
In this system, no distinction is made between the specification 
and the reasoning level.
The approach taken by Twelf is a computational view of reasoning
where properties of the object system are also encoded as LF types, 
with constructors for these types essentially
being proof scripts.
An external analysis is used to check properties of this encoding and from
this results about the specification are extracted.
In particular, typically the LF type encoding a property of interest is
shown to be total through the use of coverage and termination checking.

Due to this view of reasoning, the sort of statements which can
be checked in Twelf are all of functional structure; some collection of derivations are given
as input and some other collection of derivations are constructed as output.
Thus, only formulas of a $\forall\exists$ form are checkable in Twelf
which limits the expressiveness of reasoning in this way.
Reasoning in Twelf also does not construct proofs explicitly,
instead the validity of these formulas are extracted from the specification
along with a positive result from the totality checking procedure.

Another aspect to consider in this comparison is the treatment of contexts in
Twelf reasoning.
Twelf permits definitions similar to context schemas, though they are
called regular worlds in this setting, but the contexts are kept implicit
in the reasoning process.
Thus there is a single implicit context across an instance of the analysis, which limits the
flexibility of the system.
Properties which require typing derivation in different contexts cannot be expressed
in Twelf.
Keeping the contexts implicit also obscures the true structure of 
the analysis performed by the system.

\section{The Logic $\mathcal{M}_2^+$}
The logics $\mathcal{M}_2$~\cite{schurmann98cade} and 
$\mathcal{M}_2^+$~\cite{schurmann00phd} were developed to formalize reasoning
as performed in Twelf.
The logic provides a logical foundation to the sort of reasoning
performed by Twelf which can address some of the shortcomings of using
an external process like totality checking for reasoning.
However, because they are focused specifically on capturing reasoning as
it is viewed in Twelf, the logics $\mathcal{M}_2$ and $\mathcal{M}_2^+$ 
have a different flavor than the logic presented in this thesis.
These logics are based on an understanding of Twelf rather than directly
on understanding derivability and reasoning in LF, and does not illuminate the
structure of reasoning steps as we aim to do.
This view of reasoning leads to a more complex form of induction which
is realized through proof terms corresponding to recursive total functions.
Thus, while the logic $\mathcal{M}_2^+$ provides a formalization for
Twelf reasoning it is still rooted in the computational approach to reasoning
used by Twelf of extracting proofs through checking
totality of functions of dependent types.

Due to the approach of its design, the logic $\mathcal{M}_2^+$ also has 
the same limitations in expressivity as we have seen in Twelf.
Proofs in this logic are functions, and the logic ensures that only
total functions can be constructed.
Properties which are not expressed as function types cannot then be reasoned about
with this logic.
The logic also maintains a single context across a formula, thus again restricting
the ability to express relations over derivations in LF which use different contexts.

\section{Beluga}
Another well established system for reasoning about LF is
Beluga~\cite{pientka10ijcar}.
Beluga uses a two level approach, providing a dependently typed functional
programming language on top of LF.
The reasoning level of Beluga is based on contextual modal type
theory~\cite{nanevski08tocl} which allows an explicit treatment of
contexts in typing. 
Beluga is, like Twelf, based on a computational view of reasoning.
In Beluga one writes dependently-typed recursive functions and the
type system ensures that the admitted programs are ones which are
total and thus encode sound proofs.
These functions are not proofs in themselves, as the type checking is
critical to reasoning and only together do they constitute a proof.
A distinguishing feature of Beluga is that it is intended as a programming 
language and so provides the ability to execute functions written in this way.

The use of contextual modal type theory as the reasoning level
permits a more expressive syntax for formulas than is available in Twelf,
in particular the role of contexts is much more expressive in Beluga.
However, the properties which can be encoded as types in the system
remain restricted to the same functional structure relating some
set of input derivations to a set of output derivations.
Reasoning in Beluga also relies on some external analysis, in this case
type checking, and so proofs are not constructed directly as can be done
for the logic we have defined.

\section{Harpoon}
Harpoon~\cite{errington2020lfmtp} is an interactive tool which assists
in the construction of Beluga functions using a collection of tactics
for generating proof scripts which then generate valid Beluga functions.
Programming in Beluga in this way provides the feel of theorem proving
and many of the tactics have a similar feel to those which we
include in Adelfa from Chapter~\ref{ch:adelfa}.
Harpoon however, determines the structure and collection of these tactics from an 
understanding of Beluga programming rather than from inspection of how 
reasoning proceeds in LF
These derivations are also not proofs directly, as they remain dependent on
type checking of Beluga to constitute a proof.
The expressivity of Harpoon reasoning also
remains limited by the underlying mapping
into Beluga functions.

\section{Cocon}
There has recently been some interesting work done in developing the 
Cocon Type Theory~\cite{pientka2019,pientka2020}.
Cocon is a powerful type theory which
unifies LF methodology with dependent type theories such as Martin-Löf
Type Theory.
A system based on this type theory would be quite powerful for reasoning
about dependently-typed specifications, but to our knowledge such an 
implementation has not yet been explored.
It would be interesting to consider a system based on
Cocon because of the ability to mix representations and computations
in the type system, but this is a consideration which is outside
the scope of current work.

\section{Abella-LF}
A different view of reasoning about LF derivability 
is taken in~\cite{southern2014}, which utilizes a translation-based approach.
Reasoning is realized through translating LF judgements into 
relations in a predicate logic and reasoning is performed over the
translated form.
The result of derivations constructed over the translation are also 
lifted to LF using the translation.
This approach has been implemented as Abella-LF, a variant of the Abella
Prover~\cite{gacek08ijcar, abella.website}.

Abella uses a two-level approach to reasoning where derivability is 
encoded as a definition in the reasoning logic~\cite{gacek09phd}.
Abella-LF builds in a translation and decoding of LF terms 
using syntactic sugar on Abella terms to provide users with an illusion of working 
directly with LF.
The Abella system is very expressive, and in fact may be too expressive
for reasoning about LF specifications as it is unclear
what meaning reasoning about the translation has in LF.
The tactics available in Abella are not based on conceptual steps of
LF reasoning and so developments in Abella-LF do not always correspond with
reasoning steps of LF.
Because of this, it is possible for the translation to be exposed 
during reasoning and the illusion of working in an LF setting to be lost.
This is especially apparent in the treatment of contexts; the behaviour of
Abella-LF does not align with intuition obtained from understanding LF
derivability.

The translation underlying Abella-LF is also not correct
in the form it is used; the derivability of a translated LF judgement does
not ensure the derivability of the original LF judgement unless restricted
to the Canonical LF system.
This issue will have to be addressed for it to be possible to make
claims about the lifting of any results to LF.
Our work provides an understanding of derivability in LF which could
be used as a foundation for an implementation of Abella-LF which
addresses these issues.
Specifically, having identified proof rules capturing specific and well-understood
reasoning steps we could use this collection to constrain reasoning about
the translation in a way which will be meaningful in the LF setting.

%% file: conclusion/conclusion.tex
\chapter{Conclusion}
\label{ch:future-work}

In this thesis we have considered the construction of a logic for
reasoning about LF specifications.
The logic we have developed is based on a semantical approach
using a substitution based interpretation of quantification and
relying on checking the derivability of LF judgements represented
by the atomic formulas.
This logic is also given a corresponding proof system which formalizes
arguments based on the semantics.
This proof system builds in an understanding of LF typing judgements,
incorporates a means for arguing inductively based on the heights of
derivations for such judgements and encodes meta-theoretic
observations about LF derivability  through axioms that reflect the
contents of these observations.   
We have also mechanized the construction of proofs in the proof system
in the proof assistant Adelfa.
The usefulness of this implementation, and by extension the proof system, has 
been shown through a collection of examples of reasoning about
a variety of object systems.

The logic and the proof system that we have developed for it display a
fair amount of flexibility.
For example, it is possible to express in the logic all the properties
that are expressible in a system like Twelf~\cite{pfenning99cade} or Beluga~\cite{pientka10ijcar}.
Further, it is possible to construct proofs for such properties using
the proof system.
Finally, going beyond these systems, the logic allows for the
expression of properties that have a disjunctive character or that
embody an alternation of term-level universal and existential
quantifiers and formal proofs for such properties can also be
constructed.
There are, nevertheless, particular aspects of the logic and the proof
system that we would like to enhance or improve so as to provide
greater flexibility in reasoning.
We discuss some of these aspects below as avenues for further
research. 

\section{Encoding Properties of Contexts in the Logic}
The first extension we consider is encoding properties of 
contexts in the logic.
Such properties would permit reasoning about relations between one or more contexts,
such as subsumption of context schemas.
These relations would extend the reasoning capabilities of the logic through the
ability to make use of properties of the contexts appearing in formulas.

One property of this sort might be context schema subsumption, allowing 
context expressions known to satisfy one context schema to also be identified as
satisfying any other schema which subsumes it.
Such a property could be used to lift a context expression satisfying the more
restrictive schema to the more expressive schema in order to apply, for example,
lemmas proved about contexts of the more general form.

These relations may also take other forms, for example they might identify contexts
which arise from derivations for distinct types in LF pertaining to different relations
involving a particular term.
For example, suppose that we define a reduction relation between STLC terms which
holds for terms $t_1$ and $t_2$ if $t_2$ is obtained from $t_1$ by reducing a
$\beta$-redex somewhere in the term.
This might be encoded by the following LF declarations.
\[
\begin{array}{lcl}
\redty & : & \tmty\rightarrow \tmty\rightarrow\type.\\
\redbetatm & : & \typedpi{T}{\tpty}{\typedpi{R}{\arr{\tmty}{\tmty}}{\typedpi{M}{\tmty}
                         {\typedpi{M'}{\tmty}{\arr{\redty\app (R\app M)\app M'}{}}}}}\\
& & \qquad
    \redty\app(\apptm\app(\lamtm\app T\app(\lflam{x}{R\app x}))\app M)\app M'\\
\redlamtm & : & \typedpi{T}{\tpty}{\typedpi{R}{\arr{\tmty}{\tmty}}
                        {\typedpi{R'}{\arr{\tmty}{\tmty}}
                        {\arr{\typedpi{x}{\tmty}{\redty\app (R\app x)\app(R'\app x)}}{}}}}\\
& & \qquad
    \redty\app(\lamtm\app T\app (\lflam{x}{R\app x}))\app
                           (\lamtm\app T\app (\lflam{x}{R'\app x}))\\
\redappltm & : & \typedpi{M}{\tmty}{\typedpi{N}{\tmty}{\typedpi{M'}{\tmty}
                         {\arr{\redty\app M\app M'}{}}}}\\
& & \qquad
    \redty\app(\apptm\app M\app N)\app(\apptm\app M'\app N)\\
\redapprtm & : & \typedpi{M}{\tmty}{\typedpi{N}{\tmty}{\typedpi{N'}{\tmty}
                         {\arr{\redty\app N\app N'}{}}}}\\
& & \qquad
    \redty\app(\apptm\app M\app N)\app(\apptm\app M\app N')
\end{array}
\]
Contexts relevant to derivations of $\redty$ will have the structure given by a
single block schema $(z:\tmty)$, as we can see from the type of $\redlamtm$.
We can see that for a particular term $t$, we might define a relation which holds
between a context expression satisfying the schema for typing in the STLC
we have seen previously, $\{T:\oty\}(x:\tmty,y:\ofty\app x\app T)$, used in typing $t$
in the STLC, and a context expression satisfying this reduction schema for reducing $t$, 
which includes only the instances of $x:\tmty$ from the typing context.
Thus we could express formulas in the logic which relate derivations using
different context expressions in a way which permits the structure of some contexts
to be informed by others.
For example, subject reduction might be stated as the following theorem where 
$\mbox{ctx-rel}$ represents the relation described above and $\mbox{ty-ctx}$ and
$\mbox{red-ctx}$ denote the STLC typing and reduction context schemas respectively.
\begin{tabbing}
\hspace{1cm}\=\qquad\=\qquad\qquad\=\kill
\>$\fctx{G_1}{\mbox{ty-ctx}}{\fctx{G_2}{\mbox{red-ctx}}{
 \fall{M:\oty}{\fall{N:\oty}{\fall{T:\oty}{\fall{D_1:\oty}{\fall{D_2:\oty}{}}}}}}}$\\
\>\>$\fimp{\mbox{ctx-rel}\app G_1\app G_2}{
 \fimp{\fatm{G_1}{D_1:\ofty\app M\app T}}{
 \fimp{\fatm{G_2}{D_2:\redty\app M\app N}}{}}}$\\
\>\>\>$\fexists{D_3:\oty}{\fatm{G_1}{D_3:\ofty\app N\app T}}$
\end{tabbing}

At the outset we identify two things which will be needed to realize this
extension to the logic:
\begin{enumerate}
\item A means to express the properties of contexts in the logic and
\item The ability to interpret the properties of contexts inductively.
\end{enumerate}
We might take inspiration from the context definitions of Abella for describing
arbitrary properties of contexts within the logic, and think of building into the 
proof system an ability to reason inductively about contexts based on the definitions.

\section{Schematic Polymorphism}
Another interesting direction is to introduce a form of polymorphism
called schematic polymorphism,
as has been done for Abella in~\cite{wang2018corr}.
This form of polymorphism would simplify construction of arguments which
have the same structure regardless of type instance, reducing repetition in 
reasoning about systems which require structures to be repeated at many types 
within an encoding.

For example, we might encode in LF an $\appty$ relation for lists of 
integers.
Properties such as the functionality of $\appty$ have a structure which
does not make use of the knowledge that the elements in the list are integers, and 
thus would have the same structure for lists containing elements of any other type
we might include in the specification.
Since this proof is general with respect to the particular element type, it would
be useful to construct the proof of this property once and use it for proving the
functionality of append for other sorts of lists such as lists
of pairs of integers or lists of booleans.
To realize this in Adelfa we will need to extend both the declaration syntax 
and the proof system with relations and proof rules that are generic with respect 
to a particular type.
These schematic type families and proofs would provide a means of expressing structure 
that is identical regardless of a particular instance.

For the LF declaration syntax, we would require a means for describing a collection 
of type families which follow a particular structure.
Thus we might permit schematic type variables in these declarations, and any particular
LF declaration could be generated by replacing this schematic variable with a
ground type in LF.
For example we might describe the $\appty$ relation using something like the following.
\[\begin{array}{lcl}
\appty_{A} & : & \listty_{A} \rightarrow \listty_{A} \rightarrow 
             \listty_{A} \rightarrow \type.\\
\mbox{\sl app-nil} & : & \typedpi{L}{\listty_{A}}{\appty\app \niltm\app L\app L}.\\
\mbox{\sl app-cons} & : & \typedpi{X}{A}{\typedpi{L_1}{\listty_{A}}{\typedpi{L_2}{\listty_{A}}{\typedpi{L_3}{\listty_{A}}{}}}}\\
& & \qquad \appty\app L_1\app L_2\app L_3 \rightarrow\appty\app(\constm\app X\app L_1)\app L_2\app (\constm\app X\app L_3).
\end{array}\]
In this description, $A$ is a schematic variable which can be replaced with any 
particular ground LF type to generate that particular instance of these declarations
encoding the append relation for lists of that sort.

In the logic, we would extend the formula syntax
by permitting quantification over schematic type variables.
Derivations in the proof system for such formulas would 
describe schematic proofs which are invariant under the choice of type
instance, and will require new proof rules describing sound reasoning steps
of this kind.
The observation underlying these proof rules will be that the derivations should 
not identify particular instances of the schematic type variables, 
and so must be such that the proofs which can be constructed are the same for
every type.
For example, we might state the functionality of append in a generalized form as 
the following.
\begin{tabbing}
\qquad\=\qquad\=\kill
\>$\fall{L_1:\oty}{\fall{L_2:\oty}{}}$\\
\>\>$\fimp{\fatm{\emptyce}{L_1:\listty_{A}}}{\fimp{\fatm{\emptyce}{L_2:\listty_{A}}}{\fexists{L_3:\oty}{\fexists{D:\oty}{\fatm{\emptyce}{D:\appty_{A}\app L_1\app L_2\app L_3}}}}}$
\end{tabbing}
A proof for this formula in the refined proof system would be a schematic proof 
as all instances must have the same structure regardless of the particular 
instance for the generic type $A$.
Any explicit instance of the proof could then be generated as needed by
choosing a particular ground type to instantiate the schematic variable.

%% file: root.bbl
\newcommand{\etalchar}[1]{$^{#1}$}
\begin{thebibliography}{WCPW03}

\bibitem[ABF{\etalchar{+}}05]{aydemir05tphols}
Brian~E. Aydemir, Aaron Bohannon, Matthew Fairbairn, J.~Nathan Foster,
  Benjamin~C. Pierce, Peter Sewell, Dimitrios Vytiniotis, Geoffrey Washburn,
  Stephanie Weirich, and Steve Zdancewic.
\newblock Mechanized metatheory for the masses: The {POPLmark} challenge.
\newblock In {\em Theorem Proving in Higher Order Logics: 18th International
  Conference}, number 3603 in LNCS, pages 50--65. Springer, 2005.

\bibitem[ARCH]{cmusolution}
Michael Ashley-Rollman, Karl Crary, and Robert Harper.
\newblock {POPL}mark group from {CMU}'s solution.
\newblock \url{https://www.seas.upenn.edu/~plclub/poplmark/cmu.html}.

\bibitem[BC04]{bertot04book}
Yves Bertot and Pierre Cast\'eran.
\newblock {\em Interactive Theorem Proving and Program Development. Coq'Art:
  The Calculus of Inductive Constructions}.
\newblock Texts in Theoretical Computer Science. Springer, 2004.

\bibitem[BCG{\etalchar{+}}14]{baelde14jfr}
David Baelde, Kaustuv Chaudhuri, Andrew Gacek, Dale Miller, Gopalan Nadathur,
  Alwen Tiu, and Yuting Wang.
\newblock Abella: {A} system for reasoning about relational specifications.
\newblock {\em Journal of Formalized Reasoning}, 7(2), 2014.

\bibitem[EJP20]{errington2020lfmtp}
Jacob Errington, Junyoung~Clare Jang, and Brigitte Pientka.
\newblock {Mechanizing Metatheory Interactively}.
\newblock In {\em Fifteenth Workshop on Logical Frameworks and Meta-Languages:
  Theory and Practice}, 2020.

\bibitem[Gac08]{gacek08ijcar}
Andrew Gacek.
\newblock The {A}bella interactive theorem prover (system description).
\newblock In A.~Armando, P.~Baumgartner, and G.~Dowek, editors, {\em Fourth
  International Joint Conference on Automated Reasoning}, volume 5195 of {\em
  LNCS}, pages 154--161. Springer, 2008.

\bibitem[Gac09a]{abella.website}
Andrew Gacek.
\newblock The {A}bella system and homepage.
\newblock \url{http://abella.cs.umn.edu/}, 2009.

\bibitem[Gac09b]{gacek09phd}
Andrew Gacek.
\newblock {\em A Framework for Specifying, Prototyping, and Reasoning about
  Computational Systems}.
\newblock PhD thesis, University of Minnesota, 2009.

\bibitem[GMN11]{gacek11ic}
Andrew Gacek, Dale Miller, and Gopalan Nadathur.
\newblock Nominal abstraction.
\newblock {\em Information and Computation}, 209(1):48--73, 2011.

\bibitem[HHP93]{harper93jacm}
Robert Harper, Furio Honsell, and Gordon Plotkin.
\newblock A framework for defining logics.
\newblock {\em Journal of the ACM}, 40(1):143--184, 1993.

\bibitem[HL07]{harper07jfp}
Robert Harper and Daniel~R. Licata.
\newblock Mechanizing metatheory in a logical framework.
\newblock {\em Journal of Functional Programming}, 17(4--5):613--673, July
  2007.

\bibitem[LCH07]{lee07popl}
Daniel~K. Lee, Karl Crary, and Robert Harper.
\newblock Towards a mechanized metatheory of standard {ML}.
\newblock In {\em POPL '07: Proceedings of the 34th annual ACM SIGPLAN-SIGACT
  symposium on Principles of programming languages}, pages 173--184, New York,
  NY, USA, 2007. ACM Press.

\bibitem[Ler09]{leroy09cacm}
Xavier Leroy.
\newblock Formal verification of a realistic compiler.
\newblock {\em Communications of the ACM}, 52(7):107--115, 2009.

\bibitem[Mil91]{miller91jlc}
Dale Miller.
\newblock A logic programming language with lambda-abstraction, function
  variables, and simple unification.
\newblock {\em J. of Logic and Computation}, 1(4):497--536, 1991.

\bibitem[Mil92]{miller92jsc}
Dale Miller.
\newblock Unification under a mixed prefix.
\newblock {\em Journal of Symbolic Computation}, 14(4):321--358, 1992.

\bibitem[MN87]{miller87slp}
Dale Miller and Gopalan Nadathur.
\newblock A logic programming approach to manipulating formulas and programs.
\newblock In Seif Haridi, editor, {\em IEEE Symposium on Logic Programming},
  pages 379--388, San Francisco, September 1987.

\bibitem[Nad19]{nadathur2019}
Gopalan Nadathur.
\newblock Personal communication, August 2019.

\bibitem[NL05]{nadathur05iclp}
Gopalan Nadathur and Natalie Linnell.
\newblock Practical higher-order pattern unification with on-the-fly raising.
\newblock In {\em {ICLP 2005: 21st International Logic Programming
  Conference}}, volume 3668 of {\em LNCS}, pages 371--386, Sitges, Spain,
  October 2005. Springer.

\bibitem[NPP08]{nanevski08tocl}
Aleksandar Nanevski, Frank Pfenning, and Brigitte Pientka.
\newblock Contextual model type theory.
\newblock {\em ACM Trans.\ on Computational Logic}, 9(3):1--49, 2008.

\bibitem[NW18]{wang2018corr}
Gopalan Nadathur and Yuting Wang.
\newblock Schematic polymorphism in the {A}bella proof assistant.
\newblock {\em CoRR}, abs/1806.07523, 2018.

\bibitem[PAF{\etalchar{+}}19]{pientka2019}
Brigitte Pientka, Andreas Abel, Francisco Ferreira, David Thibodeau, and
  Rebecca Zucchini.
\newblock {Cocon: Computation in Contextual Type Theory}.
\newblock 2019.

\bibitem[PD10]{pientka10ijcar}
Brigitte Pientka and Joshua Dunfield.
\newblock Beluga: {A} framework for programming and reasoning with deductive
  systems (system description).
\newblock In J.~Giesl and R.~H{\"a}hnle, editors, {\em Fifth International
  Joint Conference on Automated Reasoning}, number 6173 in LNCS, pages 15--21,
  2010.

\bibitem[PE88]{pfenning88pldi}
Frank Pfenning and Conal Elliott.
\newblock Higher-order abstract syntax.
\newblock In {\em Proceedings of the {ACM}-{SIGPLAN} Conference on Programming
  Language Design and Implementation}, pages 199--208. ACM Press, June 1988.

\bibitem[Pie07]{pientka07poplmark}
Brigitte Pientka.
\newblock Proof pearl: The power of higher-order encodings in the logical
  framework {LF}.
\newblock In Klaus Schneider and Jens Brandt, editors, {\em Theorem Proving in
  Higher Order Logics, 20th International Conference, TPHOLs 2007,
  Kaiserslautern, Germany, September 10-13, 2007, Proceedings}, volume 4732 of
  {\em Lecture Notes in Computer Science}, pages 246--261. Springer, 2007.

\bibitem[PS99]{pfenning99cade}
Frank Pfenning and Carsten Sch{\"u}rmann.
\newblock System description: Twelf --- {A} meta-logical framework for
  deductive systems.
\newblock In H.~Ganzinger, editor, {\em 16th Conf.\ on Automated Deduction
  (CADE)}, number 1632 in LNAI, pages 202--206, Trento, 1999. Springer.

\bibitem[PS02]{Pfenning02guide}
Frank Pfenning and Carsten Sch{\"u}rmann.
\newblock {\em Twelf User's Guide}, 1.4 edition, December 2002.

\bibitem[PS20]{pientka2020}
Brigitte Pientka and Ulrich Sch{\"o}pp.
\newblock Semantical analysis of contextual types.
\newblock In Jean Goubault-Larrecq and Barbara K{\"o}nig, editors, {\em
  Foundations of Software Science and Computation Structures}, pages 502--521,
  Cham, 2020. Springer International Publishing.

\bibitem[SC14]{southern2014}
Mary Southern and Kaustuv Chaudhuri.
\newblock {A Two-Level Logic Approach to Reasoning About Typed Specification
  Languages}.
\newblock In Venkatesh Raman and S.~P. Suresh, editors, {\em 34th International
  Conference on Foundation of Software Technology and Theoretical Computer
  Science (FSTTCS 2014)}, volume~29 of {\em Leibniz International Proceedings
  in Informatics (LIPIcs)}, pages 557--569, Dagstuhl, Germany, 2014. Schloss
  Dagstuhl--Leibniz-Zentrum fuer Informatik.

\bibitem[Sch00]{schurmann00phd}
Carsten Sch{\"{u}}rmann.
\newblock {\em Automating the Meta Theory of Deductive Systems}.
\newblock PhD thesis, Carnegie Mellon University, October 2000.
\newblock {CMU-CS-00-146}.

\bibitem[SP98]{schurmann98cade}
Carsten Sch{\"u}rmann and Frank Pfenning.
\newblock Automated theorem proving in a simple meta-logic for {LF}.
\newblock In Claude Kirchner and H{\'e}l{\`e}ne Kirchner, editors, {\em 15th
  Conf.\ on Automated Deduction (CADE)}, volume 1421 of {\em Lecture Notes in
  Computer Science}, pages 286--300. Springer, 1998.

\bibitem[Tiu06]{tiu06lfmtp}
Alwen Tiu.
\newblock A logic for reasoning about generic judgments.
\newblock In A.~Momigliano and B.~Pientka, editors, {\em Int. Workshop on
  Logical Frameworks and Meta-Languages: Theory and Practice (LFMTP'06)},
  volume 173 of {\em ENTCS}, pages 3--18, 2006.

\bibitem[WCPW03]{watkins03tr}
Kevin Watkins, Iliano Cervesato, Frank Pfenning, and David Walker.
\newblock A concurrent logical framework {I}: Judgments and properties.
\newblock Technical Report CMU-CS-02-101, Carnegie Mellon University, 2003.
\newblock Revised, May 2003.

\end{thebibliography}
